\documentclass{article}

     \usepackage[preprint,nonatbib]{neurips_2020}

\usepackage[utf8]{inputenc} 
\usepackage[T1]{fontenc}    
\usepackage{url}            
\usepackage{booktabs}       
\usepackage{nicefrac}       
\usepackage{microtype}      
\usepackage{color} 

\usepackage{mathtools}
\usepackage{amssymb,amsthm}
\usepackage{bm}
\usepackage{dsfont}
\usepackage{enumitem}
\usepackage{pgf}
\newtheorem{theorem}{Theorem}
\newtheorem{lemma}{Lemma}
\newtheorem{proposition}{Proposition}

\newtheorem*{proposition_derivative}{Proposition~\ref{prop:derivative_i_n_main} (extended)}
\newtheorem*{proposition_ode}{Proposition~\ref{prop:ode} (extended)}

\usepackage{hyperref}
\usepackage{minitoc}

\newcommand{\smallO}{{\scriptstyle\mathcal{O}}}
\DeclareMathOperator*{\argmin}{arg\,min}

\newcommand{\N}{\mathbb{N}}
\newcommand{\R}{\mathbb{R}}
\newcommand{\cN}{\mathcal{N}}
\newcommand{\cH}{\mathcal{H}}
\newcommand{\cZ}{\mathcal{Z}}
\newcommand{\cC}{{\mathcal{C}}}
\newcommand{\E}{\mathbb{E}}
\newcommand{\Var}{\mathbb{V}\!\mathrm{ar}}
\newcommand{\simiid}[1][0pt]{\mathrel{\raisebox{#1}{$\sim$}}}
\newcommand{\iid}{\overset{\text{\tiny iid}}{\simiid[-2pt]}}
\newcommand{\MMSE}{\mathrm{MMSE}}
\newcommand{\bA}{\mathbf{A}}
\newcommand{\bS}{\mathbf{S}}
\newcommand{\bU}{\mathbf{U}}
\newcommand{\bV}{\mathbf{V}}
\newcommand{\bW}{\mathbf{W}}
\newcommand{\bX}{\mathbf{X}}
\newcommand{\bY}{\mathbf{Y}}
\newcommand{\bZ}{\mathbf{Z}}
\newcommand{\ba}{\mathbf{a}}
\newcommand{\bp}{\mathbf{p}}
\newcommand{\bu}{\mathbf{u}}
\newcommand{\bv}{\mathbf{v}}
\newcommand{\bw}{\mathbf{w}}
\newcommand{\bx}{\mathbf{x}}
\newcommand{\by}{\mathbf{y}}

\title{Information theoretic limits of learning a sparse rule}

\author{
Cl\'ement Luneau\thanks{Corresponding author: clement.luneau@epfl.ch}\hspace{1.5mm}, Nicolas Macris\\
Ecole Polytechnique F\'ed\'erale de Lausanne\\
Suisse\\
\And
Jean Barbier\\
International Center for Theoretical Physics\\
Trieste, Italy\\
}

\begin{document}
\maketitle

\begin{abstract}
	We consider generalized linear models in regimes where the number of nonzero components of the signal and accessible data points are \textit{sublinear} with respect to the size of the signal.
	We prove a variational formula for the asymptotic mutual information per sample when the system size grows to infinity.
	This result allows us to derive an expression for the minimum mean-square error (MMSE) of the Bayesian estimator when the signal entries have a discrete distribution with finite support.
	We find that, for such signals and suitable vanishing scalings of the sparsity and sampling rate, the MMSE is nonincreasing piecewise constant.
	In specific instances the MMSE even displays an \textit{all-or-nothing} phase transition, that is, the MMSE sharply jumps from its maximum value to zero at a critical sampling rate.
	The all-or-nothing phenomenon has previously been shown to occur in high-dimensional linear regression.
	Our analysis goes beyond the linear case and applies to learning the weights of a perceptron with general activation function in a teacher-student scenario.
	In particular, we discuss an all-or-nothing phenomenon for the generalization error with a sublinear set of training examples.
\end{abstract}

\doparttoc 
\faketableofcontents 
%
%
%
\section{Introduction}\label{section:introduction}
Modern tasks in statistical analysis, signal processing and learning require solving high-dimensional inference problems with a very large number of parameters.
This arises in areas as diverse as learning with neural networks \cite{LeCun2015Deep},
high-dimensional regression \cite{Buehlmann2011Statistics} or compressed sensing \cite{Donoho2009Message,Candes2006Optimal}.
In many situations, there appear barriers to what is possible to estimate or learn when the data becomes too scarce or too noisy. Such barriers can be of \textit{algorithmic} nature, but they can also be \textit{intrinsic} to the very nature of the problem. A celebrated example is the impossibility of reconstructing a noisy signal when the noise is beyond the so-called Shannon capacity of the communication channel \cite{Shannon1948Mathematical}. A large amount of interdisciplinary work has shown that these intrinsic barriers can be understood as \textit{static phase transitions} (in the sense of physics) when the system size tends to infinity (see \cite{Engel2001Statistical,Nishimori2001Statistical,Mezard2009Information}). 

When the problem can be formulated as an (optimal) Bayesian inference problem the mathematically rigorous theory of these phase transitions is now quite well developed.
Progress initially came from applications of the Guerra-Toninelli interpolation method (developed for the Sherrington-Kirkpatrick spin-glass model \cite{Guerra2006Introduction}) to coding and communication theory \cite{Montanari2005Tight,Macris2007GriffithKellySherman,Macris2007Sharp,Kudekar2009Sharp,Korada2010Tight,Giurgiu2016Spatial},
and more recently to low-rank matrix and tensor estimation \cite{Korada2009Exact,Barbier2016Mutual,Barbier2017Layered,Miolane2017Fundamental,Lelarge2019Fundamental,Barbier2019Mutual,Mourrat2019HamiltonJacobi,Barbier2020Information,Reeves2020Information},
compressive sensing and high-dimensional regression \cite{Barbier2016Mutuala,Barbier2020Mutual,Barbier2018Mutual,Reeves2019Replica},
and generalized linear models \cite{Barbier2019Optimal}.
In particular, for all these problems it has been possible to reduce the asymptotic mutual information to a low-dimensional variational expression, and deduce from its solution relevant error measures (e.g., minimum mean-square and generalization errors).
All these works consider the \textit{traditional regime} of statistical mechanics where the system size goes to infinity while relevant control parameters (such as signal sparsity, sampling rate, or signal-to-noise ratio) are kept fixed.

However, there exist \textit{other interesting regimes} for which many of the above mentioned problems also display fundamental intrinsic limits akin to phase transitions.
Consider for example the problem of compressive sensing. An interesting regime is one where both the number of nonzero components and of samples scale in a \textit{sublinear} manner as the system size tends to infinity.
In this case we would like to identify the phase transition, if there is any, and its nature.
This question has first been addressed recently in the framework of compressed sensing for binary Bernoulli signals by \cite{Gamarnik2017High, Reeves2019All, Reeves2019Alla}.
An \textit{all-or-nothing phenomenon} is identified, that is, in an appropriate sparse regime, the minimum mean-square error (MMSE) sharply drops from its maximum possible value (no reconstruction) for ``too small'' sampling rates to zero (perfect reconstruction) for ``large enough'' sampling rates.
The interest of such regime is not limited to estimation problems.
It is also relevant from a learning point of view, e.g., it corresponds to learning scenarios where we have access to a high number of features but only a sublinear number of them -- unknown to us -- are relevant for the learning task at hand.

Examples abound where the ``bet on sparsity principle'' \cite{Hastie2009Elements, Rish2014Sparse} is of utmost importance for the interpretability of a high-dimensional model.
Let us mention the MNIST handwritten digit database, where each digit can be seen as a $784=28\times 28$-dimensional binary vector representing the pixels whereas the digits effectively live in a space of the order of tens of dimensions \cite{Costa2004Learning, Hein2005Intrinsic}.
Another example of effective sparsity comes from natural images which are often sparse in a wavelet basis \cite{Mallat2009Wavelet}. Then, a fundamental question is {\it ``when is it possible to achieve a low estimation or generalization error with a sublinear amount of samples (sublinear with respect to the total number of features)?''}

In this contribution we address this question for a mathematically simple, but precise and tractable, setting. We consider generalized linear models in the regime of vanishing sparsity and sample rate, or equivalently, of sublinear number of data samples and nonzero signal components. As explained below these models can be used for estimation as well as learning, and we uncover in the sublinear regime intrinsic statistical barriers to these tasks in the form of sharp phase transitions. These statistical barriers are computed exactly and thus provide precise benchmarks to which algorithmic performance can be compared.

Let us outline the mathematical setting (further detailed in Section~\ref{section:settings}).
In a probabilistic setting the \textit{unknown} signal vector $\bX^* \in \mathbb{R}^n$ has entries drawn independently at random from a distribution $P_{0,n} \coloneqq \rho_nP_0 + (1-\rho_n) \delta_0$ with $P_0$ a fixed distribution.
The parameter $\rho_n$ controls the sparsity of the signal so that $\bX^*$ has $k_n \coloneqq n \rho_n$ nonzero components on average.
We observe the data $\bY = \varphi\big(\nicefrac{\bm{\Phi} \bX^*}{\sqrt{k_n}}\,\big)\in \mathbb{R}^{m_n}$ obtained by first multiplying the signal with a \textit{known} $m_n\times n$ random matrix $\bm{\Phi}$ whose entries are independent standard Gaussian random variables, and then applying $\varphi$ component-wise.
The number of data points is controlled by the sampling rate $\alpha_n$, i.e., $m_n \coloneqq \alpha_n n$.
We consider the regime $(\rho_n,\alpha_n) \to (0,0)$ as $n$ goes to infinity with $\alpha_n =\gamma \rho_n\vert\ln\rho_n\vert$, for which sharp phase transitions appear when $P_0$ is discrete with finite support.
Note that both $m_n$ and $k_n$ scale sublinearly as $n\to +\infty$.

The model can be interpreted as either an estimation problem or a learning problem:
\begin{itemize}[leftmargin=0.3in]
\item
In the \textit{estimation interpretation}, we assume a purely Bayesian (or optimal) setting. We know the model, the activation function $\varphi$, the prior $P_{0,n}$ as well as the measurement matrix $\bm{\Phi}$.
Our goal is then to determine what is the lowest reconstruction error that we can achieve, i.e., what is the average minimum mean-square error
$
k_n^{-1}\E\,\Vert \bX^* - \mathbb{E}[\bX^* \vert \bY, \bm{\Phi} ]\Vert^2
$
when $n$ gets large.
\item
In the \textit{learning interpretation}, we consider a teacher-student scenario in which a teacher hands out training samples $\{(Y_\mu, (\Phi_{\mu i})_{i=1}^n)\}_{\mu=1}^{m_n}$ to a student.
The teacher produces the output label $Y_\mu$ by feeding the input $(\Phi_{\mu i})_{i=1}^n$ to its own one-layer neural network with activation function $\varphi$ and weights $\bX^*=(X_i^*)_{i=1}^n$.
The student -- who is given the model and the prior -- has to learn the weights $\bX^*$ of the teacher's one-layer neural network by minimizing the empirical training error of the $m_n$ training samples.
For example, the binary perceptron corresponds to $\varphi = \mathrm{sign}$ and $Y_\mu\in \{\pm 1\}$.
Of particular interest is the generalization error.
Given a new -- previously unseen -- random pattern $\bm{\Phi}_{\scriptscriptstyle \mathrm{new}} \coloneqq (\Phi_{{\scriptscriptstyle \mathrm{new}}, i})_{i=1}^n$ whose true label is $Y_{\scriptscriptstyle \mathrm{new}}$ (generated by the teacher's neural network),
the optimal generalization error is
$\E[(Y_{\scriptscriptstyle \mathrm{new}} - \E[\varphi(\nicefrac{\bm{\Phi}_{\scriptscriptstyle \mathrm{new}}^{\mathsf{T}} \bX^*}{\sqrt{k_n}})\vert \bY, \bm{\Phi}, \bm{\Phi}_{\scriptscriptstyle \mathrm{new}}])^{2}]$;
the error made when estimating $Y_{\mathrm{new}}$ in a purely Bayesian way.
\end{itemize}
Let us summarize informally our results.
We set $\alpha_n = \gamma \rho_n\vert \ln\rho_n\vert$ where $\gamma$ is fixed and $\rho_n$ vanishes as $n$ diverges.
We first rigorously determine the mutual information $m_n^{-1}I(\bX^*; \bY \vert \bm{\Phi})$ in terms of a low-dimensional variational problem, see Theorem \ref{th:RS_1layer} which also provides a precise control of the finite size fluctuations.
Remarkably, when $P_0$ is a discrete distribution with finite support, this variational problem simplifies to a minimization problem over a finite set of values, see Theorem~\ref{theorem:limit_MI_discrete_prior}.
For such signals, using I-MMSE type formulas \cite{Guo2005Mutual}, we can deduce from the solution to this minimization problem the asymptotic MMSE and optimal generalization error, see Theorem~\ref{theorem:asymptotic_mmse}.
Our analysis shows that both errors are nonincreasing piecewise constant functions of $\gamma$.
In particular, if the entries of $\vert \bX^* \vert$ are either $0$ or some $a>0$ then both errors display an all-or-nothing behavior as $n\to +\infty$, with a sharp transition at a threshold $\gamma = \gamma_c$ explicitly computed.
These findings are illustrated, and their significance discussed, in Section~\ref{section:all-or-nothing-phenomenon}.

In our work the generalized linear model is treated by entirely different methods than the linear model in \cite{Gamarnik2017High, Reeves2019All}.
Importantly, the sparsity regime treated by our method requires the sparsity $\rho_n$ to go to zero slower than $n^{\nicefrac{-1}{9}}$, while it has to go to zero faster than $n^{\nicefrac{-1}{2}}$ in the results of \cite{Reeves2019All} for the linear case.
From this angle, both results complement each other.
Our proof technique for Theorem~\ref{th:RS_1layer} exploits the adaptive interpolation method (see \cite{Barbier2019Adaptivea,Barbier2019Adaptive}) that is a powerful improvement over the Guerra-Toninelli interpolation and allows to prove replica symmetric formulas for Bayesian inference problems.
We adapt the analysis of \cite{Barbier2019Optimal} in a non-trivial way in order to consider the new scaling regime of our problem where
$\alpha_n = \gamma \rho_n \vert \ln \rho_n \vert$, and $\rho_n \to 0$ as $n$ gets large instead of being fixed.
We show that the adaptive interpolation can still be carried through, which requires a more refined control of the error terms compared to \cite{Barbier2019Optimal}.
It is interesting, and not a priori obvious, that this can be done since this is \textit{not} the usual statistical mechanics extensive regime.
For example, the mutual information has to be normalized by the subextensive quantity $m_n=\smallO(n)$. Quite remarkably, with this suitable normalization, the asymptotic mutual information, MMSE and generalization error have a similar form to those famously found in ordinary thermodynamic regimes in physics \cite{Gardner1989Three,Gyoergyi1990First,Seung1992Statistical,Opper1991Generalization}.

In Section~\ref{section:settings} we present the setting and state our theoretical results on the mutual information and the MMSE in the sublinear regime.
We use these results in Section~\ref{section:all-or-nothing-phenomenon} to uncover the all-or-nothing phenomenon for general activation functions.
In Section~\ref{section:overview_interpolation} we give an overview of the adaptive interpolation method used to prove Theorem~\ref{th:RS_1layer}.
The full proofs of our results are given in the appendices.

\section{Problem setting and main results}\label{section:settings}
\subsection{Generalized linear estimation of low sparsity signals at low sampling rates}
Let $n \in \N^*$ and $m_n \coloneqq \alpha_n n$ with $(\alpha_n)_{n \in \N^*}$ a decreasing sequence of positive sampling rates. 
Let $P_{0}$ be a probability distribution with finite second moment $\E_{X \sim P_0}\,[X^2]$.
Let $(X^*_i)_{i=1}^n \iid P_{0,n}$ be the components of a signal vector $\bX^*$ (this is also denoted $\bX^*\iid P_{0,n}$), where 
\begin{equation}
P_{0,n} \coloneqq \rho_nP_0 + (1-\rho_n) \delta_0\,.
\end{equation}
The parameter $\rho_n \in (0,1)$ controls the sparsity of the signal; the latter being made of $k_n \coloneqq \rho_n n$ nonzero components in expectation.
We will be interested in low sparsity regimes where $k_n= \smallO(n)$.
Let $k_A \in \N$.
We consider a measurable function $\varphi: \R \times \R^{k_A} \to \R$ and a probability distribution $P_A$ over $\R^{k_A}$.
The $m_n$ data points $\bY \coloneqq (Y_\mu)_{\mu=1}^{m_n}$ are generated as
\begin{align}\label{measurements}
	Y_\mu \coloneqq \varphi\Big(\frac{1}{\sqrt{k_n}} (\bm{\Phi} \bX^*)_{\mu}, \bA_\mu\Big) + \sqrt{\Delta} Z_\mu\:, \quad 1\leq \mu \leq m_n\:,
\end{align}
where $(\bA_\mu)_{\mu =1}^{m_n} \iid P_A$, $(Z_\mu)_{\mu=1}^{m}\iid \mathcal{N}(0,1)$ is an additive white Gaussian noise (AWGN), $\Delta > 0$ is the noise variance, and $\bm{\Phi}$ is a $m_n \times n$ measurement (or data) matrix with independent entries having zero mean and unit variance. Note that the noise $(Z_\mu)_{\mu=1}^{m}$ can be considered as part of the model, or as a ``regularising noise'' needed for the analysis but that can be set arbitrarily small. 
Typically, and as $n$ gets large, $\nicefrac{(\bm{\Phi} \bX^*)_{\mu}}{\sqrt{k_n}}=\Theta(1)$.
%
The estimation problem is to recover $\bX^*$ from the knowledge of $\bY$, $\bm{\Phi}$, $\Delta$, $\varphi$, $P_{0,n}$ and $P_A$ (the realization of the random stream $(\bA_\mu)_{\mu =1}^{m_n}$ itself, if present in the model, is unknown). 
It will be helpful to think of the measurements as the outputs of a \textit{channel}:
\begin{equation}\label{eq:channel}
	Y_\mu \sim P_{\mathrm{out}}\Big( \cdot \, \Big\vert \frac{1}{\sqrt{k_n}} (\bm{\Phi} \bX^*)_{\mu} \Big)\:, \quad 1\leq \mu \leq m_n\:.
\end{equation}
The transition kernel $P_{\mathrm{out}}$ admits a transition density with respect to Lebesgue's measure given by:
\begin{equation}\label{transition-kernel}
P_{\mathrm{out}}( y \vert x ) 
= \frac{1}{\sqrt{2\pi \Delta}} \int dP_A(\ba) \,e^{-\frac{1}{2 \Delta}( y - \varphi(x,\ba))^2}\;.
\end{equation}
The random stream $(\bA_\mu)_{\mu =1}^{m_n}$ represents any source of randomness in the model.
For example, the logistic regression $\mathbb{P}(Y_\mu=1)=f(\nicefrac{(\bm{\Phi} \bX^*)_{\mu}}{\sqrt{k_n}})$ with $f(x)=(1+e^{-\lambda x})^{-1}$ is modeled by considering a teacher that draws i.i.d. uniform numbers $A_\mu\sim {\cal U}[0,1]$, and then obtains the labels through
$Y_\mu = \bm{1}_{\{A_\mu \le f(\nicefrac{(\bm{\Phi} \bX^*)_{\mu}}{\sqrt{k_n}})\}} - \bm{1}_{\{A_\mu \ge f(\nicefrac{(\bm{\Phi} \bX^*)_{\mu}}{\sqrt{k_n}})\}}$ ($\bm{1}_{\mathcal{E}}$ denotes the indicator function of an event $\mathcal{E}$).
In the absence of such a randomness in the model, the activation $\varphi: \R \to \R$ is deterministic, $k_A = 0$ and the integral $\int dP_A(\ba)$ in \eqref{transition-kernel} simply disappears.
Our numerical experiments in Section~\ref{section:all-or-nothing-phenomenon} are for deterministic activations but all of our theoretical results hold for the broader setting. 

We have presented the problem from an estimation point of view.
In this case, the important quantity to assess the performance of an algorithm estimating $\bX^*$ is the mean-square error.
Another point of view is the learning one: each row of the matrix $\bm{\Phi}$ is the input to a one-layer neural network whose weights $\bX^*$ have been sampled independently at random by a teacher.
The student is given the input/output pairs $(\bm{\Phi},\bY)$ as well as the model used by the teacher.
The student's role is then to learn the weights.
In this case, more than the mean-square error, the important quantity is the generalization error.

\subsection{Asymptotic mutual information}
The mutual information $I(\bX^*;\bY \vert \bm{\Phi})$ between the signal $\bX^*$ and the data $\bY$ given the matrix $\bm{\Phi}$ is the main quantity of interest in our work.
Before stating Theorem~\ref{th:RS_1layer} on the value of this mutual information, we first introduce two scalar denoising models that play a key role.

The first model is an additive Gaussian channel.
Let $X^* \sim P_{0,n}$ be a scalar random variable.
We observe $Y^{(r)} \coloneqq \sqrt{r} X^* + Z$ where $r \geq 0$ plays the role of a signal-to-noise ratio (SNR) and the noise $Z \sim \cN(0,1)$ is independent of $X^*$.
The mutual information $I_{P_{0,n}}(r) \coloneqq I(X^*;Y^{(r)})$ between the signal of interest $X^*$ and $Y^{(r)}$ depends on $\rho_n$ through the prior $P_{0,n}$, and it reads:
\begin{equation} \label{formula_I_P0n}
I_{P_{0,n}}(r)
= \frac{r\rho_n\E_{X \sim P_0}[X^2]}{2} - \E \ln \int dP_{0,n}(x)e^{rX^*x + \sqrt{r}Zx -\frac{rx^2}{2}} \;.
\end{equation}
The second scalar channel is linked to the transition kernel $P_{\mathrm{out}}$ defined by \eqref{transition-kernel}.
Let $V$, $W^*$ be two independent standard Gaussian random variables.
In this scalar estimation problem we want to infer $W^*$ from the knowledge of $V$ and the observation
$\widetilde{Y}^{(q,\rho)} \sim P_{\mathrm{out}}(\cdot \,\vert \sqrt{q}\, V + \sqrt{\rho - q} \,W^*)$ where $\rho > 0$ and $q \in [0,\rho]$.
The conditional mutual information $I_{P_{\mathrm{out}}}(q,\rho) \coloneqq I(W^*;\widetilde{Y}^{(q,\rho)} \vert V)$ is:
\begin{equation}\label{formula_I_Pout}
I_{P_{\mathrm{out}}}(q,\rho)
= \E \ln P_{\mathrm{out}}\big(\widetilde{Y}^{(\rho,\rho)} \vert \sqrt{\rho}\, V \big)
- \E \ln \int dw\,\frac{e^{-\frac{w^2}{2}}}{\sqrt{2\pi}} P_{\mathrm{out}}\big(\widetilde{Y}^{(q,\rho)} \vert \sqrt{q}\, V + \sqrt{\rho - q}\, w\big)\; .
\end{equation}
Both $I_{P_{0,n}}$ and $I_{P_{\mathrm{out}}}$ have nice monotonicity, Lipschitzianity and concavity properties that are important for the proof of Theorem~\ref{th:RS_1layer} (stated below).

We use the mutual informations \eqref{formula_I_P0n} and \eqref{formula_I_Pout} to define the \textit{(replica-symmetric) potential}:
\begin{equation}\label{def_i_RS}
i_{\scriptstyle{\mathrm{RS}}}(q, r; \alpha_n, \rho_n)
\coloneqq \frac{1}{\alpha_n}I_{P_{0,n}}\Big(\frac{\alpha_n}{\rho_n}r\Big) +  I_{P_{\mathrm{out}}}(q, \E_{P_{0}}[X^2])- \frac{r(\E_{P_{0}}[X^2]-q)}{2}  \;.
\end{equation}
Our first result links the extrema of this potential to the mutual information of our original problem.
\begin{theorem}[Mutual information of the GLM at sublinear sparsity and sampling rate] \label{th:RS_1layer}
Suppose that $\Delta > 0$ and that the following hypotheses hold:
\begin{enumerate}[label=(H\arabic*),noitemsep]
	\item \label{hyp:bounded} There exists $S > 0$ such that the support of $P_0$ is included in $[-S,S]$.
	\item \label{hyp:c2} $\varphi$ is bounded, and its first and second partial derivatives with respect to its first argument exist, are bounded and continuous. They are denoted $\partial_{x} \varphi$, $\partial_{xx} \varphi$.
	\item \label{hyp:phi_gauss2} $(\Phi_{\mu i}) \iid \cN(0,1)$.
\end{enumerate}
Let $\rho_n =\Theta(n^{-\lambda})$ with $\lambda \in [0,\nicefrac{1}{9})$ and $\alpha_n = \gamma \rho_n \vert\ln\rho_n\vert$ with $\gamma>0$.
Then for all $n \in \N^*$:
\begin{equation}\label{bound_mutual_info_rs_formula}
\bigg\vert \frac{I(\bX^*;\bY \vert \bm{\Phi}) }{m_n} - \adjustlimits{\inf}_{q \in [0,\E_{P_{0}}[X^2]\,]} {\sup}_{r \geq 0}\; i_{\scriptstyle{\mathrm{RS}}}(q, r; \alpha_n, \rho_n)\bigg\vert
\leq \frac{\sqrt{C}\,\vert \ln n \vert^{\nicefrac{1}{6}}}{n^{\frac{1}{12}-\frac{3\lambda}{4}}} \;,
\end{equation} 
where $C$ is a polynomial in $\big(S,\big\Vert \frac{\varphi}{\sqrt{\Delta}} \big\Vert_\infty, \big\Vert \frac{\partial_x \varphi}{\sqrt{\Delta}} \big\Vert_\infty, \big\Vert \frac{\partial_{xx} \varphi}{\sqrt{\Delta}} \big\Vert_\infty , \lambda, \gamma\big)$ with positive coefficients.
\end{theorem}
Hence, the asymptotic mutual information is given to leading order by the variational problem $\inf_{q \in [0,\E_{P_{0}}[X^2]\,]} \sup_{r \geq 0}\, i_{\scriptstyle{\mathrm{RS}}}(q, r; \alpha_n, \rho_n)$.
Note that this variational problems depends on $n$ and Theorem~\ref{th:RS_1layer} does not say anything on its value in the asymptotic regime, e.g., does it converge or diverge?
Our next theorem answers this question when $P_0$ is a discrete distribution with finite support.
%
\subsection{Specialization to discrete priors: all-or-nothing phenomenon and its generalization}
\begin{theorem}[Specialization of Theorem~\ref{th:RS_1layer} to discrete priors with finite support]\label{theorem:limit_MI_discrete_prior}
	Suppose that $\Delta > 0$ and that $P_{0,n} \coloneqq (1-\rho_n)\delta_0 + \rho_n P_0$ where $P_0$ is a discrete distribution with finite support $$
	\mathrm{supp}(P_0) \subseteq \{-v_K, -v_{K-1}, \dots, -v_1,v_1, v_2, \dots, v_K\}\;;
	$$
	where ${0 < v_1 < v_2 < \dots < v_K < v_{K+1} \coloneqq +\infty}$.
	Further assume that the hypotheses~\ref{hyp:c2} and \ref{hyp:phi_gauss2} in Theorem~\ref{th:RS_1layer} hold.
	Let $\rho_n =\Theta(n^{-\lambda})$ with $\lambda \in (0,\nicefrac{1}{9})$ and ${\alpha_n = \gamma \rho_n \vert\ln\rho_n\vert}$ with $\gamma>0$.
	Then,
	\begin{equation}\label{limit_minimization_problem}
		\lim_{n \to +\infty} \frac{I(\bX^*;\bY \vert \boldsymbol{\Phi}) }{m_n}
		= \min_{1 \leq k \leq K+1}\bigg\{I_{P_{\mathrm{out}}}\big(\E[X^2 \bm{1}_{\{\vert X \vert \geq v_k\}}],\E[X^2]\big) +\frac{\mathbb{P}(\vert X \vert \geq v_k)}{\gamma}\bigg\}\;,
	\end{equation}
where $X \sim P_0$.
\end{theorem}
The proof of Theorem~\ref{theorem:limit_MI_discrete_prior} requires computing the limit of $\inf_{q \in [0,\E X^2]} \sup_{r \geq 0} i_{\scriptstyle{\mathrm{RS}}}(q, r; \alpha_n, \rho_n)$ when $\rho_n$ vanishes.
We prove Theorem~\ref{theorem:limit_MI_discrete_prior} for $P_0 = \delta_1$ in Appendix~\ref{section:specialization_bernoulli_prior} and for a general discrete distribution with finite support $P_0$ in Appendix~\ref{appendix:specialization_discrete_prior}.

When doing estimation, one important metric to assess the quality of an estimator $\widehat{\bX}(\bY, \bm{\Phi})$ is its mean-square error $\nicefrac{\E\,\Vert \bX^* - \widehat{\bX}(\bY, \bm{\Phi})\Vert^2}{k_n}$.
The latter is always lower bounded by the mean-square error of the Bayesian estimator $\mathbb{E}[\bX^* \vert \bY, \bm{\Phi} ]$; the so-called minimum mean-square error (MMSE).
Remarkably, once we have Theorem~\ref{theorem:limit_MI_discrete_prior}, we can obtain the asymptotic MMSE with a little more work.
First, we have to introduce a modified inference problem where in addition to the observations $\bY$ we are given $\widetilde{\bY}^{(\tau)} = \sqrt{\nicefrac{\alpha_n \tau}{\rho_n}}\, \bX^* + \widetilde{\bZ}$.
When $\tau$ is close enough to $0$, 
the analysis yielding Theorem~\ref{theorem:limit_MI_discrete_prior} can be adapted to obtain the limit
\begin{multline*}
\lim_{n \to +\infty} \frac{I(\bX^*;\bY,\widetilde{\bY}^{(\tau)}\vert \boldsymbol{\Phi}) }{m_n}\\
	= \min_{1 \leq k \leq K+1}\bigg\{I_{P_{\mathrm{out}}}\big(\E[X^2 \bm{1}_{\{\vert X \vert \geq v_k\}}],\E[X^2]\big)
	+ \frac{\mathbb{P}(\vert X \vert \geq v_k)}{\gamma}
	+ \frac{\tau \E[X^2 \bm{1}_{\{\vert X \vert < v_k\}}]}{2}\bigg\}\;.
\end{multline*}
We can then apply the I-MMSE identity\footnote{
	The derivative of $\nicefrac{I(\bX^*;\bY,\widetilde{\bY}^{(\tau)}\vert \boldsymbol{\Phi}) }{m_n}$ with respect to $\tau$ at $\tau=0$ is equal to half the MMSE of the original problem.
}\cite{Guo2005Mutual,Deshpande2016Asymptotic} to obtain the asymptotic MMSE:
\begin{theorem}[Asymptotic MMSE]\label{theorem:asymptotic_mmse}
	Under the assumptions of Theorem~\ref{theorem:limit_MI_discrete_prior},
	if the minimization problem on the right-hand side of \eqref{limit_minimization_problem} has a unique solution $k^* \in \{1,\dots,K+1\}$ then
	\begin{equation}\label{eq:formula_asymptotic_mmse}
		\lim_{n \to +\infty} \frac{\E \Vert \bX^* - \E[\bX^* \vert \bY, \boldsymbol{\Phi}] \Vert^2}{k_n} 
		= \E\big[X^2 \bm{1}_{\{\vert X \vert < v_{k^*}\}}\big]\;\text{, where } X \sim P_0\;.
	\end{equation}
\end{theorem}
We prove Theorem~\ref{theorem:asymptotic_mmse} in Appendix~\ref{section:asymptotic_mmse}.
We remark that it is possible with more technical work \cite[Appendix C.2]{Barbier2019Optimal} to weaken \ref{hyp:c2} in Theorems~\ref{theorem:limit_MI_discrete_prior} and \ref{theorem:asymptotic_mmse} to the~assumption
``There exists $\epsilon > 0$ such that the sequence $\E \vert \varphi(\nicefrac{(\bm{\Phi} \bX^*)_1}{\sqrt{k_n}},\bA_1)\vert^{2+\epsilon}$ is bounded, and for almost all $\ba \sim P_A$ the function $x \mapsto \varphi(x,\ba)$ is continuous almost everywhere.''
Hence, Theorems~\ref{theorem:limit_MI_discrete_prior} and \ref{theorem:asymptotic_mmse} also apply to the linear activation $\varphi(x)=x$, the perceptron $\varphi(x) = \mathrm{sign}(x)$ and the ReLU~${\varphi(x) = \max(0,x)}$.
\section{The all-or-nothing phenomenon}\label{section:all-or-nothing-phenomenon}
We now highlight interesting consequences of our results regarding the MMSE of the estimation problem as well as the optimal generalization error of the learning problem in the teacher-student scenario.
Reeves et al. \cite{Reeves2019All} have proved the existence of an \textit{all-or-nothing phenomenon} for the linear model when $\bX^*$ is a $0\,$-$1$ vector
and here we extend their results in two ways: $i)$ for the estimation error of a generalized linear model, and $ii)$ for the generalization error of a perceptron neural network with general activation function $\varphi$.

We consider signals whose entries are either Bernoulli random variables, i.e., $P_{0,n} \coloneqq (1-\rho_n)\delta_0 + \rho_n P_0$ with $P_0 = \delta_1$, or Bernoulli-Rademacher random variables, i.e., $P_{0,n} \coloneqq (1-\rho_n)\delta_0 + \rho_n P_0$ with $P_0 = \nicefrac{(\delta_{1}+\delta_{-1})}{2}$.
In both cases $\E_{P_{0}}[X^2] = 1$ (we can always assume the latter by rescaling the noise).
We place ourselves in the regime of Theorem~\ref{theorem:asymptotic_mmse} where $\alpha_n = \gamma \rho_n \vert \ln \rho_n \vert$ for some fixed $\gamma > 0$ and $\rho_n \to 0$ in the high-dimensional limit $n \to +\infty$.

\paragraph{MMSE}
In this regime, and for such signals, Theorem~\ref{theorem:asymptotic_mmse} states that the minimum mean-square error $\mathrm{MMSE}(\bX^* \vert \bY, \bm{\Phi}) \coloneqq \frac{\E \Vert \bX^* - \E[\bX^* \vert \bY, \boldsymbol{\Phi}] \Vert^2}{k_n}$ satisfies:
\begin{equation}\label{eq:asymptotic_mmse_bernoulli}
\lim_{n \to +\infty} \mathrm{MMSE}(\bX^* \vert \bY, \bm{\Phi})
=
\begin{cases}
	0 \quad\text{if}\;\; I_{P_{\mathrm{out}}}(0,1) > \gamma^{-1}\;; \\
	1 \quad\text{if}\;\; I_{P_{\mathrm{out}}}(0,1) < \gamma^{-1}\;.
\end{cases}
\end{equation}
Therefore, we locate an \textit{all-or-nothing phase transition} at the threshold
\begin{equation}\label{gamma_transition}
	\gamma_c \coloneqq \frac{1}{I_{P_{\mathrm{out}}}(0,1)}\;.
\end{equation}
Remember that $\gamma$ controls the amount $m_n$ of training samples.
In the high-dimensional limit, perfect reconstruction is possible if $\gamma > \gamma_c$ (the asymptotic MMSE is zero) while it is impossible to do better than a random guess if $\gamma < \gamma_c$ (the asymptotic MMSE is equal to $\lim_{n \to +\infty} \nicefrac{\E \Vert \bX^* - \E \bX^* \Vert^2}{k_n} = 1$; the asymptotic MMSE in the absence of observations).
As $I_{P_{\mathrm{out}}}(0,1) \coloneqq I(W^*; \varphi(W^*,\bA) + \sqrt{\Delta} Z)$ where $W^*, Z \iid \mathcal{N}(0,1) \perp \bA \sim P_{A}$,
the threshold $\gamma_c$ is fully determined by the activation function and the amount of noise, and it can be easily evaluated in a number of cases.
In Figure~\ref{figure:potentials_and_gamma_c} we draw $\gamma_c$ for $\varphi(x) = x$, $\varphi(x) = \mathrm{sign}(x)$, $\varphi(x) = \max(0,x)$ and noise variance $\Delta \in [0,0.5]$.
We see that for $\Delta$ small enough the ReLU activation requires less training samples to learn the sparse rule than the linear one; it is the opposite once $\Delta$ becomes large enough.
When $\Delta$ diverges both the linear and sign activations have the asymptote $\gamma_c \thicksim 2\Delta$ while the ReLU activation has another steeper asymptote $\gamma_c \thicksim a \Delta$, $a \approx 5.87$.
The corresponding formulas for $\gamma_c$ are given in Table~\ref{table:gamma_c}.
Note that for the random linear model $\varphi(x) = x$,
the threshold $\alpha_c(\rho_n) \coloneqq \gamma_c \rho_n \vert \ln \rho_n \vert = \nicefrac{2\rho_n \vert \ln \rho_n \vert}{\ln(1 + \Delta^{-1})}$ is in agreement with the sample rate $n^*$ for which \cite{Reeves2019All} prove that weak recovery is impossible below it while strong recovery is possible above.
\begin{figure}[hb]
	\centering
	\resizebox{0.75\textwidth}{!}{\input{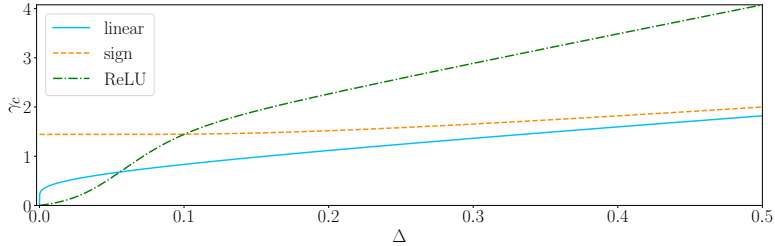}}
	\caption{\label{figure:potentials_and_gamma_c}
		\small{Threshold $\gamma_c$ of the all-or-nothing phase transition for different activation functions as a function of the noise variance $\Delta$.
		}}
\end{figure}
\begin{table}[h]
	\centering
	\begin{tabular}{@{}ccl@{}}
		\toprule
		Activation $\varphi(x)$ & $\gamma_c(\Delta = 0)$ & $\gamma_c(\Delta)$ for $\Delta > 0$ \\
		\midrule
		$x$  & $0$  & $2/\ln(1 + \Delta^{-1})$ \\
		\midrule
		$\mathrm{sign}(x)$   & $\nicefrac{1}{\ln 2}$  & $1/\big(\ln 2 - \E[\ln(1 + e^{\nicefrac{-2(1+\sqrt{\Delta}Z)}{\Delta} })]\big)$\\
		\midrule
		$\max(0,x)$   & $0$   & $4\Delta / \big(1 - 4\Delta\E[h_\Delta(Z) \ln h_\Delta(Z)]\big)$\\
		& & with $h_\Delta(Z) \coloneqq \frac12 + \sqrt{\frac{\Delta}{1 + \Delta}}e^{\frac{Z^2}{2(1+\Delta)}}\int_{-\infty}^{\frac{Z}{\sqrt{1+\Delta}}} \frac{dt}{\sqrt{2\pi}}e^{-\frac{t^2}{2}}$\\
		\bottomrule
	\end{tabular}
	\caption{\label{table:gamma_c} \small{Closed-formed formulas of $\gamma_c$ for different activation functions. We use $Z \sim \cN(0,1)$.}}
\end{table}
\paragraph{Optimal generalization error}
When learning in a (matched) teacher-student scenario, the components of $\bX^*$ correspond to the unknown weights of the teacher's one-layer neural network.
The student is given the model and training samples $\{(Y_\mu, (\Phi_{\mu, i})_{i=1}^n)\}_{\mu=1}^{m_n}$.
Then, the optimal generalization error is the MMSE for predicting the output $Y_{\scriptscriptstyle \mathrm{new}} \sim P_{\mathrm{out}}(\, \cdot \,\vert \nicefrac{\bm{\Phi}_{\scriptscriptstyle \mathrm{new}}^{\mathsf{T}} \bX^*}{\sqrt{k_n}})$
generated by a new input $\bm{\Phi}_{\scriptscriptstyle \mathrm{new}} \coloneqq (\Phi_{{\scriptscriptstyle \mathrm{new}}, i}) \iid \mathcal{N}(0,1)$.
More precisely, the optimal generalization error is
$\mathrm{MMSE}( Y_{\scriptscriptstyle \mathrm{new}} \vert \bY, \bm{\Phi}, \bm{\Phi}_{\scriptscriptstyle \mathrm{new}})
\coloneqq \E[(Y_{\scriptscriptstyle \mathrm{new}} - \E[ Y_{\scriptscriptstyle \mathrm{new}} \vert \bY, \bm{\Phi},\bm{\Phi}_{\scriptscriptstyle \mathrm{new}}]\,)^{2} ]$
where $V,W^* \sim \mathcal{N}(0,1)$ and $\bA \sim P_{A}$ are independent.
Based on our proof of Theorem~\ref{theorem:asymptotic_mmse} and the optimal generalization error when $\rho_n = \Theta(1)$ (regime of linear sparsity and sampling rate) \cite[Theorem 2]{Barbier2019Optimal}, we conjecture that under the assumptions of Theorem~\ref{theorem:asymptotic_mmse}:
\begin{equation}\label{conjecture_generalization_error}
	\lim_{n \to +\infty} \mathrm{MMSE}( Y_{\scriptscriptstyle \mathrm{new}} \vert \bY, \bm{\Phi}, \bm{\Phi}_{\scriptscriptstyle \mathrm{new}})
	= \Delta + \E\big[\big(\varphi(V,\bA)  - \E[\varphi(\sqrt{q^*}\,V + \sqrt{{\E X^2 -q^*}}\,W^*,\bA) \vert V]\big)^{2}\,\big]\,.
\end{equation}
where $\E X^2 - q^* = \E[X^2 \bm{1}_{\{\vert X \vert < v_{k^*}\}}]$ is the asymptotic MMSE \eqref{eq:formula_asymptotic_mmse}.
For Bernoulli and Bernoulli-Rademacher signals (the ones considered in this section), it simplifies to:
\begin{equation}\label{eq:asymptotic_generalization_error_bernoulli}
	\lim_{n \to +\infty} \mathrm{MMSE}( Y_{\scriptscriptstyle \mathrm{new}} \vert \bY, \bm{\Phi}, \bm{\Phi}_{\scriptscriptstyle \mathrm{new}})
	= 
	\begin{cases}
		\Delta + \E[(\varphi(V,\bA) - \E[\varphi(V,\bA)\vert V])^2] 
		\;\text{if}\;\, \gamma > \gamma_c\;; \\
		\Delta + \Var(\varphi(V,\bA))
		\qquad\qquad\qquad\quad\;\,\text{if}\;\, \gamma < \gamma_c\;.
	\end{cases}
\end{equation}
We thus find that the optimal generalization error also displays an all-or-nothing phase transition at $\gamma_c$.
More precisely, if $\gamma < \gamma_c$ then the optimal generalization error equals $\Delta + \Var(\varphi(V,\bA))$ when $n\to +\infty$.
This is the same generalization error achieved by the dumb label estimator in the Bayesian sense; the one predicting the new label to be the output value averaged over all possible inputs, weights and noise.
If instead $\gamma > \gamma_c$ then it is equal to $\Delta + \E[\Var(\varphi(V,\bA) \vert V)]$; the irreducible error due to both the noise $\bZ$ and the random stream $(\bA_\mu)_{\mu =1}^{m_n}$.

Proving \eqref{conjecture_generalization_error} entails introducing side observations in the original problem and differentiating with respect to the signal-to-noise ratio of this side channel to exploit the I-MMSE relation, in a similar fashion to what we do in the proof of Theorem~\ref{theorem:asymptotic_mmse} (see Appendix~\ref{section:asymptotic_mmse}).
The side observations have the same form than the ones used in \cite[Section 5 of SI Appendix]{Barbier2019Optimal} to determine the asymptotic optimal generalization error in the regime of linear sparsity and sampling rate.
\paragraph{Illustration of the all-or-nothing phenomenon}
In Figure~\ref{figure:all_or_nothing} we use \eqref{eq:asymptotic_mmse_bernoulli} to draw in solid black lines the asymptotic MMSE in the regime of sublinear sparsity and sampling rate, for both priors Bernoulli and Bernoulli-Rademacher
and the activation functions $\varphi(x) = x$, $\varphi(x) = \mathrm{sign}(x)$, $\varphi(x) = \max(0,x)$.
For comparison we also draw in dashed colored lines the asymptotic MMSE in regimes of linear sparsity and sampling rate, that is, $\rho_n = \rho$ and $\alpha_n = \gamma \rho \vert \ln \rho \vert$ are constant with $n$.
In this case, the asymptotic MMSE is given by \cite[Theorem 2]{Barbier2019Optimal}
\begin{equation}\label{eq:asymptotic_mmse_bernoulli_rho_n=rho}
	\lim_{n \to +\infty} \mathrm{MMSE}(\bX^* \vert \bY, \bm{\Phi})
	= 1 - q^*\;,
\end{equation}
whenever $\argmin_{q \in [0,1]} \sup_{r \geq 0} i_{\scriptstyle{\mathrm{RS}}}(q, r; \gamma \rho \vert \ln \rho \vert, \rho)$ is a singleton $\{q^*\}$.
To optimize the potential $i_{\scriptstyle{\mathrm{RS}}}(q, r; \gamma \rho \vert \ln \rho \vert, \rho)$ we initialize $q \in [0,1]$ at different values
and iterate the following fixed point equation (obtained directly by setting the gradient of the potential to zero):
\begin{equation}\label{fixed_point_equations}
r= -2 \frac{\partial I_{P_{\mathrm{out}}}}{\partial q}\bigg\vert_{q,1}\quad , \quad q = -\frac{2}{\rho_n} I'_{P_{0,n}}\bigg(\frac{\alpha_n}{\rho_n} r\bigg)\;.
\end{equation}
Finally, the fixed point $q^*$ yielding the lowest potential $\sup_{r \geq 0} i_{\scriptstyle{\mathrm{RS}}}(q^*, r; \gamma \rho \vert \ln \rho \vert, \rho)$ is used to determine the $\mathrm{MMSE}$ thanks to \eqref{eq:asymptotic_mmse_bernoulli_rho_n=rho}.
In all configurations the asymptotic MMSE jumps from a value close to $1$ to approximately $0$ as $\gamma$ increases past $\gamma_c$.
As $\rho_n = \rho$ gets closer to $0$, this jump becomes sharper with the MMSE approaching $0$ or $1$ depending on which side of $\gamma_c$ we are.
Though this jump becomes sharper, a pure all-or-nothing phase transition only occurs in the regime of sublinear sparsity and sampling rate (solid black lines).

In Figure~\ref{figure:generalization_error} we use \eqref{eq:asymptotic_generalization_error_bernoulli} to plot in solid black lines the asymptotic optimal generalization error for the Bernoulli prior and the same activation functions.
The dashed colored lines again correspond to regimes of linear sparsity and sampling rate; they are obtained using the formula for the asymptotic optimal generalization error given by \cite[Theorem 2]{Barbier2019Optimal}:
\begin{equation}\label{eq:asymptotic_generalization_error_bernoulli_rho_n=rho}
	\lim_{n \to +\infty} \mathrm{MMSE}( Y_{\scriptscriptstyle \mathrm{new}} \vert \bY, \bm{\Phi},\bm{\Phi}_{\scriptscriptstyle \mathrm{new}})
= \Delta + \E\big[\big(\varphi(V,\bA)  - \E[\varphi(\sqrt{q^*}\,V + \sqrt{1 -q^*}\,W^*,\bA) \vert V]\big)^{\!2}\,\big].
\end{equation}
\begin{figure}[t]
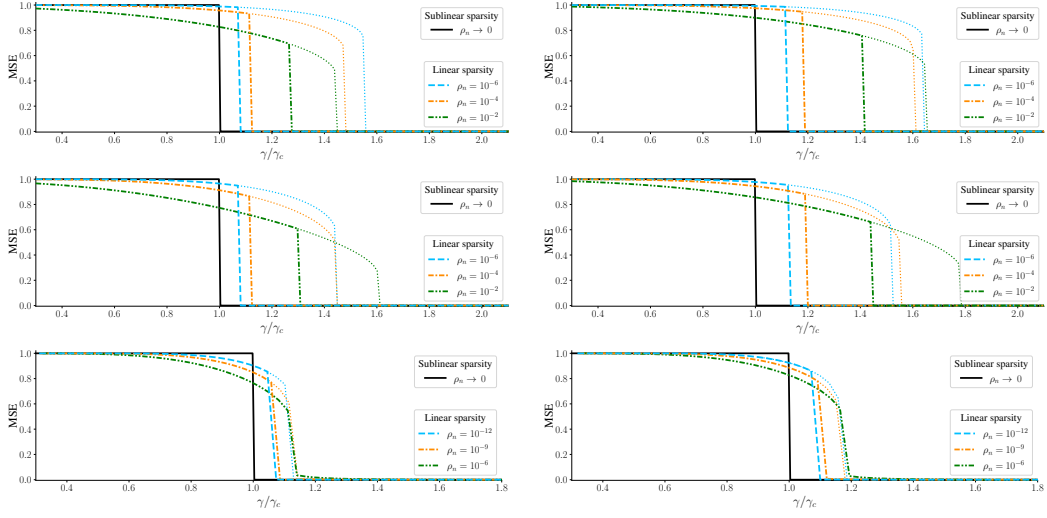

	{\resizebox{0.49\textwidth}{!}{\input{figures/mmse_bernoulli_linear_activation.pgf}}\hfill
		\resizebox{0.49\textwidth}{!}{\input{figures/mmse_bernoulli_rademacher_linear_activation.pgf}}}
	{\resizebox{0.49\textwidth}{!}{\input{figures/mmse_bernoulli_sign_activation.pgf}}\hfill
	\resizebox{0.49\textwidth}{!}{\input{figures/mmse_bernoulli_rademacher_sign_activation.pgf}}}
	{\resizebox{0.49\textwidth}{!}{\input{figures/mmse_bernoulli_relu_activation.pgf}}\hfill
	\resizebox{0.49\textwidth}{!}{\input{figures/mmse_bernoulli_rademacher_relu_activation.pgf}}}
	{\caption{\label{figure:all_or_nothing}
		\small{Asymptotic MMSE as a function of $\nicefrac{\gamma}{\gamma_c}$ in the regime of sublinear sparsity and sampling rate
			($\rho_n = \Theta(n^{-\lambda})$ with $\lambda \in (0,\nicefrac{1}{9})$, solid black line), and in the regime of linear sparsity and sampling rate ($\rho_n$ fixed, dashed colored lines).
		Dotted lines correspond to algorithmic performance in the regime of linear sparsity and sampling rate (iterating \eqref{fixed_point_equations} from $q=10^{-10}$).
		\textit{Left panels:} Bernoulli prior.
		\textit{Right panels:} Bernoulli-Rademacher prior.
		\textit{From top to bottom:} ${\varphi(x) = x, \Delta=0.1; \varphi(x) = \mathrm{sign}(x), \Delta=0; \varphi(x) = \max(0,x), \Delta=0.5}$.
		}
	}}
\end{figure}
\begin{figure}[th]
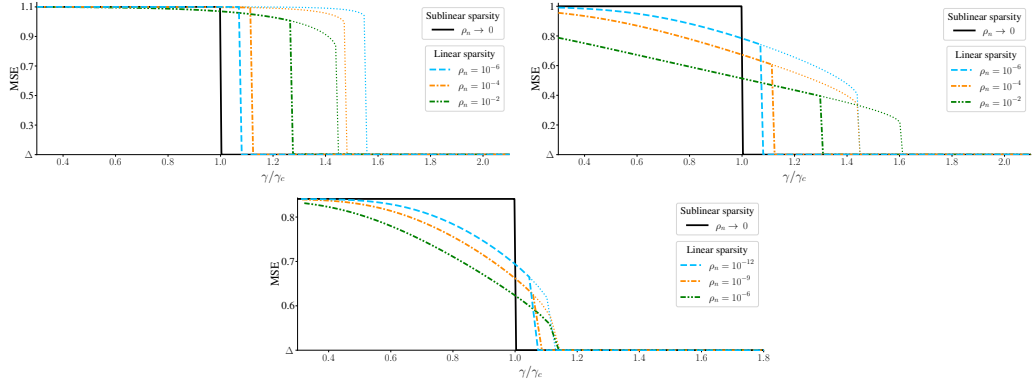

	\centering
	{\resizebox{0.49\textwidth}{!}{\input{figures/generalization_error_bernoulli_linear_activation.pgf}}
		\resizebox{0.49\textwidth}{!}{\input{figures/generalization_error_bernoulli_sign_activation.pgf}}}
	\resizebox{0.49\textwidth}{!}{\input{figures/generalization_error_bernoulli_relu_activation.pgf}}
	\caption{\label{figure:generalization_error}
		\small{
			Asymptotic optimal generalization error as a function of $\nicefrac{\gamma}{\gamma_c}$ in the regime of sublinear sparsity and sampling rate (${\rho_n = \Theta(n^{-\lambda})}$ with $\lambda \in (0,\nicefrac{1}{9})$, solid black line), and in the regime of linear sparsity and sampling rate ($\rho_n$ is fixed, dashed colored lines).
			Dotted lines correspond to algorithmic performance in the regime of linear sparsity and sampling rate (iterating \eqref{fixed_point_equations} from $q=10^{-10}$).
			\textit{Top left:} random linear model $\varphi(x) = x$, $\Delta=0.1$. \textit{Top right:} perceptron $\varphi(x) = \mathrm{sign}(x), \Delta=0$.
			\textit{Bottom:} ReLU ${\varphi(x) = \max(0,x), \Delta=0.5}$.}}
\end{figure}
In all configurations the optimal generalization error jumps from a value close to $\Delta + \Var(\varphi(V))$ to approximately $\Delta$ as $\gamma$ increases past $\gamma_c$ (note that the activations are deterministic so there is no contribution from $\bA$ in the error).
The value $\Delta$ is as good as the optimal generalization error can get, i.e., it is equal to the noise variance which is the squared error we would get if we were given the true weights $\bX^*$.
Again, the jump gets sharper as $\rho_n = \rho$ approaches $0$ but a pure all-or-nothing phase transition only occurs in the regime of sublinear sparsity and sampling rate (solid black lines).

The all-or-nothing behavior of the asymptotic MMSE and optimal generalization error is quite striking.
Indeed, in the limit of vanishing sparsity and sampling rate either estimation or learning is as good as it can get or as bad as a random guess.
This purely dichotomic behavior only occurs in the truly sparse limit, and is shown here to be pretty general in the sense that it occurs for a wide variety of activation functions.
An important aspect of our results is to provide a definitive statistical benchmark allowing to measure the quality of algorithms with respect to the minimal amount of sparse data needed to estimate or learn. This benchmark is provided by non-trivial formulas \eqref{gamma_transition} for the threshold $\gamma_c$ given for several examples in Table~\ref{table:gamma_c}.
We note that such precise benchmarks are quite rarely obtained in traditional machine learning approaches.
\paragraph{Further remarks}
Algorithmic aspects are beyond the scope of this paper.
However, we make a few remarks about generalized approximate message passing (GAMP) algorithms.
In the regime of linear sparsity and sampling rate, the state evolution equations precisely tracking the asymptotic performance of the algorithm are linked to the fixed point equation \eqref{fixed_point_equations} \cite{Rangan2011Generalized}. 
The fixed point $q^{\mathrm{alg}}$ reached by initializing \eqref{fixed_point_equations} arbitrarily close to $q=0$ can be used in \eqref{eq:asymptotic_mmse_bernoulli_rho_n=rho} and \eqref{eq:asymptotic_generalization_error_bernoulli_rho_n=rho} -- instead of $q^*$-- to obtain both the mean-square and generalization errors of GAMP algorithms.
These errors are represented with dotted colored lines in Figures~\ref{figure:all_or_nothing} and \ref{figure:generalization_error}.
We observe an algorithmic-to-statistical gap, that is, the dotted lines corresponding to the algorithmic performance do not drop to zero around $\gamma_c$ but at a higher \textit{algorithmic threshold}.
In this work we don't study the performance of GAMP algorithms in the regime of sublinear sparsity and sampling rate.
However, reference \cite{Reeves2019Alla} rigorously shows that in this regime the all-or-nothing behavior also occurs at an algorithmic level for GAMP algorithms.
It would be highly desirable to extend their results to other activations and derive the corresponding thresholds.
\section{Overview of the proof of Theorem~\ref{th:RS_1layer}}\label{section:overview_interpolation}
The interested reader will find the proof of Theorem~\ref{th:RS_1layer} in Appendix~\ref{section:interpolation}.
In this section we give an outline of the proof and its main ideas.
The proof is based on the adaptive interpolation method \cite{Barbier2019Adaptivea,Barbier2019Adaptive} whose main difference with the canonical interpolation method \cite{Guerra2002Thermodynamic,Guerra2003Broken} is the increased flexibility given to the path followed by the interpolation between its two extremes.
The method has been developed separately for symmetric rank-one tensor problems where the spike has i.i.d. components \cite{Barbier2019Adaptivea,Barbier2019Adaptive}, and for one-layer GLMs whose input signal has again i.i.d. components \cite{Barbier2019Optimal}.
The sparse regime of the problem studied in this contribution differs of the usual scaling for which such techniques have been developed.
They have been used in a regime where the number of measurements and sparsity are linear in $n$ as in \cite{Barbier2019Optimal}.
Working in the sparse regime requires writing more refined concentration bounds and proving that the key steps of the adaptive interpolation can still be carried through.
\paragraph{1. Interpolating estimation problem}
To simplify the presentation we assume that $\Delta = 1$ and $\E_{X \sim P_{0}}[X^2] = 1$.
The proof starts by introducing an interpolating inference problem that depends on a parameter $t \in [0,1]$ and two continuous interpolation functions $R_1,R_2:[0,1] \to \R_+$ with $R_1(0)=R_2(0)=0$.
Let $\bX^* \iid P_{0,n}$, $\bm{\Phi} \coloneqq (\Phi_{\mu i}) \iid \cN(0,1)$, $\bV \coloneqq (V_{\mu})_{\mu=1}^{m_n} \iid \cN(0,1)$ and $\bW^* \coloneqq (W_{\mu}^*)_{\mu=1}^{m_n} \iid \cN(0,1)$. We define for all $t \in [0,1]$ an ``interpolating pre-activation'':
\begin{equation*}
S_{\mu}^{(t)}
\coloneqq \sqrt{\nicefrac{(1-t)}{k_n}}\, (\bm{\Phi} \bX^*)_\mu  + \sqrt{R_2(t)} \,V_{\mu} + \sqrt{t - R_2(t)} \,W_{\mu}^* \;.
\end{equation*}
The inference problem at a fixed $t$ is to recover both unknowns $\bX^*, \bW^*$ from the knowledge of $\bV$, $\bm{\Phi}$
and the data
%
%
\begin{align*}
\begin{cases}
Y_{\mu}^{(t)}  & \sim \quad P_{\mathrm{out}}(\,\cdot\, \vert \, S_{\mu}^{(t)})\quad\,,\;\:  1 \leq \mu \leq m_n\,; \\
\widetilde{Y}_{i}^{(t)} & = \sqrt{R_1(t)}\, X^*_i + \widetilde{Z}_i\;\,, \;\: 1 \leq \, i \,  \leq \, n \;\;\, ;
\end{cases}
\end{align*}
where $Z_\mu, \widetilde{Z}_i \iid \cN(0,1)$.
The corresponding \textit{interpolating mutual information} is:
\begin{equation*}
i_{n}(t) \coloneqq m_n^{-1}I\big((\bX^*,\bW^*)\,;\,(\bY^{(t)},\widetilde{\bY}^{(t)})\big|\bm{\Phi},\bV\big) \;.	
\end{equation*}
\paragraph{2. Fundamental sum-rule}
Note that at $t=0$ we recover the original problem of interest and ${i_n(0) = \nicefrac{I(\bX^*;\bY \vert \bm{\Phi}) }{m_n}}$.
At the other extreme $t=1$, the mutual information can be written in terms of the simple mutual informations $I_{P_{0,n}}$ and $I_{P_{\mathrm{out}}}$, that is,
$i_{n}(1) = \nicefrac{I_{P_{0,n}}(R_1(1))}{\alpha_n} + I_{P_{\mathrm{out}}}(R_2(1),1)$.
We link the mutual information at both extremes by computing the derivative $i_{n}'(\cdot)$ of $i_{n}(\cdot)$ and then using the fundamental identity $i_{n}(0) = i_{n}(1)-\int_0^1 i_{n}'(t) dt$. It yields the sum-rule:
$$
\frac{I(\bX^*;\bY|\bm{\Phi})}{m_n}
= \frac{1}{\alpha_n}I_{P_{0,n}}(R_1(1))
+ I_{P_{\mathrm{out}}}(R_2(1),1)
-\frac{\rho_n}{2\alpha_n}\int_0^1 R'_1(t)\big(1-R'_2(t)\big) dt
+ \mathcal{R}_n \;.
$$
The last term $\mathcal{R}_n$ is a remainder whose absolute value we want to control in order to get Theorem~\ref{th:RS_1layer}.
\paragraph{3. Controlling the remainder}
This is done by plugging two different choices of interpolation functions $(R_1,R_2)$ in the sum-rule. 
One choice yields an upper bound on the difference in the left-hand side of \eqref{bound_mutual_info_rs_formula}, while another yields a lower bound.
Each choice of interpolation functions $(R_1,R_2)$ is defined implicitly as the solution to a second order ordinary differential equation.
Remarkably, under these two choices, the remainder $\mathcal{R}_n$ can be controlled using precise concentration results.

\section*{Broader Impact}
We believe that it is difficult to clearly foresee societal consequence of the present, purely theoretical, work.
The results presented inscribe themselves in the larger theme of providing guidelines for better and parsimonious use of data when possible, for example when learning a sparse rule.
On the long run, such guidelines must be taken into account for building engineering systems that are more efficient in terms of computational and energetic cost.

\begin{ack}
The work of C. L. is supported by the Swiss National Foundation for Science grant number 200021E 17554.
\end{ack}


\newpage
\appendix
\addcontentsline{toc}{section}{Appendix} 
\part{Appendix} 
\parttoc
\section{Proof of Theorem~\ref{th:RS_1layer} with the adaptive interpolation method}\label{section:interpolation}
Note that it is the same to observe \eqref{measurements} or their rescaled versions $\frac{1}{\sqrt\Delta}\varphi\big(\frac{1}{\sqrt{k_n}} (\bm{\Phi} \bX^*)_{\mu}, \bA_\mu\big) + Z_\mu$. 
Therefore, up to a rescaling of $\varphi$ by $\nicefrac{1}{\sqrt{\Delta}}$, we will suppose that $\Delta = 1$ all along the proof of Theorem~\ref{th:RS_1layer}.
For a similar reason, we can suppose that $\E_{X \sim P_{0}}[X^2] = 1$.
\subsection{Interpolating estimation problem}\label{interp-est-problem}
We fix a sequence $(s_n)_{n \in \N^*} \in (0,1/2]$ and define $\mathcal{B}_n \coloneqq [s_n,2 s_n]^2$.
Let $r_{\max} \coloneqq -2 \frac{\partial I_{P_{\mathrm{out}}}}{\partial q}\big\vert_{q=1,\rho=1}$ a positive real number.
For all $\epsilon = (\epsilon_1, \epsilon_2) \in \mathcal{B}_n$, we define the \textit{interpolation functions}
\begin{equation*}
R_{1}(\cdot,\epsilon): t \in [0,1] \mapsto \epsilon_1+\int_0^t r_{\epsilon}(v)dv
\quad \text{and} \quad
R_2(\cdot,\epsilon): t \in [0,1] \mapsto \epsilon_2 + \int_0^t q_{\epsilon}(v)dv\;,
\end{equation*}
where $q_{\epsilon}: [0,1] \to [0,1]$ and ${r_{\epsilon}: [0,1] \to [0,\frac{\alpha_n}{\rho_n}r_{\max}]}$ are two continuous functions.
We say that the families of functions $(q_{\epsilon})_{\epsilon \in \mathcal{B}_n}$ and $(r_{\epsilon})_{\epsilon \in \mathcal{B}_n}$
are \textit{regular} if $\forall t \in [0,1]: \epsilon  \mapsto  \big(R_1(t,\epsilon), R_2(t,\epsilon )\big)$ is a $\cC^1$ diffeomorphism from $\mathcal{B}_n$ onto its image whose Jacobian determinant is greater than, or equal, to one.
This property will reveal important later in our proof.
Let $\bX^* \iid P_{0,n}$, $\bm{\Phi} \coloneqq (\Phi_{\mu i}) \iid \cN(0,1)$, $\bV \coloneqq (V_{\mu})_{\mu=1}^{m_n} \iid \cN(0,1)$ and $\bW^* \coloneqq (W_{\mu}^*)_{\mu=1}^{m_n} \iid \cN(0,1)$. We define:
\begin{equation}\label{Stmu}
S_{\mu}^{(t,\epsilon)}
= S_{\mu}^{(t,\epsilon)}(\bX^*,W_\mu^*)
\coloneqq \sqrt{\frac{1-t}{k_n}}\, (\bm{\Phi} \bX^*)_\mu  + \sqrt{R_2(t,\epsilon)} \,V_{\mu} + \sqrt{t+2s_n-R_2(t,\epsilon)} \,W_{\mu}^* \;.
\end{equation}
Consider the following observations coming from two types of channels: 
\begin{align}\label{2channels}
\begin{cases}
Y_{\mu}^{(t,\epsilon)}  &\sim \quad P_{\mathrm{out}}(\,\cdot\, \vert \, S_{\mu}^{(t,\epsilon)})\quad\;\,,\;\:  1 \leq \mu \leq m_n\,; \\
\widetilde{Y}_{i}^{(t,\epsilon)} &= \sqrt{R_1(t,\epsilon)}\, X^*_i + \widetilde{Z}_i\;, \;\: 1 \leq \, i \,  \leq \, n \;\;\; ;
\end{cases}
\end{align}
where $(\widetilde{Z}_i)_{i=1}^n\iid \cN(0,1)$.
The inference problem (at time $t$) is to recover both unknowns $\bX^*, \bW^*$ from the knowledge of $\bV$, $\bm{\Phi}$ and the observations $\bY^{(t,\epsilon)} \coloneqq (Y_{\mu}^{(t,\epsilon)})_{\mu=1}^{m_n}, \widetilde{\bY}^{(t,\epsilon)} \coloneqq (\widetilde{Y}_{i}^{(t,\epsilon)})_{i=1}^n$. 
The joint posterior density of $(\bX^*,\bW^*)$ given $(\bY^{(t,\epsilon)},\widetilde{\bY}^{(t,\epsilon)},\bm{\Phi},\bV)$ reads:
\begin{multline}\label{posterior_interpolation}
dP(\bx,\bw \vert \bY^{(t,\epsilon)},\widetilde{\bY}^{(t,\epsilon)},\bm{\Phi},\bV)\\
\coloneqq \frac{1}{\cZ_{t,\epsilon}}\,\prod_{i=1}^n dP_{0,n}(x_i)\,e^{-\frac{1}{2}\big(\sqrt{R_{1}(t,\epsilon)}\,x_i -\widetilde{Y}_i^{(t,\epsilon)}\big)^2}\,
\prod_{\mu=1}^{m_n} \frac{dw_\mu}{\sqrt{2\pi}} e^{-\frac{w_\mu^2}{2}} P_{\mathrm{out}}(Y_{\mu}^{(t,\epsilon)}\vert s_{\mu}^{(t,\epsilon)}) \:,
\end{multline}
where $s_{\mu}^{(t,\epsilon)} \coloneqq S_{\mu}^{(t,\epsilon)}(\bx,w_\mu)$ and $\cZ_{t,\epsilon} \equiv \cZ_{t,\epsilon}(\bY^{(t,\epsilon)},\widetilde{\bY}^{(t,\epsilon)},\bm{\Phi},\bV)$ is the normalization.
The \textit{interpolating mutual information} is:
\begin{equation}\label{interpolating_mutual_information}
i_{n,\epsilon}(t) \coloneqq \frac{1}{m_n} I\big((\bX^*,\bW^*);(\bY^{(t,\epsilon)},\widetilde{\bY}^{(t,\epsilon)})\big|\bm{\Phi},\bV\big) \;.	
\end{equation}
The \textit{perturbation} $\epsilon$ only induces a small change in mutual information. In particular, at $t=0$:
%
\begin{lemma}
\label{lemma:perturbation_mutual_information_t=0}
Suppose that \ref{hyp:bounded},~\ref{hyp:c2},~\ref{hyp:phi_gauss2} hold, that $\Delta = \E_{P_0}[X^2] = 1$
and that there exist real positive numbers $M_\alpha, M_{\rho/\alpha}$ such that $\forall n \in \N^*$: $\alpha_n \leq M_{\alpha}$ and $\nicefrac{\rho_n}{\alpha_n}  \leq  M_{\rho/\alpha}$.
For all $\epsilon \in \mathcal{B}_n$:
\begin{equation*}
\bigg\vert i_{n,\epsilon}(0) -\frac{I(\bX^*;\bY|\bm{\Phi})}{m_n} \bigg\vert \leq \sqrt{C} \frac{s_n}{\sqrt{\rho_n}}\;,
\end{equation*}
where $C$ is a polynomial in $\big(S,\Vert\varphi\Vert_\infty, \Vert \partial_x \varphi \big\Vert_\infty, \Vert \partial_{xx} \varphi \Vert_\infty, M_\alpha, M_{\rho/\alpha}\big)$ with positive coefficients.
\end{lemma}
We prove Lemma~\ref{lemma:perturbation_mutual_information_t=0} in Appendix~\ref{appendix:perturbation_mutual_information_t=0}.
By the chain rule for mutual information and the Lipschitzianity of $I_{P_{0,n}}, I_{P_{\mathrm{out}}}$ (see Lemmas \ref{lemma:property_I_P0n} and \ref{lemma:property_I_Pout} in Appendix~\ref{appendix:scalar_channels}), at $t=1$ we have for all $\epsilon \in \mathcal{B}_n$:
\begin{align}\label{perturbation_mutual_information_t=1}
i_{n,\epsilon}(1)
&= \frac{I(\bX^*;\widetilde{\bY}^{(1,\epsilon)} \vert \bm{\Phi})+I(\bW^*;\bY^{(1,\epsilon)} \vert \bm{\Phi},\bV )}{m_n}
= \frac{I_{P_{0,n}}(R_1(1,\epsilon))}{\alpha_n} + I_{P_{\mathrm{out}}}(R_2(1,\epsilon),1+2s_n)\nonumber\\
&= \frac{1}{\alpha_n}I_{P_{0,n}}\bigg(\int_0^1 r_{\epsilon}(t)dt \bigg)
+ I_{P_{\mathrm{out}}}\bigg(\int_0^1 q_{\epsilon}(t)dt ,1\bigg) + \mathcal{O}(s_n) \;,
\end{align}
assuming there exists $M_{\rho/\alpha} > 0$ such that $\forall n \in \N^*: \nicefrac{\rho_n}{\alpha_n} \leq M_{\rho/\alpha}$.
$\mathcal{O}(s_n)$ is a quantity whose absolute value is bounded by $C s_n$ where $C$ is a polynomial in $\big(S,\Vert\varphi\Vert_\infty, \Vert \partial_x \varphi \big\Vert_\infty, \Vert \partial_{xx} \varphi \Vert_\infty, M_{\rho/\alpha}\big)$ with positive coefficients.
\subsection{Fundamental sum rule}
We want to \textit{compare} the original model of interest (model at $t=0$) to the purely scalar one ($t=1$).
To do so, we use $i_{n,\epsilon}(0) = i_{n,\epsilon}(1)-\int_0^1i_{n,\epsilon}'(t) dt$ where $i_{n,\epsilon}'(\cdot)$ is the derivative of $i_{n,\epsilon}(\cdot)$.
Once combined with Lemma~\ref{lemma:perturbation_mutual_information_t=0} and \eqref{perturbation_mutual_information_t=1}, it yields (note that $\mathcal{O}(s_n) = \mathcal{O}(\nicefrac{s_n}{\sqrt{\rho_n}})$ since $0 < \rho_n < 1$):
\begin{align}
\frac{I(\bX^*;\bY \vert \bm{\Phi})}{m_n}
= \mathcal{O}\bigg(\frac{s_n}{\sqrt{\rho_n}}\bigg) + \frac{1}{\alpha_n}I_{P_{0,n}}\bigg(\int_0^1 r_{\epsilon}(t)dt \bigg)
&+ I_{P_{\mathrm{out}}}\bigg(\int_0^1 q_{\epsilon}(t)dt ,1\bigg)\nonumber\\
&-\int_0^1 i_{n,\epsilon}'(t) dt\;.\label{link_i_n(0)_i_n(1)}
\end{align}
From now on let $(\bx, \bw) \in \R^n \times \R^{m_n}$ be a pair of random vectors sampled from the joint posterior distribution \eqref{posterior_interpolation}.
The angular brackets $\langle - \rangle_{n,t,\epsilon}$ denote an expectations w.r.t.\ the distribution \eqref{posterior_interpolation}, i.e.,
$\langle g(\bx,\bw) \rangle_{n,t,\epsilon} \coloneqq \int g(\bx,\bw) dP(\bx,\bw \vert \bY^{(t,\epsilon)},\widetilde{\bY}^{(t,\epsilon)},\bm{\Phi},\bV)$
for every integrable function $g$.
We define the scalar overlap $Q \coloneqq\frac{1}{k_n} \sum_{i=1}^{n} X_i^* x_i$.
The computation of $i_{n,\epsilon}^\prime$ is found in Appendix~\ref{appendix:derivative_interpolating_mutual_information}.
\begin{proposition}\label{prop:derivative_i_n_main}
Suppose that \ref{hyp:bounded},~\ref{hyp:c2},~\ref{hyp:phi_gauss2} hold and that $\Delta = \E_{X \sim P_0}[X^2] = 1$.
Further assume that there exist real positive numbers $M_\alpha, M_{\rho/\alpha}$ such that $\forall n \in \N^*$: $\alpha_n \leq M_{\alpha}$ and $\nicefrac{\rho_n}{\alpha_n}  \leq  M_{\rho/\alpha}$.
Define $u_y(x) \coloneqq \ln P_{\mathrm{out}}(y|x)$ and $u'_y(\cdot)$ its derivative w.r.t. $x$.
For all $(t,\epsilon) \in [0,1] \times \mathcal{B}_n$:
	\begin{multline}\label{eq:derivative_i_n(t)_main}
	i_{n,\epsilon}'(t)
	=  \mathcal{O}\bigg(\frac{1}{\rho_n \sqrt{n}}\bigg)
	+ \frac{\rho_n}{2\alpha_n}r_\epsilon(t) (1 - q_\epsilon(t))\\
	+ \frac{1}{2} \E\,\bigg\langle \big(Q - q_\epsilon(t)\big) \bigg( \frac{1}{m_n}\sum_{\mu=1}^{m_n} u_{Y_\mu^{(t,\epsilon)}}'(S_\mu^{(t,\epsilon)}) u'_{Y_\mu^{(t,\epsilon)}}(s_\mu^{(t,\epsilon)}) - \frac{\rho_n}{\alpha_n}r_\epsilon(t)\bigg)\bigg\rangle_{\!\! n,t,\epsilon}\;,
	\end{multline}
where $\big\vert \mathcal{O}\big(\frac{1}{\rho_n\sqrt{n}}\big) \big\vert \leq \frac{\sqrt{C}}{\rho_n \sqrt{n}}$,
with $C$ a polynomial in $\big(S,\Vert\varphi\Vert_\infty, \Vert \partial_x \varphi \big\Vert_\infty, \Vert \partial_{xx} \varphi \Vert_\infty, M_{\alpha}, M_{\rho/\alpha}\big)$ with positive coefficients,
uniformly in $(t,\epsilon)$.
\end{proposition}
The next key result states that the overlap concentrates on its expectation.
This behavior is called \textit{replica symmetric} in statistical physics.
Similar results have been obtained in the spin glass literature \cite{Talagrand2011Mean,CojaOghlan2017Information}.
In this work we use a formulation taylored to Bayesian inference problems as developed in the context of LDPC codes, random linear estimation \cite{Barbier2020Mutual} and Nishimori symmetric spin glasses \cite{Macris2007GriffithKellySherman,Korada2010Tight,Korada2009Exact}. 
\begin{proposition}[Overlap concentration]\label{prop:concentration_overlap}
Suppose that \ref{hyp:bounded},~\ref{hyp:c2},~\ref{hyp:phi_gauss2} hold, that $\Delta = \E_{P_0}[X^2] = 1$
and that the family of functions $(r_\epsilon)_{\epsilon \in \mathcal{B}_n}$, $(q_\epsilon)_{\epsilon \in \mathcal{B}_n}$ are \textit{regular}.
Further assume that there exist real positive numbers $M_\alpha, M_{\rho/\alpha}, m_{\rho/\alpha}$ such that $\forall n \in \N^*$:
$\alpha_n \leq M_{\alpha}$ and $\frac{m_{\rho/\alpha}}{n} < \frac{\rho_n}{\alpha_n}  \leq  M_{\rho/\alpha}$.
Let $M_n \coloneqq \Big(s_n^2\rho_n^2\big(\frac{\rho_n n}{\alpha_n m_{\rho/\alpha}}\big)^{\!\nicefrac{1}{3}}-s_n^2\rho_n^2 \Big)^{-1} > 0$.
We have for all $t \in [0,1]$:
\begin{equation}
\int_{{\cal B}_n} \frac{d\epsilon}{s_n^2}\int_0^1dt\, \E\big\langle \big(Q - \E\langle Q\rangle_{n, t, \epsilon}\big)^2\big \rangle_{n, t, \epsilon} 
\,\leq\, CM_n\;,
\end{equation}
where $C$ is a polynomial in $\big(S,\Vert\varphi\Vert_\infty, \Vert \partial_x \varphi \big\Vert_\infty, \Vert \partial_{xx} \varphi \Vert_\infty, M_\alpha, M_{\rho/\alpha}, m_{\rho/\alpha}\big)$ with positive coefficients.
\end{proposition}
We prove Proposition~\ref{prop:concentration_overlap} in Appendix~\ref{appendix-overlap}.
We can now prove the fundamental sum rule.
\begin{proposition}[Fundamental sum rule]\label{prop:cancel_remainder}
Suppose that $\forall (t,\epsilon) \in [0,1] \times \mathcal{B}_n: q_{\epsilon}(t) = \E \langle Q \rangle_{n,t,\epsilon}$.
Under the assumptions of Proposition \ref{prop:concentration_overlap}, we have:
\begin{align*}
&\frac{I(\bX^*;\bY|\bm{\Phi})}{m_n}
= \mathcal{O}\big(\sqrt{M_n}\,\big) + \mathcal{O}\bigg(\frac{s_n}{\sqrt{\rho_n}}\bigg)\\
&\qquad\;\,
+\int_{\mathcal{B}_n} \frac{d\epsilon}{s_n^2} \bigg\{
\frac{1}{\alpha_n}I_{P_{0,n}}\Big(\int_0^1 \!\! r_\epsilon(t)dt\Big)
+ I_{P_{\mathrm{out}}}\bigg(\int_0^1 \!\! q_{\epsilon}(t) dt ,1\bigg)
-\frac{\rho_n}{2\alpha_n}\int_0^1 r_\epsilon(t)\big(1-q_\epsilon(t)\big) dt \bigg\}\:.
\end{align*}
The constant factors in $\mathcal{O}\big(\sqrt{M_n}\,\big)$ and $\mathcal{O}\big(\nicefrac{s_n}{\sqrt{\rho_n}}\big)$ are $\sqrt{C_1}$ and $\sqrt{C_2}$ where $C_1, C_2$ are polynomials in $\big(S,\Vert\varphi\Vert_\infty, \Vert \partial_x \varphi \big\Vert_\infty, \Vert \partial_{xx} \varphi \Vert_\infty, M_\alpha, M_{\rho/\alpha}, m_{\rho/\alpha}\big)$ with positive coefficients.
\end{proposition}
\begin{proof}
Let $\E_{\epsilon,t} \coloneqq \int_{{\cal B}_n} \frac{d\epsilon}{s_{n}^2} \int_0^1dt$.
By Cauchy-Schwarz inequality: 
\begin{multline*}
\bigg\vert
\int_{{\cal B}_n} \frac{d\epsilon}{s_{n}^2} \int_0^1 dt\,
\E\,\bigg\langle \big(Q - q_\epsilon(t)\big) \bigg( \frac{1}{m_n}\sum_{\mu=1}^{m_n} u_{Y_\mu^{(t,\epsilon)}}'(S_\mu^{(t,\epsilon)}) u'_{Y_\mu^{(t,\epsilon)}}(s_\mu^{(t,\epsilon)}) - \frac{\rho_n}{\alpha_n}r_\epsilon(t)\bigg)\bigg\rangle_{\!\! n,t,\epsilon}
\bigg\vert^2\\
\leq 
\int_{{\cal B}_n} \frac{d\epsilon}{s_{n}^2} \int_0^1 dt\,
\E\,\bigg\langle \bigg(\frac{1}{m_n}\sum_{\mu=1}^{m_n} u_{Y_\mu^{(t,\epsilon)}}'(S_\mu^{(t,\epsilon)}) u'_{Y_\mu^{(t,\epsilon)}}(s_\mu^{(t,\epsilon)}) - \frac{\rho_n}{\alpha_n}r_\epsilon(t)\bigg)^{\!\! 2}\bigg\rangle_{\!\! n,t,\epsilon}\\
\cdot \int_{{\cal B}_n} \frac{d\epsilon}{s_{n}^2} \int_0^1 dt\,
\E \big\langle \big(Q - q_{\epsilon}(t)\big)^2\,\big\rangle_{n, t, \epsilon}\;.
\end{multline*}
The first factor on the right-hand side of this inequality is bounded by a constant that depends polynomially on $\Vert \varphi \Vert_\infty, \Vert \partial_x \varphi \Vert_\infty$
\footnote{
Remember that $r_{\epsilon}$ takes its values in $[0,\frac{\alpha_n}{\rho_n}r_{\max}]$.
Besides, under \ref{hyp:c2}, $u'_{Y_\mu^{(t,\epsilon)}}$ is upper bounded by
$(\vert Y_\mu^{(t,\epsilon)} \vert + \Vert \varphi \Vert_\infty)\Delta^{-1} \Vert \partial_x \varphi \Vert_\infty
= (\sqrt{\Delta} \vert Z_\mu \vert + 2\Vert \varphi \Vert_\infty)\Delta^{-1} \Vert \partial_x \varphi \Vert_\infty$ (see the inequality \eqref{upperbound_u_y(x)} in Appendix~\ref{appendix:scalar_channels}).
The noise $Z_\mu$ is averaged over thanks to the expectation.
}
Since $\forall (t,\epsilon) \in [0,1]\times \mathcal{B}_n: q_{\epsilon}(t) = \E \langle Q \rangle_{n,t,\epsilon}$, the second term is in
$\mathcal{O}(M_n)$ (see Proposition~\ref{prop:concentration_overlap}).
Therefore, by Proposition~\ref{prop:derivative_i_n_main}:
\begin{equation}\label{eq:id_fluctuation}
\E_{\epsilon,t}\,i'_{n,\epsilon}(t)
= \mathcal{O}\big(\sqrt{M_n}\big)+\mathcal{O}\Big(\frac{1}{\rho_n\sqrt{n}}\Big) + \E_{\epsilon,t}\,
	\frac{\rho_n}{2\alpha_n}r_\epsilon(t)\big(1-q_\epsilon(t)\big)\;.
\end{equation}
Note that $\nicefrac{1}{\rho_n\sqrt{n}} = \mathcal{O}(\sqrt{M_n})$.
Integrating~\eqref{link_i_n(0)_i_n(1)} over $\epsilon \in \mathcal{B}_n$ and making use of \eqref{eq:id_fluctuation} give the result.
\end{proof}
\subsection{Matching bounds}
To prove Theorem~\ref{th:RS_1layer}, we will lower and upper bound $\nicefrac{I(\bX^*;\bY \vert \bm{\Phi})}{m_n}$ by the same quantity, up to a small error.
To do so we will plug two different choices of interpolation functions $R_1(\cdot,\epsilon), R_2(\cdot,\epsilon)$ in the sum-rule of Proposition~\ref{prop:cancel_remainder}.
In both cases, the interpolation functions will be the solutions of a second-order ordinary differential equation (ODE).
We now describe these ODEs.

Fix $t \in [0,1]$ and $R = (R_1, R_2) \in [0,+\infty) \times [0,t+2s_n]$.
Consider the observations:
\begin{align}
\begin{cases}
Y_{\mu}^{(t,R_2)}  &\sim P_{\mathrm{out}}(\,\cdot\, \vert \, S_{\mu}^{(t,R_2)})\;,\;\:  1 \leq \mu \leq m_n\,; \\
\widetilde{Y}_{i}^{(t,R_1)} &= \;\sqrt{R_1}\, X^*_i + \widetilde{Z}_i\;\;\,, \;\: 1 \leq \, i \,  \leq \, n \;\;\,;
\end{cases}
\end{align}
where $S_{\mu}^{(t,R_2)} = S_{\mu}^{(t,R_2)}(\bX^*,W_\mu^*) \coloneqq \sqrt{\nicefrac{(1-t)}{k_n}}\, (\bm{\Phi} \bX^*)_\mu  + \sqrt{R_2} \,V_{\mu} + \sqrt{t+2s_n-R_2} \,W_{\mu}^*$.
The joint posterior density of $(\bX^*,\bW^*)$ given $(\bY^{(t,R_2)},\widetilde{\bY}^{(t,R_1)},\bm{\Phi},\bV)$ is:
\begin{multline*}
dP(\bx,\bw \vert \bY^{(t,R_2)},\widetilde{\bY}^{(t,R_1)},\bm{\Phi},\bV)\\
\propto \prod_{i=1}^n dP_{0,n}(x_i)\,e^{-\frac{1}{2} \big( \sqrt{R_{1}}x_i -\widetilde{Y}_i^{(t,R_1)}\big)^2}\,
\prod_{\mu=1}^{m_n} \frac{dw_\mu}{\sqrt{2\pi}} e^{-\frac{w_\mu^2}{2}} P_{\mathrm{out}}(Y_{\mu}^{(t,R_2)}\vert S_{\mu}^{(t,R_2)}(\bx,w_\mu)) \;.
\end{multline*}
The angular brackets $\langle - \rangle_{n,t,R}$ denotes the expectation w.r.t.\ this posterior.
Let $r \in [0,r_{\max}]$, $F_2^{(n)}(t, R) \coloneqq \E \langle Q \rangle_{n,t,R}$ and
$F_1^{(n)}(t, R) \coloneqq -2\frac{\alpha_n}{\rho_n} \frac{\partial I_{P_{\mathrm{out}}}}{\partial q}\big\vert_{q = \E \langle Q \rangle_{n,t,R}, \rho=1}$. We will consider the two following second-order ODEs with initial value $\epsilon \in [s_n,2s_n]^2$:
\begin{align}
&y'(t) = \quad\: \Big(\frac{\alpha_n}{\rho_n}r\,, F_2^{(n)}(t, y(t))\Big)\qquad, \: y(0)=\epsilon \;; \label{def:ode_upperbound}\\
&y'(t) = \big(F_1^{(n)}(t,y(t))\,, F_2^{(n)}(t, y(t))\big) \:, \: y(0)=\epsilon\;.\label{def:ode_lowerbound}
\end{align}
The next proposition sums up useful properties on the solutions of these two ODEs, i.e., our two kinds of interpolation
functions. The proof is given in Appendix~\ref{appendix:properties_ode}.
\begin{proposition}\label{prop:ode}
Suppose that \ref{hyp:bounded},~\ref{hyp:c2},~\ref{hyp:phi_gauss2} hold and that $\Delta = \E_{X \sim P_0}[X^2] = 1$.
For all ${\epsilon \in \mathcal{B}_n}$, there exists a unique global solution $R(\cdot,\epsilon): [0,1] \to [0,+\infty)^2 $ to \eqref{def:ode_lowerbound}.
This solution is continuously differentiable and its derivative $R'(\cdot,\epsilon)$ satisfies $R'([0,1],\epsilon) \subseteq [0, \nicefrac{\alpha_n r_{\max}}{\rho_n}] \times [0,1]$.
Besides, for all $t \in [0,1]$, $R(t,\cdot)$ is a $\mathcal{C}^1$-diffeomorphism from $\mathcal{B}_n$ onto its image whose Jacobian determinant is greater than, or equal to, one.
Finally, the same statement holds if we consider \eqref{def:ode_upperbound} instead.
\end{proposition}
\begin{proposition}[Upper bound]\label{prop:upper_bound}
Suppose that \ref{hyp:bounded},~\ref{hyp:c2},~\ref{hyp:phi_gauss2} hold, that $\Delta = \E_{P_0}[X^2] = 1$
and that $\forall n \in \N^*$:
$\alpha_n \leq M_{\alpha},\,\frac{m_{\rho/\alpha}}{n} < \frac{\rho_n}{\alpha_n}  \leq  M_{\rho/\alpha}\;$ for positive numbers $M_\alpha, M_{\rho/\alpha}, m_{\rho/\alpha}$. Then:
\begin{equation}\label{upperbound_mutual_info}
\forall n \in \N^*:\quad
\frac{I(\bX^*;\bY \vert \bm{\Phi}) }{m_n}
\leq {\adjustlimits \inf_{r \in[0,r_{\max}]} \sup_{q \in [0,1]}} i_{\scriptstyle{\mathrm{RS}}}\big(q, r; \alpha_n, \rho_n \big)
+ \mathcal{O}(\sqrt{M_n}) +  \mathcal{O}\bigg(\frac{s_n}{\sqrt{\rho_n}}\bigg) \;.
\end{equation}
\end{proposition}
\begin{proof}
Fix $r \in r_{\max}$.
For all $\epsilon \in \mathcal{B}_n$, $(R_1(\cdot,\epsilon),R_2(\cdot,\epsilon))$ is the unique solution to the ODE \eqref{def:ode_upperbound} (see Proposition~\ref{prop:ode}).
Let $q_{\epsilon}(t) \coloneqq R_2'(t,\epsilon) = \E\langle Q \rangle_{n,t,\epsilon}$, $r_{\epsilon}(t) \coloneqq R_1'(t,\epsilon) = \frac{\alpha_n r}{\rho_n}$.
By Proposition~\ref{prop:ode}, the families of functions $(q_\epsilon)_{\epsilon \in \mathcal{B}_n}$, $(r_\epsilon)_{\epsilon \in \mathcal{B}_n}$ are \textit{regular}.
We can now apply Proposition~\ref{prop:cancel_remainder} to get:
\begin{align}
\frac{I(\bX^*;\bY|\bm{\Phi})}{m_n}
&= \int_{\mathcal{B}_n} \frac{d\epsilon}{s_n^2} \, i_{\scriptstyle{\mathrm{RS}}}\bigg(\int_0^1 q_{\epsilon}(t)\, dt,r; \alpha_n, \rho_n\bigg)
+ \mathcal{O}(\sqrt{M_n}) + \mathcal{O}\bigg(\frac{s_n}{\sqrt{\rho_n}}\bigg)\nonumber\\
&\leq \sup_{q \in [0,1]} i_{\scriptstyle{\mathrm{RS}}}\big(q,r ; \alpha_n, \rho_n\big)
+ \mathcal{O}(\sqrt{M_n}) + \mathcal{O}\bigg(\frac{s_n}{\sqrt{\rho_n}}\bigg)\;.\label{upperbound_i_n_r}
\end{align}
The inequality \eqref{upperbound_i_n_r} holds for all $r \in [0,r_{\max}]$ and the constant factors in the quantities $\mathcal{O}(\sqrt{M_n})$, $\mathcal{O}\big(\nicefrac{s_n}{\sqrt{\rho_n}}\big)$ are uniform in $r$. Hence the inequality \eqref{upperbound_mutual_info} with the infimum over $r$.
\end{proof}
\begin{proposition}[Lower bound]\label{prop:lower_bound}
Under the same hypotheses than Proposition~\ref{prop:upper_bound}, we have:
\begin{equation}\label{lowerbound_mutual_info}
\forall n \in \N^*:\quad
\frac{I(\bX^*;\bY \vert \bm{\Phi}) }{m_n}
\geq {\adjustlimits \inf_{r \in[0,r_{\max}]} \sup_{q \in [0,1]}} i_{\scriptstyle{\mathrm{RS}}}\big(q, r; \alpha_n, \rho_n \big)
+ \mathcal{O}(\sqrt{M_n}) + \mathcal{O}\bigg(\frac{s_n}{\sqrt{\rho_n}}\bigg) \;.
\end{equation}		
\end{proposition}
\begin{proof}
For all $\epsilon \in \mathcal{B}_n$, $(R_1(\cdot,\epsilon),R_2(\cdot,\epsilon))$ is the unique solution to the ODE \eqref{def:ode_lowerbound} (see Proposition~\ref{prop:ode}).
We define $q_{\epsilon}(t) \coloneqq R_2'(t,\epsilon) = \E\langle Q \rangle_{n,t,\epsilon}$, $r_{\epsilon}(t) \coloneqq R_1'(t,\epsilon) = -\frac{2\alpha_n}{\rho_n} \frac{\partial I_{P_{\mathrm{out}}}}{\partial q}\big\vert_{q = q_\epsilon(t), \rho=1}$.
By Proposition~\ref{prop:ode}, the families of functions $(q_\epsilon)_{\epsilon \in \mathcal{B}_n}$, $(r_\epsilon)_{\epsilon \in \mathcal{B}_n}$ are \textit{regular}.
Note that $\forall \epsilon \in \mathcal{B}_n$:
\begin{align}
&\frac{1}{\alpha_n}I_{P_{0,n}}\Big(\int_0^1 \!\! r_\epsilon(t)\,dt\Big)
+ I_{P_{\mathrm{out}}}\bigg(\int_0^1 \!\! q_{\epsilon}(t)\,dt ,1\bigg)
-\frac{\rho_n}{2\alpha_n}\int_0^1 r_\epsilon(t)\big(1-q_\epsilon(t)\big)\,dt\nonumber\\
&\qquad\qquad\qquad\qquad
\geq \int_0^1 \bigg\{ \frac{1}{\alpha_n}I_{P_{0,n}}\big( r_\epsilon(t)\big)
+ I_{P_{\mathrm{out}}}\big(q_{\epsilon}(t),1\big)
-\frac{\rho_n}{2\alpha_n}r_\epsilon(t)\big(1-q_\epsilon(t)\big)\bigg\}\,dt\nonumber\\
&\qquad\qquad\qquad\qquad
=\int_0^1 \bigg\{\sup_{q \in [0,1]} \frac{1}{\alpha_n}I_{P_{0,n}}\big( r_\epsilon(t)\big)
+ I_{P_{\mathrm{out}}}(q,1)
-\frac{\rho_n}{2\alpha_n}r_\epsilon(t)(1-q)\bigg\}\,dt\nonumber\\
&\qquad\qquad\qquad\qquad
= \int_0^1 {\sup}_{q \in [0,1]} \: i_{\scriptstyle{\mathrm{RS}}}\bigg(q, \frac{\rho_n}{\alpha_n}r_\epsilon(t); \alpha_n, \rho_n\bigg)\,dt
\label{before_lowerbound_potential_term_in_sumrule}\\
&\qquad\qquad\qquad\qquad
\geq \adjustlimits{\inf}_{r \in [0,r_{\max}]} {\sup}_{q \in [0,1]} i_{\scriptstyle{\mathrm{RS}}}\big(q, r; \alpha_n, \rho_n\big)\;.
\label{lowerbound_potential_term_in_sumrule}
\end{align}
The first inequality is an application of Jensen's inequality to the concave functions $I_{P_{0,n}}, I_{P_{\mathrm{out}}}(\cdot,1)$ (see Lemmas~\ref{lemma:property_I_P0n} and \ref{lemma:property_I_Pout}).
The subsequent equality is because the global maximum of the concave function $h: q \in [0,1] \mapsto I_{P_{\mathrm{out}}}(q,1)
-\frac{\rho_n}{2\alpha_n}r_\epsilon(t)(1-q)$ is reached at $q_{\epsilon}(t)$ since $h'(q_{\epsilon}(t)) = 0$.
The equality \eqref{before_lowerbound_potential_term_in_sumrule} follows from the definition \eqref{def_i_RS} of $i_{\scriptstyle{\mathrm{RS}}}$.
Finally, the inequality \eqref{lowerbound_potential_term_in_sumrule} is because $r_{\epsilon}(t) \in \big[0,\frac{\alpha_n}{\rho_n}r_{\max}\big]$ and we simply lowerbound the integrand in \eqref{before_lowerbound_potential_term_in_sumrule} by a quantity independent of $t \in [0,1]$.
We now apply Proposition~\ref{prop:cancel_remainder} and make use of \eqref{lowerbound_potential_term_in_sumrule} to obtain the inequality \eqref{lowerbound_mutual_info}.
\end{proof}
To prove Theorem~\ref{th:RS_1layer}, it remains to combine Propositions \ref{prop:upper_bound} and \ref{prop:lower_bound} with the identity
\begin{equation}\label{identity_invert_sup_inf}
\adjustlimits{\inf}_{r \in[0,r_{\max}]} {\sup}_{q \in [0,1]} i_{\scriptstyle{\mathrm{RS}}}\big(q, r; \alpha_n, \rho_n\big)
= \adjustlimits{\inf}_{r \geq 0} {\sup}_{q \in [0,1]} i_{\scriptstyle{\mathrm{RS}}}\big(q,r; \alpha_n, \rho_n\big)
= \adjustlimits{\inf}_{q \in [0,1]} {\sup}_{r \geq 0} i_{\scriptstyle{\mathrm{RS}}}\big(q,r; \alpha_n, \rho_n\big)\,,
\end{equation}
and the choice $\rho_n = \Theta(n^{-\lambda})$, $\alpha_n = \gamma \rho_n \vert\ln\rho_n\vert$ and $s_n = \Theta(n^{-\beta})$ with $\lambda \in [0,\nicefrac{1}{9})$, $\gamma > 0$ and $\beta \in (\nicefrac{\lambda}{2},\nicefrac{1}{6}-\lambda)$.
Optimizing over $\beta$ to maximize the convergence rate of
\begin{equation*}
\mathcal{O}(\sqrt{M_n}) + \mathcal{O}\bigg(\frac{s_n}{\sqrt{\rho_n}}\bigg) = \mathcal{O}\bigg(\max\bigg\{\frac{1}{n^{\beta-\nicefrac{\lambda}{2}}},\frac{\vert \ln n \vert^{\nicefrac{1}{6}}}{n^{\nicefrac{1}{6}-\lambda-\beta}}\bigg\}\bigg)
\end{equation*}
yields Theorem~\ref{th:RS_1layer}.
The identity \eqref{identity_invert_sup_inf} has been proved in \cite[Proposition 7 and Corollary 7 in SI]{Barbier2019Optimal}.
\section{Proof of Theorem~\ref{theorem:limit_MI_discrete_prior} for a Bernoulli prior}\label{section:specialization_bernoulli_prior}
In this section, we assume that $P_{0,n} \coloneqq (1-\rho_n) \delta_0 + \rho_n \delta_1$ and we prove Theorem~\ref{theorem:limit_MI_discrete_prior} for this specific case.
The proof contains all the main ideas needed to establish Theorem~\ref{theorem:limit_MI_discrete_prior} while being technically simpler.
The interested reader can find the proof of Theorem~\ref{theorem:limit_MI_discrete_prior} for a general discrete prior with finite support in Appendix~\ref{appendix:specialization_discrete_prior}.

For $\rho_n, \alpha_n > 0$ we denote the variational problem appearing in Theorem~\ref{th:RS_1layer} by
\begin{equation}
I(\rho_n, \alpha_n) \coloneqq \adjustlimits{\inf}_{q \in [0,1]} {\sup}_{r \geq 0}\; i_{\scriptstyle{\mathrm{RS}}}(q, r; \alpha_n, \rho_n)\;,
\end{equation}
where the potential $i_{\scriptstyle{\mathrm{RS}}}$ is defined in \eqref{def_i_RS}.
Let $X^* \sim P_{0,n}$, $Z \sim \cN(0,1)$ be independent random variables.
We define for all $r \geq 0$:
\begin{equation}\label{def:psi_P0n_bernoulli}
\psi_{P_{0,n}}(r) \coloneqq \E \Big[\ln\Big(1 - \rho_n + \rho_n e^{-\frac{r}{2} + rX^* + \sqrt{r} Z}\Big) \Big]\;.
\end{equation}
Note that $I_{P_{0,n}}(r) \coloneqq I(X^*;\sqrt{r}\,X^* + Z) = \frac{r\rho_n}{2} - \psi_{P_{0,n}}(r)$ so
\begin{equation}
I(\rho_n, \alpha_n) =
\inf_{q \in [0,1]} I_{P_{\mathrm{out}}}(q, 1) + \sup_{r \geq 0} \bigg\{\frac{rq}{2} - \frac{1}{\alpha_n} \psi_{P_{0,n}}\bigg(\frac{\alpha_n}{\rho_n}r\bigg)\bigg\}\;.
\end{equation}
The latter expression for $I(\rho_n, \alpha_n)$ is easier to work with.
We point out that $\psi_{P_{0,n}}$ is twice differentiable, nondecreasing, strictly convex and $\frac{\rho_n}{2}$-Lipschitz on $[0,+\infty)$ (see Lemma~\ref{lemma:property_I_P0n}) while $I_{P_{\mathrm{out}}}(\cdot,1)$ is nonincreasing and concave on $[0,1]$ (see \cite[Appendix B.2, Proposition 18]{Barbier2019Optimal}).

Our goal is now to compute the limit of $I(\rho_n, \alpha_n)$ when $\alpha_n \coloneqq \gamma \rho_n \vert \ln \rho_n \vert$ for a fix $\gamma > 0$ and $\rho_n \to 0$.
Once we know this limit, we directly obtain Theorem~\ref{theorem:limit_MI_discrete_prior} thanks to Theorem~\ref{th:RS_1layer}.
We first show that -- for $q$ in a growing interval -- the point at which the supremum over $r$ is achieved is located in an interval shrinking on $r^* \coloneqq \nicefrac{2}{\gamma}$.
\begin{lemma}\label{lemma:location_r*(q)_bernoulli}
Let $P_{0,n} \coloneqq (1-\rho_n) \delta_0 + \rho_n \delta_1$ and $\alpha_n \coloneqq \gamma \rho_n \vert \ln \rho_n \vert$ for a fix $\gamma > 0$.
Define $g_{\rho_n}: r \in (0,+\infty) \mapsto \frac{2}{\rho_n}\psi_{P_{0,n}}^\prime\big(\frac{\alpha_n}{\rho_n} r\big)$ and $\forall \rho_n \in (0,e^{-1}):$
\begin{equation}
a_{\rho_n} \coloneqq g_{\rho_n}\bigg(\frac{2(1 - \vert \ln \rho_n \vert^{-\frac14})}{\gamma}\bigg)
\quad,\quad
b_{\rho_n} \coloneqq g_{\rho_n}\bigg(\frac{2(1 + \vert \ln \rho_n \vert^{-\frac14})}{\gamma}\bigg) \;.
\end{equation}
We have $[a_{\rho_n}, b_{\rho_n}] \subset (\rho_n,1)$ and $\lim_{\rho_n \to 0} a_{\rho_n} = 0$,  $\lim_{\rho_n \to 0} b_{\rho_n} = 1$.
Besides, for every $q \in (\rho_n,1)$ there exists a unique $r_{n}^*(q) \in (0,+\infty)$ such that
\begin{equation}
\frac{r_{n}^*(q)q}{2} - \frac{1}{\alpha_n} \psi_{P_{0,n}}\bigg(\frac{\alpha_n}{\rho_n}r_{n}^*(q)\bigg)
= \sup_{r \geq 0} \: \frac{rq}{2} - \frac{1}{\alpha_n} \psi_{P_{0,n}}\bigg(\frac{\alpha_n}{\rho_n}r\bigg)\;,
\end{equation}
and
\begin{align}
\forall q \in [a_{\rho_n}, b_{\rho_n}]:
	&\;\frac{2(1-\vert \ln \rho_n \vert^{-\frac14})}{\gamma}
	\leq r_{n}^*(q)
	\leq \frac{2(1+\vert \ln \rho_n \vert^{-\frac14})}{\gamma}\;,\\
\forall q \in [b_{\rho_n},1):&\; r_{n}^*(q) \geq \frac{2(1+\vert \ln \rho_n \vert^{-\frac14})}{\gamma}\;.
\end{align}
\end{lemma}
\begin{proof}
For every $q \in (0,1)$ we define $f_{\rho_n,q}: r \in [0,+\infty) \mapsto \frac{rq}{2} - \frac{1}{\alpha_n} \psi_{P_{0,n}}\big(\frac{\alpha_n}{\rho_n}r\big)$ whose supremum over $r$ we want to compute.
The derivative of $f_{\rho_n,q}$ with respect to $r$ reads
\begin{equation}
f'_{\rho_n,q}(r) = \frac{q}{2} - \frac{1}{\rho_n} \psi'_{P_{0,n}}\bigg(\frac{\alpha_n}{\rho_n}r\bigg)\;.
\end{equation}
The derivative $\psi'_{P_{0,n}}$ is continuously increasing and thus one-to-one from $(0,+\infty)$ onto $(\rho_n^2/2, \rho_n/2)$.
Therefore, if $q \in (0,\rho_n]$ then $f'_{\rho_n,q} \leq 0$ and the supremum of $f_{\rho_n,q}$ is achieved at $r=0$.
On the contrary, if $q \in (\rho_n,1)$ then there exists a unique solution  $r_{n}^*(q) \in (0,+\infty)$ to the critical point equation $f'_{\rho_n,q}(r) = 0$.
As $f_{\rho_n,q}$ is concave (given that $\psi_{P_0,n}$ is convex), this solution $r_{n}^*(q)$ is the global maximum of $f_{\rho_n,q}$.
We now transform the critical point equation:
\begin{equation}
f_{\rho_n,q}(r) = 0 \Leftrightarrow  \frac{2}{\rho_n} \psi'_{P_{0,n}}\bigg(\frac{\alpha_n}{\rho_n}r\bigg) = q
\Leftrightarrow g_{\rho_n}(r) = q \;,
\end{equation}
where $g_{\rho_n}: r \mapsto \frac{2}{\rho_n}\psi_{P_{0,n}}^\prime\big(\frac{\alpha_n}{\rho_n} r\big)$ is increasing and one-to-one from $(0,+\infty)$ to $(\rho_n,1)$.
For all $\rho_n \in (0,e^{-1}): \vert \ln \rho_n \vert^{-\frac14} \in (0,1)$.
By Lemma~\ref{lemma:bounds_g_rhon} (directly following the proof) applied with $\epsilon = \vert \ln \rho_n \vert^{-\frac14}$, we have:
\begin{align}
\rho_n < a_{\rho_n} &\coloneqq g_{\rho_n}\bigg(\frac{2(1 - \vert \ln \rho_n \vert^{-\frac14})}{\gamma}\bigg)
\leq \frac{\exp\Big(\!\!-\frac{\vert \ln \rho_n \vert^{\frac12}}{16(1-\vert \ln \rho_n \vert^{-\nicefrac14})}\Big)}{2}
+ \frac{\exp\Big(\!\!-\frac{\vert \ln \rho_n \vert^{\frac34}}{2} \Big)}{1-\rho_n}\;;\label{bounds_a_n}\\
1 > b_{\rho_n} &\coloneqq g_{\rho_n}\bigg(\frac{2(1 + \vert \ln \rho_n \vert^{-\frac14})}{\gamma}\bigg)
\geq \frac{1-0.5\exp\Big(-\frac{\vert \ln \rho_n \vert^{\nicefrac12}}{16}\Big)}{1 + \exp\Big(- \frac{\vert \ln \rho_n \vert^{\nicefrac34}}{2} \Big)}\;.\label{bounds_b_n}
\end{align}
It directly follows from \eqref{bounds_a_n} that $\lim_{\rho_n \to 0} a_{\rho_n} = 0$ and from \eqref{bounds_b_n} that $\lim_{\rho_n \to 0} b_{\rho_n} = 1$.
As $g_{\rho_n}$ is increasing, if $q = g_{\rho_n}(r_n^*(q)) \in [a_{\rho_n}, b_{\rho_n}]$ then
\begin{equation*}
\frac{2(1-\vert \ln \rho_n \vert^{-\frac14})}{\gamma}
\leq r_{n}^*(q)
\leq \frac{2(1+\vert \ln \rho_n \vert^{-\frac14})}{\gamma}
\end{equation*}
while if $q = g_{\rho_n}(r_n^*(q)) \in [b_{\rho_n},1)$ then $r_{n}^*(q) \geq \frac{2(1+\vert \ln \rho_n \vert^{-\frac14})}{\gamma}$.
\end{proof}
\begin{lemma}\label{lemma:bounds_g_rhon}
Let $\alpha_n \coloneqq \gamma \rho_n \vert \ln \rho_n \vert$ for a fix $\gamma > 0$ and define $g_{\rho_n}: r \mapsto \frac{2}{\rho_n}\psi_{P_{0,n}}^\prime\big(\frac{\alpha_n}{\rho_n} r\big)$. For all $(\rho_n, \epsilon) \in (0,1)^2$ we have:
\begin{align}
g_{\rho_n}\bigg(\frac{2(1 -\epsilon)}{\gamma}\bigg)
&\leq \frac{\exp\big(-\frac{\epsilon^2}{16}\frac{\vert \ln \rho_n \vert}{1-\epsilon}\big)}{2}
+ \frac{
	\exp\big(-\frac{\epsilon}{2}\vert \ln \rho_n \vert \big)
}{1-\rho_n}\;;\\
g_{\rho_n}\bigg(\frac{2(1 + \epsilon)}{\gamma}\bigg)
&\geq
\frac{1-0.5\exp\big(-\frac{\epsilon^2}{16} \vert \ln \rho_n \vert\big)}{1 + \exp\big(- \frac{\epsilon}{2} \vert \ln \rho_n \vert\big)} \;.
\end{align}
\end{lemma}
\begin{proof}
The derivative of $\psi_{P_{0,n}}$ reads
$\psi_{P_{0,n}}^\prime(r) = \frac{\rho_n}{2} \E \big[\big(1 + \frac{1-\rho_n}{\rho_n} e^{-\frac{r}{2} - \sqrt{r}Z} \big)^{-1}\big]$. Therefore:
\begin{equation}
g_{\rho_n}(r) = \E \Bigg[\frac{1}{1 + (1-\rho_n)\exp\big\{\vert \ln \rho_n \vert \big(1-\nicefrac{\gamma r}{2} - \sqrt{\frac{\gamma r}{\vert \ln \rho_n \vert}}Z\big)\big\}} \Bigg]
\in (0,1)\;.
\end{equation}
Hence for all $\epsilon \in (0,1)$ we have:
\begin{equation}
g_{\rho_n}\bigg(\frac{2(1 \pm \epsilon)}{\gamma}\bigg)
= \E \Bigg[\frac{1}{1 + (1-\rho_n)\exp\big\{\vert \ln \rho_n \vert \big(\mp \epsilon - \sqrt{\frac{2(1 \pm \epsilon)}{\vert \ln \rho_n \vert}}Z\big)\big\}} \Bigg]\;.
\end{equation}
By the dominated convergence theorem $\lim_{\rho_n \to 0}g_{\rho_n}\big(\nicefrac{2(1 + \epsilon)}{\gamma}\big) = 1$ and $\lim_{\rho_n \to 0}g_{\rho_n}\big(\nicefrac{2(1 - \epsilon)}{\gamma}\big) = 0$.
We first lower bound $g_{\rho_n}\big(\nicefrac{2(1 + \epsilon)}{\gamma}\big)$. Note that $\forall z \geq -\frac{\epsilon}{2} \sqrt{\frac{\vert \ln \rho_n \vert}{2(1+\epsilon)}}: - \epsilon - \sqrt{\frac{2(1 + \epsilon)}{\vert \ln \rho_n \vert}}z \leq -\frac{\epsilon}{2}$. Hence:
\begin{align}
g_{\rho_n}\bigg(\frac{2(1 + \epsilon)}{\gamma}\bigg)
&= \int_{-\infty}^{+\infty} \frac{dz}{\sqrt{2\pi}}\frac{e^{-\frac{z^2}{2}}}{1 + (1-\rho_n)\exp\big\{\vert \ln \rho_n \vert \big(-\epsilon - \sqrt{\frac{2(1 + \epsilon)}{\vert \ln \rho_n \vert}}z\big)\big\}} \nonumber\\
&\geq \int_{-\frac{\epsilon}{2} \sqrt{\frac{\vert \ln \rho_n \vert}{2(1+\epsilon)}}}^{+\infty} \frac{dz}{\sqrt{2\pi}}\frac{e^{-\frac{z^2}{2}}}{1 + (1-\rho_n)\exp\big(- \frac{\epsilon}{2} \vert \ln \rho_n \vert\big)} \nonumber\\
&=
\frac{1-F\Big(-\frac{\epsilon}{2} \sqrt{\frac{\vert \ln \rho_n \vert}{2(1+\epsilon)}}\Big)}{1 + (1-\rho_n)\exp\big(- \frac{\epsilon}{2} \vert \ln \rho_n \vert\big)}
\geq
\frac{1-F\Big(-\frac{\epsilon}{2} \sqrt{\frac{\vert \ln \rho_n \vert}{2}}\Big)}{1 + \exp\big(- \frac{\epsilon}{2} \vert \ln \rho_n \vert\big)}  \;,
\end{align}
where $F(x) \coloneqq \int_{-\infty}^x \frac{dz}{\sqrt{2\pi}}e^{-\frac{z^2}{2}}$ is the cumulative distribution function of the standard normal distribution.
Making use of the upper bound $F(-x) \leq \frac{e^{-\nicefrac{x^2}{2}}}{2}$ for $x > 0$ yields
\begin{equation}\label{lowerbound_g_2(1+eps)}
g_{\rho_n}\bigg(\frac{2(1 + \epsilon)}{\gamma}\bigg)
\geq
\frac{1-0.5\exp\big(-\frac{\epsilon^2}{16} \vert \ln \rho_n \vert\big)}{1 + \exp\big(- \frac{\epsilon}{2} \vert \ln \rho_n \vert\big)}  \;.
\end{equation}
Next we prove the upper bound on $g_{\rho_n}\big(\nicefrac{2(1 - \epsilon)}{\gamma}\big)$.
We denote the indicator function of an event $\mathcal{E}$ by $\bm{1}_{\mathcal{E}}$.
We have:
\begin{align}
g_{\rho_n}\bigg(\frac{2(1 -\epsilon)}{\gamma}\bigg)
&= \E \Bigg[\frac{1}{1 + (1-\rho_n)\exp\big\{\vert \ln \rho_n \vert \big(\epsilon - \sqrt{\frac{2(1 - \epsilon)}{\vert \ln \rho_n \vert}}Z\big)\big\}} \Bigg]\\
&\leq \E \Bigg[\bm{1}_{\big\{Z \geq \frac{\epsilon}{2}\sqrt{\frac{\vert \ln \rho_n \vert}{2(1-\epsilon)}}\big\}}
+ \frac{
\bm{1}_{\big\{Z < \frac{\epsilon}{2}\sqrt{\frac{\vert \ln \rho_n \vert}{2(1-\epsilon)}}\big\}}
}{1 + (1-\rho_n)\exp\big(\frac{\epsilon}{2}\vert \ln \rho_n \vert \big)} \Bigg]\nonumber\\
&= F\bigg(-\frac{\epsilon}{2}\sqrt{\frac{\vert \ln \rho_n \vert}{2(1-\epsilon)}}\bigg)
+ \frac{
1 - F\Big(-\frac{\epsilon}{2}\sqrt{\frac{\vert \ln \rho_n \vert}{2(1-\epsilon)}}\Big)
}{1 + (1-\rho_n)\exp\big(\frac{\epsilon}{2}\vert \ln \rho_n \vert \big)}\nonumber\\
&\leq F\bigg(-\frac{\epsilon}{2}\sqrt{\frac{\vert \ln \rho_n \vert}{2(1-\epsilon)}}\bigg)
+ \frac{
	\exp\big(-\frac{\epsilon}{2}\vert \ln \rho_n \vert \big)
}{1-\rho_n}\nonumber\\
&\leq \frac{\exp\big(-\frac{\epsilon^2}{16}\frac{\vert \ln \rho_n \vert}{1-\epsilon}\big)}{2}
+ \frac{
	\exp\big(-\frac{\epsilon}{2}\vert \ln \rho_n \vert \big)
}{1-\rho_n}\;.
\end{align}
The last inequality follows from the same upper bound on $F(-x)$ that we used to obtain \eqref{lowerbound_g_2(1+eps)}.
\end{proof}
Lemma~\ref{lemma:location_r*(q)_bernoulli} essentially states that the global maximum of $r \mapsto \frac{rq}{2} - \frac{1}{\alpha_n} \psi_{P_{0,n}}\big(\frac{\alpha_n}{\rho_n}r\big)$ is located in a tight interval around $\nicefrac{2}{\gamma}$ when $q \in [a_{\rho_n}, b_{\rho_n}]$.
The next step is to use this knowledge to tightly bound the maximum value $\sup_{r \geq 0} \, \frac{rq}{2} - \frac{1}{\alpha_n} \psi_{P_{0,n}}\big(\frac{\alpha_n}{\rho_n}r\big)$ for all $q \in [a_{\rho_n}, b_{\rho_n}]$.
The following lemma gives a bound on $\frac{1}{\alpha_n} \psi_{P_{0,n}}\big(\frac{\alpha_n}{\rho_n}r\big)$ for $0 \leq r \leq \nicefrac{2(1 + \epsilon)}{\gamma}$.
\begin{lemma}\label{lemma:bounds_psi_P0n}
Let $P_{0,n} \coloneqq (1-\rho_n) \delta_0 + \rho_n \delta_1$ and $\alpha_n \coloneqq \gamma \rho_n \vert \ln \rho_n \vert$ for a fix $\gamma > 0$.
For every $\epsilon \in (0,1)$ and $r \in [0,\nicefrac{2(1+\epsilon)}{\gamma}]$ we have
\begin{equation}
0 \leq \frac{1}{\alpha_n} \psi_{P_{0,n}}\bigg(\frac{\alpha_n}{\rho_n}r\bigg) \leq
	\frac{\epsilon}{\gamma} + \frac{\ln 2}{\gamma  \vert \ln \rho_n \vert}
+ \frac{1}{\gamma}\sqrt{\frac{2}{ \pi \vert \ln \rho_n \vert}} \;.
\end{equation}
\end{lemma}
\begin{proof}
The function $\psi_{P_{0,n}}$ is nondecreasing on $[0,+\infty)$ so $\forall r \in [0,\nicefrac{2(1+\epsilon)}{\gamma}]:$
	\begin{equation}\label{bounds_psiP0n_r*}
	0 \leq \frac{1}{\alpha_n} \psi_{P_{0,n}}\bigg(\frac{\alpha_n}{\rho_n}r\bigg)
	\leq \frac{1}{\alpha_n} \psi_{P_{0,n}}\bigg(\frac{\alpha_n}{\rho_n}\frac{2(1 + \epsilon)}{\gamma}\bigg)
	= \frac{\psi_{P_{0,n}}\big(2(1 + \epsilon) \vert \ln \rho_n \vert\big)}{\gamma \rho_n \vert \ln \rho_n \vert}\;.
	\end{equation}
	The upper bound on the right-hand side of \eqref{bounds_psiP0n_r*} reads (remember the definition \ref{def:psi_P0n_bernoulli} of $\psi_{P_{0,n}}$):
	\begin{align}
	\frac{\psi_{P_{0,n}}\big(2(1 + \epsilon) \vert \ln \rho_n \vert\big)}{\gamma \rho_n \vert \ln \rho_n \vert}
	&=  \frac{1-\rho_n}{\gamma \rho_n \vert \ln \rho_n \vert} \E \Big[\ln\Big(1 - \rho_n + \rho_n e^{-(1 + \epsilon) \vert \ln \rho_n \vert + \sqrt{2(1 + \epsilon) \vert \ln \rho_n \vert} Z}\Big) \Big]\nonumber\\
	&\qquad+ \frac{1}{\gamma  \vert \ln \rho_n \vert} \E \Big[\ln\Big(1 - \rho_n + \rho_n e^{(1 + \epsilon) \vert \ln \rho_n \vert + \sqrt{2(1 + \epsilon) \vert \ln \rho_n \vert} Z}\Big) \Big]\nonumber\\
	&=  \frac{1-\rho_n}{\gamma \rho_n \vert \ln \rho_n \vert} \E \Big[\ln\Big(1 - \rho_n + \rho_n e^{-(1 + \epsilon) \vert \ln \rho_n \vert + \sqrt{2(1 + \epsilon) \vert \ln \rho_n \vert} Z}\Big) \Big]\nonumber\\
	&\qquad+ \frac{1}{\gamma  \vert \ln \rho_n \vert} \E \Big[\ln\Big(1 - \rho_n + e^{\epsilon \vert \ln \rho_n \vert + \sqrt{2(1 + \epsilon) \vert \ln \rho_n \vert} Z}\Big) \Big]\;.\label{upperbound_psiP0n_r*}
	\end{align}
	To control the first term on the right-hand side of \eqref{upperbound_psiP0n_r*} we use that $\ln(1 + x) \leq x$:
	\begin{align}
	&\frac{1-\rho_n}{\gamma \rho_n \vert \ln \rho_n \vert} \E \Big[\ln\Big(1 - \rho_n + \rho_n e^{-(1 + \epsilon) \vert \ln \rho_n \vert + \sqrt{2(1 + \epsilon) \vert \ln \rho_n \vert} Z}\Big) \Big]\nonumber\\
	&\qquad\qquad\qquad\qquad\qquad\qquad\qquad\qquad
	\leq
	\frac{\E \Big[e^{-(1 + \epsilon) \vert \ln \rho_n \vert + \sqrt{2(1 + \epsilon) \vert \ln \rho_n \vert} Z} - 1\Big]}{\gamma \vert \ln \rho_n \vert}\nonumber\\
	&\qquad\qquad\qquad\qquad\qquad\qquad\qquad\qquad=
	\frac{e^{-(1 + \epsilon) \vert \ln \rho_n \vert} \E \Big[e^{\sqrt{2(1 + \epsilon) \vert \ln \rho_n \vert} Z}\Big] - 1}{\gamma \vert \ln \rho_n \vert}
	= 0 \;.\label{negativity_first_term_psiP0n_r*}
	\end{align}
	To control the second term on the right-hand side of \eqref{upperbound_psiP0n_r*}, we use that:
	\begin{align*}
	\forall z \leq 0:&\ln\Big(1 - \rho_n + e^{\epsilon \vert \ln \rho_n \vert + \sqrt{2(1 + \epsilon) \vert \ln \rho_n \vert} z}\Big)
	\leq \ln(1 + e^{\epsilon \vert \ln \rho_n \vert}) \leq \ln(2e^{\epsilon \vert \ln \rho_n \vert}) \;;\\
	\forall z \geq 0:&\ln\Big(1 - \rho_n + e^{\epsilon \vert \ln \rho_n \vert + \sqrt{2(1 + \epsilon) \vert \ln \rho_n \vert} z}\Big)
	\leq \ln(2e^{\epsilon \vert \ln \rho_n \vert + \sqrt{2(1 + \epsilon) \vert \ln \rho_n \vert} z}) \;.\\
	\end{align*}
	It directly follows that:
	\begin{equation*}
	\frac{1}{\gamma  \vert \ln \rho_n \vert} \E \Big[\ln\Big(1 - \rho_n + e^{\epsilon \vert \ln \rho_n \vert + \sqrt{2(1 + \epsilon) \vert \ln \rho_n \vert} Z}\Big) \Big]
	\leq
	\frac{\epsilon}{\gamma} + \frac{\ln 2}{\gamma  \vert \ln \rho_n \vert}
	+ \frac{1}{\gamma}\sqrt{\frac{1 + \epsilon}{ \pi \vert \ln \rho_n \vert}}\;.
	\end{equation*}
The latter combined with \eqref{upperbound_psiP0n_r*} and \eqref{negativity_first_term_psiP0n_r*} ends the proof.
\end{proof}
We can now compute the limit of $I(\rho_n, \alpha_n)$ when $\rho_n \to 0$ and $\alpha_n \coloneqq \gamma \rho_n \vert \ln \rho_n \vert$.
\begin{proposition}\label{prop:limit_I(rho_n,alpha_n)_bernoulli}
Let $P_{0,n} \coloneqq (1-\rho_n) \delta_0 + \rho_n \delta_1$ and $\alpha_n \coloneqq \gamma \rho_n \vert \ln \rho_n \vert$ for a fix $\gamma > 0$.
Then the quantity $I(\rho_n, \alpha_n) \coloneqq \adjustlimits{\inf}_{q \in [0,1]} {\sup}_{r \geq 0}\; i_{\scriptstyle{\mathrm{RS}}}(q, r; \alpha_n, \rho_n)$ converges when $\rho_n \to 0^+$ and
$$
\lim_{\rho_n \to 0^+} I(\rho_n, \alpha_n) = \min\bigg\{ I_{P_{\mathrm{out}}}(0, 1), \frac{1}{\gamma}\bigg\}\;.
$$
\end{proposition}
\begin{proof}
Let $a_{\rho_n}$, $b_{\rho_n}$ the quantities defined in Lemma~\ref{lemma:location_r*(q)_bernoulli}.
By Lemmas~\ref{lemma:location_r*(q)_bernoulli} and \ref{lemma:bounds_psi_P0n} (applied with ${\epsilon = \vert \ln \rho_n \vert^{-\frac14}}$ for $\rho_n$ small enough), we have $\forall q \in [a_{\rho_n},b_{\rho_n}]$:
\begin{multline}
\frac{(1-\vert \ln \rho_n \vert^{-\frac14})q}{\gamma}
-\frac{1}{\gamma}\Bigg(	\frac{1}{\vert \ln \rho_n \vert^{\frac14}} + \frac{\ln 2}{\vert \ln \rho_n \vert}
+ \sqrt{\frac{2}{ \pi \vert \ln \rho_n \vert}} \Bigg)\\
\leq
\frac{r_n^*(q)q}{2} - \frac{1}{\alpha_n} \psi_{P_{0,n}}\bigg(\frac{\alpha_n}{\rho_n}r_n^*(q)\bigg)
\leq
\frac{(1+\vert \ln \rho_n \vert^{-\frac14})q}{\gamma}\;.
\end{multline}
Therefore, $\forall q \in [a_{\rho_n},b_{\rho_n}]$:
\begin{multline*}
 I_{P_{\mathrm{out}}}(q, 1) + \frac{q}{\gamma}
-\frac{1}{\gamma}\Bigg(	\frac{2}{\vert \ln \rho_n \vert^{\frac14}} + \frac{\ln 2}{\vert \ln \rho_n \vert}
+ \sqrt{\frac{2}{ \pi \vert \ln \rho_n \vert}} \Bigg)\\
\leq
\sup_{r \geq 0} i_{\scriptstyle{\mathrm{RS}}}(q, r; \alpha_n, \rho_n)
\leq
 I_{P_{\mathrm{out}}}(q, 1) + \frac{q}{\gamma}
  + \frac{1}{\gamma \vert \ln \rho_n \vert^{\frac14}}\;.
\end{multline*}
It directly follows that:
\begin{multline}\label{bounds_inf_sup_irs}
-\frac{1}{\gamma}\Bigg(	\frac{2}{\vert \ln \rho_n \vert^{\frac14}} + \frac{\ln 2}{\vert \ln \rho_n \vert}
+ \sqrt{\frac{2}{ \pi \vert \ln \rho_n \vert}} \Bigg)
+ \bigg\{ \inf_{ q \in [a_{\rho_n},b_{\rho_n}]}  I_{P_{\mathrm{out}}}(q, 1) + \frac{q}{\gamma} \bigg\}\\
\leq
\inf_{q \in [a_{\rho_n},b_{\rho_n}]} \sup_{r \geq 0} i_{\scriptstyle{\mathrm{RS}}}(q, r; \alpha_n, \rho_n)
\leq
\frac{1}{\gamma \vert \ln \rho_n \vert^{\frac14}}
+ \bigg\{\inf_{ q \in [a_{\rho_n},b_{\rho_n}]} I_{P_{\mathrm{out}}}(q, 1) + \frac{q}{\gamma}\bigg\}\;.
\end{multline}
Note that $q \mapsto  I_{P_{\mathrm{out}}}(q, 1) + \frac{q}{\gamma}$ is concave on $[0,1]$ so
\begin{align}
\inf_{q \in [a_{\rho_n},b_{\rho_n}]} I_{P_{\mathrm{out}}}(q, 1) + \frac{q}{\gamma}
&= \min\bigg\{ I_{P_{\mathrm{out}}}(a_{\rho_n}, 1) + \frac{a_{\rho_n}}{\gamma}, I_{P_{\mathrm{out}}}(b_{\rho_n}, 1) + \frac{b_{\rho_n}}{\gamma}\bigg\}\nonumber\\
&\xrightarrow[\rho_n \to 0]{}
\min\Big\{ I_{P_{\mathrm{out}}}(0, 1), \frac{1}{\gamma}\Big\}\;.\label{limit_inf_an_bn}
\end{align}
Combining the bounds \eqref{bounds_inf_sup_irs} on $\inf_{q \in [a_{\rho_n},b_{\rho_n}]} \sup_{r \geq 0} i_{\scriptstyle{\mathrm{RS}}}(q, r; \alpha_n, \rho_n)$ with the limit \eqref{limit_inf_an_bn} yields:
\begin{equation}\label{limit_inf_q_in_an_bn}
\lim_{\rho_n \to 0} \inf_{q \in [a_{\rho_n},b_{\rho_n}]} \sup_{r \geq 0} i_{\scriptstyle{\mathrm{RS}}}(q, r; \alpha_n, \rho_n)
= \min\Big\{ I_{P_{\mathrm{out}}}(0, 1), \frac{1}{\gamma}\Big\} \;.
\end{equation}
\paragraph{Upper bound on the limit superior of $I(\rho_n, \alpha_n)$}
The upper bound on the limit superior of $I(\rho_n, \alpha_n) \coloneqq \inf_{q \in [0,1]} {\sup}_{r \geq 0}i_{\scriptstyle{\mathrm{RS}}}(q, r; \alpha_n, \rho_n)$ directly follows from the limit \eqref{limit_inf_q_in_an_bn} and the upper bound $I(\rho_n, \alpha_n) \leq {\inf}_{q \in [a_{\rho_n},b_{\rho_n}]} {\sup}_{r \geq 0} i_{\scriptstyle{\mathrm{RS}}}(q, r; \alpha_n, \rho_n)$:
\begin{equation}\label{limsup_I(rho,alpha)}
\limsup_{\rho_n \to 0^+} I(\rho_n, \alpha_n) \leq \min\bigg\{ I_{P_{\mathrm{out}}}(0, 1), \frac{1}{\gamma}\bigg\} \;.
\end{equation}
\paragraph{Matching lower bound on the limit inferior of $I(\rho_n, \alpha_n)$}
We first rewrite $I(\rho_n, \alpha_n)$ by splitting the segment $[0,1]=[0,a_{\rho_n}] \cup [a_{\rho_n}, b_{\rho_n}] \cup [b_{\rho_n},1]$:
\begin{multline}\label{min_split_segment}
I(\rho_n, \alpha_n) = \min \bigg\{
\inf_{q \in [0,a_{\rho_n}]} {\sup}_{r \geq 0}i_{\scriptstyle{\mathrm{RS}}}(q, r; \alpha_n, \rho_n)\,;
\inf_{q \in [a_{\rho_n}, b_{\rho_n}]} {\sup}_{r \geq 0}i_{\scriptstyle{\mathrm{RS}}}(q, r; \alpha_n, \rho_n)\,;\\
\inf_{q \in [b_{\rho_n},1]} {\sup}_{r \geq 0}i_{\scriptstyle{\mathrm{RS}}}(q, r; \alpha_n, \rho_n)
\bigg\}\;.
\end{multline}
For all $q \in [0,a_{\rho_n}]$ we have:
\begin{align*}
{\sup}_{r \geq 0}i_{\scriptstyle{\mathrm{RS}}}(q, r; \alpha_n, \rho_n)
&=
I_{P_{\mathrm{out}}}(q, 1) + \sup_{r \geq 0} \bigg\{\frac{rq}{2} - \frac{1}{\alpha_n} \psi_{P_{0,n}}\bigg(\frac{\alpha_n}{\rho_n}r\bigg)\bigg\}\\
&\geq
I_{P_{\mathrm{out}}}(q, 1) + \lim_{r \to 0^+} \bigg\{\frac{rq}{2} - \frac{1}{\alpha_n} \psi_{P_{0,n}}\bigg(\frac{\alpha_n}{\rho_n}r\bigg)\bigg\}
= I_{P_{\mathrm{out}}}(q, 1)\;.
\end{align*}
As $q \mapsto I_{P_{\mathrm{out}}}(q, 1)$ is decreasing it follows that:
\begin{equation}
\adjustlimits{\inf}_{q \in [0,a_{\rho_n}]} {\sup}_{r \geq 0}i_{\scriptstyle{\mathrm{RS}}}(q, r; \alpha_n, \rho_n)
\geq \inf_{q \in [0,a_{\rho_n}]} I_{P_{\mathrm{out}}}(q, 1) = I_{P_{\mathrm{out}}}(a_{\rho_n}, 1) \;.\label{lowerbound_inf_[0,an]}
\end{equation}
For all $q \in [b_{\rho_n},1)$ we have:
\begin{align}
{\sup}_{r \geq 0}i_{\scriptstyle{\mathrm{RS}}}(q, r; \alpha_n, \rho_n)
&=
I_{P_{\mathrm{out}}}(q, 1) + \sup_{r \geq 0} \bigg\{\frac{rq}{2} - \frac{1}{\alpha_n} \psi_{P_{0,n}}\bigg(\frac{\alpha_n}{\rho_n}r\bigg)\bigg\}\nonumber\\
&\geq \frac{q(1+\vert \ln \rho_n \vert^{-\frac14})}{\gamma} - \frac{1}{\alpha_n} \psi_{P_{0,n}}\bigg(\frac{\alpha_n}{\rho_n}\frac{2(1+\vert \ln \rho_n \vert^{-\frac14})}{\gamma}\bigg)\nonumber\\
&\geq \frac{b_{\rho_n}}{\gamma} - \frac{1}{\alpha_n} \psi_{P_{0,n}}\bigg(\frac{\alpha_n}{\rho_n}\frac{2(1+\vert \ln \rho_n \vert^{-\frac14})}{\gamma}\bigg)\nonumber\\
&\geq \frac{b_{\rho_n}}{\gamma} - \frac{1}{\gamma}\Bigg(	\frac{1}{\vert \ln \rho_n \vert^{\frac14}} + \frac{\ln 2}{\vert \ln \rho_n \vert}
+ \sqrt{\frac{2}{ \pi \vert \ln \rho_n \vert}} \,\Bigg)\;.  \label{lowerbound_iRS_[bn,1]}
\end{align}
The first inequality follows from the trivial lower bounds $I_{P_{\mathrm{out}}}(q,1) \geq 0$ and
$$
\sup_{r \geq 0}\: \frac{rq}{2} - \frac{1}{\alpha_n} \psi_{P_{0,n}}\bigg(\frac{\alpha_n}{\rho_n}r\bigg)
\geq \frac{\widetilde{r}q}{2} - \frac{1}{\alpha_n} \psi_{P_{0,n}}\bigg(\frac{\alpha_n}{\rho_n}\widetilde{r}\bigg)
\quad
\text{where}
\quad
\widetilde{r} \coloneqq \frac{2(1+\vert \ln \rho_n \vert^{-\frac14})}{\gamma} \;.
$$
The last inequality follows from Lemma~\ref{lemma:bounds_psi_P0n} applied with $\epsilon = \vert \ln \rho_n \vert^{-\frac{1}{4}}$:
\begin{equation*}
\frac{1}{\alpha_n}\psi_{P_{0,n}}\bigg(\frac{\alpha_n}{\rho_n}\frac{2(1+\vert \ln \rho_n \vert^{-\frac14})}{\gamma}\bigg)
\leq \frac{1}{\gamma}\Bigg(	\frac{1}{\vert \ln \rho_n \vert^{\frac14}} + \frac{\ln 2}{\vert \ln \rho_n \vert}
+ \sqrt{\frac{2}{ \pi \vert \ln \rho_n \vert}} \Bigg)\;.
\end{equation*}
Note that the final lower bound \eqref{lowerbound_iRS_[bn,1]} does not depend on $q \in [b_{\rho_n},1)$ so the same inequality holds for the infimum of ${\sup}_{r \geq 0}i_{\scriptstyle{\mathrm{RS}}}(q, r; \alpha_n, \rho_n)$ over $q \in [b_{\rho_n},1]$.
Combining \eqref{min_split_segment}, \eqref{lowerbound_inf_[0,an]} and \eqref{lowerbound_iRS_[bn,1]} yields:
\begin{multline*}
I(\rho_n, \alpha_n) \geq
\min \bigg\{I_{P_{\mathrm{out}}}(a_{\rho_n}, 1);
\inf_{q \in [a_{\rho_n}, b_{\rho_n}]} {\sup}_{r \geq 0}i_{\scriptstyle{\mathrm{RS}}}(q, r; \alpha_n, \rho_n)\,;\\
\frac{b_{\rho_n}}{\gamma} - \frac{1}{\gamma}\Bigg(	\frac{1}{\vert \ln \rho_n \vert^{\frac14}} + \frac{\ln 2}{\vert \ln \rho_n \vert}
+ \sqrt{\frac{2}{ \pi \vert \ln \rho_n \vert}} \,\Bigg)
\bigg\}\;.
\end{multline*}
Hence we have (remember the limit \eqref{limit_inf_q_in_an_bn} and that $a_{\rho_n} \to 0$ and $b_{\rho_n} \to 1$ when $\rho_n$ vanishes):
\begin{equation}\label{liminf_I(rho,alpha)}
\liminf_{\rho_n \to 0^+} I(\rho_n, \alpha_n)
\geq \min \bigg\{I_{P_{\mathrm{out}}}(0, 1)\,;
\min\Big\{ I_{P_{\mathrm{out}}}(0, 1), \frac{1}{\gamma}\Big\} \,;
\frac{1}{\gamma}
\bigg\}
= \min\Big\{ I_{P_{\mathrm{out}}}(0, 1), \frac{1}{\gamma}\Big\} \;.
\end{equation}
We see thanks to \eqref{limsup_I(rho,alpha)} and \eqref{liminf_I(rho,alpha)} that the superior and inferior limits of $I(\rho_n,\alpha_n)$ match each other and
$\lim_{\rho_n \to 0^+} I(\rho_n, \alpha_n) = \min\big\{ I_{P_{\mathrm{out}}}(0, 1), \frac{1}{\gamma}\big\}$.
\end{proof}
Finally, we obtain Theorem~\ref{theorem:limit_MI_discrete_prior} for the specific choice $P_{0,n} \coloneqq (1-\rho_n) \delta_0 + \rho_n \delta_1$ by combining Theorem~\ref{th:RS_1layer} and Proposition~\ref{prop:limit_I(rho_n,alpha_n)_bernoulli} together:
	\begin{equation}
	\lim_{n \to +\infty} \frac{I(\bX^*;\bY \vert \boldsymbol{\Phi}) }{m_n} = \min\bigg\{I_{P_{\mathrm{out}}}(0, 1)\;;\; \frac{1}{\gamma}\bigg\}\;.
	\end{equation}
\section{Asymptotic minimum mean-square error: proof of Theorem~\ref{theorem:asymptotic_mmse}}\label{section:asymptotic_mmse}
Let $\widehat{\bX} = \widehat{\bX}(\bY,\boldsymbol{\Phi})$ be an estimator of $\bX^*$ that is a function of the observations $\bY$ and the measurement matrix $\boldsymbol{\Phi}$.
Then the mean-square error of this estimator is $\nicefrac{\E \Vert \bX^* - \widehat{\bX} \Vert^2}{k_n} \in [0, \E_{X \sim P_0} X^2]$ where the normalization factor $k_n \coloneqq n \rho_n$ is the expected sparsity of $\bX^*$.
It is well-known that the Bayes estimator $\E[\bX^* \vert \bY, \boldsymbol{\Phi}]$ achieves the minimum mean-square error (MMSE) among all estimators of the form $\widehat{\bX}(\bY,\boldsymbol{\Phi})$.
We denote the mean-square error of the Bayes estimator by
\begin{equation}
\MMSE(\bX^* \vert \bY, \boldsymbol{\Phi}) \coloneqq \frac{\E \Vert \bX^* - \E[\bX^* \vert \bY, \boldsymbol{\Phi}] \Vert^2}{k_n} \;.
\end{equation}
The MMSE is therefore a tight lower bound on the error that we achieve when estimating $\bX^*$ from the observations $\bY$ and the known measurement matrix $\boldsymbol{\Phi}$.
For this reason a result on the MMSE is easier to interprete than a result on the normalized mutual information $\nicefrac{I(\bX^*;\bY \vert \boldsymbol{\Phi}) }{m_n}$.
In this section, we prove Theorem~\ref{theorem:asymptotic_mmse}, that is, a formula for the asymptotic MMSE when $n$ diverges to infinity while $\rho_n =\Theta(n^{-\lambda})$ with $\lambda \in (0,\nicefrac{1}{9})$ and $\alpha_n = \gamma \rho_n \vert\ln\rho_n\vert$ with $\gamma>0$.
The proof of this theorem is given at the end of this section.
The proof relies on the I-MMSE relation \cite{Guo2005Mutual} that links the MMSE to the derivative of the mutual information with respect to the signal-to-noise ratio of some well-chosen observation channel.
For this reason, we first have to determine the asymptotic mutual information of a modified inference problem in which, in addition to the observations \eqref{measurements}, we have access to the side information $\widetilde{\bY}^{(\tau)} = \sqrt{\nicefrac{\alpha_n \tau}{\rho_n}}\, \bX^* + \widetilde{\bZ}$ with $\tau > 0$ and $\widetilde{\bZ}$ an additive white Gaussian noise.
Indeed, the parameter $\tau$ is akin to a signal-to-noise ratio and the derivative of the mutual information $\nicefrac{I(\bX^*;\bY, \widetilde{\bY}^{(\tau)} \vert \boldsymbol{\Phi}) }{m_n}$ with respect to $\tau$ yields half the MMSE \cite{Guo2005Mutual}:
\begin{equation*}
\frac{\partial}{\partial \tau}\bigg(\frac{I(\bX^*;\bY, \widetilde{\bY}^{(\tau)} \vert \boldsymbol{\Phi}) }{m_n}\bigg)
= \frac{\MMSE(\bX^* \vert \bY, \widetilde{\bY}^{(\tau)}, \boldsymbol{\Phi})}{2}
\xrightarrow[\tau \to 0^+]{}
\frac{\MMSE(\bX^* \vert \bY, \boldsymbol{\Phi})}{2}\;.
\end{equation*}
\subsection{Generalized linear estimation with side information}
Let $(X^*_i)_{i=1}^n \iid P_{0,n}$ be the components of the signal vector $\bX^*$.
We now have access to the observations:
\begin{align}\label{glm_with_side_information}
\begin{cases}
Y_\mu &\sim P_{\mathrm{out}}\Big( \cdot \, \Big\vert \frac{(\boldsymbol{\Phi} \bX^*)_{\mu}}{\sqrt{k_n}} \Big)\;, \quad 1\leq \mu \leq m_n\,;\\
\widetilde{Y}_i^{(\tau)} &= \sqrt{\frac{\alpha_n}{\rho_n}\tau}\, X_i^* + \widetilde{Z}_i \quad\:,\quad 1 \leq i \leq n \;;
\end{cases}
\end{align}
where $\tau \geq 0$.
Remember that the transition kernel $P_{\mathrm{out}}$ is defined in \eqref{transition-kernel} using the activation function $\varphi$ and the probability distribution $P_{A}$.
The side information induces only a small change in the \textit{(replica-symmetric) potential} whose extremization gives the asymptotic normalized mutual information. More precisely, the potential now reads:
\begin{equation}\label{def_i_RS_extended}
i_{\scriptstyle{\mathrm{RS}}}(q, r, \tau; \alpha_n, \rho_n)
\coloneqq \frac{1}{\alpha_n}I_{P_{0,n}}\bigg(\frac{\alpha_n}{\rho_n}(r+\tau)\bigg) +  I_{P_{\mathrm{out}}}\big(q, \E X^2\big)- \frac{r(\E X^2 -q)}{2}  \;,
\end{equation}
where $X \sim P_0$.
We then have the following generalization of Theorem~\ref{th:RS_1layer}.
\begin{theorem}[Mutual information of the GLM with side information at sublinear sparsity and sampling rate]\label{theorem:RS_1layer_extended}
	Suppose that $\Delta > 0$ and that the following hypotheses hold:
	\begin{enumerate}[label=(H\arabic*),noitemsep]
		\item There exists $S > 0$ such that the support of $P_0$ is included in $[-S,S]$.
		\item $\varphi$ is bounded, and its first and second partial derivatives with respect to its first argument exist, are bounded and continuous. They are denoted $\partial_{x} \varphi$, $\partial_{xx} \varphi$.
		\item $(\Phi_{\mu i}) \iid \cN(0,1)$.
	\end{enumerate}
	Let $\rho_n =\Theta(n^{-\lambda})$ with $\lambda \in [0,\nicefrac{1}{9})$ and $\alpha_n = \gamma \rho_n \vert\ln\rho_n\vert$ with $\gamma>0$.
	Then for all $n \in \N^*$:
	\begin{equation}\label{bound_mutual_info_rs_formula_extended}
	\bigg\vert \frac{I(\bX^*;\bY, \widetilde{\bY}^{(\tau)}\vert \boldsymbol{\Phi}) }{m_n} - \adjustlimits{\inf}_{q \in [0,\E_{P_{0}}[X^2]\,]} {\sup}_{r \geq 0}\; i_{\scriptstyle{\mathrm{RS}}}(q, r, \tau; \alpha_n, \rho_n)\bigg\vert
	\leq \frac{\sqrt{C}\,\vert \ln n \vert^{\nicefrac{1}{6}}}{n^{\frac{1}{12}-\frac{3\lambda}{4}}} \;,
	\end{equation} 
	where $C$ is a polynomial in $\big(\tau, S,\big\Vert \frac{\varphi}{\sqrt{\Delta}} \big\Vert_\infty, \big\Vert \frac{\partial_x \varphi}{\sqrt{\Delta}} \big\Vert_\infty, \big\Vert \frac{\partial_{xx} \varphi}{\sqrt{\Delta}} \big\Vert_\infty , \lambda, \gamma\big)$ with positive coefficients.
\end{theorem}
\begin{proof}
The proof is similar to the proof of Theorem~\ref{th:RS_1layer} except for a small change in the adaptive interpolation method due to the side information.
More precisely, at $t \in [0,1]$ we have access to the observations
\begin{align}\label{adaptive_interpolation_extended}
\begin{cases}
Y_{\mu}^{(t,\epsilon)}  &\sim \qquad P_{\mathrm{out}}(\,\cdot\, \vert \, S_{\mu}^{(t,\epsilon)})\qquad\quad\:,\;\:  1 \leq \mu \leq m_n \, ; \\
\widetilde{Y}_{i}^{(t,\epsilon, \tau)} &= \sqrt{\frac{\alpha_n}{\rho_n}\tau + R_1(t,\epsilon)}\, X^*_i + \widetilde{Z}_i\:,\; 1 \leq \, i \,  \leq \, n \;\;\; ;
\end{cases}
\end{align}
where $X_i^* \iid P_{0,n}$, $\widetilde{Z}_i \iid \cN(0,1)$ and
\begin{equation*}
S_{\mu}^{(t,\epsilon)}
\coloneqq \sqrt{\frac{1-t}{k_n}}\, \sum_{i=1}^n \Phi_{\mu i} X_i^*  + \sqrt{R_2(t,\epsilon)} \,V_{\mu} + \sqrt{t+2s_n-R_2(t,\epsilon)} \,W_{\mu}^*
\end{equation*}
with $\Phi_{\mu i}, V_{\mu}, W_{\mu}^* \iid \cN(0,1)$.
The proof then goes by looking to the interpolating mutual information $\nicefrac{I((\bX^*, \bW^*);(\bY^{(t,\epsilon)}, \widetilde{\bY}^{(t,\epsilon, \tau)}) \vert \boldsymbol{\Phi})}{m_n}$, and follows exactly the same lines than the proof of Theorem~\ref{th:RS_1layer}.
In particular, the interpolation functions $(R_1, R_2)$ are chosen a posteriori as the solutions to the same second-order ordinary differential equations than for Theorem~\ref{th:RS_1layer}.
\end{proof}
Let $X^* \sim P_{0,n} \perp Z \sim \cN(0,1)$. We define for all $r \geq 0$:
$$
\psi_{P_{0,n}}(r)
	\coloneqq \E \Big[\ln \int dP_{0,n}(x) e^{-\frac{r}{2}x^2 + rX^*x + \sqrt{r}xZ} \Big]\;.
$$
Note that $I_{P_{0,n}}(r) \coloneqq I(X^*;\sqrt{r}\,X^* + Z) = \frac{r\rho_n\E[X^2]}{2} - \psi_{P_{0,n}}(r)$ where $X \sim P_0$.
For $\rho_n, \alpha_n > 0$ and $\tau \geq 0$, we denote the variational problem appearing in Theorem~\ref{th:RS_1layer} by
\begin{align}
I(\rho_n, \alpha_n, \tau) &\coloneqq \adjustlimits{\inf}_{q \in [0,\E X^2]} {\sup}_{r \geq 0}\; i_{\scriptstyle{\mathrm{RS}}}(q, r, \tau ; \alpha_n, \rho_n)\nonumber\\
&= \inf_{q \in [0,\E X^2]} \, I_{P_{\mathrm{out}}}\big(q, \E X^2\big) + \frac{\tau \E X^2}{2} 
+ \sup_{r \geq 0} \bigg\{\!\frac{rq}{2} -\frac{1}{\alpha_n}\psi_{P_{0,n}}\bigg(\frac{\alpha_n}{\rho_n}(r+\tau)\bigg) \!\bigg\}\nonumber\\
&=
\inf_{q \in [0,\E X^2]} \, I_{P_{\mathrm{out}}}\big(q, \E X^2\big) + \frac{\tau (\E X^2 -q)}{2}
+ \sup_{r \geq \tau} \bigg\{\!\frac{rq}{2} -\frac{1}{\alpha_n}\psi_{P_{0,n}}\bigg(\frac{\alpha_n}{\rho_n}r\bigg) \!\bigg\}\,,\label{def:I(rho_n,alpha_n,tau)}
\end{align}
where $X \sim P_0$.
Similarly to what is done in Appendix~\ref{appendix:specialization_discrete_prior}, we can compute the limit of $I(\rho_n, \alpha_n, \tau)$ for a discrete distribution with finite support $P_0$.
\begin{proposition}\label{prop:limit_I(rho_n,alpha_n,tau)}
	Let $P_{0,n} \coloneqq (1-\rho_n)\delta_0 + \rho_n P_0$ where $P_0$ is a discrete distribution with finite support $\mathrm{supp}(P_0) \subseteq \{-v_K, -v_{K-1}, \dots, -v_1,v_1, v_2, \dots, v_K\}$ where $0 < v_1 < \dots < v_K < v_{K+1}=+\infty$.
	Let $\alpha_n \coloneqq \gamma \rho_n \vert \ln \rho_n \vert$ for a fix $\gamma > 0$.
	For every $\tau \in [0,\nicefrac{2}{\gamma v_K^2})$, $I(\rho_n, \alpha_n, \tau)$ defined in \eqref{def:I(rho_n,alpha_n,tau)} converges when $\rho_n \to 0^+$ and (in what follows $X \sim P_0$):
	\begin{flalign}
	&\lim_{\rho_n \to 0^+} I(\rho_n, \alpha_n, \tau)&\nonumber\\
	&\qquad\;= \min_{1 \leq k \leq K+1}\bigg\{I_{P_{\mathrm{out}}}\big(\E[X^2 \bm{1}_{\{\vert X \vert \geq v_k\}}],\E[X^2]\big) +\frac{\mathbb{P}(\vert X \vert \geq v_k)}{\gamma}
	+ \frac{\tau \E[X^2 \bm{1}_{\{\vert X \vert < v_k\}}]}{2}\bigg\}\,.\hspace{-1em}&
	\end{flalign}
\end{proposition}
\begin{proof}
Fix $\tau \in [0,\nicefrac{2}{\gamma v_K^2})$.
Define $\widetilde{I}_{P_{\mathrm{out}}}(q, \E X^2) = I_{P_{\mathrm{out}}}(q, \E X^2)  + \frac{\tau (\E X^2 -q)}{2}$.
From \eqref{def:I(rho_n,alpha_n,tau)} we have
\begin{equation}
I(\rho_n, \alpha_n, \tau) =
\inf_{q \in [0,\E X^2]} \; \widetilde{I}_{P_{\mathrm{out}}}(q, \E X^2)
+ \sup_{r \geq \tau} \bigg\{ \frac{rq}{2} -\frac{1}{\alpha_n}\psi_{P_{0,n}}\bigg(\frac{\alpha_n}{\rho_n}r\bigg) \!\bigg\}\,.\label{def:I_tilde(rho_n,alpha_n,tau)}
\end{equation}
Note that $\widetilde{I}_{P_{\mathrm{out}}}(\cdot, \E X^2)$ is concave nonincreasing on $[0, \E X^2]$ -- exactly as $I_{P_{\mathrm{out}}}(\cdot, \E X^2)$ --, and that the variational problem \eqref{def:I_tilde(rho_n,alpha_n,tau)} has a form similar to the quantity $I(\rho_n, \alpha_n)$ whose limit is given by Proposition~\ref{prop:limit_I(rho_n,alpha_n)_general} in Appendix~\ref{appendix:specialization_discrete_prior}.
The only difference that we have to take into account in the analysis is that the supremum is over $r \in [\tau, +\infty)$ instead of $r \in [0,+\infty)$.

Remember the definition \eqref{definition_a_rhon_k_and_b_rhon_k} of $a_{\rho_n}^{(K)}$.
By Lemma~\ref{lemma:location_r*(q)_general}, for every $q \in (\rho_n \E[X]^2, \E[X^2])$ there exists a unique $r_{n}^*(q) \in (0,+\infty)$ such that
\begin{equation}
\frac{r_{n}^*(q)q}{2} - \frac{1}{\alpha_n} \psi_{P_{0,n}}\bigg(\frac{\alpha_n}{\rho_n}r_{n}^*(q)\bigg)
= \sup_{r \geq 0} \: \frac{rq}{2} - \frac{1}{\alpha_n} \psi_{P_{0,n}}\bigg(\frac{\alpha_n}{\rho_n}r\bigg)\;,
\end{equation}
and $\forall q \in [a_{\rho_n}^{(K)}, \E X^2): r_{n}^*(q) \geq \nicefrac{2(1-\vert \ln \rho_n \vert^{-\frac14})}{\gamma v_K^2}$.
By assumption $\tau < \nicefrac{2}{\gamma v_K^2}$ so, for $\rho_n$ small enough, $\forall q \in [a_{\rho_n}^{(K)}, \E X^2): r_{n}^*(q) > \tau$.
It follows that $\forall q \in [a_{\rho_n}^{(K)}, \E X^2): r_n^*(q)$ satisfies
\begin{equation}\label{sup_r>tau_r_n^*}
\frac{r_{n}^*(q)q}{2} - \frac{1}{\alpha_n} \psi_{P_{0,n}}\bigg(\frac{\alpha_n}{\rho_n}r_{n}^*(q)\bigg)
= \sup_{r \geq \tau} \: \frac{rq}{2} - \frac{1}{\alpha_n} \psi_{P_{0,n}}\bigg(\frac{\alpha_n}{\rho_n}r\bigg)\;.
\end{equation}
Thanks to the identity \eqref{sup_r>tau_r_n^*} the same analysis leading to Propositions~\ref{proposition:limits_inf_a_b} and \ref{prop:limit_I(rho_n,alpha_n)_general} can be repeated, replacing $I_{P_{\mathrm{out}}}(\cdot, \E X^2)$ by $\widetilde{I}_{P_{\mathrm{out}}}(\cdot, \E X^2)$
(this makes no difference as we only need for $\widetilde{I}_{P_{\mathrm{out}}}(\cdot,\E X^2)$ to be concave nonincreasing),
in order to obtain the limit:
\begin{multline}\label{limit_inf_[ak,EX^2]_tau}
\lim_{\rho_n \to 0^+}
\adjustlimits{\inf}_{q \in [a_{\rho_n}^{(K)}, \E X^2]} {\sup}_{r \geq \tau}\;
\widetilde{I}_{P_{\mathrm{out}}}(q, \E X^2)
+ \frac{rq}{2} -\frac{1}{\alpha_n}\psi_{P_{0,n}}\bigg(\frac{\alpha_n}{\rho_n}r\bigg)\\
= \min_{1 \leq k \leq K+1}\bigg\{\widetilde{I}_{P_{\mathrm{out}}}\big(\E[X^2 \bm{1}_{\{\vert X \vert \geq v_k\}}],\E X^2\big) +\frac{\mathbb{P}(\vert X \vert \geq v_k)}{\gamma}\bigg\}\;.
\end{multline}
Note that the limit \eqref{limit_inf_[ak,EX^2]_tau} is for the infimum over $q \in [a_{\rho_n}^{(K)}, \E X^2]$, not the infimum over ${q \in [0,\E X^2]}$. This is because, for $q \in (\rho_n \E X^2,a_{\rho_n}^{(K)})$, $r_n^*(q)$ does not necessarily satisfy \eqref{sup_r>tau_r_n^*}.
However, the limit \eqref{limit_inf_[ak,EX^2]_tau} directly implies the following upper bound on the limit superior:
\begin{equation}\label{limsup_I(rho_n,alpha_n,tau)}
\limsup_{\rho_n \to 0^+} I(\rho_n, \alpha_n, \tau)
\leq \min_{1 \leq k \leq K+1}\bigg\{\widetilde{I}_{P_{\mathrm{out}}}\big(\E[X^2 \bm{1}_{\{\vert X \vert \geq v_k\}}],\E X^2\big) +\frac{\mathbb{P}(\vert X \vert \geq v_k)}{\gamma}\bigg\}\;.
\end{equation}
In order to lower bound the limit inferior, we have to lower bound the infimum over $q \in [0,a_{\rho_n}^{(K)}]$ of
$\widetilde{I}_{P_{\mathrm{out}}}(q, \E X^2)
+ \sup_{r \geq \tau} \big\{ \frac{rq}{2} -\frac{1}{\alpha_n}\psi_{P_{0,n}}\big(\nicefrac{\alpha_n r}{\rho_n}\big) \big\}$.
Because $\widetilde{I}_{P_{\mathrm{out}}}(\cdot, \E X^2)$ is nonincreasing and $q \mapsto \sup_{r \geq \tau} \big\{ \frac{rq}{2} -\frac{1}{\alpha_n}\psi_{P_{0,n}}\big(\nicefrac{\alpha_n r}{\rho_n}\big) \big\}$ is nondecreasing (it is the supremum of nondecreasing functions), we have:
\begin{align}
&\inf_{q \in [0,a_{\rho_n}^{(K)}]}\;
\widetilde{I}_{P_{\mathrm{out}}}(q, \E X^2)
+ \sup_{r \geq \tau} \bigg\{ \frac{rq}{2} -\frac{1}{\alpha_n}\psi_{P_{0,n}}\bigg(\frac{\alpha_n}{\rho_n}r\bigg) \bigg\}\nonumber\\
&\qquad\qquad\qquad\qquad\qquad\qquad\qquad\qquad
\geq \widetilde{I}_{P_{\mathrm{out}}}(a_{\rho_n}^{(K)}, \E X^2)
+ \sup_{r \geq \tau} \biggl\{-\frac{1}{\alpha_n}\psi_{P_{0,n}}\bigg(\frac{\alpha_n}{\rho_n}r\bigg) \bigg\}\nonumber\\
&\qquad\qquad\qquad\qquad\qquad\qquad\qquad\qquad
\geq \widetilde{I}_{P_{\mathrm{out}}}(a_{\rho_n}^{(K)}, \E X^2)
-\frac{1}{\alpha_n}\psi_{P_{0,n}}\bigg(\frac{\alpha_n}{\rho_n}\tau\bigg) \;.\label{lowerbound_inf_0_aK_tau}
\end{align}
The last inequality follows from $\psi_{P_{0,n}}$ being nondecreasing (see Lemma~\ref{lemma:property_I_P0n}).
We can use the computations in the proof of Lemma~\ref{lemma:limits_psi_P0n_r*} to write $\frac{1}{\alpha_n}\psi_{P_{0,n}}\big(\frac{\alpha_n \tau}{\rho_n}\big)$ more explicitly:
\begin{flalign}
&\frac{1}{\alpha_n} \psi_{P_{0,n}}\bigg(\frac{\alpha_n \tau}{\rho_n}\bigg)
= \frac{B_{\rho_n}}{\gamma} + \frac{\tau \E X^2}{2}  -\frac{1}{\gamma}
+ \frac{1}{\gamma}\sum_{j=1}^K p_j^{\scriptscriptstyle +} \E \bigg[\frac{\ln \widetilde{h}\big(Z,\gamma \tau \vert \ln \rho_n \vert,v_j;\rho_n, \bv, \bp^{\scriptscriptstyle +},\bp^{\scriptscriptstyle -}\big)}{\vert \ln \rho_n \vert}\bigg]\nonumber\\
&\qquad\qquad\qquad\qquad\qquad\qquad\qquad
+ \frac{1}{\gamma}\sum_{j=1}^K p_j^{\scriptscriptstyle -} \E \bigg[\frac{\ln \widetilde{h}\big(Z,\gamma \tau \vert \ln \rho_n \vert,v_j;\rho_n, \bv, \bp^{\scriptscriptstyle -},\bp^{\scriptscriptstyle +}\big)}{\vert \ln \rho_n \vert}\bigg]\:,\label{formula_psi_P0n_tau}
\intertext{where}
&B_{\rho_n}
= \frac{1-\rho_n}{\rho_n \vert \ln \rho_n \vert} \E \ln\bigg(1 - \rho_n + \rho_n \sum_{i=1}^K e^{-\frac{\gamma \tau}{2 v_k^2}\vert \ln \rho_n \vert} \bigg(p_i^{\scriptscriptstyle +} e^{\sqrt{\gamma \tau \vert \ln \rho_n \vert v_i^2} Z }
+ p_i^{\scriptscriptstyle -} e^{ -\sqrt{\gamma \tau \vert \ln \rho_n \vert v_i^2} Z}\bigg)\bigg)\nonumber
\intertext{and $\forall z \in \R:$}
&\widetilde{h}\big(z,\gamma \tau \vert \ln \rho_n \vert,v_j;\rho_n, \bv, \bp^{\pm},\bp^{\mp}\big)
= (1 - \rho_n)e^{\vert \ln \rho_n \vert \Big(1-\frac{\gamma \tau v_j^2}{2} -\sqrt{\frac{\gamma \tau v_j^2}{\vert \ln \rho_n \vert}} z\Big)}\nonumber\\
&\qquad\,
+\sum_{i=1}^K e^{-\vert \ln \rho_n \vert  \big(\frac{\gamma \tau(v_i-v_j)^2}{2} - \sqrt{\frac{\gamma \tau}{\vert \ln \rho_n \vert }} (v_i-v_j) z\big)}\Big(p_i^{\pm}
+ p_i^{\mp} e^{-2\vert \ln \rho_n \vert v_i\big(\gamma \tau v_j + z\sqrt{\frac{\gamma\tau}{\vert \ln \rho_n \vert}}\big)}\bigg)\,.\label{h_tilde_simplified_tau}
\end{flalign}
We can show, exactly as it is done for $A_{\rho_n}$ in the proof of Lemma~\ref{lemma:limits_psi_P0n_r*}, that $\vert B_{\rho_n} \vert \leq \nicefrac{1}{\vert \ln \rho_n \vert}$.
As $\tau < \nicefrac{2}{\gamma v_K^2}$ we have $\forall j \in \{1,\dots,K\}: 1-\nicefrac{\gamma \tau v_j^2}{2} > 0$,
and from \eqref{h_tilde_simplified_tau} we then easily deduce that $\forall j \in \{1,\dots,K\},\forall z \in \R:$
\begin{equation}\label{limit_h_tilde_tau}
\lim_{\rho_n \to 0^+} \frac{\ln \widetilde{h}\big(z,\gamma\tau \vert \ln \rho_n \vert,v_j;\rho_n, \bv, \bp^{\pm},\bp^{\mp}\big)}{\vert \ln \rho_n \vert}
= 1 - \frac{\gamma \tau v_j^2}{2}\;.
\end{equation}
By the dominated convergence theorem, making use of the pointwise limits \eqref{limit_h_tilde_tau}, we have:
\begin{multline}\label{limit_sum_E_h_tilde_tau}
\sum_{j=1}^K p_j^{\scriptscriptstyle +} \E \bigg[\frac{\ln \widetilde{h}\big(Z,\gamma\tau \vert \ln \rho_n \vert,v_j;\rho_n, \bv, \bp^{\scriptscriptstyle +},\bp^{\scriptscriptstyle -}\big)}{\vert \ln \rho_n \vert}\bigg]
+ p_j^{\scriptscriptstyle -} \E \bigg[\frac{\ln \widetilde{h}\big(Z,\gamma\tau \vert \ln \rho_n \vert,v_j;\rho_n, \bv, \bp^{\scriptscriptstyle -},\bp^{\scriptscriptstyle +}\big)}{\vert \ln \rho_n \vert}\bigg]\\
\xrightarrow[\rho_n \to 0^+]{} \sum_{j=1}^K (p_j^{\scriptscriptstyle +} 
+ p_j^{\scriptscriptstyle -} ) \bigg(1 - \frac{\gamma \tau v_j^2}{2}\bigg)
=1 - \frac{\gamma \tau \E X^2}{2}\;.
\end{multline}
Combining the identity \eqref{formula_psi_P0n_tau}, $\lim_{\rho_n \to 0^+} B_{\rho_n} = 0$ and the limit \eqref{limit_sum_E_h_tilde_tau} yields:
\begin{equation}\label{limit_psi_P0n_tau}
\lim_{\rho_n \to 0} \frac{1}{\alpha_n} \psi_{P_{0,n}}\bigg(\frac{\alpha_n}{\rho_n} \tau \bigg)
= \frac{\tau \E X^2}{2}  -\frac{1}{\gamma} + \frac{1}{\gamma}\bigg(1 - \frac{\gamma \tau \E X^2}{2}\bigg) = 0\;.
\end{equation}
The lower bound \eqref{lowerbound_inf_0_aK_tau} together with the limits \eqref{limit_psi_P0n_tau} and $\lim_{\rho_n \to 0^+} a_{\rho_n}^{(K)} = 0$ (see Lemma~\ref{lemma:location_r*(q)_general}) implies:
\begin{equation}
\liminf_{\rho_n \to 0^+} \inf_{q \in [0,a_{\rho_n}^{(K)}]}\;
\widetilde{I}_{P_{\mathrm{out}}}(q, \E X^2)
+ \sup_{r \geq \tau} \bigg\{ \frac{rq}{2} -\frac{1}{\alpha_n}\psi_{P_{0,n}}\bigg(\frac{\alpha_n}{\rho_n}r\bigg) \bigg\}
\geq \widetilde{I}_{P_{\mathrm{out}}}(0, \E X^2)\;.
\end{equation}
Finally, we combine the latter inequality with the limit \eqref{limit_inf_[ak,EX^2]_tau} to obtain
\begin{equation}\label{liminf_I(rho_n,alpha_n,tau)}
\liminf_{\rho_n \to 0^+} I(\rho_n, \alpha_n, \tau)
\geq \min_{1 \leq k \leq K+1}\bigg\{\widetilde{I}_{P_{\mathrm{out}}}\big(\E[X^2 \bm{1}_{\{\vert X \vert \geq v_k\}}],\E X^2\big) +\frac{\mathbb{P}(\vert X \vert \geq v_k)}{\gamma}\bigg\}\;.
\end{equation}
The upper bound \eqref{limsup_I(rho_n,alpha_n,tau)} on the limit superior matches the lower bound \eqref{liminf_I(rho_n,alpha_n,tau)} on the limit inferior. Hence,
\begin{flalign*}
&\lim_{\rho_n \to 0^+} I(\rho_n, \alpha_n, \tau)
= \min_{1 \leq k \leq K+1}\bigg\{\widetilde{I}_{P_{\mathrm{out}}}\big(\E[X^2 \bm{1}_{\{\vert X \vert \geq v_k\}}],\E X^2\big) +\frac{\mathbb{P}(\vert X \vert \geq v_k)}{\gamma}\bigg\}&\\
&\qquad\qquad
= \min_{1 \leq k \leq K+1}\bigg\{I_{P_{\mathrm{out}}}\big(\E[X^2 \bm{1}_{\{\vert X \vert \geq v_k\}}],\E X^2\big) + \frac{\tau \E[X^2 \bm{1}_{\{\vert X \vert < v_k\}}]}{2} +\frac{\mathbb{P}(\vert X \vert \geq v_k)}{\gamma}\bigg\}\,;
\end{flalign*}
where the last equality follows simply from the definition of $\widetilde{I}_{P_{\mathrm{out}}}$.
\end{proof}
The next theorem is a direct corollary of Theorem~\ref{theorem:RS_1layer_extended} and Proposition~\ref{prop:limit_I(rho_n,alpha_n,tau)}.
\begin{theorem}\label{theorem:limit_mutual_info_side_info}
	Suppose that $\Delta > 0$ and that $P_{0,n} \coloneqq (1-\rho_n)\delta_0 + \rho_n P_0$ where $P_0$ is a discrete distribution with finite support $\mathrm{supp}(P_0) \subseteq \{-v_K, -v_{K-1}, \dots, -v_2, -v_1,v_1, v_2, \dots, v_{K-1}, v_K\}$ where $0 < v_1 < v_2 < \dots < v_K < v_{K+1}=+\infty$.
	Further assume that the following hypotheses hold:
	\begin{enumerate}[label=(H\arabic*),noitemsep]
		\setcounter{enumi}{1}
		\item $\varphi$ is bounded, and its first and second partial derivatives with respect to its first argument exist, are bounded and continuous. They are denoted $\partial_{x} \varphi$, $\partial_{xx} \varphi$.
		\item $(\Phi_{\mu i}) \iid \cN(0,1)$.
	\end{enumerate}
	Let $\rho_n =\Theta(n^{-\lambda})$ with $\lambda \in (0,\nicefrac{1}{9})$ and $\alpha_n = \gamma \rho_n \vert\ln\rho_n\vert$ with $\gamma>0$.
	Then $\forall \tau \in [0,\nicefrac{2}{\gamma v_K^2})$:
	\begin{flalign*}
	&\lim_{n \to +\infty} \frac{I(\bX^*;\bY,\widetilde{\bY}^{(\tau)} \vert \boldsymbol{\Phi}) }{m_n}&\\
	&\qquad\qquad\;\;= \min_{1 \leq k \leq K+1}\bigg\{I_{P_{\mathrm{out}}}\big(\E[X^2 \bm{1}_{\{\vert X \vert \geq v_k\}}],\E X^2\big) + \frac{\tau \E[X^2 \bm{1}_{\{\vert X \vert < v_k\}}]}{2} +\frac{\mathbb{P}(\vert X \vert \geq v_k)}{\gamma}\bigg\}\;.
	\end{flalign*}
\end{theorem}
\subsection{Proof of Theorem~\ref{theorem:asymptotic_mmse}}
For all $n \in \mathbb{N}^*$ and $\tau \in [0,+\infty)$ we define $i_n(\tau) \coloneqq \nicefrac{I(\bX^*;\bY,\widetilde{\bY}^{(\tau)} \vert \boldsymbol{\Phi}) }{m_n}$ the normalized conditional mutual information between $\bX^*$ and the observations $\bY,\widetilde{\bY}^{(\tau)}$ -- defined in \eqref{glm_with_side_information} -- given $\boldsymbol{\Phi}$.
We place ourselves in the regime of Theorem~\ref{theorem:asymptotic_mmse}, that is, $\rho_n =\Theta(n^{-\lambda})$ with $\lambda \in [0,\nicefrac{1}{9})$ and $\alpha_n = \gamma \rho_n \vert\ln\rho_n\vert$ with $\gamma>0$.
By Theorem~\ref{theorem:limit_mutual_info_side_info} if the side-information is low enough, namely $\tau < \nicefrac{2}{\gamma v_K^2}$, then $\lim_{n \to +\infty} i_n(\tau) = i(\tau)$ where
\begin{equation}\label{def:i(tau)}
 i(\tau) \coloneqq \min_{1 \leq k \leq K+1}\bigg\{I_{P_{\mathrm{out}}}\big(\E[X^2 \bm{1}_{\{\vert X \vert \geq v_k\}}],\E X^2\big) + \frac{\tau \E[X^2 \bm{1}_{\{\vert X \vert < v_k\}}]}{2} +\frac{\mathbb{P}(\vert X \vert \geq v_k)}{\gamma}\bigg\}\;.
\end{equation}
We first establish a few properties of the function $i_n$.
The posterior density of $\bX^*$ given the observations $(\bY,\widetilde{\bY}^{(\tau)})$ defined in \eqref{glm_with_side_information} reads:
\begin{equation}\label{posterior_density_X*_side_info}
	dP\big(\bx \big\vert \bY,\widetilde{\bY}^{(\tau)}\big) = \frac{1}{\cZ(\bY,\widetilde{\bY}^{(\tau)})} \prod_{i=1}^{n} dP_{0,n}(x_i) e^{-\frac12 \big(\widetilde{Y}_i^{(\tau)} - \sqrt{\frac{\alpha_n \tau}{\rho_n}} x_i\big)^2} \prod_{\mu=1}^{m_n} P_{\mathrm{out}}\bigg(Y_\mu \bigg\vert \frac{(\boldsymbol{\Phi} \bx)_\mu}{\sqrt{k_n}}\bigg)\;,
\end{equation}
where $\cZ(\bY,\widetilde{\bY}^{(\tau)})$ is a normalization factor.
In what follows $\bx$ denotes a $n$-dimensional random vector distributed with respect to the posterior distribution \eqref{posterior_density_X*_side_info}.
We will use the brackets $\langle - \rangle_{n,\tau}$ to denote an expectation with respect to $\bx$.
By definition of the mutual information we have:
\begin{align}
i_n(\tau)
&= -\frac{1}{m_n} \E \ln \cZ(\bY,\widetilde{\bY}^{(\tau)})
+ \frac{1}{m_n}\E\bigg[ \ln \prod_{i=1}^{n} e^{-\frac12 \big(\widetilde{Y}_i^{(\tau)} - \sqrt{\frac{\alpha_n \tau}{\rho_n}} X_i^*\big)^2} \prod_{\mu=1}^{m_n} P_{\mathrm{out}}\bigg(Y_\mu \bigg\vert \frac{(\boldsymbol{\Phi} \bX^*)_\mu}{\sqrt{k_n}}\bigg)\bigg]\nonumber\\
&= -\frac{1}{m_n} \E \ln \cZ(\bY,\widetilde{\bY}^{(\tau)})
-\frac{1}{2\alpha_n} + \E\bigg[ \ln P_{\mathrm{out}}\bigg(Y_1 \bigg\vert \frac{(\boldsymbol{\Phi} \bX^*)_1}{\sqrt{k_n}}\bigg)\bigg]\;.
\end{align}
Derivation under the expectation sign, justified by the dominated convergence theorem, yields the first derivative:
\begin{align}
i'_n(\tau)
&= \frac{1}{m_n}\sum_{i=1}^{n}\E\bigg[\bigg\langle \bigg(\widetilde{Y}_i^{(\tau)} - \sqrt{\frac{\alpha_n \tau}{\rho_n \tau}}x_i\bigg)\frac12 \sqrt{\frac{\alpha_n}{\rho_n \tau}}(X_i^* - x_i) \bigg\rangle_{\!\! n,\tau}\bigg]\nonumber\\
&= \frac{1}{m_n}\sum_{i=1}^{n}\E\bigg[\bigg\langle\bigg(\widetilde{Y}_i^{(\tau)} - \sqrt{\frac{\alpha_n \tau}{\rho_n \tau}}X_i^*\bigg)\frac12 \sqrt{\frac{\alpha_n}{\rho_n \tau}}(x_i - X_i^*)\bigg\rangle_{\!\! n,\tau}\bigg]\nonumber\\
&= \frac{1}{2m_n} \sqrt{\frac{\alpha_n}{\rho_n \tau}}\E\big[\widetilde{Z}_i (\langle x_i\rangle_{n,\tau} - X_i^*)\big]\nonumber\\
&= \frac{1}{2m_n} \sqrt{\frac{\alpha_n}{\rho_n \tau}}\sum_{i=1}^{n}\E\big[\widetilde{Z}_i \langle x_i\rangle_{n,\tau}\big]\nonumber\\
&= \frac{1}{2m_n} \frac{\alpha_n}{\rho_n}\sum_{i=1}^{n}\E\big[\langle x_i^2\rangle_{n,\tau} - \langle x_i\rangle_{n,\tau}^2\big]\nonumber\\
&= \frac{\E \Vert \bX^* - \E[\bX^* \vert \bY, \bY^{(\tau)},\boldsymbol{\Phi}] \Vert^2}{2k_n}\;.\label{i_n'(tau)_IMMSErelation}
\end{align}
The second equality above follows from Nishimori identity. The fifth equality is obtained thanks to a Gaussian integration by parts with respect to $\widetilde{Z}_i$.
The final identity \eqref{i_n'(tau)_IMMSErelation} is the I-MMSE relation previously mentioned.
Further differentiating with respect to $\tau$ and integrating by parts with respect to the Gaussian random variables $\widetilde{Z}_i$ give
\begin{equation}\label{i_n''(tau)}
i_n^{\prime\prime}(\tau)
= -\frac{1}{2 k_n}\sum_{i=1}^n \E\big[\big\langle (x_i - \langle x_i \rangle_{n,\tau})^2 \big\rangle_{n,\tau}^2\big]\;.
\end{equation}
The identity \eqref{i_n''(tau)} shows that $i_n$ is concave as its second derivative is nonpositive.
By Griffiths’ lemma it follows that whenever the pointwise limit \eqref{def:i(tau)} is differentiable at $\tau \in (0,\nicefrac{2}{\gamma v_K^2})$ we have:
\begin{equation*}
\lim_{n \to + \infty} i'_n(\tau) = i'(\tau) \;.
\end{equation*}
The final step is to determine $i'(\tau)$. Suppose that the minimization problem
\begin{equation}\label{eq:proof_minimization_problem} 
\min_{1 \leq k \leq K+1}\bigg\{I_{P_{\mathrm{out}}}\big(\E[X^2 \bm{1}_{\{\vert X \vert \geq v_k\}}],\E[X^2]\big) +\frac{\mathbb{P}(\vert X \vert \geq v_k)}{\gamma}\bigg\}
\end{equation}
has a unique solution $k^* \in \{1,\dots,K+1\}$. Then, there exists $\epsilon \in [0,\nicefrac{2}{\gamma v_K^2})$ such that $\forall \tau \in [0,\epsilon): k^*$ is the unique solution to the minimization problem
\begin{equation*}
\min_{1 \leq k \leq K+1}\bigg\{I_{P_{\mathrm{out}}}\big(\E[X^2 \bm{1}_{\{\vert X \vert \geq v_k\}}],\E[X^2]\big)
+ \frac{\tau \E[X^2 \bm{1}_{\{\vert X \vert < v_k\}}]}{2} 
+ \frac{\mathbb{P}(\vert X \vert \geq v_k)}{\gamma}\bigg\}\;.
\end{equation*}
Therefore, $\forall \tau \in [0,\epsilon):$
\begin{align*}
i(\tau) &= I_{P_{\mathrm{out}}}\big(\E[X^2 \bm{1}_{\{\vert X \vert \geq v_{k^*}\}}],\E[X^2]\big)
+ \frac{\tau \E[X^2 \bm{1}_{\{\vert X \vert < v_{k^*}\}}]}{2} 
+ \frac{\mathbb{P}(\vert X \vert \geq v_{k^*})}{\gamma}\;,\\
i'(\tau) &= \frac{\E[X^2 \bm{1}_{\{\vert X \vert < v_{k^*}\}}]}{2}\;.
\end{align*}
We conclude that whenever the minimization problem \eqref{eq:proof_minimization_problem} has a unique solution $k^*$ we have
\begin{equation*}
\lim_{n \to + \infty} \frac{\E \Vert \bX^* - \E[\bX^* \vert \bY, \boldsymbol{\Phi}] \Vert^2}{k_n} 
= \lim_{n \to + \infty} 2 i'_n(0) = 2i'(0) = \E[X^2 \bm{1}_{\{\vert X \vert < v_{k^*}\}}]\;.
\end{equation*}
\subsection{All-or-nothing phenomenon and its generalization}
We now look at the asymptotic MMSE as a function of the number of measurements, i.e., as a function of the parameter $\gamma$ that controls the number of measurements $m_n = \gamma \cdot n\rho_n \vert \log \rho_n \vert$.
Let $X \sim P_0$ and assume that $\mathrm{supp} \vert X \vert = K$. We place ourselves under the assumptions of Theorem~\ref{theorem:asymptotic_mmse}.
The functions $k \mapsto I_{P_{\mathrm{out}}}\big(\E[X^2 \bm{1}_{\{\vert X \vert \geq v_k\}}],\E[X^2]\big)$ and $k \mapsto \mathbb{P}(\vert X \vert \geq v_k)$ are nondecreasing and increasing on $\{1, 2, \dots, K+1\}$, respectively.
Hence, the minimization problem on the right-hand side of \eqref{limit_minimization_problem} has a unique solution denoted $k^*(\gamma)$ for all but $K$ or less values of $\gamma \in (0,+\infty)$, and ${\gamma_1 < \gamma_2 \Rightarrow k^*(\gamma_1) \geq k^*(\gamma_2)}$ (assuming $k^*(\gamma_1), k^*(\gamma_2)$ are well-defined).
By Theorem~\ref{theorem:asymptotic_mmse}, it implies that the asymptotic MMSE as a function of $\gamma$ is nonincreasing and piecewise constant; its image is included in $\{\E X^2, \E[X^2 \bm{1}_{\{\vert X \vert \leq v_{K-1}\}}],  \dots, \E[X^2 \bm{1}_{\{\vert X \vert \leq v_{1}\}}], 0\}$.
The asymptotic MMSE has at most $K$ discontinuities. As $\gamma$ increases past a discontinuity, the asymptotic MMSE jumps from $\E[X^2 \bm{1}_{\{\vert X \vert < v_{k_1^*}\}}]$ for some $k_1^* \in \{2,\dots,K+1\}$ down to a lower value $\E[X^2 \bm{1}_{\{\vert X \vert < v_{k_2^*}\}}]$ where $k_2^* \in \{1,\dots,k_1^*-1\}$.


Therefore, when $K=1$, the asymptotic MMSE has one discontinuity at $\gamma_c \coloneqq 1/I_{P_\mathrm{out}}(0,\E X^2)$ where it jumps down from $\E X^2$ to $0$: this is the all-or-nothing phenomenon previously observed in \cite{Gamarnik2017High, Reeves2019All, Reeves2019Alla} for a linear activation function $\varphi(x) = x$ and a deterministic distribution $P_0$.
Theorem~\ref{theorem:asymptotic_mmse} generalizes this all-or-nothing phenomenon to activation functions satisfying mild conditions and any discrete distribution $P_0$ whose support is included in $\{-v,v\}$ for some $v > 0$.

When $K > 1$, the phenomenology is more complex.
The asymptotic MMSE exhibits intermerdiate plateaus in between the plateaus ``$\mathrm{MMSE} = \E X^2$'' (no reconstruction at all) for low values of $\gamma$ and ``$\mathrm{MMSE} = 0$'' (perfect reconstruction) for large values of $\gamma$.
For illustration purposes we now define the following three discrete distributions with support size $K \geq 1$:
\begin{itemize}
	\item $P_{\mathrm{unif}}^{(K)}$ is the uniform distribution on $\{\sqrt{a}, 2\sqrt{a}, \dots, K \sqrt{a}\}$ with $a := \nicefrac{6}{(K+1)(2K+1)}$ so that $\E X^2 = 1$ for $X \sim P_0$.
	\item $P_{\mathrm{linear}}^{(K)}$ is the distribution on $\{\sqrt{b}, 2\sqrt{b}, \dots, K \sqrt{b}\}$ with $b := \sum_{j=1}^K \nicefrac{1}{Kj^2}$ and $P_{\mathrm{linear}}^{(K)}(i\sqrt{b}) = \nicefrac{1}{Ki^2b}$ so that $\E X^2 = 1$ and $\E[X^2 \bm{1}_{\{\vert X \vert < k\sqrt{b}\}}] = \nicefrac{k-1}{K}$ for $X \sim P_0$, i.e., the quantity $\E[X^2 \bm{1}_{\{\vert X \vert < v_k \}}]$ increases linearly with $k$.
	\item $P_{\mathrm{binom}}^{(K,p)}$ is the binomial distribution on $\{\sqrt{c}, 2\sqrt{c}, \dots, K \sqrt{c}\}$ with
	\begin{equation*}
		c = \nicefrac{1}{(K-1)(K-2)p^2 + 3(K-1)p + 1}
	\end{equation*}
	and
	$P_{\mathrm{binom}}^{(K,p)}(i\sqrt{c}) = \binom{K-1}{i-1} p^{i-1}(1-p)^{K-i}$ so that $\E X^2 = 1$.
\end{itemize}
In Figure~\ref{figure:asymptotic_mmse_delta_gamma_K=5} we plot the asymptotic MMSE (using Theorem~\ref{theorem:asymptotic_mmse}) as a function of the noise variance $\Delta$ and the parameter $\gamma$ for three different activation functions and $P_0 \in \{P_{\mathrm{unif}}^{(5)},P_{\mathrm{linear}}^{(5)}, P_{\mathrm{binom}}^{(5,0.2)} \}$.
\begin{figure}[hbt]
	\centering
	{\includegraphics[width=\linewidth]{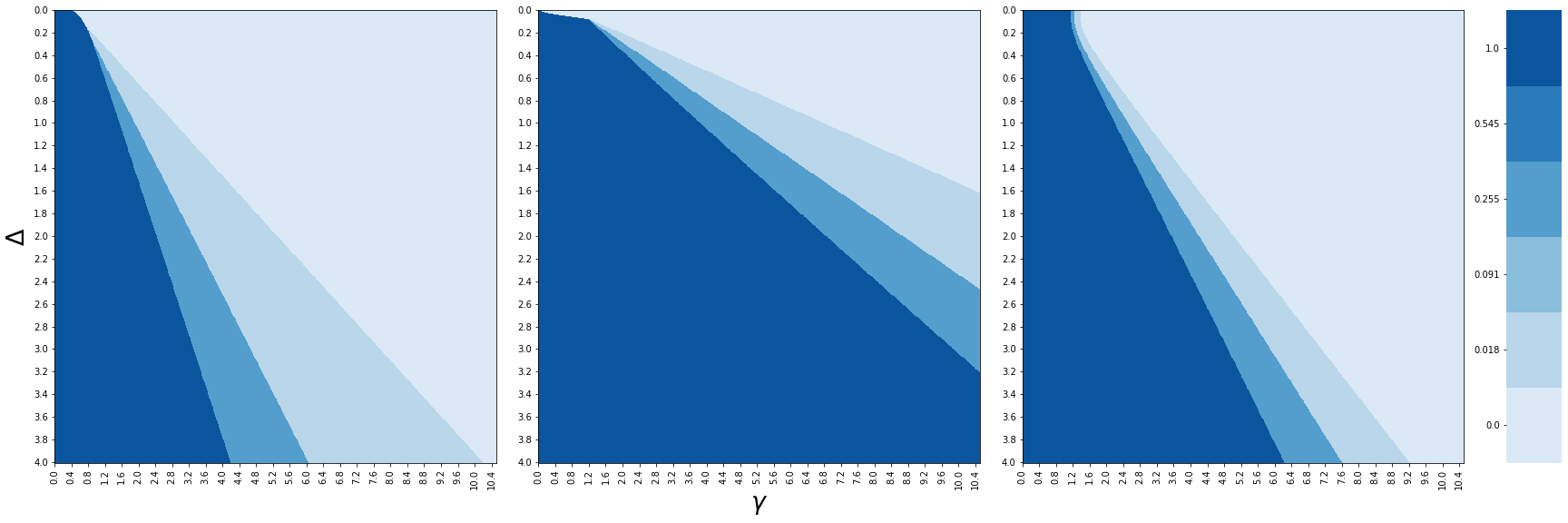}
		\includegraphics[width=\linewidth]{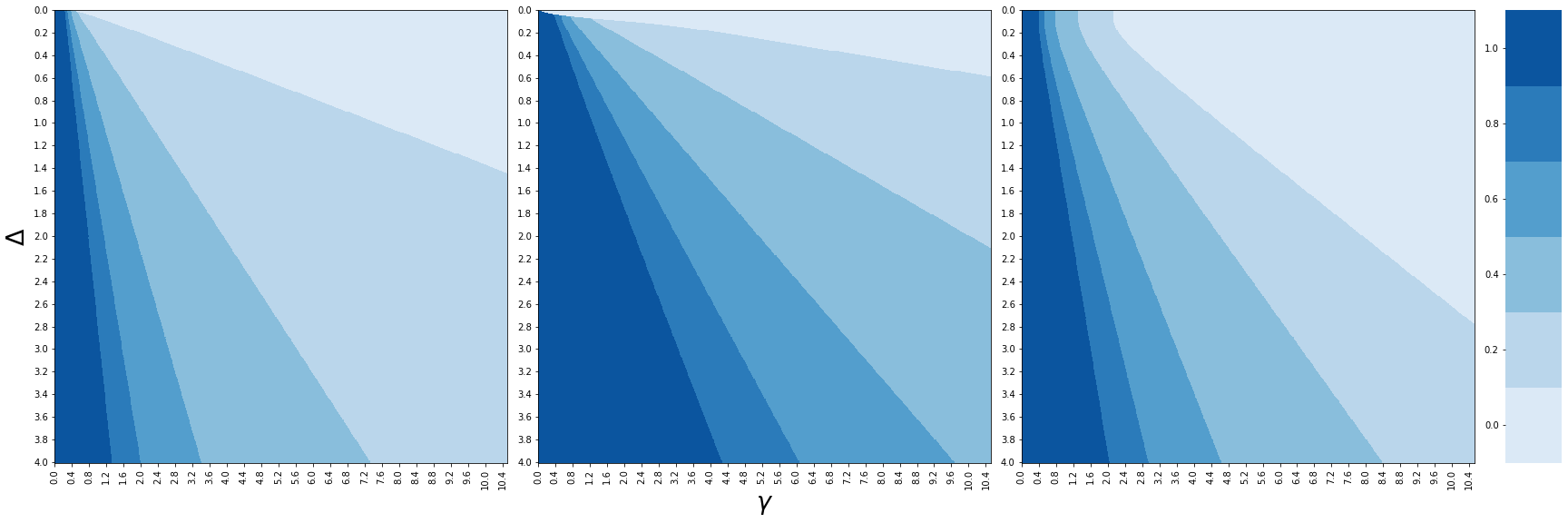}}
	\includegraphics[width=\linewidth]{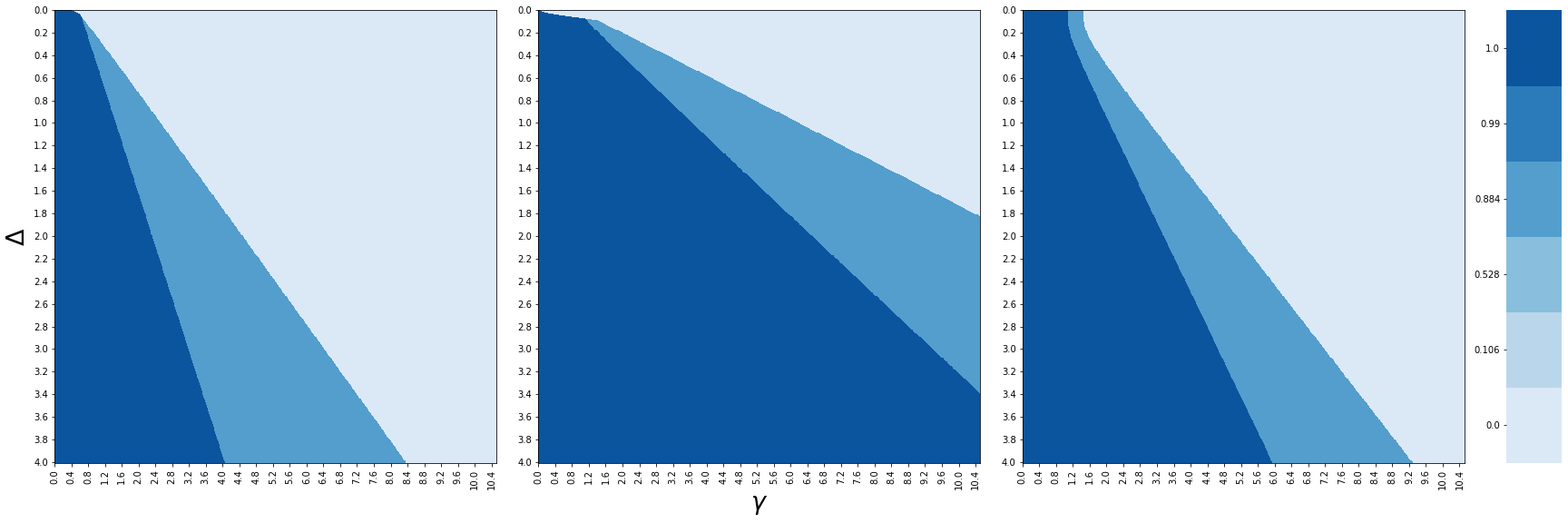}
	\caption{\label{figure:asymptotic_mmse_delta_gamma_K=5}
		\small{Minimum mean-square error in the asymptotic regime of Theorem~\ref{theorem:asymptotic_mmse} for $\Delta \in [0,4]$ and $\gamma \in (0,10.5]$.
			\textit{From left to right:} the activation function is linear $\varphi(x) = x$, the ReLU $\varphi(x)=\max(0,x)$ and the sign function $\varphi(x) = \mathrm{sign}(x)$.
			\textit{Top to bottom:} the prior distribution $P_0$ of the nonzero elements of $\bX^*$ is $P_{\mathrm{unif}}^{(5)}$, $P_{\mathrm{linear}}^{(5)}$ and $P_{\mathrm{binom}}^{(5,0.2)}$.}}
\end{figure}
\section{Properties of the mutual informations of the scalar channels}\label{appendix:scalar_channels}
This appendix gives important properties on the mutual informations of the scalar channels defined in Section~\ref{section:settings}.
We first recall the important Nishimori identity that we will use in this appendix and others as well.
\begin{lemma}[Nishimori identity] \label{lemma:nishimori}
	Let $(\bX,\bY) \in \R^{n_1} \times \R^{n_2}$ be a pair of jointly distributed random vectors.
	Let $k \geq 1$.
	Let $\bX^{(1)}, \dots, \bX^{(k)}$ be $k$ independent samples drawn from the conditional distribution $P(\bX=\cdot\, \vert \bY)$, independently of every other random variables.
	The angular brackets $\langle - \rangle$ denote the expectation operator with respect to $P(\bX= \cdot\, \vert \bY)$, while $\E$ denotes the expectation with respect to $(\bX,\bY)$.
	Then, for every integrable function $g$ the two following quantities are equal:
	\begin{align*}
	\E\,\langle g(\bY,\bX^{(1)}, \dots, \bX^{(k-1)}, \bX^{(k)}) \rangle
	&:= \E\int g(\bY,\bx^{(1)}, \dots, \bx^{(k-1)}, \bx^{(k)}) \prod_{i=1}^k dP(\bx^{(i)} \vert \bY)\;;\\
	\E\,\langle g(\bY,\bX^{(1)}, \dots,  \bX^{(k-1)}, \bX) \rangle
	&:= \E\int g(\bY,\bx^{(1)}, \dots, \bx^{(k-1)}, \bX) \prod_{i=1}^{k-1} dP(\bx^{(i)} \vert \bY)\;.
	\end{align*}
\end{lemma}
\begin{proof}
	This is a simple consequence of Bayes' formula.
	It is equivalent to sample the pair $(\bX,\bY)$ according to its joint distribution, or to first sample $\bY$ according to its marginal distribution and to then sample $\bX$ conditionally to $\bY$ from its conditional distribution $P(\bX=\cdot\,|\bY)$.
	Hence the $(k+1)$-tuple $(\bY,\bX^{(1)}, \dots,\bX^{(k)})$ is equal in law to $(\bY,\bX^{(1)},\dots,\bX^{(k-1)},\bX)$.
\end{proof}
\begin{lemma}\label{lemma:property_I_P0n}
Let $X \sim P_X$ be a real random variable with finite second moment.
Let $Z \sim \cN(0,1)$ be independent of $X$.
Define $I_{P_{X}}(r) := I(X; Y^{(r)})$ the mutual information between $X$ and ${Y^{(r)} := \sqrt{r}X + Z}$, and
$$
\psi_{P_X}(r) := \E \ln \int dP_X(x) e^{\sqrt{r}xY^{(r)} - \frac{rx^2}{2}}\;.
$$
Then, $I_{P_{X}}$ (resp.\ $\psi_{P_{X}}$) is twice continuously differentiable, nondecreasing, Lipschitz with Lipschitz constant $\nicefrac{\E[X^2]}{2}$, and concave (resp.\ convex) on $[0,+\infty)$.
Besides, if $P_X$ is not deterministic then  $I_{P_{X}}$ (resp.\ $\psi_{P_{X}}$) is strictly concave (resp.\ strictly convex).
\end{lemma}
\begin{proof}
The properties of the mutual information $I_{P_X}$ are well-known and proved in \cite{Guo2005Mutual, Guo2011Estimation}.
Note that ${\forall  r \geq 0: I_{P_{X}}(r) = \nicefrac{r\E[X^2]}{2} - \psi_{P_X}(r)}$.
The properties of $\psi_{P_X}$ follow directly from those of $I_{P_X}$ and the latter identity.
\end{proof}
\begin{lemma}\label{lemma:property_I_Pout}
Let $\Delta \in (0,+\infty)$.
Let $\varphi:\R \times \R^{k_A} \to \R$ be a bounded measurable function.
Further assume that the first and second partial derivatives of $\varphi$ with respect to its first argument, denoted $\partial_x \varphi$ and $\partial_{xx} \varphi$,  exist and are bounded.\\
Let $W^*,V,Z \sim \cN(0,1)$ and $\bA \sim P_A$ -- $P_A$ is a probability distribution over $\R^{k_A}$ -- be independent random variables.
Define $I_{P_{\mathrm{out}}}(q,\rho) := I(W^*; \widetilde{Y}^{(q,\rho)} \vert V)$ the conditional mutual information between $W^*$ and $\widetilde{Y}^{(q,\rho)} := \varphi(\sqrt{\rho-q}\,W^* + \sqrt{q}\,V,\bA) + \sqrt{\Delta}\,Z$ given $V$. Then:
\begin{itemize}
	\item $\forall \rho \in (0,+\infty)$ the function $q \mapsto I_{P_{\mathrm{out}}}(q,\rho)$ is continuously twice differentiable, concave and nonincreasing on $[0,\rho]$;
	\item
	For all $\rho \in (0,+\infty)$, the function $q \mapsto I_{P_{\mathrm{out}}}(q,\rho)$ is Lipschitz on $[0,\rho]$ with Lipschitz constant $C_1\big(\big\Vert \frac{\varphi}{\sqrt{\Delta}} \big\Vert_\infty, \big\Vert \frac{\partial_x \varphi}{\sqrt{\Delta}} \big\Vert_\infty\big)$ where:
		\begin{equation*}
			C_1(a,b) := (4a^2 + 1) b^2 \;.
		\end{equation*}
	\item For all $q \in [0,+\infty)$, the function $\rho \mapsto I_{P_{\mathrm{out}}}(q,\rho)$ is Lipschitz on $[q,+\infty)$ with Lipschitz constant $C_2\big(\big\Vert \frac{\varphi}{\sqrt{\Delta}} \big\Vert_\infty, \big\Vert \frac{\partial_x \varphi}{\sqrt{\Delta}} \big\Vert_\infty,
	\big\Vert \frac{\partial_{xx} \varphi}{\sqrt{\Delta}} \big\Vert_\infty\big)$ where:
		\begin{equation*}
			C_2(a,b,c) := b^2(128a^4 + 12a^2 + 27) + c\big(16a^3 + 4 \sqrt{\nicefrac{2}{\pi}}\,\big) \;.
		\end{equation*}
\end{itemize}
\end{lemma}
\begin{proof}
Let $P_{\mathrm{out}}(y \vert x) = \int  \frac{dP_A(\ba)}{\sqrt{2\pi \Delta}} e^{-\frac{1}{2 \Delta}(y - \varphi(x,\ba))^2}$.
The posterior density of $W^*$ given $(V,\widetilde{Y}^{(q,\rho)})$ is
\begin{equation}\label{posterior_2nd_scalar_channel}
dP(w \vert V,\widetilde{Y}^{(q,\rho)} ) := \frac{1}{\cZ_{q,\rho}(V,\widetilde{Y}^{(q,\rho)})} \frac{dw}{{\sqrt{2\pi}}}e^{-\frac{w^2}{2}} P_{\mathrm{out}}(\widetilde{Y}^{(q,\rho)} \vert \sqrt{\rho-q}\,w + \sqrt{q}\,V) \;,
\end{equation}
where $\cZ(q,\rho) := \int \frac{dw}{{\sqrt{2\pi}}}e^{-\frac{w^2}{2}}P_{\mathrm{out}}(\widetilde{Y}^{(q,\rho)} \vert \sqrt{\rho-q}\,w + \sqrt{q}\,V)$ is the normalization factor.
Then:
\begin{align}
I_{P_{\mathrm{out}}}(q,\rho)
&=  \E\big[\ln P_{\mathrm{out}}(\widetilde{Y}^{(q,\rho)} \vert \sqrt{\rho-q}\,W^* + \sqrt{q}\,V)\big] -\E \ln \cZ(q,\rho)\nonumber\\
&=  \E \ln \cZ(\rho,\rho) - \E \ln \cZ(q,\rho)\;.\label{identity_I_Pout}
\end{align}
It is shown in \cite[Appendix B.2, Proposition 18]{Barbier2019Optimal} that, for all $\rho \in (0,+\infty)$, $q \mapsto \E \ln \cZ(q,\rho)$ is continuously twice differentiable, convex and nondecreasing on $[0,\rho]$, i.e.,  $q \mapsto I_{P_{\mathrm{out}}}(q,\rho)$ is continuously twice differentiable, concave and nonincreasing on $[0,\rho]$.

We prove the second point of the lemma by upper bounding the partial derivative of $I_{P_{\mathrm{out}}}$ with respect to $q$.
The Lipschitzianity will then follow directly from the mean-value theorem.
We denote an expectation with respect to the posterior distribution \eqref{posterior_2nd_scalar_channel} using the angular brackets $\langle - \rangle_{q,\rho}$, i.e.,
$\langle g(w) \rangle_{q,\rho} := \int g(w) dP(w \vert V,\widetilde{Y}^{(q,\rho)})$.
Let $u_y(x) := \ln P_{\mathrm{out}}(y \vert x)$.
We know from \cite[Appendix B.2, Proposition 18]{Barbier2019Optimal} that $\forall \rho \in (0,+\infty), \forall q \in [0,\rho]$:
\begin{equation}\label{formula_dI_Pout(q,rho)/dq}
\frac{\partial\, I_{P_{\mathrm{out}}}}{\partial q}\Big\vert_{q,\rho}
= -\frac{\partial\, \E \ln \cZ}{\partial q}\Big\vert_{q,\rho}
= -\frac{1}{2} \E\Big[\Big\langle u_{\widetilde{Y}^{(q,\rho)}}^\prime\big(\sqrt{\rho-q}\,w + \sqrt{q}\,V\big)\Big\rangle_{\! q,\rho}^2\,\Big]\;.
\end{equation}
By Jensen's inequality and Nishimory identity, it directly follows from \eqref{formula_dI_Pout(q,rho)/dq}:
\begin{equation}\label{upperbound_dI_Pout(q,rho)/dq}
\bigg\vert \frac{\partial\, I_{P_{\mathrm{out}}}}{\partial q}\Big\vert_{q,\rho} \bigg\vert
\leq \frac{1}{2} \E\Big[\Big\langle u_{\widetilde{Y}^{(q,\rho)}}^\prime\big(\sqrt{\rho-q}\,w + \sqrt{q}\,V\big)^2\Big\rangle_{\! q,\rho}\,\Big]
= \frac{1}{2} \E\Big[u_{\widetilde{Y}^{(q,\rho)}}^\prime\big(\sqrt{\rho-q}\,W^* + \sqrt{q}\,V\big)^2\Big]\;.
\end{equation}
Remember that $\partial_x \varphi,\partial_{xx} \varphi$ denote the first and second partial derivatives of $\varphi$ with respect to its first coordinate.
The infinity norms $\Vert \varphi \Vert_\infty$ and $\Vert \partial_x \varphi \Vert_\infty$ are finite by assumptions.
Note that $\forall x \in \R$:
\begin{align}
u_{y}^\prime(x)
&= \frac{\int  \frac{y - \varphi(x,\ba)}{\Delta}\partial_x\varphi(x,\ba)\frac{dP_A(\ba)}{\sqrt{2\pi \Delta}} e^{-\frac{1}{2 \Delta}(y - \varphi(x,\ba))^2}}{\int \frac{dP_A(\ba)}{\sqrt{2\pi \Delta}} e^{-\frac{1}{2 \Delta}(y - \varphi(x,\ba))^2}}\;;\\
\vert u_{y}^\prime(x) \vert &\leq \frac{\vert y \vert + \Vert \varphi \Vert_\infty}{\Delta} \Vert \partial_x \varphi \Vert_\infty \label{upperbound_u_y(x)}
\end{align}
Then $\vert u_{\widetilde{Y}^{(q,\rho)}}^\prime(x) \vert \leq \frac{2\Vert \varphi \Vert_\infty + \sqrt{\Delta} \vert Z \vert}{\Delta} \Vert \partial_x \varphi \Vert_\infty$.
This upper bound combined with \eqref{upperbound_dI_Pout(q,rho)/dq} yields:
\begin{equation}\label{final_upperbound_dI_Pout(q,rho)/dq}
\bigg\vert \frac{\partial\, I_{P_{\mathrm{out}}}}{\partial q}\Big\vert_{q,\rho} \bigg\vert
\leq \frac{4\Vert \varphi \Vert_\infty^2 + \Delta}{\Delta^2} \Vert \partial_x \varphi \Vert_\infty^2\;,
\end{equation}
which implies the second point of the lemma thanks to the mean-value theorem.

To prove the third, and last, point of the lemma we will now upper bound the partial derivative of $I_{P_{\mathrm{out}}}$ with respect to $\rho$.
Note that
\begin{equation*}
\E \ln \cZ(q,\rho) = \E\bigg[\int dy \, e^{u_y(\sqrt{\rho-q}\,W^* + \sqrt{q}\, V)} \ln \int \frac{dw}{\sqrt{2\pi}}e^{u_y(\sqrt{\rho-q}\,w + \sqrt{q}\,V)-\frac{w^2}{2}}\bigg]\;.
\end{equation*}
Therefore:
\begin{align}
\frac{\partial\, \E \ln \cZ}{\partial \rho}\bigg\vert_{q,\rho}
&= \E\bigg[\frac{W^*}{2\sqrt{\rho-q}}\int dy \, \big(u_y^\prime(x) e^{u_y(x)}\big)\big\vert_{x = \sqrt{\rho-q}\,W^* + \sqrt{q}\,V}  \ln \int \frac{dw}{\sqrt{2\pi}}e^{u_y(\sqrt{\rho-q}\,w + \sqrt{q}\,V)-\frac{w^2}{2}}\bigg]\nonumber\\
&\qquad\qquad\qquad
+ \E\bigg[\bigg\langle\frac{w}{2\sqrt{\rho-q}}u_{\widetilde{Y}^{(q,\rho)}}^\prime(\sqrt{\rho-q}\,w + \sqrt{q}\,V)\bigg\rangle_{\!\! q,\rho}\,\bigg]\nonumber\\
&= \E\bigg[\frac{W^*}{2\sqrt{\rho-q}}\int dy \, \big(u_y^\prime(x) e^{u_y(x)}\big)\big\vert_{x = \sqrt{\rho-q}\,W^* + \sqrt{q}\,V}  \ln \int \frac{dw}{\sqrt{2\pi}}e^{u_y(\sqrt{\rho-q}\,w + \sqrt{q}\,V)-\frac{w^2}{2}}\bigg]\nonumber\\
&\qquad\qquad\qquad
+ \E\bigg[\frac{W^*}{2\sqrt{\rho-q}}u_{\widetilde{Y}^{(q,\rho)}}^\prime(\sqrt{\rho-q}\,W^* + \sqrt{q}\,V)\bigg]\nonumber\\
&= \frac{1}{2}\E\bigg[ \Big(u_{\widetilde{Y}^{(q,\rho)}}^{\prime\prime}(x) + u_{\widetilde{Y}^{(q,\rho)}}^{\prime}(x)^2\Big)\Big\vert_{x=\sqrt{\rho-q}\,W^* + \sqrt{q}\,V} \ln \cZ(q,\rho)\bigg]\nonumber\\
&\qquad\qquad\qquad
+ \frac{1}{2}\E\bigg[u_{\widetilde{Y}^{(q,\rho)}}^{\prime\prime}(\sqrt{\rho-q}\,W^* + \sqrt{q}\,V)\bigg]\nonumber\\
&= \frac{1}{2}\E\bigg[ \Big(u_{\widetilde{Y}^{(q,\rho)}}^{\prime\prime}(x) + u_{\widetilde{Y}^{(q,\rho)}}^{\prime}(x)^2\Big)\Big\vert_{x=\sqrt{\rho-q}\,W^* + \sqrt{q}\,V} (\ln \cZ(q,\rho)+1)\bigg]\nonumber\\
&\qquad\qquad\qquad
-\frac{1}{2}\E\bigg[u_{\widetilde{Y}^{(q,\rho)}}^{\prime}(\sqrt{\rho-q}\,W^* + \sqrt{q}\,V)^2\bigg] \;.\label{identity_dElnZ(q,rho)drho}
\end{align}
The second equality follows from Nishimori identity and the third one from integrating by parts with respect to $W^*$.
We now define $\forall \rho \in [0,+\infty): h(\rho) := \E\ln \cZ(\rho,\rho) = \E[\int dy \, e^{u_y(\sqrt{\rho}\,V)} u_y(\sqrt{\rho}\,V)]$.
We have:
\begin{align}
h'(\rho)
&= \E\bigg[\frac{V}{2\sqrt{\rho}}\int dy \, e^{u_y(\sqrt{\rho}\,V)} \big(u_y(\sqrt{\rho}\,V) + 1\big)u_y^\prime(\sqrt{\rho}\,V) \bigg]\nonumber\\
&= \frac{1}{2}\E\bigg[\int dy \, e^{u_y(\sqrt{\rho}\,V)} \big(u_y^{\prime\prime}(\sqrt{\rho}\,V) + u_y^\prime(\sqrt{\rho}\,V)^2\big)\big( u_y(\sqrt{\rho}\,V) + 1\big) \bigg]\nonumber\\
&\qquad\qquad\qquad
+ \frac{1}{2}\E\bigg[\int dy \, e^{u_y(\sqrt{\rho}\,V)} u_y^\prime(\sqrt{\rho}\,V)^2 \bigg]\nonumber\\
&= \frac{1}{2}\E\bigg[ \Big(u_{\widetilde{Y}^{(\rho,\rho)}}^{\prime\prime}(x) + u_{\widetilde{Y}^{(\rho,\rho)}}^{\prime}(x)^2\Big)\Big\vert_{x=\sqrt{\rho}\,V} (\ln \cZ(\rho,\rho)+1)\bigg]
+ \frac{1}{2}\E\big[u_{\widetilde{Y}^{(\rho,\rho)}}^{\prime}(\sqrt{\rho}\,V)^2\big] \;.\label{derivative_h(rho)}
\end{align}
Combining \eqref{identity_I_Pout}, \eqref{identity_dElnZ(q,rho)drho} and \eqref{derivative_h(rho)} yields
\begin{align}
\frac{\partial\, I_{P_{\mathrm{out}}}}{\partial \rho}\Big\vert_{q,\rho}
&= \frac{1}{2}\E\bigg[ \Big(u_{\widetilde{Y}^{(\rho,\rho)}}^{\prime\prime}(x) + u_{\widetilde{Y}^{(\rho,\rho)}}^{\prime}(x)^2\Big)\Big\vert_{x=\sqrt{\rho}\,V} (\ln \cZ(\rho,\rho)+1)\bigg]\nonumber\\
&\qquad\qquad
-\frac{1}{2}\E\bigg[ \Big(u_{\widetilde{Y}^{(q,\rho)}}^{\prime\prime}(x) + u_{\widetilde{Y}^{(q,\rho)}}^{\prime}(x)^2\Big)\Big\vert_{x=\sqrt{\rho-q}\,W^* + \sqrt{q}\,V} (\ln \cZ(q,\rho)+1)\bigg]\nonumber\\
&\qquad\qquad\qquad\qquad
+ \frac{1}{2}\E\big[u_{\widetilde{Y}^{(q,\rho)}}^{\prime}(\sqrt{\rho}\,V)^2\big]
+ \frac{1}{2}\E\bigg[u_{\widetilde{Y}^{(\rho,\rho)}}^{\prime}(\sqrt{\rho-q}\,W^* + \sqrt{q}\,V)^2\bigg]\;.
\label{formula_dI_Pout(q,rho)/drho}
\end{align}
The last two summands on the right-hand side of \eqref{formula_dI_Pout(q,rho)/drho} are upper bounded by $\frac{4\Vert \varphi \Vert_\infty^2 + \Delta}{\Delta^2} \Vert \partial_x \varphi \Vert_\infty^2$ (see the proof of the second point of the lemma).
The first two summands on the right-hand side of \eqref{formula_dI_Pout(q,rho)/drho} involve the function $(x,y) \mapsto u_{y}^{\prime\prime}(x) + u_{y}^{\prime}(x)^2$.
We have:
\begin{equation}\label{formula_u''+u'^2}
u_{y}^{\prime\prime}(x) + u_{y}^{\prime}(x)^2
= \frac{\int \frac{(y- \varphi(x,\ba))^2\partial_x\varphi(x,\ba)^2 
		-\Delta \partial_x\varphi(x,\ba)^2
		+ \Delta \partial_{xx}\varphi(x,\ba)(y - \varphi(x,\ba))}{\Delta^2}
\frac{dP_A(\ba)}{\sqrt{2\pi \Delta}} e^{-\frac{1}{2 \Delta}(y - \varphi(x,\ba))^2}}{\int \frac{dP_A(\ba)}{\sqrt{2\pi \Delta}} e^{-\frac{1}{2 \Delta}(y - \varphi(x,\ba))^2}}\;.
\end{equation}
Then, by a direct computation, we obtain:
\begin{align}
&\int_{-\infty}^{+\infty} (u_{y}^{\prime\prime}(x) + u_{y}^{\prime}(x)^2) e^{u_y(x)}dy\nonumber\\
&= \! \int \! dP_A(\ba) \!\int_{-\infty}^{+\infty} \!\! \frac{\big((y - \varphi(x,\ba))^2
	-\Delta \big)\partial_x\varphi(x,\ba)^2
	+ \Delta \partial_{xx}\varphi(x,\ba)(y - \varphi(x,\ba))}{\Delta^2}
\frac{e^{-\frac{(y- \varphi(x,\ba))^2}{2 \Delta}}dy}{\sqrt{2\pi \Delta}} \nonumber\\
&= \int dP_A(\ba) \int_{-\infty}^{+\infty} \frac{\big(\widetilde{y}^2
	-1\big)\partial_x\varphi(x,\ba)^2
	+ \sqrt{\Delta} \partial_{xx}\varphi(x,\ba)\widetilde{y}}{\Delta}
\frac{e^{-\frac{\widetilde{y}^2}{2}}d\widetilde{y}}{\sqrt{2\pi}}\nonumber\\
&= 0 \;.\label{int_u''+u'^2=0}
\end{align}
Therefore:
\begin{multline*}
\E\bigg[ \Big(u_{\widetilde{Y}^{(q,\rho)}}^{\prime\prime}(x) + u_{\widetilde{Y}^{(q,\rho)}}^{\prime}(x)^2\Big)\Big\vert_{x=\sqrt{\rho-q}\,W^* + \sqrt{q}\,V}\bigg]\\
= \E\bigg[ \bigg(\int_{-\infty}^{+\infty} (u_{y}^{\prime\prime}(x) + u_{y}^{\prime}(x)^2) e^{u_y(x)}dy\bigg)\bigg\vert_{x = \sqrt{\rho-q}\,W^* + \sqrt{q}\,V}\bigg]
= 0 \;.
\end{multline*}
This directly implies:
\begin{multline}\label{1st_term_dIPoutdrho_rewritten}
\E\bigg[ \Big(u_{\widetilde{Y}^{(q,\rho)}}^{\prime\prime}(x) + u_{\widetilde{Y}^{(q,\rho)}}^{\prime}(x)^2\Big)\Big\vert_{x=\sqrt{\rho-q}\,W^* + \sqrt{q}\,V} (\ln \cZ(q,\rho)+1)\bigg]\\
= \E\bigg[ \Big(u_{\widetilde{Y}^{(q,\rho)}}^{\prime\prime}(x) + u_{\widetilde{Y}^{(q,\rho)}}^{\prime}(x)^2\Big)\Big\vert_{x=\sqrt{\rho-q}\,W^* + \sqrt{q}\,V} \bigg(\ln \cZ(q,\rho)+\frac{\ln(2\pi \Delta)}{2}\bigg)\bigg] \;.
\end{multline}
We use the formula \eqref{formula_u''+u'^2} for $u_{y}^{\prime\prime}(x) + u_{y}^{\prime}(x)^2$ to get the upper bound:
\begin{multline}\label{upperbound_u''+u'^2}
\big\vert u_{\widetilde{Y}^{(q,\rho)}}^{\prime\prime}(x) + u_{\widetilde{Y}^{(q,\rho)}}^{\prime}(x)^2\big\vert
\leq \frac{\big((2 \Vert \varphi \Vert_\infty + \sqrt{\Delta} \vert Z \vert )^2 + \Delta\big) \Vert \partial_x\varphi \Vert_\infty^2
	+ \Delta \Vert \partial_{xx}\varphi \Vert_\infty (2 \Vert \varphi \Vert_\infty + \sqrt{\Delta} \vert Z \vert )}{\Delta^2}\;.
\end{multline}
Trivially, $P_{\mathrm{out}}(y \vert x) \leq 1/\sqrt{2\pi \Delta}$. This implies
\begin{equation*}
\ln \cZ(q,\rho) = \ln \int \frac{dw}{{\sqrt{2\pi}}}e^{-\frac{w^2}{2}}P_{\mathrm{out}}(\widetilde{Y}^{(q,\rho)} \vert \sqrt{\rho-q}\,w + \sqrt{q}\,V) \leq -\frac{\ln(2\pi \Delta)}{2}\;,
\end{equation*}
while, by Jensen's inequality, we have
\begin{align*}
\ln \cZ(q,\rho)
&= \ln \int \frac{dw}{{\sqrt{2\pi}}}e^{-\frac{w^2}{2}} dP_A(\ba) \frac{1}{\sqrt{2\pi \Delta}} e^{-\frac{1}{2 \Delta}(\widetilde{Y}^{(q,\rho)}  - \varphi(x,\ba))^2}\\
&\geq \int \frac{dw}{{\sqrt{2\pi}}}e^{-\frac{w^2}{2}} dP_A(\ba) \bigg(-\frac{\ln(2\pi \Delta)}{2} -\frac{(\widetilde{Y}^{(q,\rho)}  - \varphi(x,\ba))^2}{2 \Delta}\bigg)\\
&\geq -\frac{\ln(2\pi \Delta)}{2} -\frac{(2 \Vert \varphi \Vert_\infty + \sqrt{\Delta} \vert Z \vert)^2}{2 \Delta}\;.
\end{align*}
Hence
\begin{equation}\label{upperbound_free_entropy_Pout}
\bigg\vert \ln \cZ(q,\rho) + \frac{\ln(2\pi \Delta)}{2}  \bigg\vert \leq \frac{(2 \Vert \varphi \Vert_\infty + \sqrt{\Delta} \vert Z \vert)^2}{2 \Delta} \;.
\end{equation}
Combining \eqref{1st_term_dIPoutdrho_rewritten}, \eqref{upperbound_u''+u'^2}, \eqref{upperbound_free_entropy_Pout} yields the following upper bound of the second term on the right-hand side of \eqref{formula_dI_Pout(q,rho)/drho}:
\begin{multline}
\bigg\vert \frac{1}{2}\E\bigg[ \Big(u_{\widetilde{Y}^{(q,\rho)}}^{\prime\prime}(x) + u_{\widetilde{Y}^{(q,\rho)}}^{\prime}(x)^2\Big)\Big\vert_{x=\sqrt{\rho-q}\,W^* + \sqrt{q}\,V} (\ln \cZ(q,\rho)+1)\bigg]\bigg\vert\\
\leq C\bigg(\bigg\Vert \frac{\varphi}{\sqrt{\Delta}} \bigg\Vert_\infty, \bigg\Vert \frac{\partial_x \varphi}{\sqrt{\Delta}} \bigg\Vert_\infty,
\bigg\Vert \frac{\partial_{xx} \varphi}{\sqrt{\Delta}} \bigg\Vert_\infty\bigg)\;,
\end{multline}
where $C(a,b,c) := b^2(64a^4 + 6a^2 + 13.5) + c\big(8a^3 + 2 \sqrt{\frac{2}{\pi}}\,\big)$.
This upper bound holds for all $q \in [0,\rho]$.
In particular, it holds for the first term on the right-hand side of \eqref{formula_dI_Pout(q,rho)/drho} where $q=\rho$.
We now have an upper bound for each summand on the right-hand side of \eqref{formula_dI_Pout(q,rho)/drho} and we can combine them to get:
\begin{equation*}
\frac{\partial\, I_{P_{\mathrm{out}}}}{\partial \rho}\Big\vert_{q,\rho}
\leq 2C\bigg(\bigg\Vert \frac{\varphi}{\sqrt{\Delta}} \bigg\Vert_\infty, \bigg\Vert \frac{\partial_x \varphi}{\sqrt{\Delta}} \bigg\Vert_\infty,
\bigg\Vert \frac{\partial_{xx} \varphi}{\sqrt{\Delta}} \bigg\Vert_\infty\bigg)
+ 2\bigg(4\bigg\Vert \frac{\varphi}{\sqrt{\Delta}} \bigg\Vert_\infty^2 + 1 \bigg)\bigg\Vert \frac{\partial_x \varphi}{\sqrt{\Delta}} \bigg\Vert_\infty^2\;.
\end{equation*}
We can conclude the proof of the third point of the lemma using this last upper bound and the mean-value theorem.
\end{proof}
\section{Properties of the interpolating mutual information}
We recall that $u_{y}(x) := \ln P_{\mathrm{out}}(y|x)$, and that $u'_{y}(\cdot)$ and $u''_{y}(\cdot)$ are the first and second derivatives of $u_y(\cdot)$.
We denote $P_{\mathrm{out}}'(y|x)$ and $P_{\mathrm{out}}''(y|x)$ the first and second derivatives of $x\mapsto P_{\mathrm{out}}(y|x)$.
Finally, the scalar overlap is $Q := \frac{1}{k_n}\sum_{i=1}^{n} X^*_i x_i$.
\subsection{Derivative of the interpolating mutual information}\label{appendix:derivative_interpolating_mutual_information}
\begin{proposition_derivative}
Suppose that $\Delta > 0$ and that all of \ref{hyp:bounded},~\ref{hyp:c2} and \ref{hyp:phi_gauss2} hold.
Further assume that $\E_{X \sim P_0}[X^2] = 1$.
The derivative of the interpolating mutual information \eqref{interpolating_mutual_information} with respect to $t$ satisfies for all $(t,\epsilon) \in [0,1] \times \mathcal{B}_n$:
\begin{multline}\label{eq:derivative_i_n(t)}
i_{n,\epsilon}'(t)
=  \mathcal{O}\bigg(\frac{1}{\sqrt{n \rho_n}}\bigg)
+ \mathcal{O}\bigg(\sqrt{\frac{\alpha_n}{\rho_n}\Var\,\frac{\ln {\cal Z}_{t,\epsilon}}{m_n}}\bigg)
+ \frac{\rho_n}{2\alpha_n}r_\epsilon(t) (1 - q_\epsilon(t))\\
+ \frac{1}{2} \E\,\bigg\langle \big(Q - q_\epsilon(t)\big) \bigg( \frac{1}{m_n}\sum_{\mu=1}^{m_n} u_{Y_\mu^{(t,\epsilon)}}'(S_\mu^{(t,\epsilon)}) u'_{Y_\mu^{(t,\epsilon)}}(s_\mu^{(t,\epsilon)}) - \frac{\rho_n}{\alpha_n}r_\epsilon(t)\bigg)\bigg\rangle_{\!\! n,t,\epsilon}\;,
\end{multline}
where
\begin{equation*}
\bigg\vert \mathcal{O}\bigg(\frac{1}{\sqrt{n \rho_n}}\bigg) \bigg\vert
\leq \frac{S^2 C}{\sqrt{n \rho_n}}
\quad\text{and}\quad
\bigg\vert \mathcal{O}\bigg(\sqrt{\frac{\alpha_n}{\rho_n}\Var \, \frac{\ln {\cal Z}_{t,\epsilon}}{m_n}}\bigg) \bigg\vert
\leq S^2 \sqrt{D \frac{\alpha_n}{\rho_n}\Var \, \frac{\ln {\cal Z}_{t,\epsilon}}{m_n}}\;;
\end{equation*}
with ($\partial_x \varphi$ and $\partial_{xx} \varphi$ denote the first and second partial derivatives of $\varphi$ with respect to its first argument):
\begin{align*}
C &:= \bigg\Vert \frac{\partial_x \varphi}{\sqrt{\Delta}} \bigg\Vert_\infty^2\bigg(64\bigg\Vert \frac{\varphi}{\sqrt{\Delta}} \bigg\Vert_\infty^4 + 2\bigg\Vert \frac{\varphi}{\sqrt{\Delta}} \bigg\Vert_\infty^2 + 12.5\bigg)
+ \bigg\Vert \frac{\partial_{xx} \varphi}{\sqrt{\Delta}} \bigg\Vert_\infty\bigg(8\bigg\Vert \frac{\varphi}{\sqrt{\Delta}} \bigg\Vert_\infty^3 + 2 \sqrt{\frac{2}{\pi}}\,\bigg)\;;\\
D &:= \bigg\Vert \frac{\partial_x \varphi}{\sqrt{\Delta}} \bigg\Vert_\infty^4 + \frac{1}{2}\bigg\Vert \frac{\partial_{xx} \varphi}{\sqrt{\Delta}} \bigg\Vert_\infty^2\;.
\end{align*}
In addition, if both sequences $(\alpha_n)_n$ and $(\nicefrac{\rho_n}{\alpha_n})_n$ are bounded, i.e., if there exist real positive numbers $M_\alpha, M_{\rho/\alpha}$ such that $\forall n \in \N^*: \alpha_n \leq M_{\alpha}, \nicefrac{\rho_n}{\alpha_n}  \leq  M_{\rho/\alpha}$ then for all $(t,\epsilon) \in [0,1] \times \mathcal{B}_n$:
\begin{multline}\label{eq:derivative_i_n(t)_bis}
i_{n,\epsilon}'(t)
=  \mathcal{O}\bigg(\frac{1}{\sqrt{n}\,\rho_n}\bigg)
+ \frac{\rho_n}{2\alpha_n}r_\epsilon(t) (1 - q_\epsilon(t))\\
+ \frac{1}{2} \E\,\bigg\langle \big(Q - q_\epsilon(t)\big) \bigg( \frac{1}{m_n}\sum_{\mu=1}^{m_n} u_{Y_\mu^{(t,\epsilon)}}'(S_\mu^{(t,\epsilon)}) u'_{Y_\mu^{(t,\epsilon)}}(s_\mu^{(t,\epsilon)}) - \frac{\rho_n}{\alpha_n}r_\epsilon(t)\bigg)\bigg\rangle_{\!\! n,t,\epsilon}\;,
\end{multline}
where
\begin{equation*}
\bigg\vert \mathcal{O}\bigg(\frac{1}{\sqrt{n}\,\rho_n}\bigg) \bigg\vert
\leq \frac{S^2 C + S^2 \sqrt{ D \big( \widetilde{C}_1 + M_{\rho/\alpha}\widetilde{C}_2 + M_\alpha \widetilde{C}_3\big)}}{\sqrt{n}\,\rho_n}\;.
\end{equation*}
Here $\widetilde{C}_1,\widetilde{C}_2, \widetilde{C}_3$ are the polynomials in $\big(S,\big\Vert \frac{\varphi}{\sqrt{\Delta}} \big\Vert_\infty, \big\Vert \frac{\partial_x \varphi}{\sqrt{\Delta}} \big\Vert_\infty, \big\Vert \frac{\partial_{xx} \varphi}{\sqrt{\Delta}} \big\Vert_\infty\big)$ defined in Proposition~\ref{prop:concentration_free_entropy}.
\end{proposition_derivative}
%
\begin{proof}
We recall that $\cZ_{t,\epsilon}$ is the normalization to the joint posterior density of $(\bX^*,\bW^*)$ given $(\bY^{(t,\epsilon)},\widetilde{\bY}^{(t,\epsilon)},\bm{\Phi},\bV)$.
We define the average interpolating free entropy $f_{n,\epsilon}(t) := \nicefrac{\E\ln\cZ_{t,\epsilon}}{m_n}$.
Note that $i_{n,\epsilon}(t) := \nicefrac{I((\bX^*,\bW^*);(\bY^{(t,\epsilon)},\widetilde{\bY}^{(t,\epsilon)})\vert\bm{\Phi},\bV)}{m_n}$ satisfies:
\begin{align*}
i_{n,\epsilon}(t)
&= -\frac{\E\ln\cZ_{t,\epsilon}}{m_n} + \frac{1}{m_n}\E\big[\ln\big(e^{-\frac{\Vert \widetilde{\bZ} \Vert^2}{2}} P_{\mathrm{out}}(Y_{\mu}^{(t,\epsilon)} \vert S_{\mu}^{(t,\epsilon)}) \big)\big]\\
&= -f_{n,\epsilon}(t) - \frac{1}{2\alpha_n}+ \E\big[\ln P_{\mathrm{out}}(Y_{1}^{(t,\epsilon)} \vert S_{1}^{(t,\epsilon)}) \big]
\end{align*}
Given $\bX^*$, $S_{1}^{(t,\epsilon)} \sim \mathcal{N}(0,V^{(t)})$ where $\rho^{(t)} := \frac{1-t}{k_n} \Vert \bX^* \Vert^2 + t + 2s_n$. Then:
\begin{equation*}
\E\,\ln P_{\mathrm{out}}(Y_{1}^{(t,\epsilon)} \vert S_{1}^{(t,\epsilon)})
= \E\big[\E[\ln P_{\mathrm{out}}(Y_{1}^{(t,\epsilon)} \vert S_{1}^{(t,\epsilon)}) \vert \bX^* ]\big]
= \E[h(\rho^{(t)})]\;,
\end{equation*}
where $h: \rho \in [0,+\infty) \mapsto \E_{V \sim \mathcal{N}(0,1)} \int u_y(\sqrt{\rho}\,V)e^{u_y(\sqrt{\rho}\,V)}dy$.
All in all, we have:
\begin{equation}\label{link_mutual_info_free_entropy}
i_{n,\epsilon}(t)
= \E[h(\rho^{(t)})] - f_{n,\epsilon}(t) - \frac{1}{2\alpha_n} \;.
\end{equation}
We directly obtain for the derivative of $i_{n,\epsilon}(\cdot)$:
\begin{equation}\label{link_derivatives_mutual_info_free_entropy}
i'_{n,\epsilon}(t)
= -\E\bigg[h'(\rho^{(t)})\bigg(\frac{\Vert \bX^* \Vert^2}{k_n} - 1\bigg)\bigg] - f'_{n,\epsilon}(t)\;,
\end{equation}
where $h', f'_{n,\epsilon}$ are the derivatives of $h, f_{n,\epsilon}$.
In Lemma~\ref{lemma:property_I_Pout} of Appendix~\ref{appendix:scalar_channels}, we compute $h'$ and show:
\begin{equation*}
\forall \rho \in [0,+\infty): \vert h'(\rho) \vert \leq C := C\bigg(\bigg\Vert \frac{\varphi}{\sqrt{\Delta}} \bigg\Vert_\infty, \bigg\Vert \frac{\partial_x \varphi}{\sqrt{\Delta}} \bigg\Vert_\infty, \bigg\Vert \frac{\partial_{xx} \varphi}{\sqrt{\Delta}} \bigg\Vert_\infty\bigg)
\end{equation*}
with
$C(a,b,c) := b^2(64a^4 + 2a^2 + 12.5) + c\big(8a^3 + 2 \sqrt{\frac{2}{\pi}}\,\big)$. The first term on the right-hand side of \eqref{link_derivatives_mutual_info_free_entropy} thus satisfies:
\begin{equation}\label{upperbound_new_term_derivative}
\bigg\vert \E\bigg[h'(\rho^{(t)})\bigg(\frac{\Vert \bX^* \Vert^2}{k_n} - 1\bigg)\bigg] \bigg\vert
\leq C \sqrt{\Var\bigg(\frac{\Vert \bX^* \Vert^2}{k_n}\bigg)}
= \frac{C}{k_n} \sqrt{n\Var\big((X_1^*)^2\big)}
= \frac{C S^2}{\sqrt{n \rho_n}} \;.
\end{equation}
We now turn to the computation of $f'_{n,\epsilon}$.
\paragraph{Derivative of the average interpolating free entropy}
Note that
\begin{equation}\label{formula_f_n(t)_for_derivative}
f_{n,\epsilon}(t)
= \frac{1}{m_n} \E\bigg[\int \frac{d\by d\widetilde{\by}}{\sqrt{2\pi}^n}  e^{-\cH_{t,\epsilon}(\bX^*,\bW^*;\by,\widetilde{\by},\bm{\Phi},\bV)}
\ln \int dP_{0,n}(\bx) \mathcal{D}\bw \, e^{-\cH_{t,\epsilon}(\bx,\bw;\by,\widetilde{\by},\bm{\Phi},\bV)}\bigg]
\end{equation}
where the expectation is over $\bX^*, \bm{\Phi}, \bV, \bW^*$, $\mathcal{D}\bw := \frac{d\bw e^{-\frac{\Vert \bw \Vert^2}{2}}}{\sqrt{2\pi}^{m_n}}$ and the Hamiltonian ${\cal H}_{t,\epsilon}$ is:
\begin{equation}
\cH_{t,\epsilon}(\bx,\bw;\by,\widetilde{\by},\bm{\Phi},\bV)
:=-\sum_{\mu=1}^{m_n} \ln P_{\mathrm{out}} ( y_\mu  \vert s_\mu^{(t, \epsilon)})
+ \frac{1}{2} \sum_{i=1}^{n}\big(\widetilde{y}_i  - \sqrt{R_{1}(t,\epsilon)}\, x_i\big)^2 \;.
\end{equation}
We will need its derivative ${\cal H}_{t,\epsilon}'$ with respect to $t$:
\begin{equation}\label{derivative_hamiltonian}
\cH_{t,\epsilon}'(\bx,\bw;\by,\widetilde{\by},\bm{\Phi},\bV)
:=
- \sum_{\mu=1}^{m_n}
\frac{\partial s_{\mu}^{(t,\epsilon)}}{\partial t} u'_{y_\mu}( s_\mu^{(t,\epsilon)})
- \frac{r_\epsilon(t)}{2\sqrt{R_1(t,\epsilon)}} \sum_{i=1}^{n} x_i (\widetilde{y}_{i} - \sqrt{R_1(t,\epsilon)} x_i)\:.
\end{equation}
The derivative of $f_{n,\epsilon}$ can be obtained by differentiating \eqref{formula_f_n(t)_for_derivative} under the expectation:
\begin{align}
f_{n,\epsilon}'(t)
&= -\frac{1}{m_n}  \E \big[\cH_{t,\epsilon}'(\bX^*,\bW^*;\bY^{(t,\epsilon)},\widetilde{\bY}^{(t,\epsilon)},\bm{\Phi},\bV)\ln \cZ_{t,\epsilon}
	\big]\nonumber\\
&\qquad\qquad\qquad\qquad\qquad\qquad\qquad\qquad\qquad
-\frac{1}{m_n} \E \big\langle \cH_{t,\epsilon}'(\bx,\bw;\bY^{(t,\epsilon)},\widetilde{\bY}^{(t,\epsilon)},\bm{\Phi},\bV) \big\rangle_{n,t,\epsilon}\nonumber\\
&= -\frac{1}{m_n}  \E \big[\cH_{t,\epsilon}'(\bX^*,\bW^*;\bY^{(t,\epsilon)},\widetilde{\bY}^{(t,\epsilon)},\bm{\Phi},\bV)\ln \cZ_{t,\epsilon}
\big]\nonumber\\
&\qquad\qquad\qquad\qquad\qquad\qquad\qquad\qquad\qquad
-\frac{1}{m_n} \E[\cH_{t,\epsilon}'(\bX^*,\bW^*;\bY^{(t,\epsilon)},\widetilde{\bY}^{(t,\epsilon)},\bm{\Phi},\bV)] \;.\label{formula_derivative_free_entropy}
\end{align}
The last equality follows from the Nishimory identity
\begin{equation*}
\E\,\langle \cH_{t,\epsilon}'(\bx,\bw;\bY^{(t,\epsilon)},\widetilde{\bY}^{(t,\epsilon)},\bm{\Phi},\bV) \rangle_{n,t,\epsilon} = \E[\cH_{t,\epsilon}'(\bX^*,\bW^*;\bY^{(t,\epsilon)},\widetilde{\bY}^{(t,\epsilon)},\bm{\Phi},\bV)] \;.
\end{equation*}
Evaluating \eqref{derivative_hamiltonian} at $(\bx,\bw;\by,\widetilde{\by},\bm{\Phi},\bV) = (\bX^*,\bW^*;\bY^{(t,\epsilon)},\widetilde{\bY}^{(t,\epsilon)},\bm{\Phi},\bV)$ yields:
\begin{equation}\label{derivative_hamiltonian_at_ground_truth}
\cH_{t,\epsilon}'(\bX^*,\bW^*;\bY^{(t,\epsilon)},\widetilde{\bY}^{(t,\epsilon)},\bm{\Phi},\bV)
= -\!\sum_{\mu=1}^{m_n}
\frac{\partial S_{\mu}^{(t,\epsilon)}}{\partial t} u'_{Y_\mu^{(t,\epsilon)}}(S_\mu^{(t,\epsilon)})
- \frac{r_\epsilon(t)}{2\sqrt{R_1(t,\epsilon)}} \sum_{i=1}^{n} X_i^* \widetilde{Z}_{i} \,.
\end{equation}
The expectation of \eqref{derivative_hamiltonian_at_ground_truth} is zero:
\begin{align*}
\E\,\cH_{t,\epsilon}'(\bX^*,\bW^*;\bY^{(t,\epsilon)},\widetilde{\bY}^{(t,\epsilon)},\bm{\Phi},\bV)
&= - \sum_{\mu=1}^{m_n} \E\bigg[\frac{\partial S_{\mu}^{(t,\epsilon)}}{\partial t} u'_{Y_\mu^{(t,\epsilon)}}(S_\mu^{(t,\epsilon)})\bigg]\\
&= - \sum_{\mu=1}^{m_n} \E\bigg[\frac{\partial S_{\mu}^{(t,\epsilon)}}{\partial t} \E\Big[u'_{Y_\mu^{(t,\epsilon)}}(S_\mu^{(t,\epsilon)}) \Big\vert \bX^*, \bW^*, \bV, \bm{\Phi}\Big]\bigg]\\
&= - \sum_{\mu=1}^{m_n} \E\bigg[\frac{\partial S_{\mu}^{(t,\epsilon)}}{\partial t} \int u'_{y}(S_\mu^{(t,\epsilon)}) P_{\mathrm{out}}(y \,\vert\, S_\mu^{(t,\epsilon)})dy\bigg]\\
&= - \sum_{\mu=1}^{m_n} \E\bigg[\frac{\partial S_{\mu}^{(t,\epsilon)}}{\partial t} \int P'_{\mathrm{out}}(y \,\vert\, S_\mu^{(t,\epsilon)})dy\bigg]\\
&=0\;.
\end{align*}
The last equality is because for all $x$:
\begin{equation*}
\int P'_{\mathrm{out}}(y \,\vert\, x)dy = \int dP_A(\ba)\partial_x\varphi(x,\ba) \int \frac{y - \varphi(x,\ba)}{\Delta}\frac{ e^{-\frac{(y - \varphi(x,\ba))^2}{2 \Delta}} }{\sqrt{2\pi \Delta}}dy = 0 \;.
\end{equation*}
The expectation of \eqref{derivative_hamiltonian_at_ground_truth} being zero, the identity \eqref{formula_derivative_free_entropy} reads:
\begin{equation}\label{second_formula_derivative_free_entropy}
f_{n,\epsilon}'(t)
= \frac{1}{m_n} \sum_{\mu=1}^{m_n} \E\bigg[\frac{\partial S_{\mu}^{(t,\epsilon)}}{\partial t} u'_{Y_\mu^{(t,\epsilon)}}(S_\mu^{(t,\epsilon)})\ln \cZ_{t,\epsilon}\bigg]\\
+\frac{1}{m_n}\frac{r_\epsilon(t)}{2\sqrt{R_1(t,\epsilon)}} \sum_{i=1}^{n} \E\big[X_i^* \widetilde{Z}_i\ln \cZ_{t,\epsilon}\big] \;.
\end{equation}
First, we compute the first kind of expectation on the right-hand side of \eqref{second_formula_derivative_free_entropy}. $\forall \mu \in \{1,\dots,m_n\}$:
\begin{multline}
\E\bigg[\frac{\partial S_{\mu}^{(t,\epsilon)}}{\partial t} u'_{Y_\mu^{(t,\epsilon)}}(S_\mu^{(t,\epsilon)})\ln \cZ_{t,\epsilon}\bigg]\\
=\frac{1}{2}
	\E \Big[\bigg(
			- \frac{(\bm{\Phi} \bX^*)_{\mu}}{\sqrt{k_n (1-t)}}
			+ \frac{q_\epsilon(t)V_{\mu}}{ \sqrt{R_2(t,\epsilon)}} 
			+ \frac{(1 - q_\epsilon(t)) W^*_{\mu}}{\sqrt{t + 2s_n - R_2(t,\epsilon) }}
	\bigg) u'_{Y_\mu^{(t,\epsilon)}}(S_\mu^{(t,\epsilon)}) \ln \cZ_{t,\epsilon}\bigg]\,. \label{1st_expectation_derivative_f_n(t)}
\end{multline}
An integration by parts w.r.t.\ the independent standard Gaussians $(\Phi_{\mu i})_{i=1}^n$ yields:
\begin{align} 
&\E\bigg[\frac{(\bm{\Phi} \bX^*)_{\mu}}{\sqrt{k_n(1-t)}}
			u'_{Y_\mu^{(t,\epsilon)}}(S_\mu^{(t,\epsilon)})
			\ln \cZ_{t,\epsilon}\bigg] \nonumber\\
&=\sum_{i=1}^n\E \bigg[ \frac{\Phi_{\mu i} X_i^*}{\sqrt{k_n(1-t)}} \!\int\!\! d\by d\widetilde{\by} \,
			u_{y_\mu}' (S_\mu^{(t,\epsilon)}) e^{-\cH_{t,\epsilon}(\bX^*,\bW^*;\by,\widetilde{\by},\bm{\Phi},\bV)}
			\ln\! \int\!\! dP_{0,n}(\bx) \mathcal{D}\bw \, e^{-\cH_{t,\epsilon}(\bx,\bw;\by,\widetilde{\by},\bm{\Phi},\bV)} \bigg]\nonumber\\
&= \sum_{i=1}^{n} 
	\E\bigg[\frac{(X_i^*)^2}{k_n}\big(u_{Y_\mu^{(t,\epsilon)}}''(S_\mu^{(t,\epsilon)}) + u_{Y_\mu^{(t,\epsilon)}}' (S_\mu^{(t,\epsilon)})^2\big) \ln \cZ_{t,\epsilon}			
		+ \frac{X_i^* u'_{Y_\mu^{(t,\epsilon)}}(S_\mu^{(t,\epsilon)})}{k_n}
	\big\langle x_i u'_{Y_\mu^{(t,\epsilon)}} (s_\mu^{(t,\epsilon)})\big\rangle_{n,t,\epsilon}\bigg]\nonumber\\
&= \E\bigg[ \frac{\Vert \bX^* \Vert^2}{k_n}
			\frac{P_{\mathrm{out}}''(Y_\mu^{(t,\epsilon)} | S_\mu^{(t,\epsilon)})}{P_{\mathrm{out}}(Y_\mu^{(t,\epsilon)} \vert S_\mu^{(t,\epsilon)})}
			\ln \cZ_{t,\epsilon}\bigg]
+ \E\,\Big\langle Q \, u_{Y_\mu^{(t,\epsilon)}}'(S_\mu^{(t,\epsilon)}) u'_{Y_\mu^{(t,\epsilon)}}(s_\mu^{(t,\epsilon)}) \Big\rangle_{n,t,\epsilon}\;,\label{eq:compA1}
\end{align}
where, in the last equality, we used the identity $u_{y}'' ( x ) + u_{y}' ( x )^2 = \frac{P_{\mathrm{out}}''(y | x)}{P_{\mathrm{out}}(y | x)}$.
Another Gaussian integration by parts, this time with respect to \ $V_\mu \sim {\cal N}(0,1)$, gives:
\begin{align} 
&\E\bigg[\frac{q_\epsilon(t)V_{\mu}}{ \sqrt{R_2(t,\epsilon)}} u'_{Y_\mu^{(t,\epsilon)}}(S_\mu^{(t,\epsilon)})
\ln \cZ_{t,\epsilon}\bigg] \nonumber\\
&\quad
= \E \bigg[ \frac{q_\epsilon(t)V_{\mu}}{ \sqrt{R_2(t,\epsilon)}}  \int d\by d\widetilde{\by} \,
u_{y_\mu}' (S_\mu^{(t,\epsilon)}) e^{-\cH_{t,\epsilon}(\bX^*,\bW^*;\by,\widetilde{\by},\bm{\Phi},\bV)}
\ln \int dP_{0,n}(\bx) \mathcal{D}\bw \, e^{-\cH_{t,\epsilon}(\bx,\bw;\by,\widetilde{\by},\bm{\Phi},\bV)} \bigg]\nonumber\\
&\quad
= \E\bigg[q_\epsilon(t)\big(u_{Y_\mu^{(t,\epsilon)}}''(S_\mu^{(t,\epsilon)}) + u_{Y_\mu^{(t,\epsilon)}}' (S_\mu^{(t,\epsilon)})^2\big) \ln \cZ_{t,\epsilon}			
+ q_\epsilon(t) u'_{Y_\mu^{(t,\epsilon)}}(S_\mu^{(t,\epsilon)})
\big\langle u'_{Y_\mu^{(t,\epsilon)}} (s_\mu^{(t,\epsilon)})\big\rangle_{n,t,\epsilon}\bigg]\nonumber\\
&\quad
= \E\bigg[ q_\epsilon(t)
\frac{P_{\mathrm{out}}''(Y_\mu^{(t,\epsilon)} | S_\mu^{(t,\epsilon)})}{P_{\mathrm{out}}(Y_\mu^{(t,\epsilon)} \vert S_\mu^{(t,\epsilon)})}
\ln \cZ_{t,\epsilon}\bigg]
+ \E\,\Big\langle q_\epsilon(t)  u_{Y_\mu^{(t,\epsilon)}}'(S_\mu^{(t,\epsilon)}) u'_{Y_\mu^{(t,\epsilon)}}(s_\mu^{(t,\epsilon)}) \Big\rangle_{n,t,\epsilon}\;,\label{eq:compA2}
\end{align}
Finally, a Gaussian integration by part w.r.t.\ $W_\mu^* \sim {\cal N}(0,1)$ gives:
\begin{equation}\label{eq:compA3}
\E\bigg[\frac{(1 - q_\epsilon(t)) W^*_{\mu}}{\sqrt{t + 2s_n - R_2(t,\epsilon) }} u'_{Y_\mu^{(t,\epsilon)}}(S_\mu^{(t,\epsilon)})
\ln \cZ_{t,\epsilon}\bigg]
= \E\bigg[ (1 - q_\epsilon(t))
\frac{P_{\mathrm{out}}''(Y_\mu^{(t,\epsilon)} | S_\mu^{(t,\epsilon)})}{P_{\mathrm{out}}(Y_\mu^{(t,\epsilon)} \vert S_\mu^{(t,\epsilon)})}
\ln \cZ_{t,\epsilon}\bigg]\;.
\end{equation}
Plugging \eqref{eq:compA1}, \eqref{eq:compA2} and \eqref{eq:compA3} back in \eqref{1st_expectation_derivative_f_n(t)}, we obtain:
\begin{multline}\label{1st_expectation_derivative_f_n(t)_simplified}
\E\bigg[\frac{\partial S_{\mu}^{(t,\epsilon)}}{\partial t} u'_{Y_\mu^{(t,\epsilon)}}(S_\mu^{(t,\epsilon)})\ln \cZ_{t,\epsilon}\bigg]
= -\frac{1}{2}\E\bigg[
\frac{P_{\mathrm{out}}''(Y_\mu^{(t,\epsilon)} | S_\mu^{(t,\epsilon)})}{P_{\mathrm{out}}(Y_\mu^{(t,\epsilon)} \vert S_\mu^{(t,\epsilon)})}
\bigg( \frac{\Vert \bX^* \Vert^2}{k_n} - 1\bigg)\ln \cZ_{t,\epsilon}\bigg]\\
-\frac{1}{2} \E\,\Big\langle \big(Q - q_\epsilon(t)\big)  u_{Y_\mu^{(t,\epsilon)}}'(S_\mu^{(t,\epsilon)}) u'_{Y_\mu^{(t,\epsilon)}}(s_\mu^{(t,\epsilon)}) \Big\rangle_{\! n,t,\epsilon}\;.
\end{multline}
It remains to compute the first kind of expectation on the right-hand side of \eqref{second_formula_derivative_free_entropy}, i.e.,
\begin{align}
\E\big[X_i^* \widetilde{Z}_i \ln \cZ_{t,\epsilon}\big]
&= \E \Big[X_i^* \widetilde{Z}_i  \ln \int dP_{0,n}(\bx) {\cal D}\bw\, 
P_{\mathrm{out}}(Y_\mu^{(t,\epsilon)} | s_\mu^{(t,\epsilon)}) e^{-\sum_{i=1}^{n}\frac{(\sqrt{R_1(t,\epsilon)}(X^*_{i}-x_i) + \widetilde{Z}_i )^2}{2} }		
\Big]\nonumber\\
&= -\E \big[X_i^* \big\langle \sqrt{R_1(t,\epsilon)}(X_i^* - x_i) + \widetilde{Z}_i \big\rangle_{n,t,\epsilon} \big]\nonumber\\
&= -\sqrt{R_1(t,\epsilon)} \E\,\big\langle (\rho_n - X_i^* x_i) \big\rangle_{n,t,\epsilon}\;.\label{2nd_expectation_derivative_f_n(t)_simplified}
\end{align}
The second equality follows from a Gaussian integration by parts w.r.t.\ $\widetilde{Z}_i \sim \mathcal{N}(0,1)$.
Plugging the two simplified expectations \eqref{1st_expectation_derivative_f_n(t)_simplified} and \eqref{2nd_expectation_derivative_f_n(t)_simplified} back in \eqref{second_formula_derivative_free_entropy} yields:
\begin{multline}\label{final_formula_f'_n(t)}
f_{n,\epsilon}'(t)
= -\frac{\rho_n}{2\alpha_n}r_\epsilon(t) (1 - q_\epsilon(t))
-\frac{1}{2}\E\bigg[
\sum_{\mu=1}^{m_n}\frac{P_{\mathrm{out}}''(Y_\mu^{(t,\epsilon)} | S_\mu^{(t,\epsilon)})}{P_{\mathrm{out}}(Y_\mu^{(t,\epsilon)} \vert S_\mu^{(t,\epsilon)})}
\bigg( \frac{\Vert \bX^* \Vert^2}{k_n} - 1\bigg)\frac{\ln \cZ_{t,\epsilon}}{m_n} \bigg]\\
-\frac{1}{2} \E\,\bigg\langle \big(Q - q_\epsilon(t)\big) \bigg( \frac{1}{m_n}\sum_{\mu=1}^{m_n} u_{Y_\mu^{(t,\epsilon)}}'(S_\mu^{(t,\epsilon)}) u'_{Y_\mu^{(t,\epsilon)}}(s_\mu^{(t,\epsilon)}) - \frac{\rho_n}{\alpha_n}r_\epsilon(t)\bigg)\bigg\rangle_{\!\! n,t,\epsilon}\;.
\end{multline}
The last step to end the proof of the proposition is to upper bound
\begin{equation}\label{def:A_n}
A_n^{(t,\epsilon)}
:= \E\bigg[
\sum_{\mu=1}^{m_n}\frac{P_{\mathrm{out}}''(Y_\mu^{(t,\epsilon)} | S_\mu^{(t,\epsilon)})}{P_{\mathrm{out}}(Y_\mu^{(t,\epsilon)} \vert S_\mu^{(t,\epsilon)})}
\bigg( \frac{\Vert \bX^* \Vert^2}{k_n} - 1\bigg)\frac{\ln \cZ_{t,\epsilon}}{m_n} \bigg]
\end{equation}
which appears on the right-hand side of \eqref{final_formula_f'_n(t)}.
\paragraph{Upper bouding the quantity \eqref{def:A_n}}
Remember that $u_{y}'' ( x ) + u_{y}' ( x )^2 = \frac{P_{\mathrm{out}}''(y | x)}{P_{\mathrm{out}}(y | x)}$ and $P_{\mathrm{out}}(y | x) = e^{u_y(x)}$.
Therefore, $\forall x$:
\begin{equation*}
\int_{-\infty}^{+\infty} P_{\mathrm{out}}''(y|x)dy = \int_{-\infty}^{+\infty} (u_{y}^{\prime\prime}(x) + u_{y}^{\prime}(x)^2) e^{u_y(x)}dy = 0 \;,
\end{equation*}
where the second equality follows from the direct computation \eqref{int_u''+u'^2=0} in Lemma~\ref{lemma:property_I_Pout} of Appendix~\ref{appendix:scalar_channels}.
Consequently, using the tower property of the conditionnal expectation, for all $\mu \in \{1 ,\dots, m \}$:
\begin{align}
\E\bigg[
\sum_{\mu=1}^{m_n}\frac{P_{\mathrm{out}}''(Y_\mu^{(t,\epsilon)} | S_\mu^{(t,\epsilon)})}{P_{\mathrm{out}}(Y_\mu^{(t,\epsilon)} \vert S_\mu^{(t,\epsilon)})}
\bigg(\!\frac{\Vert \bX^* \Vert^2}{k_n} - 1 \!\bigg) \bigg]
&=
\E\bigg[\bigg(\!\frac{\Vert \bX^* \Vert^2}{k_n} - 1\!\bigg)
\sum_{\mu=1}^{m_n} \E\bigg[\frac{P_{\mathrm{out}}''(Y_\mu^{(t,\epsilon)} | S_\mu^{(t,\epsilon)})}{P_{\mathrm{out}}(Y_\mu^{(t,\epsilon)} \vert S_\mu^{(t,\epsilon)})} \, \bigg\vert \bX^*, \bS^{(t,\epsilon)} \bigg] \bigg]\nonumber\\
&=
\E\bigg[\bigg( \frac{\Vert \bX^* \Vert^2}{k_n} - 1\bigg)
\sum_{\mu=1}^{m_n} \int_{-\infty}^{+\infty} P_{\mathrm{out}}''(y | S_\mu^{(t,\epsilon)}) dy \bigg]
=0 \;.\label{A_n_witout_free_entropy}
\end{align}
Making use of \eqref{A_n_witout_free_entropy} and Cauchy-Schwarz inequality, we have:
\begin{align}
\vert A_n^{(t,\epsilon)} \vert
&= \bigg\vert \E\bigg[
\sum_{\mu=1}^{m_n}\frac{P_{\mathrm{out}}''(Y_\mu^{(t,\epsilon)} | S_\mu^{(t,\epsilon)})}{P_{\mathrm{out}}(Y_\mu^{(t,\epsilon)} \vert S_\mu^{(t,\epsilon)})}
\bigg( \frac{\Vert \bX^* \Vert^2}{k_n} - 1\bigg)\bigg(\frac{\ln \cZ_{t,\epsilon}}{m_n}-f_{n,\epsilon}(t)\bigg) \bigg] \bigg\vert \nonumber\\
&\leq
\E\bigg[
\bigg(\sum_{\mu=1}^{m_n}\frac{P_{\mathrm{out}}''(Y_\mu^{(t,\epsilon)} | S_\mu^{(t,\epsilon)})}{P_{\mathrm{out}}(Y_\mu^{(t,\epsilon)} \vert S_\mu^{(t,\epsilon)})}\bigg)^{\!\! 2}
\bigg( \frac{\Vert \bX^* \Vert^2}{k_n} - 1\bigg)^{\!\! 2}\,\bigg]^{\frac{1}{2}}\,
\sqrt{\Var\,\frac{\ln \cZ_{t,\epsilon}}{m_n} }\;.\label{1st_upperbound_An}
\end{align}
Using again the tower property of the conditional expectation gives:
\begin{multline}
\E\bigg[
\bigg(\sum_{\mu=1}^{m_n}\frac{P_{\mathrm{out}}''(Y_\mu^{(t,\epsilon)} | S_\mu^{(t,\epsilon)})}{P_{\mathrm{out}}(Y_\mu^{(t,\epsilon)} \vert S_\mu^{(t,\epsilon)})}\bigg)^{\!\! 2}
\bigg( \frac{\Vert \bX^* \Vert^2}{k_n} - 1\bigg)^{\!\! 2}\,\bigg]\\
= \E\Bigg[\bigg( \frac{\Vert \bX^* \Vert^2}{k_n} - 1\bigg)^{\!\! 2}
\E\bigg[
\bigg(\sum_{\mu=1}^{m_n}\frac{P_{\mathrm{out}}''(Y_\mu^{(t,\epsilon)} | S_\mu^{(t,\epsilon)})}{P_{\mathrm{out}}(Y_\mu^{(t,\epsilon)} \vert S_\mu^{(t,\epsilon)})}\bigg)^{\!\! 2}
\, \bigg\vert \, \bX^*, \bS^{(t,\epsilon)} \bigg]\Bigg]\;. \label{eq:tower_property}
\end{multline}
Note that conditionally on $\bS^{(t,\epsilon)}$ the random variables $\big(\nicefrac{P_{\mathrm{out}}''(Y_\mu^{(t,\epsilon)} | S_\mu^{(t,\epsilon)})}{P_{\mathrm{out}}(Y_\mu^{(t,\epsilon)} \vert S_\mu^{(t,\epsilon)})}
\big)_{\mu =1}^{m_n}$ are i.i.d.\ and centered.
Therefore:
\begin{align}
\E\bigg[
\bigg(\sum_{\mu=1}^{m_n}\frac{P_{\mathrm{out}}''(Y_\mu^{(t,\epsilon)} | S_\mu^{(t,\epsilon)})}{P_{\mathrm{out}}(Y_\mu^{(t,\epsilon)} \vert S_\mu^{(t,\epsilon)})}\bigg)^{\!\! 2}
\, \bigg\vert \bX^*, \bS^{(t,\epsilon)} \bigg]
&=	\E\bigg[
\bigg(\sum_{\mu=1}^{m_n}\frac{P_{\mathrm{out}}''(Y_\mu^{(t,\epsilon)} | S_\mu^{(t,\epsilon)})}{P_{\mathrm{out}}(Y_\mu^{(t,\epsilon)} \vert S_\mu^{(t,\epsilon)})}\bigg)^{\!\! 2}
\, \bigg\vert  \bS^{(t,\epsilon)} \bigg]\nonumber\\
&= m_n \E\bigg[
\bigg(\frac{P_{\mathrm{out}}''(Y_1^{(t,\epsilon)} | S_1^{(t,\epsilon)})}{P_{\mathrm{out}}(Y_1^{(t,\epsilon)} \vert S_1^{(t,\epsilon)})}\bigg)^{\!\! 2}
\, \bigg\vert  \bS^{(t,\epsilon)} \bigg]\nonumber\\
&= m_n \E\bigg[\int_{-\infty}^{+\infty} \frac{P_{\mathrm{out}}''(y| S_1^{(t,\epsilon)})^2}{P_{\mathrm{out}}(y \vert S_1^{(t,\epsilon)})} dy\bigg] \,.\label{eq:var_simp}
\end{align}
We now use the formula \eqref{formula_u''+u'^2} for $u_{y}'' ( x ) + u_{y}' ( x )^2 = \nicefrac{P_{\mathrm{out}}''(y | x)}{P_{\mathrm{out}}(y | x)}$ (obtained in Lemma~\ref{lemma:property_I_Pout} of Appendix~\ref{appendix:scalar_channels}) with Jensen's equality to show that for all $x$:
\begin{align*}
\bigg(\frac{P_{\mathrm{out}}''(y | x)}{P_{\mathrm{out}}(y | x)}\bigg)^{\!\! 2}
&\leq \frac{\int \!\big(\frac{(y- \varphi(x,\ba))^2\partial_x\varphi(x,\ba)^2 
		-\Delta \partial_x\varphi(x,\ba)^2
		+ \Delta \partial_{xx}\varphi(x,\ba)(y - \varphi(x,\ba))}{\Delta^2}\big)^2
	\frac{dP_A(\ba)}{\sqrt{2\pi \Delta}} e^{-\frac{(y - \varphi(x,\ba))^2}{2 \Delta}}}{\int \frac{dP_A(\ba)}{\sqrt{2\pi \Delta}} e^{-\frac{(y - \varphi(x,\ba))^2}{2 \Delta}}}\\
&= \frac{\int \!\big(\frac{(y- \varphi(x,\ba))^2\partial_x\varphi(x,\ba)^2 
		-\Delta \partial_x\varphi(x,\ba)^2
		+ \Delta \partial_{xx}\varphi(x,\ba)(y - \varphi(x,\ba))}{\Delta^2}\big)^2
	\frac{dP_A(\ba)}{\sqrt{2\pi \Delta}} e^{-\frac{(y - \varphi(x,\ba))^2}{2 \Delta}}}{P_{\mathrm{out}}(y | x)}\;.
\end{align*}
It follows that for all $x$:
\begin{align*}
\int_{-\infty}^{+\infty} \frac{P_{\mathrm{out}}''(y | x)^2}{P_{\mathrm{out}}(y | x)} dy
&= \int dP_A(\ba) \int_{-\infty}^{+\infty} \bigg(\frac{(u^2 - 1)\partial_x\varphi(x,\ba)^2 
		+ \sqrt{\Delta} \partial_{xx}\varphi(x,\ba)u}{\Delta}\bigg)^{\!\! 2}
	\frac{du}{\sqrt{2\pi}} e^{-\frac{u^2}{2}}\\
&\leq  4\bigg\Vert \frac{\partial_x\varphi}{\sqrt{\Delta}} \bigg\Vert_{\infty}^4
		+ 2\bigg\Vert \frac{\partial_{xx}\varphi}{\sqrt{\Delta}} \bigg\Vert_{\infty}^2 \;.
\end{align*}
Let $D := \big\Vert \frac{\partial_x\varphi}{\sqrt{\Delta}} \big\Vert_{\infty}^4
+ \frac{1}{2}\big\Vert \frac{\partial_{xx}\varphi}{\sqrt{\Delta}} \big\Vert_{\infty}^2$. Combining this last upper bound with \eqref{eq:var_simp} and \eqref{eq:tower_property} yields:
\begin{align}
\E\bigg[
\bigg(\sum_{\mu=1}^{m_n}\frac{P_{\mathrm{out}}''(Y_\mu^{(t,\epsilon)} | S_\mu^{(t,\epsilon)})}{P_{\mathrm{out}}(Y_\mu^{(t,\epsilon)} \vert S_\mu^{(t,\epsilon)})}\bigg)^{\!\! 2}
\bigg( \frac{\Vert \bX^* \Vert^2}{k_n} - 1\bigg)^{\!\! 2}\,\bigg]
\leq 4D\,m_n \Var\bigg(\frac{\Vert \bX^* \Vert^2}{k_n} \bigg)
= \frac{4D \alpha_n S^4}{\rho_n}
\label{eq:borne_ddp}
\end{align}
Going back to \eqref{1st_upperbound_An}, we have $\forall (t,\epsilon) \in [0,1] \times \mathcal{B}_n$:
\begin{equation}\label{final_upperbound_A_n}
\vert A_n^{(t,\epsilon)} \vert
\leq
2S^2
\sqrt{D \frac{\alpha_n}{\rho_n} \Var\,\frac{\ln \cZ_{t,\epsilon}}{m_n} }\;.
\end{equation}
\paragraph{Putting everything together: proofs of \eqref{eq:derivative_i_n(t)} and \eqref{eq:derivative_i_n(t)_bis}}
Combining \eqref{link_derivatives_mutual_info_free_entropy} and \eqref{final_formula_f'_n(t)} yields the following formula for the derivative of $i_{n,\epsilon}$ (remember the definition \eqref{def:A_n} of $A_n^{(t,\epsilon)}$):
\begin{multline}\label{formula_derivative_i_{n,epsilon}}
i'_{n,\epsilon}(t)
= \frac{A_n^{(t,\epsilon)}}{2}-\E\bigg[h'(\rho^{(t)})\bigg(\frac{\Vert \bX^* \Vert^2}{k_n} - 1\bigg)\bigg]
+\frac{\rho_n}{2\alpha_n}r_\epsilon(t) (1 - q_\epsilon(t))\\
+\frac{1}{2} \E\,\bigg\langle \big(Q - q_\epsilon(t)\big) \bigg( \frac{1}{m_n}\sum_{\mu=1}^{m_n} u_{Y_\mu^{(t,\epsilon)}}'(S_\mu^{(t,\epsilon)}) u'_{Y_\mu^{(t,\epsilon)}}(s_\mu^{(t,\epsilon)}) - \frac{\rho_n}{\alpha_n}r_\epsilon(t)\bigg)\bigg\rangle_{\!\! n,t,\epsilon}\;.
\end{multline}
Combining the identity \eqref{formula_derivative_i_{n,epsilon}} with the upper bounds \eqref{upperbound_new_term_derivative} and \eqref{final_upperbound_A_n} yields \eqref{eq:derivative_i_n(t)}.

It remains to prove the identity \eqref{eq:derivative_i_n(t)_bis} that holds under the additional assumption that $\forall n: \alpha_n \leq M_{\alpha}, \nicefrac{\rho_n}{\alpha_n}  \leq  M_{\rho/\alpha}$.
Combining \eqref{final_upperbound_A_n} with the upper bound \eqref{eq:concentration_free_entropy_bis} on the variance of $\Var(\nicefrac{\ln \cZ_{t,\epsilon}}{m_n})$
(see Proposition~\ref{prop:concentration_free_entropy} of Appendix~\ref{appendix:concentration_free_entropy}) gives:
\begin{equation*}
\bigg\vert \frac{A_n^{(t,\epsilon)}}{2}\bigg\vert
\leq 
\frac{S^2\sqrt{D(\widetilde{C}_1 + M_{\rho/\alpha}\widetilde{C}_2 + M_{\alpha}\widetilde{C}_3)}}{\sqrt{n}\rho_n} \;.
\end{equation*}
The constants $\widetilde{C}_1,\widetilde{C}_2,\widetilde{C}_3$ are defined in Proposition~\ref{prop:concentration_free_entropy} while $D$ has been defined earlier in the proof.
Besides, as $\rho_n \leq 1$, we have $\frac{1}{\sqrt{n \rho_n}} \leq \frac{1}{\sqrt{n} \rho_n}$ and we can loosen the upper bound \eqref{upperbound_new_term_derivative}:
$\Big\vert \E\Big[h'(\rho^{(t)})\big(\frac{\Vert \bX^* \Vert^2}{k_n} - 1\big)\Big] \Big\vert
\leq \frac{C S^2}{\sqrt{n} \rho_n}$.
Then, the term $\nicefrac{A_n^{(t,\epsilon)}}{2}-\E\big[h'(\rho^{(t)})(\nicefrac{\Vert \bX^* \Vert^2}{k_n} - 1)\big]$ on the right-hand side of \eqref{formula_derivative_i_{n,epsilon}} is in $\mathcal{O}(\nicefrac{1}{\sqrt{n}\rho_n})$
and this proves the identity \eqref{eq:derivative_i_n(t)_bis}.
\end{proof}
\subsection{Proof of Lemma~\ref{lemma:perturbation_mutual_information_t=0}}\label{appendix:perturbation_mutual_information_t=0}
\begin{proof}
At $t=0$ the functions $r_\epsilon$ and $q_\epsilon$ do not play any role in the observations \eqref{2channels} since $R_1(t,\epsilon) = \epsilon_1$ and $R_2(t,\epsilon) = \epsilon_2$.
While in the main text we restricted $\epsilon$ to be in $\mathcal{B}_n := [s_n, 2s_n]^2$, we can define observations $(\bY^{(0,\epsilon)}, \widetilde{\bY}^{(0,\epsilon)})$ using \eqref{2channels} for $t=0$ and $\epsilon \in [0, 2s_n]^2$.
We then extend the \textit{interpolating mutual information} at $t=0$ to all $\epsilon \in [0,2s_n]^2$:
\begin{equation*}
i_{n,\epsilon}(0) := \frac{1}{m_n} I\big((\bX^*,\bW^*);(\bY^{(0,\epsilon)},\widetilde{\bY}^{(0,\epsilon)})\big|\bm{\Phi},\bV\big) \;.
\end{equation*}
Note that the variation we want to control in this lemma satisfies:
\begin{equation}\label{variation_t=0_split}
\bigg\vert i_{n,\epsilon}(0) -\frac{I(\bX^*;\bY|\bm{\Phi})}{m_n} \bigg\vert
\leq \bigg\vert i_{n,\epsilon}(0) - i_{n,\epsilon=(0,0)}(0) \bigg\vert + \bigg\vert i_{n,\epsilon=(0,0)}(0) -\frac{I(\bX^*;\bY|\bm{\Phi})}{m_n} \bigg\vert\;.
\end{equation}
We will upper bound the two terms on the right-hand side of \eqref{variation_t=0_split} separately.\\
\textbf{1. }By the I-MMSE relation (see \cite{Guo2005Mutual}), we have for all $\epsilon \in [0,2s_n]^2$:
\begin{equation}\label{upperbound_partial_derivative_epsilon1}
\bigg\vert \frac{\partial i_{n,\epsilon}(0)}{\partial \epsilon_1} \bigg\vert
= \frac{1}{2\alpha_n}\E\big[\big(X_1^* - \langle x_1\rangle_{n,0,\epsilon}\big)^2\,\big]
\leq  \frac{\E[(X_1^*)^2]}{2\alpha_n} = \frac{\rho_n}{2\alpha_n}\;.
\end{equation}
To upper bound the absolute value of the partial derivative with respect to $\epsilon_2$, we use that $\epsilon \in [0,2s_n]^2$:
\begin{equation*}
\frac{\partial i_{n,\epsilon}(0)}{\partial \epsilon_2} 
= -\frac{1}{2}\E\big[u'_{Y_{1}^{(0,\epsilon)}}(S_{1}^{(0,\epsilon)})\big\langle u'_{Y_{1}^{(0,\epsilon)}}(s_{1}^{(0,\epsilon)}) \big\rangle_{n,0,\epsilon}\,\big]\;.
\end{equation*}
This identity is obtained in a similar fashion to the computation of the derivative of $i_{n,\epsilon}(\cdot)$ in Appendix~\ref{appendix:derivative_interpolating_mutual_information} (see \eqref{eq:compA2} and \eqref{eq:compA3} in particular).
Under the hypothesis \ref{hyp:c2}, we obtain in the proof of Lemma~\ref{lemma:property_I_Pout} the upper bound \eqref{upperbound_u_y(x)} on $\vert u_y'(x) \vert$ for all $x \in \R$. Making use of this upper bound yields $\forall x \in \R: \big\vert u'_{Y_{1}^{(0,\epsilon)}}(x) \big\vert
\leq (2\Vert \varphi \Vert_\infty + \vert Z_1 \vert) \Vert \partial_x \varphi \Vert_\infty$.
Therefore:
\begin{equation}\label{upperbound_partial_derivative_epsilon2}
\bigg\vert\frac{\partial i_{n,\epsilon}(0)}{\partial \epsilon_2}\bigg\vert 
\leq  \frac{1}{2}\E\big[(2\Vert \varphi \Vert_\infty + \vert Z_1 \vert)^2 \Vert \partial_x \varphi \Vert_\infty^2\big]
\leq  (4\Vert \varphi \Vert_\infty^2 + 1) \Vert \partial_x \varphi \Vert_\infty^2\;.
\end{equation}
By the mean value theorem, and the upper bounds \eqref{upperbound_partial_derivative_epsilon1} and \eqref{upperbound_partial_derivative_epsilon2}, we have:
\begin{align}
\Big\vert i_{n,\epsilon}(0) - i_{n,\epsilon=(0,0)}(0) \Big\vert
&\leq \frac{\rho_n}{2\alpha_n} \vert \epsilon_1 \vert + (4\Vert \varphi \Vert_\infty^2 + 1) \Vert \partial_x \varphi \Vert_\infty^2 \vert \epsilon_2 \vert\nonumber\\
&\leq \bigg(\frac{\rho_n}{2\alpha_n} + (4\Vert \varphi \Vert_\infty^2 + 1) \Vert \partial_x \varphi \Vert_\infty^2\bigg) 2s_n\nonumber\\
&\leq \Big(M_{\rho/\alpha} + 2(4\Vert \varphi \Vert_\infty^2 + 1) \Vert \partial_x \varphi \Vert_\infty^2\Big) s_n \;.\label{upperbound_1st_term_variation_t=0}
\end{align}
\textbf{2. }It remains to upper bound the second term on the right-hand side of \eqref{variation_t=0_split}.
Define the following observations where $\bX^* \iid P_{0,n}$, $\bm{\Phi} := (\Phi_{\mu i}) \iid \cN(0,1)$, $\bW^* := (W_{\mu}^*)_{\mu=1}^{m_n} \iid \cN(0,1)$ and $\eta \in [0,+\infty)$:
\begin{equation}
Y_{\mu}^{(\eta)}  \sim P_{\mathrm{out}}\bigg(\,\cdot\, \bigg\vert \, \frac{(\bm{\Phi} \bX^*)_\mu}{\sqrt{k_n}} +  \sqrt{\eta} \,W_{\mu}^*\bigg) + Z_\mu \; , \;1 \leq \mu \leq m_n\;.
\end{equation}
The joint posterior density of $(\bX^*,\bW^*)$ given $(\bY^{(\eta)},\bm{\Phi})$ reads:
\begin{equation}
dP(\bx,\bw \vert \bY^{(\eta)},\bm{\Phi})
:= \frac{1}{\cZ_{\eta}}\,dP_{0,n}(\bx)\,
\prod_{\mu=1}^{m_n} \frac{dw_\mu}{\sqrt{2\pi}} e^{-\frac{w_\mu^2}{2}} P_{\mathrm{out}}\bigg(Y_{\mu}^{(\eta)} \bigg\vert \frac{(\bm{\Phi} \bx)_\mu}{\sqrt{k_n}} +  \sqrt{\eta} \,w_{\mu}\bigg) \;,
\end{equation}
where $\cZ_{\eta}$ is the normalization factor.
Define the average free entropy $f_n(\eta) := \nicefrac{\E \ln \cZ_\rho}{m_n}$. 
The mutual information $i_{n}(\eta) := \frac{1}{m_n} I\big((\bX^*,\bW^*);\bY^{(\eta)}\big|\bm{\Phi}\big)$ satisfies:
\begin{equation}\label{link_mutual_info_free_entropy_eta}
i_{n}(\rho)
= \E\bigg[h\bigg(\frac{\Vert \bX^* \Vert^2}{k_n} + \eta\bigg)\bigg] - f_{n}(\rho) - \frac{1}{2\alpha_n} \;.
\end{equation}
where $h: \rho \in [0,+\infty) \mapsto \E_{V \sim \mathcal{N}(0,1)} \int u_y(\sqrt{\rho}\,V)e^{u_y(\sqrt{\rho}\,V)}dy$.
The identity \eqref{link_mutual_info_free_entropy_eta} can be obtained exactly as the identity \eqref{link_mutual_info_free_entropy} in Appendix~\ref{appendix:derivative_interpolating_mutual_information}.
Under the assumptions of the lemma, all the hypotheses of domination are reunited to make sure that $\eta \mapsto i_n(\eta)$ is continuous on $[0,2s_n]$ and differentiable on $(0,2s_n)$.
Therefore, by the mean-value theorem, there exists $\eta^* \in (0,2s_n)$ such that:
\begin{equation}\label{diff_i_n(2s_n)_i_n(0)}
\bigg\vert i_{n,\epsilon=(0,0)}(0) -\frac{I(\bX^*;\bY|\bm{\Phi})}{m_n} \bigg\vert
= \big\vert i_{n}(2s_n) - i_n(0) \big\vert
= \vert i'_n(\eta^*) \vert 2s_n \;.
\end{equation}
Again, in a similar fashion to the computation of the derivative of $i_{n,\epsilon}(\cdot)$ in Appendix~\ref{appendix:derivative_interpolating_mutual_information}, we can show that $\forall \eta \in [0,+\infty)$:
\begin{align}
i'_{n}(\rho)
&= \E\bigg[h'\bigg(\frac{\Vert \bX^* \Vert^2}{k_n} + \eta\bigg)\bigg] - f'_{n}(\rho)\;; \label{derivative_i_n(eta)}\\
f'_{n}(\rho)
&= \frac{1}{2}\E\Bigg[
\sum_{\mu=1}^{m_n}\frac{P_{\mathrm{out}}''\big(Y_\mu^{(\rho)} \big\vert \frac{(\bm{\Phi} \bX^*)_\mu}{\sqrt{k_n}} +  \sqrt{\eta} \,W_{\mu}^*\big)}{P_{\mathrm{out}}\big(Y_\mu^{(\rho)} \big\vert \frac{(\bm{\Phi} \bX^*)_\mu}{\sqrt{k_n}} +  \sqrt{\eta} \,W_{\mu}^*\big)}
\frac{\ln \cZ_{\rho}}{m_n} \Bigg]\;.
\end{align}
In Lemma~\ref{lemma:property_I_Pout} of Appendix~\ref{appendix:scalar_channels}, we compute $h'$ and show:
\begin{equation*}
\forall \rho \in [0,+\infty): \vert h'(\rho) \vert \leq C := C\bigg(\bigg\Vert \frac{\varphi}{\sqrt{\Delta}} \bigg\Vert_\infty, \bigg\Vert \frac{\partial_x \varphi}{\sqrt{\Delta}} \bigg\Vert_\infty, \bigg\Vert \frac{\partial_{xx} \varphi}{\sqrt{\Delta}} \bigg\Vert_\infty\bigg)
\end{equation*}
with
$C(a,b,c) := b^2(64a^4 + 2a^2 + 12.5) + c\big(8a^3 + 2 \sqrt{\frac{2}{\pi}}\,\big)$.
The first term on the right-hand side of \eqref{derivative_i_n(eta)} thus satisfies:
\begin{equation}\label{upperbound_1st_term_derivative_i_n(eta)}
\bigg\vert \E\bigg[h'\bigg(\frac{\Vert \bX^* \Vert^2}{k_n} + \eta\bigg)\bigg] \bigg\vert
\leq C \;.
\end{equation}
The second term, i.e., $f'_{n}(\rho)$ is similar to the quantity $A_n^{(t,\epsilon)}$ defined in \eqref{def:A_n}.
We upper bound $A_n^{(t,\epsilon)}$ in the last part of the proof in Appendix~\ref{appendix:derivative_interpolating_mutual_information}.
We can follow the same steps than for upper bounding $A_n^{(t,\epsilon)}$ and obtain:
\begin{equation}\label{upperbound_2nd_term_derivative_i_n(eta)}
\vert f'_n(\eta) \vert \leq \sqrt{D m_n \Var\,\frac{\ln \cZ_\eta}{m_n}}\;.
\end{equation}
Note that $\cZ_{\eta=2s_n} = \cZ_{t=0,\epsilon=(0,0)}$.
By Proposition \ref{prop:concentration_free_entropy} in Appendix~\ref{appendix:concentration_free_entropy} we have $\Var\,\frac{\ln \cZ_{\eta=2s_n}}{m_n} \leq \frac{\widetilde{C}}{n\alpha_n\rho_n}$ where $\widetilde{C}$ is a polynomial in $\big(S,\Vert\varphi\Vert_\infty, \Vert \partial_x \varphi \Vert_\infty, \Vert \partial_{xx} \varphi \Vert_\infty, M_{\alpha}, M_{\rho/\alpha}\big)$ with positive coefficients.
In fact, this upper bound holds for all $\eta \in [0,2s_n]$, i.e.,
\begin{equation*}
\forall \eta \in [0,2s_n]: \Var\bigg(\frac{\ln \cZ_{\eta}}{m_n}\bigg) \leq \frac{\widetilde{C}}{n\alpha_n\rho_n}\;.
\end{equation*}
The proof of this uniform bound on $\Var\big(\nicefrac{\ln \cZ_{\eta}}{m_n}\big)$ is the same as the one of Proposition  \ref{prop:concentration_free_entropy}, only that it is simpler because there is no second channel similar to $\widetilde{\bY}^{(t,\epsilon)}$.
We now combine \eqref{diff_i_n(2s_n)_i_n(0)}, \eqref{derivative_i_n(eta)}, \eqref{upperbound_1st_term_derivative_i_n(eta)}, \eqref{upperbound_2nd_term_derivative_i_n(eta)} to finally obtain:
\begin{equation}\label{upperbound_2nd_term_variation_t=0}
\bigg\vert i_{n,\epsilon=(0,0)}(0) -\frac{I(\bX^*;\bY|\bm{\Phi})}{m_n} \bigg\vert
\leq \Bigg(C + \sqrt{\frac{D\widetilde{C}}{\rho_n}}\,\Bigg) 2s_n \;.
\end{equation}
\textbf{3. } We now plug \eqref{upperbound_1st_term_variation_t=0} and \eqref{upperbound_2nd_term_variation_t=0} back in \eqref{variation_t=0_split} and use that $\rho_n \in (0,1]$ to end the proof of the lemma:
\begin{equation*}
\bigg\vert i_{n,\epsilon}(0) -\frac{I(\bX^*;\bY|\bm{\Phi})}{m_n} \bigg\vert
\leq \big(M_{\rho/\alpha} + 2(4\Vert \varphi \Vert_\infty^2 + 1) \Vert \partial_x \varphi \Vert_\infty^2 + 2C + \sqrt{D\widetilde{C}}\:\big)\frac{s_n}{\sqrt{\rho_n}}\;.
\end{equation*}
\end{proof}
\section{Concentration of the free entropy}\label{appendix:concentration_free_entropy}
In this appendix we show that the log-partition function per data point, or \textit{free entropy}, of the interpolating model studied in Section~\ref{interp-est-problem} concentrates around its expectation. 
\begin{proposition}[Free entropy concentration]\label{prop:concentration_free_entropy} 
Suppose that $\Delta > 0$ and that all of \ref{hyp:bounded},~\ref{hyp:c2} and \ref{hyp:phi_gauss2} hold.
Further assume that $\E_{X \sim P_0}[X^2] = 1$.
We have for all $(t,\epsilon) \in [0,1] \times \mathcal{B}_n$:
\begin{equation}\label{eq:concentration_free_entropy}
\Var\bigg(\frac{\ln \mathcal{Z}_{t,\epsilon}}{m_n}\bigg)
\leq 
\frac{1}{n \alpha_n\rho_n} \Big(\widetilde{C}_1 + \frac{\rho_n}{\alpha_n} \widetilde{C}_2 + \alpha_n \widetilde{C}_3\Big)\;,
\end{equation}
where ($\partial_x \varphi$ and $\partial_{xx} \varphi$ denote the first and second partial derivatives of $\varphi$ with respect to its first argument):
\begin{align*}
\widetilde{C}_1
&:= 1.5 + 4\bigg\Vert \frac{\varphi}{\sqrt{\Delta}} \bigg\Vert_{\infty}^2 + 8S^2 \bigg(4 \bigg\Vert \frac{\varphi}{\sqrt{\Delta}} \bigg\Vert_\infty^2 + 1\bigg)
\bigg\Vert \frac{\partial_x \varphi}{\sqrt{\Delta}}  \bigg\Vert_\infty^2\\
&\qquad\qquad\qquad\qquad\qquad\;\;
+ \bigg(2\bigg\Vert \frac{\varphi}{\sqrt{\Delta}}  \bigg\Vert_\infty + \sqrt{\frac{2}{\pi}} \bigg)^{\!\! 2}
\bigg(\bigg\Vert \frac{\varphi}{\sqrt{\Delta}}  \bigg\Vert_\infty^2
+ (16 + 4S^2) \bigg\Vert \frac{\partial_x \varphi}{\sqrt{\Delta}}  \bigg\Vert_\infty^2\bigg)\;;\\
\widetilde{C}_2 &:= 1.5 + 12S^2 \;;\\
\widetilde{C}_3 &:= 8S^2\Bigg( 3\bigg\Vert \frac{\partial_x \varphi}{\sqrt{\Delta}}\bigg\Vert_{\infty}^2 \!+ \bigg\Vert \frac{\varphi}{\sqrt{\Delta}}\bigg\Vert_{\infty}\bigg\Vert \frac{\partial_{xx} \varphi}{\sqrt{\Delta}}\bigg\Vert_{\infty}
\!+ 12 \bigg\Vert \frac{\partial_x \varphi}{\sqrt{\Delta}}\bigg\Vert_{\infty}^2 \bigg\Vert \frac{\varphi}{\sqrt{\Delta}}\bigg\Vert_{\infty}^2
+ 2 \sqrt{\frac{2}{\pi}} \bigg\Vert \frac{\varphi}{\sqrt{\Delta}}\bigg\Vert_{\infty}\bigg\Vert \frac{\partial_x \varphi}{\sqrt{\Delta}}\bigg\Vert_{\infty}^2\Bigg)^{\!\! 2}\;.
\end{align*}
In addition, if both sequences $(\alpha_n)_n$ and $(\nicefrac{\rho_n}{\alpha_n})_n$ are bounded, i.e., if there exist real positive numbers $M_\alpha, M_{\rho/\alpha}$ such that $\forall n \in \N^*: \alpha_n \leq M_{\alpha}, \nicefrac{\rho_n}{\alpha_n}  \leq  M_{\rho/\alpha}$ then for all $(t,\epsilon) \in [0,1] \times \mathcal{B}_n$:
\begin{equation}\label{eq:concentration_free_entropy_bis}
\Var\bigg(\frac{\ln \mathcal{Z}_{t,\epsilon}}{m_n}\bigg)
\leq 
\frac{C}{n \alpha_n\rho_n} \;,
\end{equation}
where $C := \widetilde{C}_1 + M_{\rho/\alpha}\widetilde{C}_2 + M_{\alpha}\widetilde{C}_3$.
\end{proposition}
To lighten notations, we define $k_1 := \sqrt{R_2(t,\epsilon)}$, $k_2 := \sqrt{t + 2s_n- R_2(t,\epsilon)}$.
Let $\bX^* \iid P_{0,n}$, $\bm{\Phi} := (\Phi_{\mu i}) \iid \cN(0,1)$, $\bV := (V_{\mu})_{\mu=1}^{m_n} \iid \cN(0,1)$ and $\bW^* := (W_{\mu}^*)_{\mu=1}^{m_n} \iid \cN(0,1)$.
Remember that
\begin{equation}
S_{\mu}^{(t,\epsilon)} := \sqrt{\frac{1-t}{k_n}}\, (\bm{\Phi} \bX^*)_\mu  + k_1 \,V_{\mu} + k_2 \,W_{\mu}^* \;,
\end{equation}
and that, in the interpolation problem, we observe:
\begin{align}
\begin{cases}
Y_{\mu}^{(t,\epsilon)}  &\sim \quad \varphi\big(S_{\mu}^{(t,\epsilon)} , \bA_\mu \big) + \sqrt{\Delta} Z_\mu \,,\;\:  1 \leq \mu \leq m_n\,; \\
\widetilde{Y}_{i}^{(t,\epsilon)} &= \sqrt{R_1(t,\epsilon)}\, X^*_i \,+\, \widetilde{Z}_i \qquad\;\;, \;\: 1 \leq i \,  \leq \, n \;\;\:\, ;
\end{cases}
\end{align}
where $(Z_\mu)_{\mu=1}^{m_n},(\widetilde{Z}_i)_{i=1}^n \iid \cN(0,1)$ and $(\bA_\mu)_{\mu=1}^{m_n} \iid P_A$.
$\cZ_{t,\epsilon}$ is the normalization to the joint posterior density of $(\bX^*,\bW^*)$ given $(\bY^{(t,\epsilon)},\widetilde{\bY}^{(t,\epsilon)},\bm{\Phi},\bV)$, i.e.,
\begin{align*}
\cZ_{t,\epsilon} := \int dP_{0,n}(\bx)\mathcal{D}\bw \,e^{-\frac{\Vert \sqrt{R_{1}(t,\epsilon)}\bx - \widetilde{\bY}^{(t,\epsilon)} \Vert^2}{2}}\, P_{\mathrm{out}}(Y_{\mu}^{(t,\epsilon)}\vert s_{\mu}^{(t,\epsilon)}) \;,
\end{align*}
where $\mathcal{D}\bw := \frac{d\bw e^{-\frac{\Vert \bw \Vert^2}{2}}}{\sqrt{2\pi}^{m_n}}$ and
$s_{\mu}^{(t,\epsilon)} := \sqrt{\frac{1-t}{k_n}}\, (\bm{\Phi} \bx)_\mu  + k_1 \,V_{\mu} + k_2 \,w_{\mu}$.
We define:
\begin{equation*}
\Gamma_{\mu}^{(t,\epsilon)} := \frac{\varphi\big(S_{\mu}^{(t,\epsilon)} , \bA_\mu \big)  
- \varphi\big(s_{\mu}^{(t,\epsilon)} , \ba_\mu \big)}{\Delta}\;.
\end{equation*}
By definition,
$P_{\mathrm{out}}(Y_\mu^{(t,\epsilon)} | s_{\mu}^{(t,\epsilon)})
= \int dP_A(\ba_\mu)\, \frac{1}{\sqrt{2\pi\Delta}}e^{-\frac{1}{2}(\Gamma_{\mu}^{(t,\epsilon)} +  Z_{\mu})^2 }$.
Therefore, the interpolating free entropy satisfies:
\begin{equation}\label{link_Z_newZ}
\frac{\ln \cZ_{t,\epsilon}}{m_n}
= \frac{1}{2}\ln(2\pi\Delta)
- \frac{1}{2m_n}\sum_{\mu=1}^{m_n} Z_\mu^2 
- \frac{1}{2m_n}\sum_{i=1}^n \widetilde{Z}_i^2
+ \frac{\ln \widehat{\cZ}_{t,\epsilon}}{m_n}
\end{equation}
where
\begin{align}
\widehat{\cZ}_{t,\epsilon}
&:= \int dP_{0,n}(\bx)\mathcal{D}\bw dP_A(\ba_\mu)\,e^{-\widehat{\cH}_{t,\epsilon}(\bx,\bw, \ba)}\;;\label{newZ}\\
\widehat{\cH}_{t,\epsilon}(\bx,\bw, \ba)
&:= \frac{1}{2}\sum_{\mu=1}^{m_n} (\Gamma_{\mu}^{(t,\epsilon)})^2 + 2 Z_\mu \Gamma_{\mu}^{(t,\epsilon)}\nonumber\\
&\qquad\qquad\qquad\qquad
+ \frac{1}{2} \sum_{i=1}^{n} R_1(t,\epsilon)( X_i^*- x_i)^2 + 2 Z_i^\prime \sqrt{R_1(t,\epsilon)}( X_i^*-  x_i)\;.\label{explicit-ham}
\end{align}
From \eqref{link_Z_newZ}, it follows directly that:
\begin{align}
\Var\bigg(\frac{\ln \cZ_{t,\epsilon}}{m_n}\bigg)
&\leq 3 \Var\bigg(\frac{1}{2m_n}\sum_{\mu=1}^{m_n} Z_\mu^2 \bigg)
+ 3 \Var\bigg(\frac{1}{2m_n}\sum_{i=1}^n \widetilde{Z}_i^2\bigg)
+ 3\Var\bigg(\frac{\ln \widehat{\cZ}_{t,\epsilon}}{m_n}\bigg)\nonumber\\
&= \frac{3}{2\alpha_n n} 
+ \frac{3}{2 \alpha_n^2 n}
+ 3\Var\bigg(\frac{\ln \widehat{\cZ}_{t,\epsilon}}{m_n}\bigg)\label{upperbound_var_free_entropy}
\end{align}
In order to prove Proposition~\ref{prop:concentration_free_entropy}, it remains to show that $\ln \widehat{\cZ}_{t,\epsilon}/{m_n}$ concentrates. 
We recall here the classical variance bounds that we will use.
We refer to \cite[Chapter 3]{Boucheron2013Concentration} for detailed proofs of these statements.
\begin{proposition}[Gaussian Poincar\'e inequality]\label{poincare}
Let $\bU = (U_1, \dots, U_N)$ be a vector of $N$ independent standard normal random variables. Let $g: \mathbb{R}^N \to \mathbb{R}$ be a ${\cal C}^1$ function. Then
 \begin{align}
	 \Var(g(\bU)) \leq \E \big[\| \nabla g (\bU) \|^2 \big] \,.
 \end{align}
\end{proposition}
\begin{proposition}[Bounded difference]\label{bounded_diff}
	Let $\mathcal{U} \subset \R$.
	Let $g: \mathcal{U}^N \to \mathbb{R}$ a function 
that satisfies the bounded difference property, i.e., there exists some constants $c_1, \dots, c_N \geq 0$ such that
\begin{equation*}
\sup_{\substack{(u_1, \ldots, u_N)\in \mathcal{U}^N \\ u_i' \in \mathcal{U}}}
\vert g(u_1, \dots, u_i, \ldots, u_N) - g(u_1, \dots, u_i', \ldots, u_N)\vert \leq c_i \quad \text{for all} \quad 1 \leq i \leq N \,.
\end{equation*}
Let $\bU=(U_1, \dots, U_N)$ be a vector of $N$ independent random variables that take values in $\mathcal{U}$. Then
\begin{align}
	\Var(g(\bU)) \leq \frac{1}{4} \sum_{i=1}^N c_i^2 \,.
\end{align}
\end{proposition}
\begin{proposition}[Efron-Stein inequality]\label{efron_stein}
Let $\mathcal{U} \subset \R$, and a function $g: \mathcal{U}^N \to \R$. Let $\bu=(U_1, \dots, U_N)$ be a vector of $N$ independent random variables with law $P_U$ that take values in $\mathcal{U}$. Let $\bU^{(i)}$ a vector which differs from $\bU$ only by its $i$-th component, which is replaced by $U_i'$ drawn from $P_U$ independently of $\bU$. Then
\begin{align}
	\Var(g(\bU)) \leq \frac{1}{2} \sum_{i=1}^N \E_{\bU}\E_{U_i'}\big[(g(\bU)-g(\bU^{(i)}))^2\big] \,.
\end{align}
\end{proposition}
We first show the concentration w.r.t.\ all Gaussian variables $\bm{\Phi}, \bV, \bZ, \bZ^\prime, \bW^*$, then the concentration w.r.t.\ $\bA$ and finally the one w.r.t.\ $\bX^*$. The order in which we prove the concentrations does matter.

We will denote $\partial_x \varphi$ and $\partial_{xx} \varphi$ the first and second partial derivatives of $\varphi$ with respect to its first argument.
Note that $\vert R_1 \vert \leq 2s_n + \frac{\alpha_n}{\rho_n}r_{\max}$ and, by the inequality \eqref{final_upperbound_dI_Pout(q,rho)/dq} in Lemma~\ref{lemma:property_I_Pout} of Appendix~\ref{appendix:scalar_channels},
$r_{\max} := 2 \big\vert \frac{\partial I_{P_{\mathrm{out}}}}{\partial q}\big\vert_{1,1} \big\vert \leq 2 C_1(\Vert \frac{\varphi}{\sqrt{\Delta}} \Vert_\infty, \Vert \frac{\partial_x \varphi}{\sqrt{\Delta}}  \Vert_\infty)$
with $C_1(a,b) := (4 a^2 + 1) b^2$. Then, the quantity
\begin{equation*}
K_{n} := 2\bigg(s_n+\frac{\alpha_n}{\rho_n}C_1\bigg(\bigg\Vert \frac{\varphi}{\sqrt{\Delta}} \bigg\Vert_\infty, \bigg\Vert \frac{\partial_x \varphi}{\sqrt{\Delta}}  \bigg\Vert_\infty\bigg)\bigg)
\end{equation*}
upper bounds $\vert R_1 \vert$. Besides, $\vert R_2 \vert$ is upper bounded by $2$.
%
%
%
\paragraph{Concentration with respect to the Gaussian random variables}
%
\begin{lemma} \label{lem:concentration_gauss1}
Let $\mathbb{E}_{\bZ, \widetilde{\bZ}}$ be the expectation w.r.t.\ $(\bZ, \widetilde{\bZ})$ \textit{only}.
Under the assumptions of Theorem~\ref{th:RS_1layer}, we have for all $(t,\epsilon) \in [0,1] \times \mathcal{B}_n$:
\begin{equation}
\E \Big[\Big(\frac{\ln \hat{\cZ}_{t,\epsilon}}{m_n}  - \frac{1}{m_n} \E_{\bZ, \bZ'} \ln \hat{\cZ}_{t,\epsilon} \Big)^2\Big]
\leq \frac{C_2}{n \alpha_n\rho_n} + \frac{C_3}{n \alpha_n^2} \;,
\end{equation}
where $C_2 := 4\big\Vert \frac{\varphi}{\sqrt{\Delta}} \big\Vert_{\infty}^2 + 8S^2 C_1\big(\big\Vert \frac{\varphi}{\sqrt{\Delta}} \big\Vert_\infty, \big\Vert \frac{\partial_x \varphi}{\sqrt{\Delta}}  \big\Vert_\infty\big)$ and $C_3 = 4S^2$.
\end{lemma}
\begin{proof}
In this proof we see $g := \ln\hat{\mathcal{Z}}_{t,\epsilon}/m_n$ as a function of $\bZ$ and $\widetilde{\bZ}$, and  we work conditionally on all other random variables. 
We have $\Vert \nabla g\Vert^2=\Vert \nabla_{\! \bZ}\, g\Vert^2+\Vert \nabla_{\! \widetilde{\bZ}}\, g\Vert^2$.
Each partial derivative has the form $\partial_u g = m_n^{-1} \langle \partial_u \widehat{\mathcal{H}}_{t,\epsilon}\rangle_{t,\epsilon}$.
We find:
\begin{align*}
& 
\Vert \nabla_{\! \bZ}\, g\Vert^2 
= m_n^{-2} \sum_{\mu=1}^{m_n}\langle \Gamma_{\mu}^{(t,\epsilon)} \rangle_{t,\epsilon}^2
\leq 
4 m_n^{-1} \bigg\Vert \frac{\varphi}{\sqrt{\Delta}} \bigg\Vert_{\infty}^2\;,
\nonumber\\
&\Vert \nabla_{\!\widetilde{\bZ}}\,g \Vert^2 
= m_n^{-2}R_1(t,\epsilon)\sum_{i=1}^n ( X_i^*-\langle  x_i \rangle_{t,\epsilon})^2
\leq 
4K_{n}S^2 m_n^{-2}n\;.
\end{align*}
So
$ \Vert \nabla g\Vert^2 \leq 4 m_n^{-1}\big(\big\Vert \frac{\varphi}{\sqrt{\Delta}} \big\Vert_{\infty}^2 + \frac{K_{n}S^2}{\alpha_n}\big)$.
Applying Proposition~\ref{poincare} yields:
\begin{align*}
\E_{\bZ,\widetilde{\bZ}} \Big[\Big(\frac{\ln \hat{\cZ}_{t,\epsilon}}{m_n}  - \frac{\E_{\bZ,\widetilde{\bZ}} \ln \hat{\cZ}_{t,\epsilon}}{m_n} \Big)^2\Big]
&\leq \frac{4}{n \alpha_n} \bigg(\bigg\Vert \frac{\varphi}{\sqrt{\Delta}} \bigg\Vert_{\infty}^2  + \frac{K_{n}S^2}{\alpha_n}\bigg)\\
&= \frac{4}{n \alpha_n} \bigg(\bigg\Vert \frac{\varphi}{\sqrt{\Delta}} \bigg\Vert_{\infty}^2  + \frac{2S^2s_n}{\alpha_n}
+\frac{2S^2}{\rho_n}C_1\bigg(\bigg\Vert \frac{\varphi}{\sqrt{\Delta}} \bigg\Vert_\infty, \bigg\Vert \frac{\partial_x \varphi}{\sqrt{\Delta}}  \bigg\Vert_\infty\bigg) \bigg)\\
&\leq \frac{4}{n \alpha_n \rho_n} \bigg(\bigg\Vert \frac{\varphi}{\sqrt{\Delta}} \bigg\Vert_{\infty}^2
+ 2S^2 C_1\bigg(\bigg\Vert \frac{\varphi}{\sqrt{\Delta}} \bigg\Vert_\infty, \bigg\Vert \frac{\partial_x \varphi}{\sqrt{\Delta}}  \bigg\Vert_\infty\bigg) \bigg)
+ \frac{4S^2}{n \alpha_n^2}\;.
\end{align*}
The last inequality follows from $\rho_n \leq 1$ and $2s_n \leq 1$.
Taking the expectation on both sides of this last inequality gives the lemma.
\end{proof}
\begin{lemma}\label{lem:concentration_gauss2}
Let $\mathbb{E}_{G}$ denotes the expectation w.r.t.\ $(\bZ, \widetilde{\bZ}, \bV, \bW^*, \bm{\Phi})$ \textit{only}.
Under the assumptions of Theorem~\ref{th:RS_1layer}, we have for all $(t,\epsilon) \in [0,1] \times \mathcal{B}_n$:
\begin{align}
\E \bigg[\bigg(
\frac{\E_{\bZ, \widetilde{\bZ}} \ln \widehat{\cZ}_{t,\epsilon}}{m_n} - \frac{\E_{G} \ln \widehat{\cZ}_{t,\epsilon} }{m_n}
\bigg)^2\bigg]
&\leq \frac{C_4}{n \alpha_n\rho_n}\,.
\end{align}
where $C_4 := \big(4\big\Vert \frac{\varphi}{\sqrt{\Delta}}  \big\Vert_\infty + 2\sqrt{\frac{2}{\pi}} \big)^{2}
\,(4 + S^2) \big\Vert \frac{\partial_x \varphi}{\sqrt{\Delta}}  \big\Vert_\infty^2$.
\end{lemma}
\begin{proof}
In this proof we see $g = \mathbb{E}_{\bZ, \widetilde{\bZ}}\ln\widehat{\mathcal{Z}}_{t,\epsilon}/m_n$ as a function of $\bV$, $\bW^*$, $\bm{\Phi}$ and we work conditionally on $\bA$, $\bX^*$.
Once again each partial derivative has the form $\partial_u g = m_n^{-1} \langle \partial_u \widehat{\mathcal{H}}_{t,\epsilon}\rangle_{t,\epsilon}$.
We first compute the partial derivatives of $g$ w.r.t.\ $\{V_\mu\}_{\mu=1}^{m_n}$:
\begin{align*}
\bigg\vert\frac{\partial g}{\partial V_\mu}\bigg\vert
= m_n^{-1}\bigg\vert \E_{\bZ, \widetilde{\bZ}}\bigg\langle (\Gamma_{\mu}^{(t,\epsilon)} +Z_\mu)\frac{\partial \Gamma_{\mu}^{(t,\epsilon)}}{\partial V_\mu} \bigg\rangle_{\!\! t,\epsilon}\bigg\vert
&\leq  m_n^{-1}\E_{\bZ, \widetilde{\bZ}}\bigg[\bigg((2\bigg\Vert \frac{\varphi}{\sqrt{\Delta}}  \bigg\Vert_\infty +\vert Z_\mu\vert\bigg) \, 2\sqrt{2}\bigg\Vert \frac{\partial_x \varphi}{\sqrt{\Delta}}  \bigg\Vert_\infty \bigg]\\
&=m_n^{-1} \bigg(4\bigg\Vert \frac{\varphi}{\sqrt{\Delta}}  \bigg\Vert_\infty + 2\sqrt{\frac{2}{\pi}}\bigg) \sqrt{2}\bigg\Vert \frac{\partial_x \varphi}{\sqrt{\Delta}}  \bigg\Vert_\infty\;.
\end{align*}
The same inequality holds for $\vert\frac{\partial g}{\partial W_\mu^*}\vert$.
To compute the derivative w.r.t.\ $\Phi_{\mu i}$, we first remark that:
\begin{multline*}
\frac{\partial \Gamma_{\mu}^{(t,\epsilon)}}{\partial \Phi_{\mu i}}
= \sqrt{\frac{1-t}{\Delta k_n}} \Big\{X_i^*\,\partial_x\varphi\Big(\sqrt{\frac{1-t}{k_n}}(\bm{\Phi} \bX^*)_\mu+ k_1 V_{\mu} + k_2 W_\mu^* , \bA_\mu \Big)\\
- x_i\,\partial_x\varphi\Big(\sqrt{\frac{1-t}{k_n}}(\bm{\Phi} \bx)_\mu+ k_1 V_{\mu} + k_2 w_\mu , \ba_\mu \Big)  \Big\}\;.
\end{multline*}
Therefore:
\begin{align*}
\bigg\vert\frac{\partial g}{\partial \Phi_{\mu i}}\bigg\vert
= m_n^{-1}\bigg\vert \E_{\bZ, \widetilde{\bZ}}\bigg\langle (\Gamma_{\mu}^{(t,\epsilon)} +Z_\mu)\frac{\partial \Gamma_{\mu}^{(t,\epsilon)}}{\partial \Phi_{\mu i}} \bigg\rangle_{\!\! t,\epsilon}\bigg\vert
&\leq 
\frac{1}{m_n\sqrt{k_n}}\E_{\bZ, \widetilde{\bZ}}\bigg[\bigg(2\bigg\Vert \frac{\varphi}{\sqrt{\Delta}}  \bigg\Vert_\infty + \vert Z_\mu\vert \bigg)
\,2S \bigg\Vert \frac{\partial_x \varphi}{\sqrt{\Delta}}  \bigg\Vert_\infty \bigg]\\
&=
\frac{1}{m_n\sqrt{k_n}} \bigg(4\bigg\Vert \frac{\varphi}{\sqrt{\Delta}}  \bigg\Vert_\infty + 2\sqrt{\frac{2}{\pi}} \bigg)
\,S \bigg\Vert \frac{\partial_x \varphi}{\sqrt{\Delta}}  \bigg\Vert_\infty\;.
\end{align*}
Putting together these inequalities on the partial derivatives of $g$, we find:
\begin{align*}
\Vert \nabla g\Vert^2 & = \sum_{\mu=1}^{m_n} \Big\vert\frac{\partial g}{\partial V_\mu}\Big\vert^2
+
\sum_{\mu=1}^{m_n} \Big\vert\frac{\partial g}{\partial W_\mu^*}\Big\vert^2
+
\sum_{\mu=1}^{m_n}\sum_{i=1}^n \Big\vert\frac{\partial g}{\partial \Phi_{\mu i}}\Big\vert^2\\
&\leq 
\frac{4}{m_n} \bigg(4 \bigg\Vert \frac{\varphi}{\sqrt{\Delta}}  \bigg\Vert_\infty + 2\sqrt{\frac{2}{\pi}}\bigg)^{\!\! 2}\bigg\Vert \frac{\partial_x \varphi}{\sqrt{\Delta}}  \bigg\Vert_\infty^2
+ \frac{1}{m_n \rho_n} \bigg(4\bigg\Vert \frac{\varphi}{\sqrt{\Delta}}  \bigg\Vert_\infty + 2\sqrt{\frac{2}{\pi}} \bigg)^{\!\! 2}
\,S^2 \bigg\Vert \frac{\partial_x \varphi}{\sqrt{\Delta}}  \bigg\Vert_\infty^2\\
&\leq 
\frac{1}{m_n \rho_n} \bigg(4\bigg\Vert \frac{\varphi}{\sqrt{\Delta}}  \bigg\Vert_\infty + 2\sqrt{\frac{2}{\pi}} \bigg)^{\!\! 2}
\,(4 + S^2) \bigg\Vert \frac{\partial_x \varphi}{\sqrt{\Delta}}  \bigg\Vert_\infty^2
\end{align*}
In the last inequality we used that $\rho_n \leq 1$.
To end the proof of the lemma it remains to apply Proposition~\ref{poincare} as we did in Lemma~\ref{lem:concentration_gauss1}.
\end{proof}
\paragraph{Concentration with respect to the random stream}
We now apply the variance bound of Proposition~\ref{bounded_diff} to show that $\E_{G}\ln\widehat{\mathcal{Z}}_{t,\epsilon}/m_n$ concentrates w.r.t.\ $\bA$.
\begin{lemma}\label{lem:concentration_A}
Let $\E_{\bA}$ denotes the expectation w.r.t.\ $\bA$ only. 
Under the assumptions of Theorem~\ref{th:RS_1layer}, we have for all $(t,\epsilon) \in [0,1] \times \mathcal{B}_n$:
\begin{equation}
\E\bigg[\bigg(
	\frac{\E_{G} \ln \widehat{\cZ}_{t,\epsilon}}{m_n}  - \frac{\E_{G,\bA} \ln \widehat{\cZ}_{t,\epsilon} }{m_n}\bigg)^{\!\! 2}\,\bigg]
\leq \frac{C_5}{n \alpha_n} \,.
\end{equation}
where $C_5 := \Big(2\big\Vert \frac{\varphi}{\sqrt{\Delta}}  \big\Vert_\infty + \sqrt{\frac{2}{\pi}} \Big)^2\big\Vert \frac{\varphi}{\sqrt{\Delta}}  \big\Vert_\infty^2$.
\end{lemma}
\begin{proof}
We see $g=\E_{G}\ln\widehat{\mathcal{Z}}_{t,\epsilon}/m_n$ as a function of $\bA$ only.
Let $\nu \in \{1, \dots, m_n \}$.
We want to estimate the difference $g(\bA) - g(\bA^{(\nu)})$ corresponding to two configurations $\bA$ and $\bA^{(\nu)}$
such that $A_{\mu}^{(\nu)} = A_{\mu}$ for $\mu \neq \nu$ and $A_{\nu}^{(\nu)} \sim P_A$ independently of everything else. 
We will denote $\widehat{\mathcal{H}}_{t,\epsilon}^{(\nu)}$ and $\Gamma_{\mu}^{(t,\epsilon)(\nu)}$
the quantities $\widehat{\mathcal{H}}_{t,\epsilon}$ and $\Gamma_{\mu}^{(t,\epsilon)}$ when $\bA$ is replaced by $\bA^{(\nu)}$.
By Jensen's inequality, we have:
\begin{align}\label{jensen1}
\frac{1}{m_n} \mathbb{E}_G \langle \widehat{\mathcal{H}}_{t,\epsilon}^{(\nu)} - \widehat{\mathcal{H}}_{t,\epsilon} \rangle_{t,\epsilon}^{(\nu)}
\leq 
g(\bA) - g(\bA^{(\nu)}) 
\leq 
\frac{1}{m_n} \mathbb{E}_G\langle \widehat{\mathcal{H}}_{t,\epsilon}^{(\nu)} - \widehat{\mathcal{H}}_{t,\epsilon}\rangle_{t,\epsilon} 
\end{align}
where the angular brackets $\langle - \rangle_{t,\epsilon}$ and $\langle - \rangle_{t,\epsilon}^{(\nu)}$ denote expectation with respect to the distributions $\propto  dP_{0,n}(\bx)\mathcal{D}\bw dP_A(\ba_\mu)\,e^{-\widehat{\cH}_{t,\epsilon}(\bx,\bw, \ba)}$
and $\propto  dP_{0,n}(\bx)\mathcal{D}\bw dP_A(\ba_\mu)\,e^{-\widehat{\cH}_{t,\epsilon}^{(\nu)}(\bx,\bw, \ba)}$, respectively.
From the definition \eqref{explicit-ham} of $\widehat{\mathcal{H}}_{t,\epsilon}$,
\begin{align*}
\widehat{\mathcal{H}}_{t,\epsilon}^{(\nu)} - \widehat{\mathcal{H}}_{t,\epsilon} 
=  \frac{1}{2} \Big( \big(\Gamma_{\nu}^{(t,\epsilon)(\nu)}\big)^2  - \big(\Gamma_{\nu}^{(t,\epsilon)}\big)^2 + 2 Z_\nu \big(\Gamma_{\nu}^{(t,\epsilon)(\nu)} - \Gamma_{\nu}^{(t,\epsilon)}\big) \Big) \;.
\end{align*}
Note that:
\begin{equation*}
\Big\vert \big(\Gamma_{\nu}^{(t,\epsilon)(\nu)}\big)^2  - \big(\Gamma_{\nu}^{(t,\epsilon)}\big)^2 + 2 Z_\nu \big(\Gamma_{\nu}^{(t,\epsilon)(\nu)} - \Gamma_{\nu}^{(t,\epsilon)}\big) \Big\vert
\leq 
8 \bigg\Vert \frac{\varphi}{\sqrt{\Delta}}  \bigg\Vert_\infty^2 + 4 \vert Z_\nu\vert \bigg\Vert \frac{\varphi}{\sqrt{\Delta}}  \bigg\Vert_\infty\;.
\end{equation*}
We thus conclude that $g$ satisfies the bounded difference property:
\begin{align}
\forall \nu \in \{1,\dots,m_n\}:
\big\vert g(\bA) - g(\bA^{(\nu)}) \vert
\leq \frac{2}{m_n}\bigg(2\bigg\Vert \frac{\varphi}{\sqrt{\Delta}}  \bigg\Vert_\infty + \sqrt{\frac{2}{\pi}} \bigg)\bigg\Vert \frac{\varphi}{\sqrt{\Delta}}  \bigg\Vert_\infty\;.
\end{align}
To end the proof of Lemma~\ref{lem:concentration_A}, we just need to apply Proposition~\ref{bounded_diff}.
\end{proof}
\paragraph{Concentration with respect to the signal}
Let $\E_{\sim \bX^*} \equiv \E_{\bA,G}$ denote the expectation w.r.t.\ all quenched variables except $\bX^*$. 
It remains to bound the variance of $\E_{\sim \bX^*} \ln \hat{\cZ}_{t,\epsilon}/m_n$ (which only depends on $\bX^*$).
\begin{lemma}\label{lem:concentration_X}
Under the assumptions of Theorem~\ref{th:RS_1layer}, we have for all $(t,\epsilon) \in [0,1] \times \mathcal{B}_n$:
\begin{align*}
\E \bigg[\bigg(\frac{\E[\ln \widehat{\cZ}_{t,\epsilon} \vert \bX^*]}{m_n} - \frac{\E \ln \widehat{\cZ}_{t,\epsilon} }{m_n}\bigg)^{\!\! 2}\,\bigg]
&\leq \frac{C_6}{n\rho_n} + \frac{C_7 \rho_n}{n \alpha_n^2}
\end{align*}
where $C_7 := 8S^2$ and
\begin{equation*}
C_6 := 8S^2\Bigg( 3\bigg\Vert \frac{\partial_x \varphi}{\sqrt{\Delta}}\bigg\Vert_{\infty}^2 \!+ \bigg\Vert \frac{\varphi}{\sqrt{\Delta}}\bigg\Vert_{\infty}\bigg\Vert \frac{\partial_{xx} \varphi}{\sqrt{\Delta}}\bigg\Vert_{\infty}
\!+ 12 \bigg\Vert \frac{\partial_x \varphi}{\sqrt{\Delta}}\bigg\Vert_{\infty}^2 \bigg\Vert \frac{\varphi}{\sqrt{\Delta}}\bigg\Vert_{\infty}^2
+ 2 \sqrt{\frac{2}{\pi}} \bigg\Vert \frac{\varphi}{\sqrt{\Delta}}\bigg\Vert_{\infty}\bigg\Vert \frac{\partial_x \varphi}{\sqrt{\Delta}}\bigg\Vert_{\infty}^2\Bigg)^{\!\! 2}\:.
\end{equation*}

\end{lemma}
\begin{proof}
$g=\nicefrac{\E[\ln\widehat{\mathcal{Z}}_{t,\epsilon} \vert \bX^*]}{m_n}$ is a function of $\bX^*$.
For $j \in \{1,\dots,n\}$, we have:
\begin{align}
\frac{\partial g}{\partial X_j^*}
&= -\frac{1}{m_n}\E\bigg[\bigg\langle \frac{\partial \widehat{\cH}_{t,\epsilon}}{\partial X_j^*}\bigg\rangle_{\!\! n,t,\epsilon}\bigg\vert \bX^*\bigg]\nonumber\\
&= -\frac1{m_n} \sqrt{\frac{1-t}{\Delta k_n}} \sum_{\mu=1}^{m_n} \E\Big[\Phi_{\mu j} \partial_x \varphi(S_\mu^{(t,\epsilon)}, \bA_\mu) \big(\langle \Gamma_{\mu}^{(t,\epsilon)} \rangle_{n,t,\epsilon} + Z_\mu) \Big\vert \bX^*\Big]\nonumber\\
&\qquad\qquad\qquad\qquad\qquad\qquad\qquad
+\frac{1}{m_n}\E\big[\big\langle R_1(t,\epsilon)(X^*_j- x_j)+\sqrt{R_1(t,\epsilon)} \widetilde{Z}_j\big\rangle_{n,t,\epsilon} \big\vert \bX^* \big]\nonumber\\
&= -\frac1{m_n} \sqrt{\frac{1-t}{\Delta k_n}} \sum_{\mu=1}^{m_n} \E\Big[\Phi_{\mu j} \partial_x \varphi(S_\mu^{(t,\epsilon)}, \bA_\mu) \langle \Gamma_{\mu}^{(t,\epsilon)} \rangle_{n,t,\epsilon}\Big\vert \bX^*\Big]\nonumber\\
&\qquad\qquad\qquad\qquad\qquad\qquad\qquad
+ \frac{R_1(t,\epsilon)}{m_n}\E\big[(X^*_j- \langle x_j \rangle_{n,t,\epsilon}) \big\vert \bX^* \big]
\end{align}
To get the last equality we use
${\E[\Phi_{\mu j} \partial_x \varphi(S_\mu^{(t,\epsilon)}\!\!,\! \bA_\mu)  Z_\mu \vert \bX^*]
	\! = \!\E[\Phi_{\mu j} \partial_x \varphi(S_\mu^{(t,\epsilon)}\!\!,\! \bA_\mu) \vert \bX^*]\E[Z_\mu]\!=\!0}$
and $\E\sqrt{R_1(t,\epsilon)} \widetilde{Z}_j\vert \bX^*] = 0$.
An integration by parts with respect to $\Phi_{\mu j}$ yields:
\begin{align*}
&\E\Big[\Phi_{\mu j} \partial_x \varphi(S_\mu^{(t,\epsilon)}, \bA_\mu) \langle \Gamma_{\mu}^{(t,\epsilon)} \rangle_{n,t,\epsilon}\Big\vert \bX^*\Big]\\
&= \sqrt{\frac{1-t}{k_n\Delta }}\E\Big[X_j^*(\partial_x \varphi^2 + \varphi \, \partial_{xx}\varphi)(S_\mu^{(t,\epsilon)}, \bA_\mu)\Big\vert \bX^*\Big]\\
&\qquad\qquad-\sqrt{\frac{1-t}{\Delta k_n}}\E\Big[ X_j^*\partial_{xx} \varphi(S_\mu^{(t,\epsilon)}, \bA_\mu) \langle \varphi(s_\mu^{(t,\epsilon)}, \ba_\mu) \rangle_{n,t,\epsilon}\Big\vert \bX^*\Big]\\
&\qquad\qquad-\sqrt{\frac{1-t}{\Delta k_n}}\E\Big[ \partial_{x} \varphi(S_\mu^{(t,\epsilon)}, \bA_\mu) \langle x_j \partial_{x} \varphi(s_\mu^{(t,\epsilon)}, \ba_\mu) \rangle_{n,t,\epsilon}\Big\vert \bX^*\Big]\\
&\qquad\qquad+\sqrt{\frac{1-t}{\Delta k_n}}\E\Big[\partial_{x} \varphi(S_\mu^{(t,\epsilon)}, \bA_\mu) \big\langle \varphi(s_\mu^{(t,\epsilon)}, \ba_\mu)\\
&\qquad\qquad\qquad\qquad\qquad\qquad
\big(X_j^*\partial_{x} \varphi(S_\mu^{(t,\epsilon)}, \bA_\mu)
-x_j\partial_{x} \varphi(s_\mu^{(t,\epsilon)}, \ba_\mu)\big)
(\Gamma_{\mu}^{(t,\epsilon)} + Z_\mu)\big\rangle_{n,t,\epsilon}\Big\vert \bX^*\Big]\\
&\qquad\qquad-\sqrt{\frac{1-t}{\Delta k_n}}\E\Big[\partial_{x} \varphi(S_\mu^{(t,\epsilon)}, \bA_\mu) \langle \varphi(s_\mu^{(t,\epsilon)}, \ba_\mu)\rangle_{n,t,\epsilon}\\
&\qquad\qquad\qquad\qquad\qquad\qquad
\big\langle \big(X_j^*\partial_{x} \varphi(S_\mu^{(t,\epsilon)}, \bA_\mu)
-x_j\partial_{x} \varphi(s_\mu^{(t,\epsilon)}, \ba_\mu)\big)
(\Gamma_{\mu}^{(t,\epsilon)} + Z_\mu)\big\rangle_{n,t,\epsilon}\Big\vert \bX^*\Big]\\
\end{align*}
It directly follows that:
$\big\vert\E\big[\Phi_{\mu j} \partial_x \varphi(S_\mu^{(t,\epsilon)}, \bA_\mu) \langle \Gamma_{\mu}^{(t,\epsilon)} \rangle_{n,t,\epsilon}\big\vert \bX^*\big]\big\vert
\leq \sqrt{\frac{\Delta}{k_n}} \widetilde{C}_6$ where:
\begin{equation*}
\widetilde{C}_6 := 2S\bigg(\bigg\Vert \frac{\partial_x \varphi}{\sqrt{\Delta}}\bigg\Vert_{\infty}^2 + \bigg\Vert \frac{\varphi}{\sqrt{\Delta}}\bigg\Vert_{\infty}\bigg\Vert \frac{\partial_{xx} \varphi}{\sqrt{\Delta}}\bigg\Vert_{\infty}
+ 4 \bigg\Vert \frac{\partial_x \varphi}{\sqrt{\Delta}}\bigg\Vert_{\infty}^2 \bigg\Vert \frac{\varphi}{\sqrt{\Delta}}\bigg\Vert_{\infty}^2
+ 2 \sqrt{\frac{2}{\pi}} \bigg\Vert \frac{\varphi}{\sqrt{\Delta}}\bigg\Vert_{\infty}\bigg\Vert \frac{\partial_x \varphi}{\sqrt{\Delta}}\bigg\Vert_{\infty}^2\bigg)\;.
\end{equation*}
Making use of this upper bound, we obtain for all $j \in \{1,\dots,n\}$:
\begin{align}
\bigg\vert \frac{\partial g}{\partial X_j^*} \bigg\vert
\leq \frac{\widetilde{C}_6}{k_n} + \frac{2S K_n}{m_n}
&=\frac{\widetilde{C}_6}{k_n} + \frac{2S}{m_n}\bigg(2s_n+2\frac{\alpha_n}{\rho_n}C_1\bigg(\bigg\Vert \frac{\varphi}{\sqrt{\Delta}} \bigg\Vert_\infty, \bigg\Vert \frac{\partial_x \varphi}{\sqrt{\Delta}}  \bigg\Vert_\infty\bigg)\bigg)\nonumber\\
&=\frac{1}{n\rho_n}\bigg(\widetilde{C}_6 + 4SC_1\bigg(\bigg\Vert \frac{\varphi}{\sqrt{\Delta}} \bigg\Vert_\infty, \bigg\Vert \frac{\partial_x \varphi}{\sqrt{\Delta}}  \bigg\Vert_\infty\bigg)\bigg)
+ \frac{2S}{n \alpha_n}\;.\label{upperbound_dgdX_j^*}
\end{align}
For a fixed $j \in \{1, \dots, n\}$, let $\bX^{(j)}$ be a vector such that $X_i^{(j)}= X_i^*$ for $i \neq j$ and $X_j^{(j)} \sim P_{0,n}$ independently of everything else.
By the mean-value theorem and thanks to \eqref{upperbound_dgdX_j^*}, we have:
\begin{align*}
&\E_{\bX^*}\E_{X_j^{(j)}}\big[\big(g(\bX^*)-g(\bX^{*(j)})\big)^2\,\big]\\
&\qquad\qquad\qquad\qquad
\leq \Bigg(\frac{1}{n\rho_n}\Bigg(\widetilde{C}_6 + 4SC_1\bigg(\bigg\Vert \frac{\varphi}{\sqrt{\Delta}} \bigg\Vert_\infty, \bigg\Vert \frac{\partial_x \varphi}{\sqrt{\Delta}}  \bigg\Vert_\infty\bigg)\Bigg)
+ \frac{2S}{n \alpha_n}\Bigg)^{\!\! 2}\E\big[\big(X_j^* - X_j^{(j)}\big)^2\big]\\
&\qquad\qquad\qquad\qquad
\leq \frac{4}{n^2\rho_n}\Bigg(\widetilde{C}_6 + 4SC_1\bigg(\bigg\Vert \frac{\varphi}{\sqrt{\Delta}} \bigg\Vert_\infty, \bigg\Vert \frac{\partial_x \varphi}{\sqrt{\Delta}}  \bigg\Vert_\infty\bigg)\Bigg)^{\!\! 2}
+ \frac{16 S^2\rho_n}{n^2 \alpha_n^2}\;.
\end{align*}
We used $\E\big[\big(X_j^* - X_j^{(j)}\big)^2\big] = 2\rho_n\E_{X \sim P_0}[X^2] - 2\rho_n^2 \E_{X \sim P_0}[X]^2 \leq 2\rho_n\E_{X \sim P_0}[X^2] = 2\rho_n$ and Jensen's inequality $(a+b)^2 \leq 2a^2 + 2b^2$ to get the last inequality.
To end the proof it now suffices to apply Proposition~\ref{efron_stein}.
\end{proof}
\noindent\textbf{Proof of Proposition~\ref{prop:concentration_free_entropy}:}
Combining Lemmas~\ref{lem:concentration_gauss1}, \ref{lem:concentration_gauss2}, ~\ref{lem:concentration_A} and ~\ref{lem:concentration_X} yields:
\begin{equation}\label{var_free_entropy_hat}
\Var\bigg(\frac{\ln \widehat{\mathcal{Z}}_{t,\epsilon}}{m_n}\bigg) 
\leq 
\frac{C_2 + C_4}{n \alpha_n\rho_n} + \frac{C_3 + C_7 \rho_n}{n \alpha_n^2} + \frac{C_5}{n \alpha_n}  + \frac{C_6}{n\rho_n}\;.
\end{equation}
Plugging \eqref{var_free_entropy_hat} back in \eqref{upperbound_var_free_entropy} gives:
\begin{align}
\Var\bigg(\frac{\ln \mathcal{Z}_{t,\epsilon}}{m_n}\bigg) 
&\leq 
\frac{C_2 + C_4}{n \alpha_n\rho_n} + \frac{C_3 + C_7 \rho_n +1.5}{n \alpha_n^2} + \frac{C_5 + 1.5}{n \alpha_n}  + \frac{C_6}{n\rho_n}\nonumber\\
&\leq 
\frac{C_2 + C_4 + C_5 + 1.5}{n \alpha_n\rho_n} + \frac{C_3 + C_7+1.5}{n \alpha_n^2} + \frac{C_6}{n\rho_n}\nonumber\\
&= 
\frac{1}{n \alpha_n\rho_n} \Big(C_2 + C_4 + C_5 + 1.5+ \frac{\rho_n}{\alpha_n} (C_3 + C_7 + 1.5) + \alpha_n C_6\Big)\;.
\end{align}
The second inequality follows from $\rho_n \leq 1$.
It ends the proof of Proposition~\ref{prop:concentration_free_entropy}.
\section{Concentration of the overlap}\label{appendix-overlap}
In this appendix we prove Proposition~\ref{prop:concentration_overlap}.
Define the average free entropy $f_{n,\epsilon}(t) := \frac1{m_n}\E\ln \mathcal{Z}_{t,\epsilon}$.
In this section we think of it as a function of $R_1=R_1(t,\epsilon)$ and $R_2=R_2(t,\epsilon)$, i.e., $(R_1,R_2)\mapsto f_{n, \epsilon}(t)$.
Similarly, we also view the free entropy for a realization of the quenched variables as a function
\begin{equation*}
(R_1,R_2) \mapsto F_{n, \epsilon}(t) \equiv \frac1{m_n}\ln \cZ_{t,\epsilon}(\bY_t,\bY_t',\bm{\Phi},\bV)\;.
\end{equation*}
In this appendix, to lighten the notations, we drop the indices of the angular brackets $\langle -\rangle_{n,t,\epsilon}$ and simply write $\langle - \rangle$.
We denote with $\cdot$ the scalar product between two vectors. We define: 
\begin{equation*}
\mathcal{L} := \frac{1}{k_n}\bigg(\frac{\|\bx\|^2}{2} - \bx \cdot \bX^* - \frac{\bx\cdot \widetilde{\bZ}}{2\sqrt{R_1}} \bigg)\;.
\end{equation*}
The fluctuations of the overlap $Q := \frac1{k_n} \bX^* \cdot \bx$ and those of $\mathcal{L}$ are related through the inequality:
\begin{align}\label{boundFLuctLQ}
\frac{1}{4}\mathbb{E}\big\langle (Q - \mathbb{E}\langle Q \rangle)^2\big\rangle
\leq \mathbb{E}\big\langle (\mathcal{L} - \mathbb{E}\langle \mathcal{L}\rangle)^2\big\rangle\;.
\end{align}
The proof of $\eqref{boundFLuctLQ}$ is based on integrations by parts with respect to $\widetilde{Z}$ and a repeated use of the Nishimori identity (see Lemma \ref{lemma:nishimori}).
Proposition~\ref{prop:concentration_overlap} is then a direct consequence of the following:
\begin{proposition}[Concentration of $\mathcal{L}$ on $\mathbb{E}\langle \mathcal{L}\rangle$]\label{prop:concentration_L}
Suppose that $\Delta > 0$, that all of \ref{hyp:bounded},~\ref{hyp:c2},~\ref{hyp:phi_gauss2} hold, that $\E_{X \sim P_0}[X^2] = 1$ and that the family of functions $(r_\epsilon)_{\epsilon \in \mathcal{B}_n}$, $(q_\epsilon)_{\epsilon \in \mathcal{B}_n}$ are \textit{regular}.
Further assume that there exist real positive numbers $M_\alpha, M_{\rho/\alpha}, m_{\rho/\alpha}$ such that $\forall n \in \N^*$:
\begin{equation*}
\alpha_n \leq M_{\alpha}
\quad \text{and} \quad
\frac{m_{\rho/\alpha}}{n} < \frac{\rho_n}{\alpha_n}  \leq  M_{\rho/\alpha} \;.
\end{equation*}
Let $(s_n)_{n \in \N^*}$ be a sequence of real numbers in $(0,\nicefrac{1}{2}]$.
Define $\mathcal{B}_n := [s_n,2s_n]^2$.
We have $\forall t \in [0,1]$:
\begin{equation}
\int_{{\cal B}_n} d\epsilon\, 
\E\big\langle (\mathcal{L} - \mathbb{E}\langle \mathcal{L}\rangle_{n, t, \epsilon})^2\big\rangle_{n, t, \epsilon}
\leq 
\frac{C}{\rho_n^2\Big(\frac{\rho_n n}{\alpha_n m_{\rho/\alpha}}\Big)^{\!\frac{1}{3}}-\rho_n^2 }\;,
\end{equation}
where $C$ is a polynomial in $\big(S,\big\Vert \frac{\varphi}{\sqrt{\Delta}} \big\Vert_\infty, \big\Vert \frac{\partial_x \varphi}{\sqrt{\Delta}} \big\Vert_\infty, \big\Vert \frac{\partial_{xx} \varphi}{\sqrt{\Delta}} \big\Vert_\infty, M_\alpha, M_{\rho/\alpha}, m_{\rho/\alpha}\big)$ with positive coefficients.
\end{proposition}
Because
$\E\big\langle (\mathcal{L} - \mathbb{E}\langle \mathcal{L}\rangle)^2\big\rangle
= \E\big\langle (\mathcal{L} - \langle \mathcal{L}\rangle)^2\big\rangle
+ \E\big[(\langle \mathcal{L}\rangle - \mathbb{E}\langle \mathcal{L}\rangle)^2\big]$,
Proposition~\ref{prop:concentration_overlap} follows directly from the next two lemmas.
\begin{lemma}[Concentration of $\mathcal{L}$ on $\langle \mathcal{L}\rangle$]\label{lemma:thermal-fluctuations}
Under the assumptions of Proposition \ref{prop:concentration_L}, $\forall t \in [0,1]$:
\begin{equation*}
 \int_{{\cal B}_n} d\epsilon\, 
  \E \big\langle (\mathcal{L} - \langle \mathcal{L}\rangle_{n,t,\epsilon})^2 \big\rangle_{n,t,\epsilon} 
  \leq \frac{1}{n\rho_n} \,.
\end{equation*}
\end{lemma}
The second lemma states that $\mathcal{L}$ concentrates w.r.t.\ the realizations of quenched disorder variables.
It is a consequence of the concentration of the free entropy (see Proposition~\ref{prop:concentration_free_entropy} in Appendix~\ref{appendix:concentration_free_entropy}).
\begin{lemma}[Concentration of $\langle\mathcal{L}\rangle$ on $\mathbb{E}\langle \mathcal{L}\rangle$]\label{lemma:disorder-fluctuations}
Under the assumptions of Proposition \ref{prop:concentration_overlap}, $\forall t \in [0,1]$:
\begin{equation}\label{integral-form2}
 \int_{{\cal B}_n} d\epsilon\, 
  \E\big[ (\langle \mathcal{L}\rangle_{n,t,\epsilon} - \mathbb{E}\langle \mathcal{L}\rangle_{n,t,\epsilon})^2 \big]
\leq 
\frac{C}{\rho_n^2\Big(\frac{\rho_n n}{\alpha_n m_{\rho/\alpha}}\Big)^{\!\frac{1}{3}}-\rho_n^2 }\;,
\end{equation}
where $C$ is a polynomial in $\big(S,\big\Vert \frac{\varphi}{\sqrt{\Delta}} \big\Vert_\infty, \big\Vert \frac{\partial_x \varphi}{\sqrt{\Delta}} \big\Vert_\infty, \big\Vert \frac{\partial_{xx} \varphi}{\sqrt{\Delta}} \big\Vert_\infty, M_\alpha, M_{\rho/\alpha}, m_{\rho/\alpha}\big)$ with positive coefficients.
\end{lemma}
We now turn to the proof of Lemmas \ref{lemma:thermal-fluctuations} and \ref{lemma:disorder-fluctuations}.
The main ingredient will be a set of formulas for the first two partial derivatives of the free entropy w.r.t.\ $R_1=R_1(t,\epsilon)$.  
For any given realisation of the quenched disorder:
\begin{align}
\frac{dF_{n, \epsilon}(t)}{dR_1}
	&= -\frac{\rho_n}{\alpha_n}\langle \mathcal{L} \rangle-\frac{1}{2m_n}\Big(\|\bX^*\|^2+\frac{\bX^*\cdot \widetilde{\bZ}}{\sqrt{R_1}}\Big) \,,\label{first-derivative}\\
\frac{1}{m_n}\frac{d^2F_{n, \epsilon}(t)}{dR_1^2} 
	&= \Big(\frac{\rho_n}{\alpha_n}\Big)^2(\langle \mathcal{L}^2 \rangle - \langle \mathcal{L} \rangle^2)
	+ \frac{1}{4 m_n^2R_1^{3/2}} \widetilde{\bZ}\cdot(\bX^*-\langle \bx\rangle) \,\label{second-derivative}.
\end{align}
Averaging \eqref{first-derivative} yields:
\begin{equation}
\frac{df_{n,\epsilon}(t)}{d R_1}
= -\frac{\rho_n}{\alpha_n}\Big(\mathbb{E}\langle \mathcal{L} \rangle+\frac{1}{2}\Big)
= \frac{\rho_n}{2\alpha_n}\Big(\frac{\mathbb{E}\|\langle \bx\rangle\|^2}{k_n} -1\Big)\;.\label{first-derivative-average}
\end{equation}
To obtain the second equality we simplified $\E \langle \mathcal{L} \rangle$ by using an integration by parts w.r.t.\ the standard Gaussian random vector $\widetilde{\bZ}$ and $\E\langle \bx\cdot \bX^*\rangle  =  \E\Vert\langle \bx\rangle\Vert^2$ (by Nishimori identity, see Lemma \ref{lemma:nishimori}).
Averaging \eqref{second-derivative} and integrating by parts w.r.t.\ the standard Gaussian random vector $\widetilde{\bZ}$ gives:
\begin{equation}
\frac{1}{m_n}\frac{d^2f_{n,\epsilon}(t)}{dR_1^2}
= \Big(\frac{\rho_n}{\alpha_n}\Big)^2\mathbb{E}[\langle \mathcal{L}^2 \rangle - \langle \mathcal{L} \rangle^2]
-\frac{1}{4m_n^2R_1}  \mathbb{E}\big[\langle \|\bx\|^2\rangle - \|\langle \bx\rangle\|^2\big]\;.\label{average-second-derivative}
\end{equation}
\paragraph{Proof of Lemma \ref{lemma:thermal-fluctuations}}
From \eqref{average-second-derivative} we have:
\begin{align}
\E\big\langle (\mathcal{L} - \langle \mathcal{L} \rangle)^2\big\rangle
& = 
\Big(\frac{\alpha_n}{\rho_n}\Big)^2\frac{1}{m_n}\frac{d^2f_{n, \epsilon}(t)}{dR_1^2}
+\Big(\frac{\alpha_n}{\rho_n}\Big)^2\frac{1}{4m_n^2R_1}\mathbb{E}\big[\langle \|\bx\|^2\rangle - \|\langle \bx\rangle\|^2\big] 
\nonumber \\
&\leq 
\frac{\alpha_n}{\rho_n^2 n}\frac{d^2f_{n,\epsilon}(t)}{dR_1^2} + \frac{1}{4\epsilon_1 n\rho_n}\; ,\label{upperbound_E<(L-<L>)^2>}
\end{align}
where we used $\mathbb{E}\langle \|\bx\|^2\rangle = \mathbb{E}\|\bX^*\|^2=n\rho_n$ by the Nishimori identity and $R_1\ge \epsilon_1$.
Recall ${\cal B}_n := [s_n,2s_n]^2$.
By assumption the families of functions $(q_\epsilon)_{\epsilon \in \mathcal{B}_n}$ and $(r_\epsilon)_{\epsilon \in \mathcal{B}_n}$ are regular.
Therefore, $R^t:(\epsilon_1,\epsilon_2)\mapsto (R_1(t,\epsilon),R_2(t,\epsilon))$ is a $\mathcal{C}^1$-diffeomorphism whose Jacobian determinant $\vert J_{R^t} \vert$ satisfies $\forall \epsilon \in \mathcal{B}_n: \vert J_{R^t}(\epsilon) \vert \geq 1$.
Integrating \eqref{upperbound_E<(L-<L>)^2>} over $\epsilon\in{\cal B}_n$ yields:
\begin{align}
\int_{{\cal B}_n} d\epsilon\,\mathbb{E}\big\langle (\mathcal{L} - \langle \mathcal{L} \rangle)^2\big\rangle
&\leq \frac{\alpha_n}{\rho_n^2 n}\int_{R^t({\cal B}_n)} \frac{dR_1dR_2}{\vert J_{R^t}((R^t)^{-1}(R_1,R_2)) \vert}\,\frac{d^2f_{n,\epsilon}(t)}{dR_1^2} +\frac{1}{4 n\rho_n} \int_{{\cal B}_n}\frac{d\epsilon_1}{\epsilon_1}d\epsilon_2\nonumber\\
&\leq \frac{\alpha_n}{\rho_n^2 n}\int_{R^t({\cal B}_n)} dR_1dR_2\,\frac{d^2f_{n,\epsilon}(t)}{dR_1^2}+\frac{s_n}{4 n\rho_n}\ln 2\;.
\end{align}
Note that $R^t(\mathcal{B}_n)\subset \big[s_n, 2s_n + \frac{\alpha_n}{\rho_n}r_{\max}\big] \times [s_n, 2s_n + 1]$ (by definition of the interpolation functions). Thus:
\begin{align}
\int_{{\cal B}_n} d\epsilon\,\mathbb{E}\big\langle (\mathcal{L} - \langle \mathcal{L} \rangle)^2\big\rangle
&\leq \frac{\alpha_n}{\rho_n^2 n} \int_{s_n}^{2s_n + 1}dR_2
\bigg[\frac{df_{n, \epsilon}(t)}{dR_1}\bigg]_{R_1=s_n}^{2s_n + \frac{\alpha_n}{\rho_n}r_{\max}} +\frac{s_n}{4 n\rho_n}\ln 2\nonumber\\
&\leq \frac{1+s_n}{2\rho_n n}
+\frac{s_n}{4 n\rho_n}\ln 2
\leq \frac{1}{n\rho_n} \;.\label{upperbound_int_E<(L-<L>)^2>}
\end{align}
The last inequality follows from $s_n \leq 1/2$ and $(\ln 2)/2 < 1$.
To obtain the second inequality we bounded the partial derivative of the free entropy using \eqref{first-derivative-average} and $\E\|\langle \bx \rangle \|^2\rangle \leq \E\langle \|\bx\|^2\rangle =n\rho_n$ (again by the Nishimori identity):
\begin{equation}\label{upperbound_df/dR_1}
\bigg\vert \frac{df_{n,\epsilon}(t)}{d R_1} \bigg\vert
= - \frac{df_{n,\epsilon}(t)}{d R_1}
= \frac{\rho_n}{2\alpha_n}\Big(1-\frac{\mathbb{E}\|\langle \bx\rangle\|^2}{k_n}\Big)
\leq \frac{\rho_n}{2\alpha_n}\;.
\end{equation}
\hfill$\blacksquare$
\paragraph{Proof of Lemma~\ref{lemma:disorder-fluctuations}}
We define the two functions:
\begin{align}\label{new-free}
 \widetilde F(R_1) := F_{n, \epsilon}(t) -\frac{\sqrt{R_1}}{m_n} 
 2S\sum_{i=1}^n\vert \widetilde{Z}_i\vert\quad,
 \quad 
 \widetilde f(R_1) := \E \widetilde F(R_1)= f_{n, \epsilon}(t) - \frac{\sqrt{R_1}}{\alpha_n} 2S\, \E \vert \widetilde{Z}_1\vert\;.
\end{align}
Because of \eqref{second-derivative}, we see that the second derivative of $\widetilde F(R_1)$ is positive so that it is convex.
Without the extra term $F_{n,\epsilon}(t)$ is not necessarily convex in $R_1$, although $f_{n,\epsilon}(t)$ is (it can be shown easily).
Note that $\widetilde f(R_1)$ is convex too.
Convexity allows us to use the following standard lemma:
\begin{lemma}[A convexity bound]\label{lemmaConvexity}
Let $G$ and $g$ be two convex functions. Let $\delta>0$ and define $C_\delta(x) \equiv g'(x+\delta) - g'(x-\delta) \geq 0$. Then:
\begin{equation*}
|G'(x) - g'(x)| \leq \delta^{-1} \sum_{u \in \{x-\delta, x, x+\delta\}} |G(u)-g(u)| + C_\delta(x) \;.
\end{equation*}
\end{lemma}
Define $A := \frac{1}{m_n}\sum_{i=1}^n \vert \widetilde{Z}_i\vert -\E\vert \widetilde{Z}_i\vert$.
From \eqref{new-free}, we directly obtain:
\begin{equation}\label{fdiff}
\widetilde F(R_1) - \widetilde f(R_1) = F_{n, \epsilon}(t) - f_{n, \epsilon}(t) - \sqrt{R_1} 2S A \;.
\end{equation} 
Thanks to \eqref{first-derivative} and \eqref{first-derivative-average} the difference of derivatives (w.r.t. $R_1$) reads:
\begin{align}\label{derdiff}
\widetilde F'(R_1) - \widetilde f'(R_1)
= \frac{\rho_n}{\alpha_n}\big(
\E\langle \mathcal{L} \rangle -\langle \mathcal{L} \rangle\big)
-\frac{\rho_n}{2\alpha_n}\bigg(\frac{\|\bX^*\|^2}{k_n} - 1 + \frac{\bX^*\cdot \widetilde{\bZ}}{k_n\sqrt{R_1}}\bigg)
-\frac{SA}{\sqrt{R_1}} \;.
\end{align}
Let $\delta \in (0,s_n)$. Define $C_\delta(R_1) := \widetilde f'(R_1+\delta)-\widetilde f'(R_1-\delta) \geq 0$ (this is well-defined because $\delta < s_n \leq R_1$).
Combining \eqref{fdiff} and \eqref{derdiff} with Lemma~\ref{lemmaConvexity} gives:
\begin{align}
\frac{\rho_n}{\alpha_n} \big\vert \langle \mathcal{L}\rangle - \mathbb{E}\langle \mathcal{L}\rangle\big\vert
&\leq 
\delta^{-1} \sum_{u\in \{R_1 -\delta, R_1, R_1+\delta\}}
 \big\vert \big(F_{n, \epsilon}(t) - f_{n, \epsilon}(t)\big)_{R_1 = u} \big\vert + 2S\vert A \vert \sqrt{u} \nonumber\\
&\qquad\quad
+ C_\delta(R_1) + \frac{S\vert A\vert}{\sqrt R_1}
 + \frac{\rho_n}{2\alpha_n}\bigg\vert \frac{\|\bX^*\|^2}{k_n} - 1 + \frac{\bX^*\cdot \widetilde{\bZ}}{k_n\sqrt{R_1}}\bigg\vert\nonumber\\
&\leq 
 \delta^{-1} \sum_{u\in \{R_1 -\delta, R_1, R_1+\delta\}}
 \big\vert \big(F_{n, \epsilon}(t) - f_{n, \epsilon}(t)\big)_{R_1 = u} \big\vert\nonumber\\
&\qquad\quad
 + C_\delta(R_1) + S\vert A\vert\bigg(\frac{1}{\sqrt R_1} + \frac{6\sqrt{R_1}}{\delta}\bigg)
 + \frac{\rho_n}{2\alpha_n}\bigg\vert \frac{\|\bX^*\|^2}{k_n} - 1 + \frac{\bX^*\cdot \widetilde{\bZ}}{k_n\sqrt{R_1}}\bigg\vert\;.\label{usable-inequ}
\end{align}
The last inequality follows from $\sqrt{R_1 + \delta} + \sqrt{R_1 - \delta} \leq 2\sqrt{R_1}$.
Taking the square and then the expectation on both sides of the inequality \eqref{usable-inequ}, and making use of $(\sum_{i=1}^6 v_i)^2 \leq 6\sum_{i=1}^6 v_i^2$ (by convexity) yields:
\begin{align}
\E\big[\big(\langle \mathcal{L}\rangle - \E\langle \mathcal{L}\rangle\big)^2\big]
&\leq 
\frac{6}{\delta^2}\bigg(\frac{\alpha_n}{\rho_n}\bigg)^{\!\! 2} \sum_{u\in \{R_1 -\delta, R_1, R_1+\delta\}}
\Var\Big(F_{n, \epsilon}(t)\big\vert_{R_1 = u}\Big)
+ \bigg(\frac{\alpha_n}{\rho_n}\bigg)^{\!\! 2} C_\delta(R_1)^2\nonumber\\
&\qquad
+\bigg(\frac{\alpha_n}{\rho_n}\bigg)^{\!\! 2} S^2 \E[A^2]\bigg(\frac{1}{R_1} + \frac{12}{\delta} + \frac{36 R_1}{\delta^2}\bigg)
+\frac{1}{4}\Var\bigg(\frac{\|\bX^*\|^2}{k_n} + \frac{\bX^*\cdot \widetilde{\bZ}}{k_n\sqrt{R_1}}\bigg)\;.\label{upperbound_E(<L> - E<L>)^2}
\end{align}
By Proposition~\ref{prop:concentration_free_entropy}, under our assumptions, the free entropy $F_{n,\epsilon}(t) = \nicefrac{\ln \cZ_{t,\epsilon}}{m_n}$ concentrates such that:
\begin{equation}\label{concentration_free_entropy_proof_overlap}
	\Var\Big(F_{n, \epsilon}(t)\Big)\leq \frac{C}{n \alpha_n\rho_n}
\end{equation}
where $C$ is a polynomial in $\big(S,\big\Vert \frac{\varphi}{\sqrt{\Delta}} \big\Vert_\infty, \big\Vert \frac{\partial_x \varphi}{\sqrt{\Delta}} \big\Vert_\infty, \big\Vert \frac{\partial_{xx} \varphi}{\sqrt{\Delta}} \big\Vert_\infty\big)$ with positive coefficients.
Remark that, by independence of the noise variables, we have:
\begin{equation}
\label{E[A^2]}
\mathbb{E}[A^2]  \le  \frac{1-2/\pi}{n\alpha_n^2}<\frac{1}{n\alpha_n^2} \;.
\end{equation}
Also, the last term on the right hand side of \eqref{upperbound_E(<L> - E<L>)^2} satisfies:
\begin{align}
\Var\bigg(\frac{\|\bX^*\|^2}{k_n} + \frac{\bX^*\cdot \widetilde{\bZ}}{k_n\sqrt{R_1}}\bigg)
=
\Var\bigg(\frac{\|\bX^*\|^2}{k_n}\bigg) + \Var\bigg(\frac{\bX^*\cdot \widetilde{\bZ}}{k_n\sqrt{R_1}}\bigg)
&= \frac{n}{k_n^2}
\Var\big((X_1^*)^2\big) + \frac{n}{k_n^2R_1}\Var\big(X_1^* \widetilde{Z}_1\big)\nonumber\\
&\leq \frac{S^4}{n\rho_n} + \frac{1}{n \rho_n R_1}\;.\label{variance_last_term}
\end{align}
Plugging \eqref{concentration_free_entropy_proof_overlap}, \eqref{E[A^2]} and \eqref{variance_last_term} back in \eqref{upperbound_E(<L> - E<L>)^2} yields:
\begin{equation}\label{final_upperbound_E(<L> - E<L>)^2}
\E\big[\big(\langle \mathcal{L}\rangle - \E\langle \mathcal{L}\rangle\big)^2\big]
\leq 
\frac{18C\alpha_n}{n \rho_n^3 \delta^2}+\frac{S^4}{4n\rho_n}
+ \bigg(\frac{\alpha_n}{\rho_n}\bigg)^{\!\! 2} C_\delta(R_1)^2
+\frac{S^2}{n\rho_n^2}\bigg(\frac{12}{\delta} + \frac{36 R_1}{\delta^2}\bigg)+ \frac{S^2 + 0.25}{n \rho_n R_1}\,.
\end{equation}
The next step is to integrate both sides of \eqref{final_upperbound_E(<L> - E<L>)^2} over $\mathcal{B}_n := [s_n,2s_n]^2$.
By assumption the families of functions $(q_\epsilon)_{\epsilon \in \mathcal{B}_n}$ and $(r_\epsilon)_{\epsilon \in \mathcal{B}_n}$ are regular.
Therefore, $R^t:(\epsilon_1,\epsilon_2)\mapsto (R_1(t,\epsilon),R_2(t,\epsilon))$ is a $\mathcal{C}^1$-diffeomorphism whose Jacobian determinant $\vert J_{R^t} \vert$ satisfies $\forall \epsilon \in \mathcal{B}_n: \vert J_{R^t}(\epsilon) \vert \geq 1$.
Besides, ${R^t(\mathcal{B}_n)\subseteq \big[s_n, K_n\big] \times [s_n, 2s_n + 1]}$ where $K_n :=  2s_n + \frac{\alpha_n}{\rho_n}r_{\max}$.
Therefore:
\begin{align}
\int_{\mathcal{B}_n} d\epsilon \frac{S^2}{n\rho_n^2}\bigg(\frac{12}{\delta} + \frac{36 R_1(t,\epsilon)}{\delta^2}\bigg)
&\leq 
\frac{12S^2}{n\rho_n^2}\int_{\mathcal{B}_n} d\epsilon \bigg(\frac{1}{\delta} + \frac{3K_n}{\delta^2}\bigg)\nonumber\\
&=
\frac{12S^2}{n\rho_n^2}s_n^2\bigg(\frac{1}{\delta} + \frac{3K_n}{\delta^2}\bigg)
\leq 12S^2\big(3.5 M_{\rho/\alpha} + 3r_{\max}\big)
\frac{\alpha_n s_n^2}{n\rho_n^3 \delta^2}\;. \label{1st_int_B_n}
\end{align}
To get the last equality we used that
$\delta + 3K_n = \big((\delta + 6s_n)\frac{\rho_n}{\alpha_n} + 3r_{\max}\big)\frac{\alpha_n}{\rho_n} \leq (3.5 M_{\rho/\alpha} + 3r_{\max})\frac{\alpha_n}{\rho_n}$
because $\delta < s_n \leq \frac{1}{2}$ and $\frac{\rho_n}{\alpha_n} \leq M_{\rho/\alpha}$.
By the change of variables $\epsilon \to (R_1,R_2) = R^t(\epsilon)$, we get:
\begin{align}
\int_{\mathcal{B}_n} d\epsilon \frac{S^2 + 0.25}{n \rho_n R_1(t,\epsilon)}
&=
\frac{S^2 + 0.25}{n \rho_n}\int_{R^t({\cal B}_n)} \frac{dR_1dR_2}{\vert J_{R^t}((R^t)^{-1}(R_1,R_2)) \vert} \frac{1}{R_1}\nonumber\\
&\leq \frac{(S^2 + 0.25)(1+s_n)}{n \rho_n} \int_{s_n}^{2s_n + 1} dR_2 \int_{s_n}^{2s_n + \frac{\alpha_n}{\rho_n}r_{\max}} \frac{dR_1}{R_1}\nonumber\\
&= \frac{(S^2 + 0.25)(1+s_n)}{n \rho_n} \ln(K_n)\nonumber\\
&\leq \frac{1.5(S^2 + 0.25)r_{\max} \alpha_n}{n \rho_n^2}\;.\label{2nd_int_B_n}
\end{align}
The last inequality follows from $\ln K_n \leq \ln(1 + \nicefrac{r_{\max} \alpha_n}{\rho_n}) \leq  \nicefrac{r_{\max} \alpha_n}{\rho_n}$.
It remains to upper bound the integral of $C_\delta(R_1)^2$.
We recall that $\vert C_\delta(R_1) \vert = C_\delta(R_1) = \widetilde{f}'(R_1 + \delta)-\widetilde{f}'(R_1-\delta)$. We have:
\begin{align}\label{bound_fprime}
\vert\widetilde f'(R_1)\vert
\leq \frac{\rho_n}{2\alpha_n}  + \frac{S}{\alpha_n\sqrt R_1} \E\vert \widetilde{Z}_1\vert
\leq \frac{\rho_n}{2\alpha_n}  + \frac{S}{\alpha_n\sqrt s_n}\;.	
\end{align}
The first inequality uses the definition \eqref{new-free} and the upper bound \eqref{upperbound_df/dR_1}.
The second inequality uses ${R_1\geq \epsilon_1 \geq s_n}$ and $\E\vert \widetilde{Z}_1\vert \leq 1$.
This implies $|C_\delta(R_1)| \leq (\rho_n  +2S/\sqrt{s_n -\delta})/\alpha_n$.
Then:
\begin{align*}
&\int_{{\cal B}_n} d\epsilon\,C_\delta(R_1(t,\epsilon))^2\\
&\qquad\leq \frac{1}{\alpha_n}\bigg(\rho_n  + \frac{2S}{\sqrt{s_n-\delta}}\bigg)
 \int_{{\cal B}_n} d\epsilon\, C_\delta(R_1(t,\epsilon))\\
&\qquad=\frac{1}{\alpha_n}\bigg(\rho_n  + \frac{2S}{\sqrt{s_n-\delta}}
 \int_{R^t({\cal B}_n)} \frac{dR_1dR_2}{\vert J_{R^t}((R^t)^{-1}(R_1,R_2)) \vert}\, C_\delta(R_1)\\
&\qquad\leq \frac{1}{\alpha_n}\bigg(\rho_n  +  \frac{2S}{\sqrt{s_n-\delta}}\bigg)
\int_{s_n}^{2s_n + 1} dR_2 \int_{s_n}^{2s_n + \frac{\alpha_n}{\rho_n}r_{\max}} dR_1 C_\delta(R_1)\\
&\qquad\leq  
\frac{1}{\alpha_n}\bigg(\rho_n  +  \frac{2S}{\sqrt{s_n-\delta}}\bigg)\!\int_{s_n}^{2s_n + 1} \!dR_2 \Big(\widetilde f(K_n+\delta) - \widetilde f(K_n-\delta)
+ \widetilde f(s_n-\delta) - \widetilde f(s_n+\delta)\Big)\:.
\end{align*}
By the mean value theorem and the upper bound \eqref{upperbound_df/dR_1}, we have (uniformly in $R_2$):
\begin{equation*}
|\widetilde f(R_1-\delta) - \widetilde f(R_1+\delta)| \leq \frac{2\delta}{\alpha_n}\bigg(\rho_n  + \frac{2S}{\sqrt{s_n-\delta}}\bigg)\;.
\end{equation*}
Therefore:
\begin{align}
\int_{{\cal B}_n} d\epsilon\,\bigg(\frac{\alpha_n}{\rho_n}\bigg)^{\!\! 2}C_\delta(R_1(t,\epsilon))^2
\leq
\frac{4(1+s_n)\delta}{\alpha_n^2}\bigg(\rho_n  + \frac{2S}{\sqrt{s_n-\delta}}\bigg)^{\!\! 2}
&\leq
\frac{4(1+s_n)\delta}{\alpha_n^2}\bigg(\frac{1+2S}{\sqrt{s_n-\delta}}\bigg)^{\!\! 2}\nonumber\\
&\leq
\frac{6(1+2S)^2\delta}{\alpha_n^2(s_n-\delta)}\;.\label{3rd_int_B_n}
\end{align}
Integrating \eqref{final_upperbound_E(<L> - E<L>)^2} over $\epsilon\in {\cal B}_n$ and making use of \eqref{1st_int_B_n}, \eqref{2nd_int_B_n}, \eqref{3rd_int_B_n} yields
(using $\int_{\mathcal{B}_n} d\epsilon = s_n^2$):
\begin{flalign*}
&\int_{{\cal B}_n} d\epsilon\, \E\big[\big(\langle \mathcal{L}\rangle - \E\langle \mathcal{L}\rangle\big)^2\big]&\\
&\;\,\leq
\frac{\alpha_n s_n^2}{n\rho_n^3 \delta^2}
\bigg(18C
+ 12S^2\big(3.5 M_{\rho/\alpha} + 3r_{\max}\big)
+ \frac{S^4}{4}\frac{\delta^2\rho_n^2}{\alpha_n}
+ 1.5(S^2 + 0.25)r_{\max} \frac{\rho_n\delta^2}{s_n^2}\bigg)
+ \frac{6(1+2S)^2}{\rho_n^2\big(\frac{s_n}{\delta}-1 \big)}\;.
\end{flalign*}
Note that $\nicefrac{\delta^2\rho_n^2}{\alpha_n} \leq M_{\rho/\alpha}$ (because $\nicefrac{\rho_n}{\alpha_n} \leq M_{\rho/\alpha}$, $\rho_n \leq 1$ and $\delta\leq 1$) and $\nicefrac{\rho_n \delta^2}{s_n^2} \leq 1$ (because $\rho_n \leq 1$ and $\nicefrac{\delta}{s_n} \leq 1$). Hence, the last upper bound implies:
\begin{equation}\label{upperbound_int_E(<L> - E<L>)^2}
\int_{{\cal B}_n} d\epsilon\, \E\big[\big(\langle \mathcal{L}\rangle - \E\langle \mathcal{L}\rangle\big)^2\big]
\leq
C_1\frac{\alpha_n s_n^2}{n\rho_n^3 \delta^2}
+ C_2\frac{1}{\rho_n^2\big(\frac{s_n}{\delta}-1 \big)}\;,
\end{equation}
where $C_1 := 18C
+ 12S^2\big(3.5 M_{\rho/\alpha} + 3r_{\max}\big)
+ \frac{S^4}{4}M_{\rho/\alpha}
+ 1.5(S^2 + 0.25)r_{\max}$
and $C_2 := 6(1+2S)^2$.
If $\nicefrac{\delta}{s_n}$ vanishes when $n$ goes to infinity (which is required if we want the second term on the right-hand side of \eqref{upperbound_int_E(<L> - E<L>)^2} to vanish) then $\frac{1}{\rho_n^2\big(\frac{s_n}{\delta}-1 \big)} = \Theta\big(\frac{\delta}{\rho_n^2s_n}\big)$.
Further choosing $\delta \propto \big(\frac{\alpha_n}{n\rho_n}\big)^{\frac{1}{3}}s_n$ yields $\frac{\delta}{\rho_n^2s_n} = \Theta\big(\frac{\alpha_n s_n^2}{n\rho_n^3 \delta^2}\big)$, i.e., both terms on the right-hand side of \eqref{upperbound_int_E(<L> - E<L>)^2}  are equivalent.
Note that we can choose $\delta \propto \big(\frac{\alpha_n}{n\rho_n}\big)^{\frac{1}{3}}s_n$ and make sure that $\forall n \in \N^*: \delta \in (0, s_n)$ because there exists $m_{\rho/\alpha}$ such that $\forall n \in \N^*: \nicefrac{\rho_n}{\alpha_n} > \nicefrac{m_{\rho/\alpha}}{n}$.
Plugging the choice $\delta = \big(\frac{m_{\rho/\alpha}\alpha_n}{n\rho_n}\big)^{\frac{1}{3}}s_n$ back in \eqref{upperbound_int_E(<L> - E<L>)^2} ends the proof of the lemma:
\begin{align*}
\int_{{\cal B}_n} d\epsilon\, \E\big[\big(\langle \mathcal{L}\rangle - \E\langle \mathcal{L}\rangle\big)^2\big]
&\leq
\frac{C_1}{m_{\rho/\alpha}}\frac{1}{\rho_n^2\Big(\frac{\rho_n n}{\alpha_n m_{\rho/\alpha}}\Big)^{\!\frac{1}{3}}}
+ C_2\frac{1}{\rho_n^2\Big(\frac{\rho_n n}{\alpha_n m_{\rho/\alpha}}\Big)^{\!\frac{1}{3}}-\rho_n^2 }\\
&\leq
\bigg(\frac{C_1}{m_{\rho/\alpha}} + C_2\bigg)
\frac{1}{\rho_n^2\Big(\frac{\rho_n n}{\alpha_n m_{\rho/\alpha}}\Big)^{\!\frac{1}{3}}-\rho_n^2 }\;.
\end{align*}
\hfill$\blacksquare$
\section{Proof of Proposition \ref{prop:ode}}\label{appendix:properties_ode}
Before proving the proposition, we recall a few definitions for reader's convenience.
We suppose that \ref{hyp:bounded},~\ref{hyp:c2},~\ref{hyp:phi_gauss2} hold and that $\Delta = \E_{X \sim P_0}[X^2] = 1$.
For all $n \in \N^*$, we define the interval $\mathcal{B}_n := [s_n, 2s_n]$ where $(s_n)_{n \in \N^*}$ is a sequence that takes its values in $(0,\nicefrac{1}{2}]$.
Let $r_{\max} := -2\,\nicefrac{\partial I_{P_{\mathrm{out}}}}{\partial q}\big\vert_{q=1,\rho=1}$ a nonnegative real number.
We have $X_i^* \iid P_{0,n}$, $\bA_\mu \iid P_A$ and $\Phi_{\mu i}, V_{\mu}, W_{\mu}^*, Z_\mu, \widetilde{Z}_i \iid \cN(0,1)$ for $i=1\dots n$ and $\mu=1\dots m_n$.
For fixed $t \in [0,1]$ and $R = (R_1, R_2) \in [0,+\infty) \times [0,t+2s_n]$, consider the observations:
\begin{align*}
\begin{cases}
Y_{\mu}^{(t,R_2)}  &= \varphi\big(S_{\mu}^{(t,R_2)}, \bA_\mu\big) + Z_\mu \,,\; 1 \leq \mu \leq m_n\\
&\sim \;P_{\mathrm{out}}\Big(\,\cdot\, \Big\vert \, S_{\mu}^{(t,R_2)}\,\Big)\\
\widetilde{Y}_{i}^{(t,R_1)} &= \;\sqrt{R_1}\, X^*_i + \widetilde{Z}_i\qquad\;\;\:, \; 1 \leq \, i \,  \leq \, n
\end{cases}\;\;;
\end{align*}
where $S_{\mu}^{(t,R_2)} = S_{\mu}^{(t,R_2)}(\bX^*,W_\mu^*) := \sqrt{\frac{1-t}{k_n}}\, (\bm{\Phi} \bX^*)_\mu  + \sqrt{R_2} \,V_{\mu} + \sqrt{t+2s_n-R_2} \,W_{\mu}^*$.
The joint posterior density of $(\bX^*,\bW^*)$ given $(\bY^{(t,R_2)},\widetilde{\bY}^{(t,R_1)},\bm{\Phi},\bV)$ is:
\begin{multline*}
dP(\bx,\bw \vert \bY^{(t,R_2)},\widetilde{\bY}^{(t,R_1)},\bm{\Phi},\bV)\\
= \frac{1}{\cZ_{t,R}}\prod_{i=1}^{n}dP_{0,n}(x_i)\,e^{-\frac{1}{2}\big(\sqrt{R_{1}}x_i -\widetilde{Y}_i^{(t,R_1)}\big)^2}\,
\prod_{\mu=1}^{m_n} \frac{dw_\mu}{\sqrt{2\pi}} e^{-\frac{w_\mu^2}{2}} P_{\mathrm{out}}(Y_{\mu}^{(t,R_2)}\vert S_{\mu}^{(t,R_2)}(\bx,w_\mu)) \;,
\end{multline*}
where $\cZ_{t,R}$ is the normalization.
The angular brackets $\langle - \rangle_{n,t,R}$ denotes the expectation w.r.t.\ this posterior.
The scalar overlap is the quantity $Q := \frac{1}{k_n} \sum_{i=1}^{n} X_i^* x_i$. We define:
\begin{equation*}
F_2^{(n)}(t, R) := \E \langle Q \rangle_{n,t,R}
\quad \text{and} \quad
F_1^{(n)}(t, R) := -2\frac{\alpha_n}{\rho_n} \frac{\partial I_{P_{\mathrm{out}}}}{\partial q}\bigg\vert_{q = \E \langle Q \rangle_{n,t,R}, \rho=1}\;.
\end{equation*}
We now repeat and prove Proposition~\ref{prop:ode}.
\begin{proposition_ode}
Suppose that \ref{hyp:bounded},~\ref{hyp:c2},~\ref{hyp:phi_gauss2} hold and that $\Delta = \E_{X \sim P_0}[X^2] = 1$.
For all $\epsilon \in \mathcal{B}_n$, there exists a unique global solution $R(\cdot,\epsilon): [0,1] \to [0,+\infty)^2 $ to the second-order ODE:
\begin{equation*}
y'(t) = \big(F_1^{(n)}(t,y(t)), F_2^{(n)}(t, y(t))\big) \quad, \quad y(0)=\epsilon\;.
\end{equation*}
This solution is continuously differentiable and its derivative $R'(\cdot,\epsilon)$ satisfies:
\begin{equation*}
R'([0,1],\epsilon) \subseteq \bigg[0, \frac{\alpha_n}{\rho_n}r_{\max}\bigg] \times [0,1]\;.
\end{equation*}
Besides, for all $t \in [0,1]$, $R(t,\cdot)$ is a $\mathcal{C}^1$-diffeomorphism from $\mathcal{B}_n$ onto its image whose Jacobian determinant is greater than, or equal to, one:
\begin{equation*}
\forall \,\epsilon \in \mathcal{B}_n: \det J_{R(t,\cdot)}(\epsilon) \geq 1 \:,
\end{equation*}
where $J_{R(t,\cdot)}$ denotes the Jacobian matrix of $R(t,\cdot)$.\\
Finally, the same statement holds if, for a fixed $r \in [0,r_{\max}]$, we instead consider the second-order ODE:
\begin{equation*}
y'(t) = \bigg(\frac{\alpha_n}{\rho_n}r\,, F_2^{(n)}(t, y(t))\bigg) \quad, \quad y(0)=\epsilon\;.
\end{equation*}
\end{proposition_ode}
\begin{proof}
We only give the proof for the ODE $y' = \big(F_1^{(n)}(t,y), F_2^{(n)}(t, y)\big)$ since the one for the ODE $y' = \big(\nicefrac{\alpha_nr}{\rho_n} , F_2^{(n)}(t,y)\big)$ is simpler and follows the same arguments.

By Jensen's inequality and Nishimori identity (see Lemma~\ref{lemma:nishimori}):
\begin{equation*}
\E\langle Q \rangle_{n,t,R} = \frac{\E\Vert\langle \bx\rangle_{n,t,R}\Vert^2}{k_n}
\leq \frac{\E\langle \Vert\bx\Vert^2 \rangle_{n,t,R}}{k_n} = \frac{\E\,\Vert\bX^* \Vert^2}{k_n} = 1 \;,
\end{equation*}
i.e., $\E\langle Q \rangle_{n,t,R} \in [0,1]$.
By Lemma \ref{lemma:property_I_Pout}, the function $q \mapsto I_{P_{\mathrm{out}}}(q,1)$ is continuously twice differentiable, concave and nonincreasing on $[0,1]$.
Therefore, $q \mapsto -2\nicefrac{\partial I_{P_{\mathrm{out}}}}{\partial q}\big\vert_{q, \rho=1}$ is nonnegative and nondecreasing on $[0,1]$, which implies
$-2\nicefrac{\partial I_{P_{\mathrm{out}}}}{\partial q}\big\vert_{q, \rho=1} \in [0,r_{\max}]$. 
We have thus shown that the function $F: (t, R) \mapsto (F_1^{(n)}(t,R), F_2^{(n)}(t,R))$ is defined on all
\begin{equation*}
\mathcal{D}_n := \Big\{(t,R_1,R_2) \in [0,1] \times [0,+\infty)^2: R_2 \leq t+2s_n \Big\}\;,
\end{equation*}
and takes its values in $[0,\nicefrac{\alpha_n r_{\max}}{\rho_n}] \times [0,1]$.

To invoke Cauchy-Lipschitz theorem, we have to check that $F$ is continuous in $t$ and uniformly Lipschitz continuous in $R$ (meaning the Lipschitz constant is independent of $t$).
We can show that $F$ is continuous on $\mathcal{D}_n$ and that, for all  $t \in [0,1]$, $F(t,\cdot)$ is differentiable on ${(0,+\infty)} \times {(0,t+2s_n)}$ thanks to the standard theorems of continuity and differentiation under the integral sign.
The domination hypotheses are indeed verified because we assume that \ref{hyp:bounded},~\ref{hyp:c2} hold.
To check the uniform Lipschitzianity, we show that the Jacobian matrix $J_{F(t,\cdot)}(R)$ of $F(t,\cdot)$ is uniformly bounded in $(t, R)$. For all $(R_1, R_2) \in {(0,+\infty)} \times {(0,t+2s_n)}$, we have:
\begin{equation}\label{jacobian_matrix_F}
J_{F(t,\cdot)}(R) =
\begin{bmatrix}
c(t,R) & c(t,R)\\
1 & 1
\end{bmatrix}
\begin{bmatrix}
\frac{\partial F_2^{(n)}}{\partial R_1}\Big\vert_{t,R} & 0\\
0 & \frac{\partial F_2^{(n)}}{\partial R_2}\Big\vert_{t,R}
\end{bmatrix}\:,
\end{equation}
with $c(t,R) := -2\frac{\alpha_n}{\rho_n} \frac{\partial^2 I_{P_{\mathrm{out}}}}{\partial q^2}\Big\vert_{q=F_2^{(n)}(t,R), \rho=1}$ and
\begin{align}
\frac{\partial F_2^{(n)}}{\partial R_1}\bigg\vert_{t,R}
&= \frac{1}{k_n} \sum_{i,j=1}^{n} \E\big[\big(\langle x_i x_j \rangle_{n,t,R} - \langle x_i \rangle_{n,t,R} \langle x_j \rangle_{n,t,R}\big)^2\,\big]\;;\label{dF_2/dR_1}\\
\frac{\partial F_2^{(n)}}{\partial R_2}\bigg\vert_{t,R}
&= \frac{1}{k_n} \sum_{\mu=1}^{m_n} \E\Big[ \Big\Vert \big\langle u'_{Y_\mu^{(t,R)}}(s_\mu^{(t,R)})  \bx \big\rangle_{n,t,R} - \big\langle u'_{Y_\mu^{(t,R)}}(s_\mu^{(t,R)}) \big\rangle_{n,t,R} \big\langle \bx \big\rangle_{n,t,R}\Big\Vert^2 \,\Big]  \;.\label{dF_2/dR_2}
\end{align}
The function $u'_y(\cdot)$ is the derivative of $u_y:x \mapsto \ln P_{\mathrm{out}}(y \vert x)$.
Both $\nicefrac{\partial F_2^{(n)}}{\partial R_1}$ and $\nicefrac{\partial F_2^{(n)}}{\partial R_2}$ are clearly nonnegative.
Using the assumption \ref{hyp:bounded}, we easily obtain from \eqref{dF_2/dR_1} that
\begin{equation}\label{bound_dF_2dR_1}
0 \leq \frac{\partial F_2^{(n)}}{\partial R_1}\bigg\vert_{t,R} \leq \frac{4S^4n}{\rho_n} \;.
\end{equation}
In the proof of Lemma~\ref{lemma:property_I_Pout}, under the hypothesis \ref{hyp:c2} we obtain the upper bound \eqref{upperbound_u_y(x)} on $\vert u_y'(x) \vert$.
It yields $\forall x \in \R: \big\vert u_{Y_{\mu}^{(t,R)}}^\prime(x) \big\vert
\leq (2\Vert \varphi \Vert_\infty + \vert Z_\mu \vert) \Vert \partial_x \varphi \Vert_\infty$.
Then, we easily see from \eqref{dF_2/dR_1} that
\begin{equation}\label{bound_dF_2dR_2}
0 \leq \frac{\partial F_2^{(n)}}{\partial R_2}\bigg\vert_{t,R}
\leq 8S^2(4\Vert \varphi \Vert_\infty^2 + 1) \Vert \partial_x \varphi \Vert_\infty^2 \frac{\alpha_n n}{\rho_n} \;.
\end{equation}
Finally, by Lemma \ref{lemma:property_I_Pout}, $q \mapsto -\frac{\partial^2 I_{P_{\mathrm{out}}}}{\partial q^2}\big\vert_{q,\rho = 1}$ is nonnegative continuous on the interval $[0,1]$, so it is bounded by a constant $C$ and $c(t,R) \in [0, 2\nicefrac{C \alpha_n}{\rho_n}]$.
Combining the later with \eqref{jacobian_matrix_F}, \eqref{bound_dF_2dR_1} and \eqref{bound_dF_2dR_2} shows that $J_{F(t,\cdot)}(R)$ is uniformly bounded in $(t,R) \in \big\{(t,R_1,R_2) \in [0,1] \times (0,+\infty)^2: R_2 < t+2s_n \big\}$. By the mean-value theorem, this implies that $F$ is uniformly Lipschitz continuous in $R$.

By the Cauchy-Lipschitz theorem, for all $\epsilon \in \mathcal{B}_n$ there exists a unique solution to the initial value problem $y' = F(t, y)$, $y(0)=\epsilon$ that we denote ${R(\cdot, \epsilon): [0,\delta] \to [0,+\infty)^2}$.
Here $\delta \in [0,1]$ is such that $[0,\delta]$ is the maximal interval of existence of the solution.
Because $F$ has its image in $[0, \nicefrac{\alpha_n r_{\max}}{\rho_n}] \times [0, 1]$, we have that $\forall t \in [0,\delta]: R(t,\epsilon) \in [s_n, 2s_n + t \alpha_n r_{\max}/\rho_n] \times [s_n, 2s_n + t]$, which means that $\delta = 1$ (the solution never leaves the domain of definition of $F$).

Each initial condition $\epsilon \in \mathcal{B}_n$ is tied to a unique solution $R(\cdot,\epsilon)$. This implies that the function $\epsilon \mapsto R(t,\epsilon)$ is injective. Its Jacobian determinant is given by Liouville's formula \cite[Chapter V, Corollary 3.1]{Hartman2002Ordinary}:
\begin{align*}
\det J_{R(t,\cdot)}(\epsilon)
&= \exp \int_0^t ds \, \bigg(\frac{\partial F_1^{(n)}}{\partial R_1} + \frac{\partial F_2^{(n)}}{\partial R_2}\bigg)\bigg\vert_{s,R(s,\epsilon)}\\
&= \exp \int_0^t ds \, \Bigg(c\big(s,R(s,\epsilon)\big)  \frac{\partial F_2^{(n)}}{\partial R_1}\bigg\vert_{s,R(s,\epsilon)} + \frac{\partial F_2^{(n)}}{\partial R_2}\bigg\vert_{s,R(s,\epsilon)}\Bigg).
\end{align*}
This Jacobian determinant is greater than, or equal to, one since we saw that all of $c(t,R)$, $\nicefrac{\partial F_1^{(n)}}{\partial R_1}$ and $\nicefrac{\partial F_2^{(n)}}{\partial R_2}$ are nonnegative.
The fact that the Jacobian determinant is bounded away from $0$ uniformly in $\epsilon$ implies by the inverse function theorem that the injective function  $\epsilon \mapsto R(t,\epsilon)$ is a $\mathcal{C}^1$-diffeomorphism from $\mathcal{B}_n$ onto its image.
\end{proof}
\section{Proof of Theorem~\ref{theorem:limit_MI_discrete_prior} for a general discrete prior with finite support}\label{appendix:specialization_discrete_prior}
In the whole appendix we assume that $P_{0,n} \coloneqq (1-\rho_n) \delta_0 + \rho_n P_{0}$ where $P_0$ is a discrete distribution with finite support $\mathrm{supp}(P_0) \subseteq \{\pm v_1,\dots, \pm v_K\}$ with $0 < v_1 < v_2 < \dots < v_K$. For all $i$, $P_0(v_i) = p_i^{\scriptscriptstyle +}, P_0(-v_i) = p_i^{\scriptscriptstyle -}$ with $p_i^{\scriptscriptstyle +},p_i^{\scriptscriptstyle -} \geq 0$ and $p_i \coloneqq p_i^{\scriptscriptstyle +} + p_i^{\scriptscriptstyle -} > 0$.
Of course, $\sum_{i=1}^{K} p_i = 1$.
Note that the second moment of $X \sim P_0$ is $\E[X^2] = \sum_{j=1}^K p_j v_j^2$.
 
For $\rho_n, \alpha_n > 0$ we denote the variational problem appearing in Theorem~\ref{th:RS_1layer} by
\begin{equation}
I(\rho_n, \alpha_n) \coloneqq \adjustlimits{\inf}_{q \in [0,\E X^2]} {\sup}_{r \geq 0}\; i_{\scriptstyle{\mathrm{RS}}}(q, r; \alpha_n, \rho_n)\;,
\end{equation}
where the potential $i_{\scriptstyle{\mathrm{RS}}}$ is defined in \eqref{def_i_RS}.
Let $X^* \sim P_{0,n} \perp Z \sim \cN(0,1)$.
We define for all $r \geq 0$:
\begin{align}
\psi_{P_{0,n}}(r)
&\coloneqq \E \Big[\ln \int dP_{0,n}(x) e^{-\frac{r}{2}x^2 + rX^*x + \sqrt{r}xZ} \Big]\label{def:psi_P0n}\\
&\;= \E \Big[\ln\Big(1 - \rho_n + \rho_n \sum_{i=1}^K e^{-\frac{rv_i^2}{2}} \big(p_i^{\scriptscriptstyle +} e^{rX^*v_i + \sqrt{r} Zv_i}
+ p_i^{\scriptscriptstyle -} e^{-rX^*v_i -\sqrt{r} Zv_i}\big)\Big)\Big]\;.\nonumber
\end{align}
Note that $I_{P_{0,n}}(r) \coloneqq I(X^*;\sqrt{r}\,X^* + Z) = \frac{r\rho_n\E[X^2]}{2} - \psi_{P_{0,n}}(r)$ where $X \sim P_0$ so
\begin{equation}
I(\rho_n, \alpha_n) =
\inf_{q \in [0,\E X^2]} I_{P_{\mathrm{out}}}(q, \E X^2) + \sup_{r \geq 0} \bigg\{\frac{rq}{2} - \frac{1}{\alpha_n} \psi_{P_{0,n}}\bigg(\frac{\alpha_n}{\rho_n}r\bigg)\bigg\}\;.
\end{equation}
The latter expression for $I(\rho_n, \alpha_n)$ is easier to work with.
We point out that $\psi_{P_{0,n}}$ is twice differentiable, nondecreasing, strictly convex and $\frac{\rho_n \E X^2}{2}$-Lipschitz on $[0,+\infty)$ (see Lemma~\ref{lemma:property_I_P0n}) while $I_{P_{\mathrm{out}}}(\cdot,\E X^2)$ is nonincreasing and concave on $[0,\E X^2]$ (see \cite[Appendix B.2, Proposition 18]{Barbier2019Optimal}).

Our goal is now to compute the limit of $I(\rho_n, \alpha_n)$ when $\alpha_n \coloneqq \gamma \rho_n \vert \ln \rho_n \vert$ for a fix $\gamma > 0$ and $\rho_n \to 0$.
We first look where the supremum over $r$ is reached depending on the value of $q \in [0,\E X^2]$.
%
\begin{lemma}\label{lemma:location_r*(q)_general}
Let  $P_{0,n} \coloneqq (1-\rho_n)\delta_0 + \rho_n P_0$ where $P_0$ is a discrete distribution with finite support $\mathrm{supp}(P_0) \subseteq \{\pm v_1, \pm v_2, \dots, \pm v_K\}$ with $0 < v_1 < v_2 < \dots < v_K$.
Let $\alpha_n \coloneqq \gamma \rho_n \vert \ln \rho_n \vert$ for a fix $\gamma > 0$.
Define $g_{\rho_n}: r \in (0,+\infty) \mapsto \frac{2}{\rho_n}\psi_{P_{0,n}}^\prime\big(\frac{\alpha_n}{\rho_n} r\big)$ and $\forall \rho_n \in (0,e^{-1}),\forall j \in \{1,\dots,K\}:$
\begin{equation}\label{definition_a_rhon_k_and_b_rhon_k}
a_{\rho_n}^{(j)} \coloneqq g_{\rho_n}\bigg(\frac{2(1 - \vert \ln \rho_n \vert^{-\frac14})}{\gamma v_j^2}\bigg)
\quad,\quad
b_{\rho_n}^{(j)} \coloneqq g_{\rho_n}\bigg(\frac{2(1 + \vert \ln \rho_n \vert^{-\frac14})}{\gamma v_j^2}\bigg) \;.
\end{equation}
Let  $X \sim P_0$. For $\rho_n$ small enough we have
\begin{equation}\label{ordering_a_and_b}
\rho_n \E[X]^2 < a_{\rho_n}^{(K)} < b_{\rho_n}^{(K)} < a_{\rho_n}^{(K-1)} < b_{\rho_n}^{(K-1)} < \dots < a_{\rho_n}^{(1)} < b_{\rho_n}^{(1)} < \E[X^2]\;,
\end{equation}
and for all $j \in \{1,\dots,K\}:$
\begin{equation}\label{limits_a_and_b}
\lim_{\rho_n \to 0} a_{\rho_n}^{(j)} = \E[X^2 \bm{1}_{\{\vert X \vert > v_j\}}]
\quad;\quad
\lim_{\rho_n \to 0} b_{\rho_n}^{(j)} = \E[X^2 \bm{1}_{\{\vert X \vert \geq v_j\}}] \;.
\end{equation}
Besides, for every $q \in (\rho_n \E[X]^2, \E[X^2])$ there exists a unique $r_{n}^*(q) \in (0,+\infty)$ such that
\begin{equation}
\frac{r_{n}^*(q)q}{2} - \frac{1}{\alpha_n} \psi_{P_{0,n}}\bigg(\frac{\alpha_n}{\rho_n}r_{n}^*(q)\bigg)
= \sup_{r \geq 0} \: \frac{rq}{2} - \frac{1}{\alpha_n} \psi_{P_{0,n}}\bigg(\frac{\alpha_n}{\rho_n}r\bigg)\;,
\end{equation}
and $\forall j \in \{1,\dots,K\}, \forall q \in [a_{\rho_n}^{(j)}, b_{\rho_n}^{(j)}]$:
\begin{equation}\label{bounds_r_n^*(q)}
\frac{2(1-\vert \ln \rho_n \vert^{-\frac14})}{\gamma v_j^2}
\leq r_{n}^*(q)
\leq \frac{2(1+\vert \ln \rho_n \vert^{-\frac14})}{\gamma v_j^2}\;.
\end{equation}
The bounds \eqref{bounds_r_n^*(q)} are tight, namely, $r_n^*(a_{\rho_n}^{(j)}) = \frac{2(1-\vert \ln \rho_n \vert^{-\frac14})}{\gamma v_j^2}, r_n^*(b_{\rho_n}^{(j)}) = \frac{2(1+\vert \ln \rho_n \vert^{-\frac14})}{\gamma v_j^2}$.
\end{lemma}
\begin{proof}
For every $q \in (0,1)$ we define $f_{\rho_n,q}: r \in [0,+\infty) \mapsto \frac{rq}{2} - \frac{1}{\alpha_n} \psi_{P_{0,n}}\big(\frac{\alpha_n}{\rho_n}r\big)$ whose supremum over $r$ we want to compute.
The derivative of $f_{\rho_n,q}$ with respect to $r$ reads
\begin{equation}
f'_{\rho_n,q}(r) = \frac{q}{2} - \frac{1}{\rho_n} \psi'_{P_{0,n}}\bigg(\frac{\alpha_n}{\rho_n}r\bigg)\;.
\end{equation}
The derivative $\psi'_{P_{0,n}}$ is continuously increasing and thus one-to-one from $(0,+\infty)$ onto $(\rho_n^2 \E[X]^2 /2, \rho_n \E[X^2] /2)$.
Therefore, if $q \in (0,\rho_n \E[X]^2 ]$ then $f'_{\rho_n,q} \leq 0$ and the supremum of $f_{\rho_n,q}$ is achieved at $r=0$.
On the contrary, if $q \in (\rho_n, \E[X^2])$ then there exists a unique solution  $r_{n}^*(q) \in (0,+\infty)$ to the critical point equation $f'_{\rho_n,q}(r) = 0$.
As $f_{\rho_n,q}$ is concave ($\psi_{P_0,n}$ is convex) this solution $r_{n}^*(q)$ is the global maximum of $f_{\rho_n,q}$.
We now transform the critical point equation:
\begin{equation}
f_{\rho_n,q}(r) = 0 \Leftrightarrow  \frac{2}{\rho_n} \psi'_{P_{0,n}}\bigg(\frac{\alpha_n}{\rho_n}r\bigg) = q
\Leftrightarrow g_{\rho_n}(r) = q \;,
\end{equation}
where $g_{\rho_n}: r \mapsto \frac{2}{\rho_n}\psi_{P_{0,n}}^\prime\big(\frac{\alpha_n}{\rho_n} r\big)$ is continuously increasing and one-to-one from $(0,+\infty)$ to $(\rho_n \E X^2, \E X^2)$.

By definition of $a_{\rho_n}^{(j)}$ and $b_{\rho_n}^{(j)}$, $r_n^*(a_{\rho_n}^{(j)}) = \nicefrac{2\big(1-\vert \ln \rho_n \vert^{-\frac14}\big)}{\gamma v_j^2}$ and $r_n^*(b_{\rho_n}^{(j)}) = \nicefrac{2\big(1+\vert \ln \rho_n \vert^{-\frac14}\big)}{\gamma v_j^2}$.
Besides, if $q = g_{\rho_n}(r_n^*(q)) \in [a_{\rho_n}^{(j)}, b_{\rho_n}^{(j)}]$ then
\begin{equation*}
	\frac{2(1-\vert \ln \rho_n \vert^{-\frac14})}{\gamma v_j^2}
	\leq r_{n}^*(q)
	\leq \frac{2(1+\vert \ln \rho_n \vert^{-\frac14})}{\gamma v_j^2}
\end{equation*}
as $g_{\rho_n}$ is increasing.
Because $g_{\rho_n}$ is increasing with $0 < v_1 < \dots < v_k$, it is clear that we have the ordering \eqref{ordering_a_and_b} provided that $\rho_n$ is close enough to $0$.

It remains to prove the limits \eqref{limits_a_and_b}.
In order to so, we first rewrite the derivative of $\psi_{P_{0,n}}$.
For all $r \geq 0$, we have:
\begin{align*}
&\psi_{P_{0,n}}^\prime(r)
= \frac{1}{2}\E \Bigg[X^* \frac{\rho_n \sum_{i=1}^K v_i e^{-\frac{rv_i^2}{2}}\big(p_i^{\scriptscriptstyle +} e^{rX^*v_i + \sqrt{r} Zv_i}
	- p_i^{\scriptscriptstyle -} e^{-rX^*v_i -\sqrt{r} Zv_i}\big)}{1 - \rho_n + \rho_n \sum_{i=1}^K e^{-\frac{rv_i^2}{2}} \big(p_i^{\scriptscriptstyle +} e^{rX^*v_i + \sqrt{r} Zv_i}
	+ p_i^{\scriptscriptstyle -} e^{-rX^*v_i -\sqrt{r} Zv_i}\big)}\Bigg]\\
&\;= \frac{\rho_n^2}{2} \sum_{j=1}^K p_j^{\scriptscriptstyle +} v_j\E\Bigg[ \frac{\sum_{i=1}^K v_i e^{-\frac{rv_i^2}{2}}\big(p_i^{\scriptscriptstyle +} e^{rv_iv_j + \sqrt{r} Zv_i}
	- p_i^{\scriptscriptstyle -} e^{-r v_i v_j -\sqrt{r} Zv_i}\big)}{1 - \rho_n + \rho_n \sum_{i=1}^K e^{-\frac{rv_i^2}{2}} \big(p_i^{\scriptscriptstyle +} e^{r v_i v_j + \sqrt{r} Zv_i}
	+ p_i^{\scriptscriptstyle -} e^{-r v_i v_j -\sqrt{r} Zv_i}\big)}\Bigg]\\
&\quad
+ \frac{\rho_n^2}{2} \sum_{j=1}^K p_j^{\scriptscriptstyle -} v_j \E\Bigg[\frac{\sum_{i=1}^K v_i e^{-\frac{rv_i^2}{2}}\big(p_i^{\scriptscriptstyle -} e^{r v_i v_j + \sqrt{r} Zv_i} - p_i^{\scriptscriptstyle +} e^{-rv_iv_j -\sqrt{r} Zv_i}\big)}{1 - \rho_n + \rho_n \sum_{i=1}^K e^{-\frac{rv_i^2}{2}} \big(p_i^{\scriptscriptstyle -} e^{r v_i v_j + \sqrt{r} Zv_i}
	+ p_i^{\scriptscriptstyle +} e^{-r v_i v_j -\sqrt{r} Zv_i}\big)}\Bigg]\\
&\;= \frac{\rho_n}{2} \sum_{j=1}^K \E\!\!\left[\frac{p_j^{\scriptscriptstyle +} v_j \sum_{i=1}^K v_i  e^{-\frac{r(v_i-v_j)^2}{2} + \sqrt{r} Z (v_i - v_j)}\big(p_i^{\scriptscriptstyle +}
	- p_i^{\scriptscriptstyle -} e^{-2r v_i v_j -2\sqrt{r} Zv_i}\big)}{\frac{1 - \rho_n}{\rho_n}e^{-\frac{r v_j^2}{2} -\sqrt{r} Z v_j} + \sum_{i=1}^K  e^{-\frac{r(v_i-v_j)^2}{2} + \sqrt{r} Z (v_i-v_j)}\big(p_i^{\scriptscriptstyle +} 
	+ p_i^{\scriptscriptstyle -} e^{-2r v_i v_j -2\sqrt{r} Zv_i}\big)}\right]\\
&\quad
+ \frac{\rho_n}{2} \sum_{j=1}^K  \E\!\!\left[\frac{p_j^{\scriptscriptstyle -} v_j\sum_{i=1}^K v_i e^{-\frac{r(v_i-v_j)^2}{2} + \sqrt{r} Z (v_i-v_j)}\big(p_i^{\scriptscriptstyle -} 
	- p_i^{\scriptscriptstyle +} e^{-2r v_i v_j -2\sqrt{r} Zv_i}\big)}{\frac{1 - \rho_n}{\rho_n}e^{-\frac{r v_j^2}{2} -\sqrt{r} Z v_j} + \sum_{i=1}^K e^{-\frac{r(v_i-v_j)^2}{2} + \sqrt{r} Z (v_i-v_j)}\big(p_i^{\scriptscriptstyle -} 
	+ p_i^{\scriptscriptstyle +} e^{-2r v_i v_j -2\sqrt{r} Zv_i}\big)}\right].
\end{align*}
The latter expression is shorten to
\begin{equation}\label{expression_derivative_psi_h}
\psi_{P_{0,n}}^\prime(r)
= \frac{\rho_n}{2} \sum_{j=1}^K p_j^{\scriptscriptstyle +} v_j \E\big[h(Z,r,v_j;\rho_n, \bv, \bp^{\scriptscriptstyle +},\bp^{\scriptscriptstyle -})\big]
+ p_j^{\scriptscriptstyle -} v_j \E\big[h(Z,r,v_j;\rho_n, \bv,\bp^{\scriptscriptstyle -},\bp^{\scriptscriptstyle +})\big]\;;
\end{equation}
where $\bv \coloneqq (v_1,v_2,\dots,v_K)$,
$\bp^{\scriptscriptstyle +} \coloneqq (p_1^{\scriptscriptstyle +},p_2^{\scriptscriptstyle +},\dots,p_K^{\scriptscriptstyle +})$,
$\bp^{\scriptscriptstyle -} \coloneqq (p_1^{\scriptscriptstyle -},p_2^{\scriptscriptstyle -},\dots,p_K^{\scriptscriptstyle -})$
and we define $\forall (z,r,u) \in \R \times [0,+\infty) \times (0,+\infty)$:
\begin{multline}
h(z,r,u;\rho_n, \bv, \bp^{\scriptscriptstyle +},\bp^{\scriptscriptstyle -})\\
\coloneqq\frac{\sum_{i=1}^K v_i e^{-\frac{r(v_i-u)^2}{2} + \sqrt{r} z (v_i-u)}\big(p_i^{\scriptscriptstyle +} 
	- p_i^{\scriptscriptstyle -} e^{-2r v_i u -2\sqrt{r} z v_i}\big)}{\frac{1 - \rho_n}{\rho_n}e^{-\frac{r u^2}{2} -\sqrt{r} z u} + \sum_{i=1}^K e^{-\frac{r(v_i-u)^2}{2} + \sqrt{r} z (v_i-u)}\big(p_i^{\scriptscriptstyle +} 
	+ p_i^{\scriptscriptstyle -} e^{-2r v_i u -2\sqrt{r} z v_i}\big)}\;.
\end{multline}
Note that $\forall z \in \R:$
\begin{align}
h\bigg(z,\frac{2(1+\vert \ln \rho_n \vert^{-\frac14})\vert \ln \rho_n \vert}{v_k^2},v_j;\rho_n, \bv, \bp^{\pm},\bp^{\mp}\bigg)
\xrightarrow[\rho_n \to 0]{}
\begin{cases}
0 \;\;\, \text{if } j < k\;;\\
v_j \; \text{if } j \geq k\;.\\
\end{cases}\label{limit_h_+}\\
h\bigg(z,\frac{2(1-\vert \ln \rho_n \vert^{-\frac14})\vert \ln \rho_n \vert}{v_k^2},v_j;\rho_n, \bv, \bp^{\pm},\bp^{\mp}\bigg)
\xrightarrow[\rho_n \to 0]{}
\begin{cases}
0 \;\;\, \text{if } j \leq k\;;\\
v_j \; \text{if } j > k\;.\\
\end{cases}\label{limit_h_-}
\end{align}
By the dominated convergence theorem, making use of the identity \eqref{expression_derivative_psi_h} and the limit \eqref{limit_h_+}, we have $\forall k \in \{1,\dots,K\}:$
\begin{align*}
a_{\rho_n}^{(k)} &\coloneqq g_{\rho_n}\bigg(\frac{2(1 - \vert \ln \rho_n \vert^{-\frac14})}{\gamma v_k^2}\bigg)
= \frac{2}{\rho_n} \psi_{P_{0,n}}^\prime\bigg(\frac{2(1 - \vert \ln \rho_n \vert^{-\frac14})\vert \ln \rho_n \vert}{v_k^2}\bigg)\\
&\;= \sum_{j=1}^K p_j^{\scriptscriptstyle +} v_j \E\bigg[h\bigg(z,\frac{2(1-\vert \ln \rho_n \vert^{-\frac14})\vert \ln \rho_n \vert}{v_k^2},v_j;\bv, \bp^{\scriptscriptstyle +},\bp^{\scriptscriptstyle -}\bigg)\bigg]\\
&\qquad\qquad
+ \sum_{j=1}^Kp_j^{\scriptscriptstyle -} v_j \E\bigg[h\bigg(z,\frac{2(1-\vert \ln \rho_n \vert^{-\frac14})\vert \ln \rho_n \vert}{v_k^2},v_j;\bv, \bp^{\scriptscriptstyle -},\bp^{\scriptscriptstyle +}\bigg)\bigg]\\
&\;\xrightarrow[\rho_n \to 0]{}
\sum_{j > k} p_j^{\scriptscriptstyle +} v_j^2 + \sum_{j > k} p_j^{\scriptscriptstyle -} v_j^2
= \E[X^2 \bm{1}_{\{\vert X \vert > v_k\}}]\;.
\end{align*}
Similarly, using this time the limit \eqref{limit_h_-}, we have $\forall k \in \{1,\dots,K\}:$
\begin{align*}
b_{\rho_n}^{(k)} &\coloneqq g_{\rho_n}\bigg(\frac{2(1 + \vert \ln \rho_n \vert^{-\frac14})}{\gamma v_k^2}\bigg)
= \frac{2}{\rho_n} \psi_{P_{0,n}}^\prime\bigg(\frac{2(1 + \vert \ln \rho_n \vert^{-\frac14})\vert \ln \rho_n \vert}{v_k^2}\bigg)\\
&\;= \sum_{j=1}^K p_j^{\scriptscriptstyle +} v_j \E\bigg[h\bigg(z,\frac{2(1+\vert \ln \rho_n \vert^{-\frac14})\vert \ln \rho_n \vert}{v_k^2},v_j;\bv, \bp^{\scriptscriptstyle +},\bp^{\scriptscriptstyle -}\bigg)\bigg]\\
&\qquad\qquad
+ \sum_{j=1}^Kp_j^{\scriptscriptstyle -} v_j \E\bigg[h\bigg(z,\frac{2(1+\vert \ln \rho_n \vert^{-\frac14})\vert \ln \rho_n \vert}{v_k^2},v_j;\bv, \bp^{\scriptscriptstyle -},\bp^{\scriptscriptstyle +}\bigg)\bigg]\\
&\;\xrightarrow[\rho_n \to 0]{}
\sum_{j \geq k} p_j^{\scriptscriptstyle +} v_j^2 + \sum_{j \geq k} p_j^{\scriptscriptstyle -} v_j^2
= \E[X^2 \bm{1}_{\{\vert X \vert \geq v_k\}}]\;.
\end{align*}
\end{proof}
Note that $\lim_{\rho_n \to 0} b_{\rho_n}^{(j)} = \lim_{\rho_n \to 0} a_{\rho_n}^{(j-1)}$.
Thus, Lemma~\ref{lemma:location_r*(q)_general} essentially states that in the limit $\rho_n \to 0$ the segment $[0,\E[X^2]]$ can be broken into $K$ subsegments $[a_{\rho_n}^{(j)},b_{\rho_n}^{(j)}]$, and for $q \in [a_{\rho_n}^{(j)},b_{\rho_n}^{(j)}]$ the point at which the supremum over $r$ is achieved is located in an interval shrinking on $r^* \coloneqq \nicefrac{2}{\gamma v_j^2}$.
The next step is then to determine what is the limit of $\frac{1}{\alpha_n} \psi_{P_{0,n}}\big(\frac{\alpha_n}{\rho_n} \frac{2}{\gamma v_j^2}\big)$.
\begin{lemma}\label{lemma:limits_psi_P0n_r*}
Let  $P_{0,n} \coloneqq (1-\rho_n)\delta_0 + \rho_n P_0$ where $P_0$ is a discrete distribution with finite support $\mathrm{supp}(P_0) \subseteq \{\pm v_1, \pm v_2, \dots, \pm v_K\}$ with $0 < v_1 < v_2 < \dots < v_K$.
Let $\alpha_n \coloneqq \gamma \rho_n \vert \ln \rho_n \vert$ for a fix $\gamma > 0$.
Then, for every $k \in \{1, \dots, K\}:$
\begin{equation}\label{eq:limits_psi_P0n_r*}
\lim_{\rho_n \to 0} \frac{1}{\alpha_n} \psi_{P_{0,n}}\bigg(\frac{\alpha_n}{\rho_n}\frac{2(1 \pm \vert \ln \rho_n \vert^{-\frac{1}{4}})}{\gamma v_k^2}\bigg)
= \frac{\E[X^2 \bm{1}_{\{\vert X \vert \geq v_k\}}]}{\gamma v_k^2}
-\frac{\mathbb{P}(\vert X \vert \geq v_k)}{\gamma}\;.
\end{equation}
\end{lemma}
\begin{proof}
Fix  $k \in \{1, \dots, K\}$.
The function $\psi_{P_{0,n}}$ is Lipschitz continuous with Lipschitz constant $\frac{\rho_n \E[X^2]}{2}$. Therefore:
\begin{multline*}
\bigg\vert \frac{1}{\alpha_n} \psi_{P_{0,n}}\bigg(\frac{\alpha_n}{\rho_n}\frac{2(1 \pm \vert \ln \rho_n \vert^{-\frac{1}{4}})}{\gamma v_k^2}\bigg)
- \frac{1}{\alpha_n} \psi_{P_{0,n}}\bigg(\frac{\alpha_n}{\rho_n}\frac{2}{\gamma v_k^2}\bigg)\bigg\vert\\
\leq \frac{\rho_n \E[X^2]}{2\alpha_n} \bigg\vert\frac{\alpha_n}{\rho_n}\frac{2 \vert \ln \rho_n \vert^{-\frac{1}{4}}}{\gamma v_k^2}\bigg\vert
= \frac{ \E[X^2]}{\gamma v_k^2} \vert \ln \rho_n \vert^{-\frac{1}{4}}\;.
\end{multline*}
The latter inequality shows that the limits of $\frac{1}{\alpha_n} \psi_{P_{0,n}}\big(\frac{\alpha_n}{\rho_n}\frac{2(1 + \vert \ln \rho_n \vert^{-\nicefrac{1}{4}})}{\gamma v_k^2}\big)$ and $\frac{1}{\alpha_n} \psi_{P_{0,n}}\big(\frac{\alpha_n}{\rho_n}\frac{2(1 -\vert \ln \rho_n \vert^{-\nicefrac{1}{4}})}{\gamma v_k^2}\big)$ are the same and equal to the limit of $\frac{1}{\alpha_n} \psi_{P_{0,n}}\big(\frac{\alpha_n}{\rho_n}\frac{2}{\gamma v_k^2}\big)$.
To compute the latter we first write $\psi_{P_{0,n}}(r)$ in a more explicit form.
We have for all $r \geq 0$:
\begin{align*}
\psi_{P_{0,n}}(r)
&\coloneqq \E \Big[\ln \int dP_{0,n}(x) e^{-\frac{r}{2}x^2 + rX^*x + \sqrt{r}xZ} \Big]\\
&\;= \E \Big[\ln\Big(1 - \rho_n + \rho_n \sum_{i=1}^K e^{-\frac{rv_i^2}{2}} \big(p_i^{\scriptscriptstyle +} e^{rX^*v_i + \sqrt{r} Zv_i}
+ p_i^{\scriptscriptstyle -} e^{-rX^*v_i -\sqrt{r} Zv_i}\big)\Big)\Big]\\
&\;= (1-\rho_n) \E \Big[\ln\Big(1 - \rho_n + \rho_n \sum_{i=1}^K e^{-\frac{rv_i^2}{2}} \big(p_i^{\scriptscriptstyle +} e^{\sqrt{r} Zv_i}
+ p_i^{\scriptscriptstyle -} e^{ -\sqrt{r} Zv_i}\big)\Big)\Big]\\
&\quad+ \rho_n\sum_{j=1}^K p_j^{\scriptscriptstyle +} \E \Big[\ln\Big(1 - \rho_n + \rho_n \sum_{i=1}^K e^{-\frac{rv_i^2}{2}} \big(p_i^{\scriptscriptstyle +} e^{r v_j v_i + \sqrt{r} Zv_i} + p_i^{\scriptscriptstyle -} e^{-r v_j v_i -\sqrt{r} Zv_i}\big)\Big)\Big]\\
&\quad+ \rho_n \sum_{j=1}^K p_j^{\scriptscriptstyle -} \E \Big[\ln\Big(1 - \rho_n + \rho_n \sum_{i=1}^K e^{-\frac{rv_i^2}{2}} \big(p_i^{\scriptscriptstyle +} e^{-r v_j v_i + \sqrt{r} Zv_i} + p_i^{\scriptscriptstyle -} e^{r v_j v_i -\sqrt{r} Zv_i}\big)\Big)\Big].
\end{align*}
By symmetry of $Z \sim \cN(0,1)$ we can replace $Z$ by $-Z$ in the expectations of the last sum. It comes:
\begin{align}
\psi_{P_{0,n}}(r)
&= (1-\rho_n) \E \Big[\ln\Big(1 - \rho_n + \rho_n \sum_{i=1}^K e^{-\frac{rv_i^2}{2}} \big(p_i^{\scriptscriptstyle +} e^{\sqrt{r} Zv_i}
+ p_i^{\scriptscriptstyle -} e^{ -\sqrt{r} Zv_i}\big)\Big)\Big]\nonumber\\
&\quad+ \rho_n\sum_{j=1}^K p_j^{\scriptscriptstyle +} \E \Big[\ln\Big(1 - \rho_n + \rho_n \sum_{i=1}^K e^{-\frac{rv_i^2}{2}} \big(p_i^{\scriptscriptstyle +} e^{r v_j v_i + \sqrt{r} Zv_i} + p_i^{\scriptscriptstyle -} e^{-r v_j v_i -\sqrt{r} Zv_i}\big)\Big)\Big]\nonumber\\
&\quad +\rho_n \sum_{j=1}^K p_j^{\scriptscriptstyle -} \E \Big[\ln\Big(1 - \rho_n + \rho_n \sum_{i=1}^K e^{-\frac{rv_i^2}{2}} \big(p_i^{\scriptscriptstyle -} e^{r v_j v_i + \sqrt{r} Zv_i} + p_i^{\scriptscriptstyle +} e^{-r v_j v_i -\sqrt{r} Zv_i}\big)\Big)\Big]\nonumber\\
&= (1-\rho_n) \E \Big[\ln\Big(1 - \rho_n + \rho_n \sum_{i=1}^K e^{-\frac{rv_i^2}{2}} \big(p_i^{\scriptscriptstyle +} e^{\sqrt{r} Zv_i}
+ p_i^{\scriptscriptstyle -} e^{ -\sqrt{r} Zv_i}\big)\Big)\Big]\nonumber\\
&\qquad+ \frac{\rho_n r \E[X^2]}{2} + \rho_n \ln \rho_n
+ \rho_n\sum_{j=1}^K p_j^{\scriptscriptstyle +} \E \big[\ln \widetilde{h}(Z,r,v_j;\rho_n, \bv, \bp^{\scriptscriptstyle +},\bp^{\scriptscriptstyle -})\big]\nonumber\\
&\qquad\qquad+ \rho_n\sum_{j=1}^K p_j^{\scriptscriptstyle -} \E \big[\ln \widetilde{h}(Z,r,v_j;\rho_n, \bv, \bp^{\scriptscriptstyle -},\bp^{\scriptscriptstyle +})\big]\;,
\label{explicit_psi_P0n}
\end{align}
where $\bv \coloneqq (v_1,v_2,\dots,v_K)$,
$\bp^{\scriptscriptstyle +} \coloneqq (p_1^{\scriptscriptstyle +},p_2^{\scriptscriptstyle +},\dots,p_K^{\scriptscriptstyle +})$,
$\bp^{\scriptscriptstyle -} \coloneqq (p_1^{\scriptscriptstyle -},p_2^{\scriptscriptstyle -},\dots,p_K^{\scriptscriptstyle -})$
and we define $\forall (z,r,u) \in \R \times [0,+\infty) \times (0,+\infty)$:
\begin{multline}
\widetilde{h}(z,r,u;\rho_n, \bv, \bp^{\pm},\bp^{\mp})\\
\coloneqq \frac{1 - \rho_n}{\rho_n}e^{-\frac{r u^2}{2} -\sqrt{r} z u} + \sum_{i=1}^K e^{-\frac{r(v_i-u)^2}{2} + \sqrt{r} z (v_i-u)}\big(p_i^{\pm} 
	+ p_i^{\mp} e^{-2r v_i u -2\sqrt{r} z v_i}\big)\;.
\end{multline}
It follows directly from \eqref{explicit_psi_P0n} that:
\begin{align}
&\frac{1}{\alpha_n} \psi_{P_{0,n}}\bigg(\frac{\alpha_n}{\rho_n} \frac{2}{\gamma v_k^2}\bigg)
= \frac{A_{\rho_n}}{\gamma} + \frac{\E[X^2]}{\gamma v_k^2}  -\frac{1}{\gamma}
+ \frac{1}{\gamma}\sum_{j=1}^K p_j^{\scriptscriptstyle +} \E \bigg[\frac{\ln \widetilde{h}\big(Z,\frac{2\vert \ln \rho_n \vert}{v_k^2},v_j;\rho_n, \bv, \bp^{\scriptscriptstyle +},\bp^{\scriptscriptstyle -}\big)}{\vert \ln \rho_n \vert}\bigg]\nonumber\\
&\qquad\qquad\qquad\qquad\qquad\qquad\qquad\quad
+ \frac{1}{\gamma}\sum_{j=1}^K p_j^{\scriptscriptstyle -} \E \bigg[\frac{\ln \widetilde{h}\big(Z,\frac{2\vert \ln \rho_n \vert}{v_k^2},v_j;\rho_n, \bv, \bp^{\scriptscriptstyle -},\bp^{\scriptscriptstyle +}\big)}{\vert \ln \rho_n \vert}\bigg]\:,\label{formula_psi_P0n_r*}
\intertext{where}
&A_{\rho_n}
= \frac{1-\rho_n}{\rho_n \vert \ln \rho_n \vert} \E \ln\bigg(1 - \rho_n + \rho_n \sum_{i=1}^K e^{-\frac{v_i^2}{v_k^2}\vert \ln \rho_n \vert} \bigg(p_i^{\scriptscriptstyle +} e^{\big(\frac{2v_i^2\vert \ln \rho_n \vert}{v_k^2}\big)^{\frac12} Z }
+ p_i^{\scriptscriptstyle -} e^{ -\big(\frac{2v_i^2\vert \ln \rho_n \vert}{v_k^2}\big)^{\frac12} Z}\bigg)\bigg)\,.\nonumber
\end{align}
Next we show that $A_{\rho_n}$ vanishes when $\rho_n \to 0$.
We can use the inequalities $\frac{x}{1+x} \leq \ln(1+x) \leq x$ valid for all $x > -1$ to get the following bounds on $A_{\rho_n}$:
\begin{align*}
A_{\rho_n} &\leq 
\frac{1-\rho_n}{\vert \ln \rho_n \vert} \Bigg( \E\bigg[\sum_{i=1}^K e^{-\frac{v_i^2}{v_k^2}\vert \ln \rho_n \vert} \bigg(p_i^{\scriptscriptstyle +} e^{\big(\frac{2v_i^2\vert \ln \rho_n \vert}{v_k^2}\big)^{\frac12} Z }
+ p_i^{\scriptscriptstyle -} e^{ -\big(\frac{2v_i^2\vert \ln \rho_n \vert}{v_k^2}\big)^{\frac12} Z}\bigg)\bigg] -1\Bigg)\\
&=\frac{1-\rho_n}{\vert \ln \rho_n \vert} \Bigg( \sum_{i=1}^K p_ie^{-2\frac{v_i^2}{v_k^2}\vert \ln \rho_n \vert} -1\Bigg)
\leq - \frac{1-\rho_n}{\vert \ln \rho_n \vert}\;;\\
A_{\rho_n} &\geq 
\frac{1-\rho_n}{\vert \ln \rho_n \vert} \E\left[\frac{\sum\limits_{i=1}^K e^{-\frac{v_i^2}{v_k^2}\vert \ln \rho_n \vert} \bigg(p_i^{\scriptscriptstyle +} e^{\big(\frac{2v_i^2\vert \ln \rho_n \vert}{v_k^2}\big)^{\frac12} Z }
+ p_i^{\scriptscriptstyle -} e^{ -\big(\frac{2v_i^2\vert \ln \rho_n \vert}{v_k^2}\big)^{\frac12} Z}\bigg) -1}{1 - \rho_n + \rho_n\sum\limits_{i=1}^K e^{-\frac{v_i^2}{v_k^2}\vert \ln \rho_n \vert} \bigg(p_i^{\scriptscriptstyle +} e^{\big(\frac{2v_i^2\vert \ln \rho_n \vert}{v_k^2}\big)^{\frac12} Z }
\!\!+ p_i^{\scriptscriptstyle -} e^{ -\big(\frac{2v_i^2\vert \ln \rho_n \vert}{v_k^2}\big)^{\frac12} Z}\bigg)}\right]\\
&\geq -\frac{1}{\vert \ln \rho_n \vert}\;.
\end{align*}
The last inequality follows from $\nicefrac{(x-1)}{(1-\rho_n + \rho_n x)} \geq \nicefrac{-1}{(1-\rho_n)}$ for $x > 0$. Together the upper bound and lower bound imply that $\vert A_{\rho_n} \vert \leq \nicefrac{1}{\vert \ln \rho_n \vert} \xrightarrow[\rho_n \to 0]{} 0$.
The last step before concluding the proof is to compute the limits of each summand in both sums over $j \in \{1,\dots,K\}$ in \eqref{formula_psi_P0n_r*}.
Note that $\forall z \in \R$:
\begin{multline}\label{h_tilde_simplified_for_limit}
\widetilde{h}\bigg(z,\frac{2 \vert \ln \rho_n \vert}{ v_k^2},v_j;\rho_n, \bv, \bp^{\pm},\bp^{\mp}\bigg)
= (1 - \rho_n)e^{\vert \ln \rho_n \vert \Big(1-\frac{v_j^2}{v_k^2} -\sqrt{\frac{2 v_j^2}{v_k^2 \vert \ln \rho_n \vert}} z\Big)}\\
+\sum_{i=1}^K e^{-\vert \ln \rho_n \vert  \Big(\frac{(v_i-v_j)^2}{v_k^2} - \sqrt{\frac{2}{\vert \ln \rho_n \vert }} \frac{v_i-v_j}{v_k} z\Big)}\Big(p_i^{\pm}
+ p_i^{\mp} e^{-4\vert \ln \rho_n \vert \frac{v_i}{v_k}\big( \frac{v_j}{v_k} + \frac{z}{\sqrt{2\vert \ln \rho_n \vert}}\big)}\bigg)\;.
\end{multline}
From \eqref{h_tilde_simplified_for_limit} we easily deduce the following pointwise limits for every $z \in \R:$
\begin{align}
\frac{\ln \widetilde{h}\bigg(z,\frac{2 \vert \ln \rho_n \vert}{v_k^2},v_j;\rho_n, \bv, \bp^{\pm},\bp^{\mp}\bigg)}{\vert \ln \rho_n \vert}
\xrightarrow[\rho_n \to 0]{}
\begin{cases}
1 - \frac{v_j^2}{v_k^2} &\text{if } j < k\;;\\
0 &\text{if } j \geq k\;.
\end{cases}\label{limit_h_tilde}
\end{align}
By the dominated convergence theorem, making use of the pointwise limits \eqref{limit_h_tilde}, we have:
\begin{multline}\label{limit_sum_E_h_tilde}
\sum_{j=1}^K p_j^{\scriptscriptstyle +} \E \bigg[\frac{\ln \widetilde{h}\big(Z,\frac{2\vert \ln \rho_n \vert}{v_k^2},v_j;\rho_n, \bv, \bp^{\scriptscriptstyle +},\bp^{\scriptscriptstyle -}\big)}{\vert \ln \rho_n \vert}\bigg]
+ p_j^{\scriptscriptstyle -} \E \bigg[\frac{\ln \widetilde{h}\big(Z,\frac{2\vert \ln \rho_n \vert}{v_k^2},v_j;\rho_n, \bv, \bp^{\scriptscriptstyle -},\bp^{\scriptscriptstyle +}\big)}{\vert \ln \rho_n \vert}\bigg]\\
\xrightarrow[\rho_n \to 0]{} \sum_{j<k} (p_j^{\scriptscriptstyle +} 
+ p_j^{\scriptscriptstyle -} ) \bigg(1 - \frac{v_j^2}{v_k^2}\bigg)
= \mathbb{P}(\vert X \vert < v_k) - \frac{\E[X^2 \bm{1}_{\{\vert X \vert < v_k\}}]}{v_k^2}\;.
\end{multline}
Combining the identity \eqref{formula_psi_P0n_r*}, $\lim_{\rho_n \to 0} A_{\rho_n} = 0$ and the limit \eqref{limit_sum_E_h_tilde} yields:
\begin{align*}
\lim_{\rho_n \to 0} \frac{1}{\alpha_n} \psi_{P_{0,n}}\bigg(\frac{\alpha_n}{\rho_n} \frac{2}{\gamma v_k^2}\bigg)
&= \frac{\E[X^2]}{\gamma v_k^2}  -\frac{1}{\gamma} + \frac{\mathbb{P}(\vert X \vert < v_k)}{\gamma} - \frac{\E[X^2 \bm{1}_{\{\vert X \vert < v_k\}}]}{\gamma v_k^2}\\
&= \frac{\E[X^2 \bm{1}_{\{\vert X \vert \geq v_k\}}]}{\gamma v_k^2} - \frac{\mathbb{P}(\vert X \vert \geq v_k)}{\gamma}\;,
\end{align*}
thus ending the proof of the proposition.
\end{proof}
We can now use Lemmas~\ref{lemma:location_r*(q)_general} and \ref{lemma:limits_psi_P0n_r*} to determine the limits when $\rho_n \to 0$ of the infimum of ${\sup}_{r \geq 0}\; i_{\scriptstyle{\mathrm{RS}}}(q, r; \alpha_n, \rho_n)$ over $q$ restrained to different subsegments of $[0,\E X^2]$.
\begin{proposition}\label{proposition:limits_inf_a_b}
Let  $P_{0,n} \coloneqq (1-\rho_n)\delta_0 + \rho_n P_0$ where $P_0$ is a discrete distribution with finite support $\mathrm{supp}(P_0) \subseteq \{\pm v_1, \pm v_2, \dots, \pm v_K\}$ with $0 < v_1 < v_2 < \dots < v_K$.
Let $\alpha_n \coloneqq \gamma \rho_n \vert \ln \rho_n \vert$ for a fix $\gamma > 0$.
Then, $\forall k \in \{1, \dots, K\}:$
\begin{flalign}
\lim_{\rho_n \to 0^+} \adjustlimits{\inf}_{q \in [a_{\rho_n}^{(k)}, b_{\rho_n}^{(k)}]} {\sup}_{r \geq 0}\; i_{\scriptstyle{\mathrm{RS}}}(q, r; \alpha_n, \rho_n)
= &\min\bigg\{I_{P_{\mathrm{out}}}(\E[X^2 \bm{1}_{\{\vert X \vert > v_k\}}],\E X^2) +\frac{\mathbb{P}(\vert X \vert > v_k)}{\gamma},&\nonumber\\
&\quad\; I_{P_{\mathrm{out}}}(\E[X^2 \bm{1}_{\{\vert X \vert \geq v_k\}}],\E X^2) +\frac{\mathbb{P}(\vert X \vert \geq v_k)}{\gamma}\bigg\}\hspace{-1.5em}&
\label{limit_inf_ak_bk_prop}
\intertext{while $\forall k \in \{2, \dots, K\}:$}
\lim_{\rho_n \to 0^+}
\adjustlimits{\inf}_{q \in [b_{\rho_n}^{(k)}, a_{\rho_n}^{(k-1)}]} \sup_{r \geq 0}\; i_{\scriptstyle{\mathrm{RS}}}(q, r; \alpha_n, \rho_n)
&= I_{P_{\mathrm{out}}}(\E[X^2 \bm{1}_{\{\vert X \vert \geq v_k\}}],\E X^2) +\frac{\mathbb{P}(\vert X \vert \geq v_k)}{\gamma}\:,\hspace{-1.5em}&
\label{limit_inf_bk_ak-1_prop}
\intertext{and}
\lim_{\rho_n \to 0^+}
\adjustlimits{\inf}_{q \in [0, a_{\rho_n}^{(K)}]} \sup_{r \geq 0}\; i_{\scriptstyle{\mathrm{RS}}}(q, r; \alpha_n, \rho_n)
&= I_{P_{\mathrm{out}}}(0,\E X^2 ) \;,&\label{limit_inf_0_aK_prop}\\
\liminf_{\rho_n \to 0^+}
\adjustlimits{\inf}_{q \in [b_{\rho_n}^{(1)},1]} \sup_{r \geq 0}\; i_{\scriptstyle{\mathrm{RS}}}(q, r; \alpha_n, \rho_n)
&\geq \frac{1}{\gamma} \;.&\label{limit_inf_b1_EX^2_prop}
\end{flalign}
\end{proposition}
\begin{proof}
In the whole proof $\rho_n$ is close enough to $0$ for the ordering \eqref{ordering_a_and_b} to hold.
First we prove \eqref{limit_inf_ak_bk_prop}.
Fix $k \in \{1,\dots,K\}$.
By Lemma~\ref{lemma:location_r*(q)_general}, for all $q \in [a_{\rho_n}^{(k)}, b_{\rho_n}^{(k)}]$ we have
$$
\sup_{r \geq 0} \frac{rq}{2} - \frac{1}{\alpha_n} \psi_{P_{0,n}}\bigg(\frac{\alpha_n}{\rho_n}r\bigg) =
\frac{r_{n}^*(q)q}{2} - \frac{1}{\alpha_n} \psi_{P_{0,n}}\bigg(\frac{\alpha_n}{\rho_n}r_{n}^*(q)\bigg)
$$
where $\frac{2(1-\vert \ln \rho_n \vert^{-\frac14})}{\gamma v_k^2}
\leq r_{n}^*(q)
\leq \frac{2(1+\vert \ln \rho_n \vert^{-\frac14})}{\gamma v_k^2}$.
This and the fact that $\psi_{P_{0,n}}$ is increasing imply that $\forall q \in [a_{\rho_n}^{(k)}, b_{\rho_n}^{(k)}]:$
\begin{align}
&I_{P_{\mathrm{out}}}(q,\E[X^2]) + \frac{q}{\gamma v_k^2}(1-\vert \ln \rho_n \vert^{-\frac14}) - \frac{1}{\alpha_n} \psi_{P_{0,n}}\bigg(\frac{\alpha_n}{\rho_n} \frac{2(1+\vert \ln \rho_n \vert^{-\frac14})}{\gamma v_k^2}\bigg)\nonumber\\
&\qquad\leq \sup_{r \geq 0}\; i_{\scriptstyle{\mathrm{RS}}}(q, r; \alpha_n, \rho_n)\nonumber\\
&\qquad\leq I_{P_{\mathrm{out}}}(q,\E[X^2]) + \frac{q}{\gamma v_k^2}(1+\vert \ln \rho_n \vert^{-\frac14}) - \frac{1}{\alpha_n} \psi_{P_{0,n}}\bigg(\frac{\alpha_n}{\rho_n} \frac{2(1-\vert \ln \rho_n \vert^{-\frac14})}{\gamma v_k^2}\bigg)\;.\label{bounds_irs_q_in_a_k_b_k}
\end{align}
These inequalities are valid for every $q \in [a_{\rho_n}^{(k)}, b_{\rho_n}^{(k)}]$ so the same inequalities will hold if we take the infimum over $q \in [a_{\rho_n}^{(k)}, b_{\rho_n}^{(k)}]$ in \eqref{bounds_irs_q_in_a_k_b_k}.
Note that $q \mapsto I_{P_{\mathrm{out}}}(q,\E[X^2]) + \frac{q}{\gamma v_k^2}(1\mp\vert \ln \rho_n \vert^{-\frac14})$ are concave functions on $[a_{\rho_n}^{(k)}, b_{\rho_n}^{(k)}]$ so the minimum of each function is achieved at either endpoint $a_{\rho_n}^{(k)}$ or $b_{\rho_n}^{(k)}$.
It comes:
\begin{align}
&\inf_{q \in [a_{\rho_n}^{(k)}, b_{\rho_n}^{(k)}]}
I_{P_{\mathrm{out}}}(q,\E[X^2]) + \frac{q}{\gamma v_k^2}(1\pm\vert \ln \rho_n \vert^{-\frac14}) - \frac{1}{\alpha_n} \psi_{P_{0,n}}\bigg(\frac{\alpha_n}{\rho_n} \frac{2(1\mp\vert \ln \rho_n \vert^{-\frac14})}{\gamma v_k^2}\bigg)\nonumber\\
&\quad\;\;= - \frac{1}{\alpha_n} \psi_{P_{0,n}}\bigg(\frac{\alpha_n}{\rho_n} \frac{2(1\mp\vert \ln \rho_n \vert^{-\frac14})}{\gamma v_k^2}\bigg)
+ \min_{q \in \{a_{\rho_n}^{(k)}, b_{\rho_n}^{(k)}\}}
I_{P_{\mathrm{out}}}(q,\E[X^2]) + \frac{q}{\gamma v_k^2}(1 \pm \vert \ln \rho_n \vert^{-\frac14})\nonumber\\
&\quad\;\;\xrightarrow[\rho_n \to 0]{} 
\frac{\mathbb{P}(\vert X \vert \geq v_k)}{\gamma}
-\frac{\E[X^2 \bm{1}_{\{\vert X \vert \geq v_k\}}]}{\gamma v_k^2}
+ \min_{q \in \Big\{\substack{\E[X^2 \bm{1}_{\{\vert X \vert > v_k\}}],\\\E[X^2 \bm{1}_{\{\vert X \vert \geq v_k\}}]}\Big\}}
I_{P_{\mathrm{out}}}(q,\E[X^2]) + \frac{q}{\gamma v_k^2}\nonumber\\
&\quad\;\;=
\min\bigg\{I_{P_{\mathrm{out}}}(\E[X^2 \bm{1}_{\{\vert X \vert > v_k\}}],\E[X^2]) +\frac{\mathbb{P}(\vert X \vert > v_k)}{\gamma},\nonumber\\
&\qquad\qquad\quad\;\; I_{P_{\mathrm{out}}}(\E[X^2 \bm{1}_{\{\vert X \vert \geq v_k\}}],\E[X^2]) +\frac{\mathbb{P}(\vert X \vert \geq v_k)}{\gamma}\bigg\}\;.
\label{limit_inf_irs_a_b}
\end{align}
The limit when $\rho_n \to 0$ follows from \eqref{limits_a_and_b} in Lemma~\ref{lemma:location_r*(q)_general} and \eqref{eq:limits_psi_P0n_r*} in Lemma~\ref{lemma:limits_psi_P0n_r*}.
Taking the infimum over $q \in [a_{\rho_n}^{(k)}, b_{\rho_n}^{(k)}]$ in \eqref{bounds_irs_q_in_a_k_b_k} and using the fact that the upper and lower bounds have the same limit \eqref{limit_inf_irs_a_b} ends the proof of \eqref{limit_inf_ak_bk_prop}.\\

We now turn to the proof of the limit \eqref{limit_inf_bk_ak-1_prop}.
Fix $k \in \{2,\dots,K\}$.
As the supremum of nondecreasing functions, the function $\widetilde{\psi}_{P_{0,n}}: q \in [0,\E X^2] \mapsto \sup_{r \geq 0} \frac{rq}{2} - \frac{1}{\alpha_n} \psi_{P_{0,n}}\big(\frac{\alpha_n}{\rho_n}r\big)$ is nondecreasing.
The fact that $I_{P_{\mathrm{out}}}(\cdot,\E X^2)$ and $\widetilde{\psi}_{P_{0,n}}$ are respectively nonincreasing and nondecreasing imply that:
\begin{multline}\label{bounds_irs_q_in_[b_k,a_k-1]}
I_{P_{\mathrm{out}}}(a_{\rho_n}^{(k-1)},\E X^2) + \widetilde{\psi}_{P_{0,n}}\big(b_{\rho_n}^{(k)}\big)\\
\leq
\adjustlimits{\inf}_{q \in [b_{\rho_n}^{(k)}, a_{\rho_n}^{(k-1)}]} \sup_{r \geq 0}\; i_{\scriptstyle{\mathrm{RS}}}(q, r; \alpha_n, \rho_n)
\leq
I_{P_{\mathrm{out}}}(b_{\rho_n}^{(k)},\E X^2) + \widetilde{\psi}_{P_{0,n}}\big(a_{\rho_n}^{(k-1)}\big)\;.
\end{multline}
By Lemma~\ref{lemma:location_r*(q)_general}, we have
\begin{align*}
\widetilde{\psi}_{P_{0,n}}(b_{\rho_n}^{(k)}) &= \frac{r_{n}^*(b_{\rho_n}^{(k)})b_{\rho_n}^{(k)}}{2} - \frac{1}{\alpha_n}\psi_{P_{0,n}}\bigg(\frac{\alpha_n}{\rho_n}r_{n}^*(b_{\rho_n}^{(k)})\bigg)\;,\\
\widetilde{\psi}_{P_{0,n}}(a_{\rho_n}^{(k-1)}) &= \frac{r_{n}^*(a_{\rho_n}^{(k-1)})a_{\rho_n}^{(k-1)}}{2} - \frac{1}{\alpha_n}\psi_{P_{0,n}}\bigg(\frac{\alpha_n}{\rho_n}r_{n}^*(a_{\rho_n}^{(k-1)})\bigg)\;,
\end{align*}
where $r_{n}^*(b_{\rho_n}^{(k)}) = \nicefrac{2\big(1+\vert \ln \rho_n \vert^{\nicefrac{-1}{4}}\big)}{\gamma v_k^2}$
and
$r_{n}^*(a_{\rho_n}^{(k-1)}) = \nicefrac{2\big(1-\vert \ln \rho_n \vert^{\nicefrac{-1}{4}}\big)}{\gamma v_{k-1}^2}$.
Making use of the limits \eqref{limits_a_and_b} in Lemma~\ref{lemma:location_r*(q)_general} and \eqref{eq:limits_psi_P0n_r*} in Lemma~\ref{lemma:limits_psi_P0n_r*} yields:
\begin{align*}
\lim_{\rho_n \to 0^+} \widetilde{\psi}_{P_{0,n}}(b_{\rho_n}^{(k)})
&= \frac{ \E[X^2 \bm{1}_{\{\vert X \vert \geq v_k\}}]}{\gamma v_k^2}
-\frac{\E[X^2 \bm{1}_{\{\vert X \vert \geq v_{k}\}}]}{\gamma v_{k}^2}
+\frac{\mathbb{P}(\vert X \vert \geq v_{k})}{\gamma}
= \frac{\mathbb{P}(\vert X \vert \geq v_{k})}{\gamma}\;;\\
\lim_{\rho_n \to 0^+} \widetilde{\psi}_{P_{0,n}}(a_{\rho_n}^{(k-1)})
&= \frac{ \E[X^2 \bm{1}_{\{\vert X \vert > v_{k-1}\}}]}{\gamma v_{k-1}^2}
-\frac{\E[X^2 \bm{1}_{\{\vert X \vert \geq v_{k-1}\}}]}{\gamma v_{k-1}^2}
+\frac{\mathbb{P}(\vert X \vert \geq v_{k-1})}{\gamma}
= \frac{\mathbb{P}(\vert X \vert \geq v_{k})}{\gamma}\;.
\end{align*}
Besides, as $\lim\limits_{\rho_n \to 0^+} b_{\rho_n}^{(k)} = \lim\limits_{\rho_n \to 0^+} a_{\rho_n}^{(k-1)} = \E[X^2 \bm{1}_{\{\vert X \vert \geq v_k\}}]$ and $I_{P_\mathrm{out}}$ is continuous, we have:
\begin{equation*}
\lim\limits_{\rho_n \to 0^+} I_{P_{\mathrm{out}}}(b_{\rho_n}^{(k)},\E X^2) = \lim\limits_{\rho_n \to 0^+} I_{P_{\mathrm{out}}}(a_{\rho_n}^{(k-1)},\E X^2)
= I_{P_{\mathrm{out}}}(\E[X^2 \bm{1}_{\{\vert X \vert \geq v_k\}}],\E X^2)\;.
\end{equation*}
Thus, the lower and upper bounds in \eqref{bounds_irs_q_in_[b_k,a_k-1]} have the same limit. It ends the proof of \eqref{limit_inf_bk_ak-1_prop}.\\

The proof of \eqref{limit_inf_0_aK_prop} is similar to the one of \eqref{limit_inf_bk_ak-1_prop}.
We have that
\begin{multline}\label{bounds_irs_q_in_[0,aK]}
I_{P_{\mathrm{out}}}(a_{\rho_n}^{(K)},\E X^2) + \widetilde{\psi}_{P_{0,n}}(0)\\
\leq
\adjustlimits{\inf}_{q \in [0, a_{\rho_n}^{(K)}]} \sup_{r \geq 0}\; i_{\scriptstyle{\mathrm{RS}}}(q, r; \alpha_n, \rho_n)
\leq
I_{P_{\mathrm{out}}}(0,\E X^2) + \widetilde{\psi}_{P_{0,n}}\big(a_{\rho_n}^{(K)}\big)\;.
\end{multline}
Clearly $\widetilde{\psi}_{P_{0,n}}(0) = 0$ while $\lim_{\rho_n \to 0^+} I_{P_{\mathrm{out}}}(a_{\rho_n}^{(K)},\E X^2) = I_{P_{\mathrm{out}}}(0,\E X^2)$ by continuity of $I_{P_{\mathrm{out}}}$ and $\lim_{\rho_n \to 0^+} a_{\rho_n}^{(K)} = 0$.
By Lemma~\ref{lemma:location_r*(q)_general},
$\widetilde{\psi}_{P_{0,n}}(a_{\rho_n}^{(K)}) = \nicefrac{r_{n}^*(a_{\rho_n}^{(K)})a_{\rho_n}^{(K)}}{2} - \nicefrac{\psi_{P_{0,n}}\big(\frac{\alpha_n}{\rho_n}r_{n}^*(a_{\rho_n}^{(K)})\big)}{\alpha_n}$
where
$r_{n}^*(a_{\rho_n}^{(K)}) = \nicefrac{2(1-\vert \ln \rho_n \vert^{\nicefrac{-1}{4}})}{\gamma v_{K}^2}$.
It follows from the limits \eqref{limits_a_and_b} in Lemma~\ref{lemma:location_r*(q)_general} and \eqref{eq:limits_psi_P0n_r*} in Lemma~\ref{lemma:limits_psi_P0n_r*} that
$\lim_{\rho_n \to 0^+} \widetilde{\psi}_{P_{0,n}}(a_{\rho_n}^{(K)}) = 0$.
Thus, the lower and upper bounds in \eqref{bounds_irs_q_in_[0,aK]} have the same limit. It ends the proof of \eqref{limit_inf_0_aK_prop}.\\

It remains to prove \eqref{limit_inf_b1_EX^2_prop}. The fact that $I_{P_{\mathrm{out}}}(\cdot,\E X^2)$ and $\widetilde{\psi}_{P_{0,n}}$ are respectively nonincreasing and nondecreasing imply that
\begin{equation}\label{bounds_irs_q_in_[b1,EX^2]}
\adjustlimits{\inf}_{q \in [b_{\rho_n}^{(k)}, \E X^2]} \sup_{r \geq 0}\; i_{\scriptstyle{\mathrm{RS}}}(q, r; \alpha_n, \rho_n)
\geq
I_{P_{\mathrm{out}}}(\E X^2,\E X^2) + \widetilde{\psi}_{P_{0,n}}\big(b_{\rho_n}^{(1)}\big)
= \widetilde{\psi}_{P_{0,n}}\big(b_{\rho_n}^{(1)}\big)\;.
\end{equation}
Hence, the inequality \eqref{limit_inf_b1_EX^2_prop} follows from taking the limit inferior on both sides of \eqref{bounds_irs_q_in_[b1,EX^2]} and the limit
\begin{align*}
\widetilde{\psi}_{P_{0,n}}\big(b_{\rho_n}^{(1)}\big)
&=
\frac{r_{n}^*(b_{\rho_n}^{(1)})b_{\rho_n}^{(1)}}{2} - \frac{\psi_{P_{0,n}}\big(\frac{\alpha_n}{\rho_n}r_{n}^*(b_{\rho_n}^{(1)})\big)}{\alpha_n}\\
&\xrightarrow[\rho_n \to 0^+]{}
\frac{\E[X^2 \bm{1}_{\{\vert X \vert \geq v_1\}}]}{\gamma v_1^2}
- \frac{\E[X^2 \bm{1}_{\{\vert X \vert \geq v_1\}}]}{\gamma v_1^2}
+ \frac{\mathbb{P}(\vert X \vert \geq v_1)}{\gamma}
= \frac{1}{\gamma}\;.
\end{align*}
\end{proof}
\begin{proposition}\label{prop:limit_I(rho_n,alpha_n)_general}
	 Let $P_{0,n} \coloneqq (1-\rho_n)\delta_0 + \rho_n P_0$ where $P_0$ is a discrete distribution with finite support $\mathrm{supp}(P_0) \subseteq \{-v_K, -v_{K-1}, \dots, -v_1,v_1, v_2, \dots, v_K\}$ with $0 < v_1 < \dots < v_K < v_{K+1}=+\infty$.
	Let $\alpha_n \coloneqq \gamma \rho_n \vert \ln \rho_n \vert$ for a fix $\gamma > 0$.
	
	Then the quantity $I(\rho_n, \alpha_n) \coloneqq \inf_{q \in [0,\E_{X \sim P_0} X^2]} \sup_{r \geq 0}\; i_{\scriptstyle{\mathrm{RS}}}(q, r; \alpha_n, \rho_n)$ converges when $\rho_n \to 0^+$ and
	\begin{equation}
	\lim_{\rho_n \to 0^+} I(\rho_n, \alpha_n)
	= \min_{1 \leq k \leq K+1}\bigg\{I_{P_{\mathrm{out}}}\big(\E[X^2 \bm{1}_{\{\vert X \vert \geq v_k\}}],\E[X^2]\big) +\frac{\mathbb{P}(\vert X \vert \geq v_k)}{\gamma}\bigg\}\;.
	\end{equation}
\end{proposition}
\begin{proof}
The proof goes in two steps. We first prove a upper bound on the limit superior of $I(\rho_n, \alpha_n)$, and then prove a lower bound on the limit inferior thats turns out to match the limit superior.
\paragraph{Upper bound on the limit superior}
Note the following trivial upper bound:
\begin{equation}\label{trivial_upperbound_I_rhon_alphan}
I(\rho_n, \alpha_n)
\leq \min_{1 \leq k \leq K}\bigg\{\adjustlimits{\inf}_{q \in [a_{\rho_n}^{(k)},b_{\rho_n}^{(k)}]} {\sup}_{r \geq 0} i_{\scriptstyle{\mathrm{RS}}}(q, r; \alpha_n, \rho_n)\bigg\}\;.
\end{equation}
The upper bound on the limit superior of $I(\rho_n, \alpha_n)$ thus directly follows from \eqref{trivial_upperbound_I_rhon_alphan} and Proposition~\ref{proposition:limits_inf_a_b} on the limits of the infimums over $q \in  [a_{\rho_n}^{(k)},b_{\rho_n}^{(k)}]$
\begin{align}
\limsup_{\rho_n \to 0^+} I(\rho_n, \alpha_n)
&\leq \min_{1\leq k \leq K} \min\bigg\{I_{P_{\mathrm{out}}}(\E[X^2 \bm{1}_{\{\vert X \vert > v_k\}}],\E[X^2]) +\frac{\mathbb{P}(\vert X \vert > v_k)}{\gamma},\nonumber\\
&\qquad\qquad\qquad\quad
I_{P_{\mathrm{out}}}(\E[X^2 \bm{1}_{\{\vert X \vert \geq v_k\}}],\E[X^2]) +\frac{\mathbb{P}(\vert X \vert \geq v_k)}{\gamma}\bigg\}\nonumber\\
&= \min_{1 \leq k \leq K+1}\bigg\{I_{P_{\mathrm{out}}}\big(\E[X^2 \bm{1}_{\{\vert X \vert \geq v_k\}}],\E[X^2]\big) +\frac{\mathbb{P}(\vert X \vert \geq v_k)}{\gamma}\bigg\}\;.\label{limsup_I(rho,alpha)_general}
\end{align}
\paragraph{Matching lower bound on the limit inferior}
The lower bound on the limit inferior is obtained by studying the infimum on each segment of the following partition:
\begin{equation}\label{partition_[0,EX^2]}
[0, \E X^2 ]
= [0,a_{\rho_n}^{(K)}] \cup \bigg(\bigcup_{k=1}^K [a_{\rho_n}^{(k)},b_{\rho_n}^{(k)}]\bigg)
\cup \bigg(\bigcup_{k=2}^K [b_{\rho_n}^{(k)}, a_{\rho_n}^{(k-1)}] \bigg)
\cup [b_{\rho_n}^{(1)},\E X^2]\;.
\end{equation}
By Proposition~\ref{proposition:limits_inf_a_b}, we directly have:
\begin{align*}
&\liminf_{\rho_n \to 0^+} \adjustlimits{\inf}_{q \in \bigcup_{k=1}^K [a_{\rho_n}^{(k)},b_{\rho_n}^{(k)}]} {\sup}_{r \geq 0}\: i_{\scriptstyle{\mathrm{RS}}}(q, r; \alpha_n, \rho_n)\\
&\qquad\qquad\qquad\qquad
= \min_{1 \leq k \leq K+1}\bigg\{I_{P_{\mathrm{out}}}\big(\E[X^2 \bm{1}_{\{\vert X \vert \geq v_k\}}],\E[X^2]\big) +\frac{\mathbb{P}(\vert X \vert \geq v_k)}{\gamma}\bigg\}\;;\\
&\liminf_{\rho_n \to 0^+} \adjustlimits{\inf}_{q \in \bigcup_{k=2}^K [b_{\rho_n}^{(k)},a_{\rho_n}^{(k-1)}]} {\sup}_{r \geq 0}\: i_{\scriptstyle{\mathrm{RS}}}(q, r; \alpha_n, \rho_n)\\
&\qquad\qquad\qquad\qquad
= \min_{2 \leq k \leq K}\bigg\{I_{P_{\mathrm{out}}}\big(\E[X^2 \bm{1}_{\{\vert X \vert \geq v_k\}}],\E[X^2]\big) +\frac{\mathbb{P}(\vert X \vert \geq v_k)}{\gamma}\bigg\}\;;\\
&\liminf_{\rho_n \to 0^+} \adjustlimits{\inf}_{q \in [0,a_{\rho_n}^{(K)}]} {\sup}_{r \geq 0}\: i_{\scriptstyle{\mathrm{RS}}}(q, r; \alpha_n, \rho_n)\\
&\qquad\qquad\qquad\qquad
= I_{P_{\mathrm{out}}}\big(0,\E[X^2]\big)
= I_{P_{\mathrm{out}}}\big(\E[X^2 \bm{1}_{\{\vert X \vert \geq +\infty\}}],\E[X^2]\big) +\frac{\mathbb{P}(\vert X \vert \geq +\infty)}{\gamma}\;;\\
&\liminf_{\rho_n \to 0^+} \adjustlimits{\inf}_{q \in [b_{\rho_n}^{(1)},1]} {\sup}_{r \geq 0}\: i_{\scriptstyle{\mathrm{RS}}}(q, r; \alpha_n, \rho_n)\\
&\qquad\qquad\qquad\qquad
\geq \frac{1}{\gamma}
= I_{P_{\mathrm{out}}}\big(\E[X^2 \bm{1}_{\{\vert X \vert \geq v_1\}}],\E[X^2]\big) +\frac{\mathbb{P}(\vert X \vert \geq v_1)}{\gamma}\;.\\
\end{align*}
Following the partition \eqref{partition_[0,EX^2]}, the limit inferior of $\inf_{q \in [0,\E X^2]} \sup_{r \geq 0}\: i_{\scriptstyle{\mathrm{RS}}}(q, r; \alpha_n, \rho_n)$ is equal to the minimum of the above four limits inferior. It comes:
\begin{equation}
\liminf_{\rho_n \to 0^+} I(\rho_n, \alpha_n)
\geq \min_{1 \leq k \leq K+1}\bigg\{I_{P_{\mathrm{out}}}\big(\E[X^2 \bm{1}_{\{\vert X \vert \geq v_k\}}],\E[X^2]\big) +\frac{\mathbb{P}(\vert X \vert \geq v_k)}{\gamma}\bigg\}\;.\label{liminf_I(rho,alpha)_general}
\end{equation}
We see that the lower bound \eqref{liminf_I(rho,alpha)_general} on the limit inferior matches the upper bound \eqref{limsup_I(rho,alpha)_general} on the limit superior, thus ending the proof.
\end{proof}
\paragraph{Proof of Theorem~\ref{theorem:limit_MI_discrete_prior}} Combining Theorem~\ref{th:RS_1layer} together with Proposition~\ref{prop:limit_I(rho_n,alpha_n)_general} ends the proof of Theorem~\ref{theorem:limit_MI_discrete_prior}.
\end{document}